\documentclass[aps,twocolumn,prx,floatfix,preprintnumbers, nofootinbib,10pt]{revtex4-2}
\usepackage{xpatch}
\makeatletter
\patchcmd{\@ssect@ltx}
    {\addcontentsline{toc}{#1}{\protect\numberline{}#8}}
    {}
    {}
    {}
\makeatother
\pdfoutput=1
\usepackage{color}
\usepackage{enumitem}

\makeatletter

\renewcommand{\thesection}{\arabic{section}}
\renewcommand{\thesubsection}{\thesection.\arabic{subsection}}
\renewcommand{\thesubsubsection}{\thesubsection.\arabic{subsubsection}}

\renewcommand{\p@section}{}
\renewcommand{\p@subsection}{}
\renewcommand{\p@subsubsection}{}

\let\oldappendix\appendix
\renewcommand{\appendix}{%
  \oldappendix
  \renewcommand{\thesection}{\Alph{section}}%
  \renewcommand{\thesubsection}{\thesection.\arabic{subsection}}%
  \renewcommand{\thesubsubsection}{\thesubsection.\arabic{subsubsection}}%
  \renewcommand{\p@section}{}%
  \renewcommand{\p@subsection}{}%
  \renewcommand{\p@subsubsection}{}%
  \renewcommand{\theHsection}{appendix.\Alph{section}}%
  \renewcommand{\theHsubsection}{\theHsection.\arabic{subsection}}%
  \renewcommand{\theHsubsubsection}{\theHsubsection.\arabic{subsubsection}}%
}

\makeatother

\usepackage{booktabs}
\usepackage{makecell}
\usepackage{cancel} 

\usepackage{mathtools}
\usepackage[dvipsnames]{xcolor}
\usepackage{float}

\usepackage{hyperref}
\definecolor{amaranth}{rgb}{0.9, 0.17, 0.31}
\definecolor{forestForestGreen(web)}{rgb}{0.13, 0.55, 0.13}
\definecolor{blue(munsell)}{HTML}{005567}
\definecolor{oxfordblue}{rgb}{0.0, 0.2, 0.4}
\definecolor{bblue}{rgb}{0.0, 0.58, 0.71}
\hypersetup{
pdfstartview={FitH}, % fits the width of the page to the window
pdftitle={}, % title
colorlinks=true, % false: boxed links; true: colored links
linkcolor=oxfordblue, % color of internal links (change box color with linkbordercolor)
citecolor=oxfordblue, % color of links to bibliography
filecolor=oxfordblue, % color of file links
urlcolor=oxfordblue% color of external links
}

\usepackage{pgfplots}
\pgfplotsset{compat=1.18}
\usepackage{amssymb}
\usepackage{amsfonts}
\usepackage{graphicx}
\usepackage{epstopdf}
\usepackage{dcolumn}
\usepackage{amsmath}
\usepackage{latexsym,bm}
\usepackage{amsthm}
\usepackage{slashed}
\usepackage{color}
\usepackage{url}
\usepackage{rotating}
\usepackage[normalem]{ulem}
\usepackage{faktor}
\usepackage{tikz-cd}
\usepackage{tensor}
\usepackage{array}
\usepackage{multirow}
\usepackage{comment}

\usepackage{tabularx}
\usepackage{makecell}
\usepackage[normalem]{ulem}
\usepackage{changepage}
\numberwithin{equation}{section}

\definecolor{ssnblue}{HTML}{006B8F}

\usepackage{makecell}
\newtheorem{theorem}{Theorem}[section]
\newtheorem{proposition}[theorem]{Proposition}
\newtheorem{lemma}[theorem]{Lemma}
\newtheorem{corollary}[theorem]{Corollary}
\theoremstyle{definition}
\newtheorem{definition}[theorem]{Definition}
\theoremstyle{remark}

\theoremstyle{claim}
\newtheorem{claim}[theorem]{Claim}

\newcommand{\Perm}{S}

\usepackage{tikz}
\usepackage{diagbox}
\usetikzlibrary{positioning}
\usetikzlibrary{calc}
\usetikzlibrary{decorations.pathreplacing,calligraphy}

\usepackage{xstring}
\usetikzlibrary{decorations.pathmorphing} 
\usetikzlibrary{decorations.markings} 
\usetikzlibrary{arrows} 
\usetikzlibrary{shapes} 
\usetikzlibrary{matrix} 
\usetikzlibrary{positioning} 
\usepackage[english]{babel} 
\usepackage[autostyle]{csquotes}
\usepackage{pifont}
\usetikzlibrary{shapes.multipart}

\tikzset{->-/.style={decoration={
  markings,
  mark=at position .6 with {\arrow{Latex[length=1.5mm,width=1.5mm]}}},postaction={decorate}}}

\usepackage{quantikz}

\newcommand{\bea}{\begin{eqnarray}}
\newcommand{\eea}{\end{eqnarray}}
\newcommand{\be}{\begin{equation}}
\newcommand{\ee}{\end{equation}}
\newcommand{\ba}{\begin{aligned}}
\newcommand{\ea}{\end{aligned}}
\newcommand{\bit}{\begin{itemize}}
\newcommand{\eit}{\end{itemize}}
\newcommand{\ben}{\begin{enumerate}}
\newcommand{\een}{\end{enumerate}}

\newcommand{\id}{\text{id}}

\newcommand{\Cl}{\text{Cl}}

\newcommand{\gr}{\operatorname{gr}}
\renewcommand{\ker}{\operatorname{Ker}}
\newcommand{\KWW}{\text{KWW}}
\newcommand{\TW}{\text{TW}}
\newcommand{\aKWW}{\alpha_\text{KWW}}
\newcommand{\aTW}{\alpha_\text{TW}}
\newcommand{\athreeF}{\alpha_\text{3F-KWW}}
\newcommand{\aKWWp}{\alpha_{\text{KWW},p}}
\newcommand{\aTWp}{\alpha_{\text{TW},p}}

\newcommand{\im}{\operatorname{Im}}
\newcommand{\Aut}{\operatorname{Aut}}
\newcommand{\Bil}{\operatorname{Bil}}
\newcommand{\ord}{\operatorname{ord}}
\newcommand{\ga}{\mathfrak{a}}
\newcommand{\gb}{\mathfrak{b}}
\newcommand{\gc}{\mathfrak{c}}
\newcommand{\gu}{\mathfrak{u}}

\DeclareMathOperator{\SL}{SL}

\newcommand{\lb}{\left(}
\newcommand{\rb}{\right)}
\newcommand{\lbb}{\left[}
\newcommand{\rbb}{\right]}
\newcommand{\lbbb}{\left\{}
\newcommand{\rbbb}{\right\}}

\newcommand{\Z}{{\mathbb Z}}

\newcommand{\cA}{\mathcal{A}}

\newcommand{\cc}{\mathcal{C}}
\newcommand{\cC}{\mathcal{C}}

\newcommand{\cP}{\mathcal{P}}

\newcommand{\cS}{\mathcal{S}}
\newcommand{\cT}{\mathcal{T}}

\newcommand{\cZ}{\mathcal{Z}}

\newcommand{\lrb}[1]{\lbbb #1 \rbbb}

\newcommand{\End}{\text{End}}

\makeatother

\makeatletter
\renewcommand\subsubsection{\@startsection{subsubsection}{3}{\z@}%
                                     {-3.25ex\@plus -1ex \@minus -.2ex}%
                                     {2.5ex \@plus .2ex}%
                                     {\centering\itshape\bfseries\small}}
\makeatother

\usepackage{physics}

\def\bra#1{{\langle{#1}|}}

\def\ket#1{{|{#1}\rangle}}

\def\unit{{1\kern-.65ex {\rm l}}}
\def\1{{1\kern-.65ex {\rm l}}}

\usepackage{makecell}

\makeatletter
\newcommand\xlabel[2][]{\phantomsection\def\@currentlabelname{#1}\label{#2}}
\makeatother

\newcommand{\cVect}{\mathbf{3Vect}}
\newcommand{\bVect}{\mathbf{2Vect}}
\newcommand{\Pic}{\text{Pic}}
\newcommand{\BrPic}{\text{BrPic}}
\newcommand{\Witt}{\text{Witt}}
\newcommand{\pt}{\text{pt}}

\usepackage[status=draft,inline,nomargin]{fixme}
\fxusetheme{color}
\FXRegisterAuthor{ki}{aki}{\color{orange}KI}
\FXRegisterAuthor{ofw}{aofw}{\color{green!40!black}OFW}
\usepackage{adjustbox}

\usepackage{faktor}
\newcommand{\mc}[1]{\mathcal{#1}}
\renewcommand{\sf}[1]{\mathsf{#1}}
\newcommand{\paren}[1]{\left( #1 \right)}
\newcommand{\sqpar}[1]{\left[ #1 \right]}
\newcommand{\curpar}[1]{\left\{ #1 \right\}}
\newcommand\restr[2]{{% we make the whole thing an ordinary symbol
  \left.\kern-\nulldelimiterspace % automatically resize the bar with \right
  #1 % the function
  \right|_{#2} % this is the delimiter
}}
\newcommand{\opaction}[1]{\restr{#1}{\text{loc.}}}
\newcommand{\abeliangrp}{A \oplus \widehat{A}}

\makeatletter
\def\l@subsubsection#1#2{}
\makeatother
\begin{document}

\title{
Non-Invertible Symmetries Mixing with Witt Non-Trivial \\
Quantum Cellular Automata
}

\author{Kansei Inamura$^{1,2}$}
\author{Oskar Wojdel$^1$}
\author{Lukasz Fidkowski$^{3}$}
\author{Sakura Sch\"afer-Nameki$^{1}$}

\affiliation{$^1$ Mathematical Institute, University
of Oxford, Woodstock Road, Oxford, OX2 6GG, United Kingdom}

\affiliation{$^2$ Rudolf Peierls Centre for Theoretical Physics, University of Oxford, Parks Road, Oxford, OX1 3PU, UK}

\affiliation{$^3$ Department of Physics, University of Washington, Seattle, WA 98195, USA}

\begin{abstract}
\noindent 
Self-dualities  and the stacking of symmetry-protected topological (SPT) phases are basic operations on quantum many-body systems. For a $\mathbb{Z}_p$ one-form symmetry in 3+1d these correspond to the Kramers-Wannier-Wegner duality $S$, which is  the gauging operation underlying non-invertible duality symmetries, and the stacking of a 1-form symmetry SPT $T$. In the continuum, they form a central extension of $PSL(2,\mathbb{Z}_4)$ for $p=2$, and of $SL(2,\mathbb{Z}_p)$ for odd primes $p$, whose central elements are invertible theories with purely gravitational response. These central extensions are governed by a twisted, graded generalization of the Witt group of abelian anyon theories, which we determine. For $p=2$ the resulting group is the single-qubit Clifford group, with duality and entangler acting as the Hadamard and phase gates. 
We realize this entire structure microscopically as quantum cellular automata (QCA) acting on a certain local operator algebra associated with a spin lattice Hilbert space on a cubic lattice.  Specifically, our local operator algebra is built by starting with all local operators commuting with a $\mathbb{Z}_p$ 1-form symmetry, and taking the quotient by all the (local) 1-form symmetry generators.
The central elements can always be extended to the full tensor product algebra with a uniquely defined QCA class. 
For $p=2$ they are generated by the non-trivial semion QCA, and for odd prime $p$ they are generated by the non-trivial $\mathbb{Z}_p$ Clifford QCA. 
Consequently the lattice fusion rules reproduce the continuum ones only up to these QCAs and lattice translations, giving rise to fusion rules refined by QCAs.

\end{abstract}

%%%%%%%%%%%%%%%%%%%%%%%%%%%%%%%%%

\maketitle
\tableofcontents

\section{Introduction and Summary}
\label{sec:Intro}

Non-invertible  symmetries are often discussed in a continuum spacetime picture, see \cite{Schafer-Nameki:2023jdn, Bhardwaj:2023kri, Shao:2023gho, Luo:2023ive} for reviews of the continuum approach. Constructing a lattice/many qubit/operator algebra framework for them is an important problem.  One approach, especially relevant for non-invertible symmetries that have a duality interpretation, is to implement them as automorphisms of constrained local spin operator algebras.  The canonical example is $1+1$d Kramers-Wannier duality, i.e. $\Z_2$ $0$-form gauging, which can be understood as an automorphism of the local algebra of operators invariant under an Ising $\Z_2$ symmetry in an Ising spin chain.  It has been rigorously shown \cite{Jones2024DHR, Ma2026quantumcellular, Jones2026QCA} that the locality-preserving transformations of this algebra are generated, up to finite depth circuits, by the Kramers-Wannier duality and translations.  This perspective also naturally ties in with the study of QCA, originally defined for tensor product local operator algebras.  Indeed, the group of QCA acting in an unconstrained, tensor product Ising spin chain is just the translation group $\Z$ \cite{GrossNesmeVogtsWerner2012}, whereas the group of QCA on the global $\Z_2$-constrained local operator algebra is a non-trivial central extension of the Kramers-Wannier duality $\Z_2$ by this translation group $\Z$: the square of the Kramers-Wannier transformation is a single translation \cite{Seiberg:2023cdc, Seiberg:2024gek}. More generally it is known that QCAs, in particular translations, can refine the fusion rules of non-invertible symmetries on the lattice \cite{Zhang:2020pco, Seiberg:2023cdc, Seiberg:2024gek, Seifnashri:2024dsd, Evans:2025msy, Lu:2026rhb, Inamura:2026hif, Jones:2026dcb, Wen:2026ncw}. 

\noindent{\bf Duality Defects in 3+1d.}
How does this picture generalize to higher dimensions?  One natural generalization is to replace the 1+1d $0$-form $\Z_2^{(0)}$ global symmetry with $1$-form $\Z_2^{(1)}$ symmetry in 3+1d.  In the continuum this gives rise to a rich non-invertible symmetry structure \cite{Kaidi:2021xfk, Choi:2021kmx, Bhardwaj:2022yxj, Choi:2022zal}, with an underlying modular $SL(2,\Z_4) / \Z_2$ \cite{BhardwajLeeTachikawa2020} group generated by $S$ - the Kramers-Wannier-Wegner (KWW) duality, corresponding to gauging the $1$-form symmetry - and $T$, corresponding to stacking a $1$-form SPT. In fact, the true symmetry group is a $\Z_8$ central extension of this group, with the central elements corresponding to stacking with an invertible theory with a purely gravitational response \cite{Witten2005SL2Z, BhardwajLeeTachikawa2020}.  For example, the modular relation $(ST)^3=1$ becomes $(ST)^3 = Y$, where $Y$ is an invertible purely gravitational theory of order $8$.  
On the lattice these duality symmetries were discussed in \cite{Gorantla:2024ocs, Koide:2021zxj}.

\noindent{\bf SymTFT Approach.}
This structure can equally be derived by more formal, categorical means, from
the perspective of the Symmetry Topological Field Theory (SymTFT)
\cite{Ji:2019jhk, Apruzzi:2021nmk, Kaidi:2022cpf, Antinucci:2022vyk}. For a
$\Z_p^{(1)}$ symmetry in $3+1$d the relevant SymTFT is a $4+1$d $\Z_p$ BF-theory of two-form gauge fields, i.e.\ a $(2,2)$ $\Z_p$ toric code, whose topological
surface operators are labelled by $\Z_p\oplus\widehat{\Z_p}$ and generate the 1-form symmetry and its dual, with the electric and magnetic surfaces braiding
through the symplectic pairing preserved by $SL(2,\Z_p)$ 
\cite{Bhardwaj:2024xcx}. 
The duality $S$ and the SPT stacking $T$ are realized as
invertible topological interfaces that permute and decorate these surfaces, and any such operation is characterized by
the abelian 2+1d anyon theory living on its codimension-one defect. The invariant classifying these interfaces is the graded, syllepsis-twisted pointed Witt group $\Witt^\pt(\Z_p\oplus\widehat{\Z_p},s)$ of 2+1d non-degenerate
braided fusion categories \cite{Bhardwaj:2024xcx}, whose grading records how each anyon terminates the bulk surfaces and whose syllepsis $s$ is their electric-magnetic braiding. This group contains the ordinary pointed Witt group $\Witt^\pt$ -- the invariant classifying ordinary abelian 3+1d QCA
\cite{Shirley_2022, HaahFidkowskiHastings2023} -- as a central subgroup, and in
Section~\ref{sec:Cats} we compute it as a central extension of its quotient by
$\Witt^\pt$, which we conjecture to be the full classification of symmetric QCA.

For $\Z_2^{(1)}$ the quotient is $S_4\cong SL(2,\Z_4)/\Z_2$ and the extension is
non-split: it is exactly the single-qubit Clifford group, with $S$ and $T$ acting as the Hadamard and phase gates and the ordinary QCA appearing as the eighth roots of unity generated by the semion, so that the relation $(ST)^3=Y$ becomes
an identity of $2\times2$ unitaries. For odd $p$ the extension splits into $SL(2,\Z_p)\times\Z_{n_p}$, yet keeping $S$ and $T$ fixed to the premetric groups that implement the KWW duality and the $\Z_p$ $1$-form SPT entangler still reproduces $(S_pT_p)^3=Y_p$, in agreement with the continuum defect fusion of
Section~\ref{sec: Fusion Rules in the Continuum} and the lattice computations
that follow.

\noindent{\bf QCAs in 3+1d with $\Z_2^{(1)}$ 1-form Symmetry.}
Likewise, the structure of QCAs on tensor product lattice Hilbert spaces becomes richer in higher dimensions.  In particular, in 3+1d there exist non-trivial QCAs, beyond simple translations \cite{HaahFidkowskiHastings2023}.  These QCAs are constructed as disentanglers of certain Walker-Wang models, corresponding to certain invertible phases with gravitational responses.  Interestingly, these QCAs  have a conjectural $\Z_8 \oplus \Z_2$ classification in terms of the $W_2$ part of the Witt group of non-degenerate braided fusion categories \cite{Shirley_2022}.  The presence of an order $8$ subgroup makes it natural to wonder if such QCAs can arise as a central extension of an appropriate operator automorphism version of the $SL(2,\Z_4)/\Z_2$ structure, in the same way that translations arise as an extension of the ordinary Kramers-Wannier duality, when it is put on the lattice.

In this work we show that this is indeed the case, by explicitly implementing the entire non-invertible symmetry structure as automorphisms of a local algebra on a spatial lattice. 
In 1+1d these local algebras were discussed in \cite{Ma2026quantumcellular}. 
In our 3+1d geometry for $\Z_2^{(1)}$ 1-form symmetry we have Ising spins on the faces of a cubic lattice, with the $1$-form symmetry {implemented by the product of Pauli $X$ operators around each cube}. The operator algebra is generated by all local operators commuting with these symmetry generators, modulo an equivalence generated by the symmetry operators themselves. The Kramers-Wannier-Wegner (KWW) duality transformation corresponding to $S$ can be implemented by a translationally-invariant automorphism $\aKWW$ of this algebra, and likewise $T$ can be implemented by an entangler $\sf{T}$ of the root $\Z_2$ $1$-form SPT phase in the $\Z_4$ classification, as we show explicitly below. This entangler also restricts to a translationally-invariant QCA $\aTW$ on the same operator algebra. By explicit computation, we then show that $(\aKWW \circ \aTW^2)^4$ - which in the continuum corresponds to $(ST^2)^4 = Y^4$ - is equivalent, up to finite depth circuits, to the $3$-fermion QCA. One has to be careful in comparing these two, because $(\aKWW \circ \aTW^2)^4$ is only defined on the $\Z_2^{(1)}$-symmetric algebra, but we show that it is extendable to the full tensor product algebra. 

We also argue that $(\aKWW \circ \aTW)^3$, which corresponds to $(ST)^3 = Y$, is the order $8$ semion QCA \cite{Shirley_2022}. Though our argument falls short of being a rigorous proof - owing partly to the fact that a QCA in the order $8$ semion class has not been fully explicitly written down - we state precise assumptions under which this can be made precise, and we provide evidence to support these. Altogether, we identify our group of locality-preserving automorphisms of the $\Z_2^{(1)}$-symmetric algebra with a $\Z_8$ central extension of $SL(2,\Z_4)/\Z_2$ isomorphic to the Clifford group of a single qubit. Thus we equate a lattice result with a categorical calculation of the pointed Witt group. It is natural to conjecture that this is the full group of locality-preserving automorphisms, modulo finite depth circuits, of the $\Z_2^{(1)}$-symmetric algebra.

\noindent{\bf QCAs in 3+1d with $\Z_p^{(1)}$ 1-form Symmetry.}
We generalize this to $\Z_p^{(1)}$ 1-form symmetries, for odd prime $p$.  In that setting, one similarly has a group of locality-preserving automorphisms modulo finite depth circuits, and a map from this group to $SL(2,\Z_p)$, which again encodes the action on the topological excitations in the SymTFT dual.  We again construct lattice automorphisms realizing this structure: $\aKWWp$ corresponding to $S$ and $\aTWp$ corresponding to $T$.  We explicitly show that $(\aKWWp \circ \aTWp)^3 = \alpha_t \circ \tilde{\beta}^{(\frac{p-1}{2})}_p \circ \aTWp$ 
, where $\tilde{\beta}^{(\frac{p-1}{2})}_p$ is a $\Z_p$-Clifford QCA \cite{Haah_2021, MengSun2026} which has order $2$ or $4$ depending on whether $p$ is $1$ or $3$ mod $4$.  This lattice level calculation again shows a mixing between non-invertible symmetries and QCA.  A more refined analysis shows that the full group of automorphisms in fact splits, as a group, as $SL(2,\Z_p) \times \Z_4$ when $p$ is $3$ mod $4$, and as $SL(2,\Z_p) \times \Z_2 \times \Z_2$ when $p$ is $1$ mod $4$.  In this purely group theoretic sense, there is no non-trivial central extension for odd $p$.  
However, we {show} that this splitting is incompatible with the canonical, unique assignment of Witt classes to those QCAs that are extendable to the full tensor product algebra, which is the more precise notion of mixing in this context.

We do not work out the case of a general finite abelian $1$-form symmetry, nor even the case of a $\Z_n^{(1)}$-symmetry for general, non-prime $n$.  However, we expect that for $n=4$, we will obtain a central extension that involves a Walker-Wang QCA for the $U(1)_4$ anyon theory.  Together with the semion QCA that appears in the $n=2$ case, these two generate the full $\Z_8 \times \Z_2$ in the pointed Witt group for $p=2$.  For odd prime $p = 3$ mod $4$ the extension is by the generator of the $\Z_4$ portion of the pointed Witt group, while for $p = 1$ mod $4$, we show how to obtain QCA corresponding to both $\Z_2$ factors in the $\Z_2 \times \Z_2$ portion of the Witt group.  Thus, we demonstrate that all QCA in the conjectured Witt classification of $3+1$d QCA can be constructed from composing elementary lattice operations, consisting of Kramers-Wannier-Wagner duality, the Tsui-Wen $1$-form SPT entangler, and (in the odd prime $p$ case) onsite generalized charge conjugation symmetries.  

\noindent
{\bf Plan of the Paper.} The rest of the paper is structured as follows:
In Section \ref{sec: Fusion Rules in the Continuum} we review the 
fusion of the KWW duality and SPT-stacking operations in 3+1d continuum field theory, for both $\Z_2^{(1)}$ and $\Z_p^{(1)}$ 1-form symmetries, respectively.  
In Section \ref{sec:Cats} we take the categorical approach and describe the braided automorphisms of the SymTFT for $A=\Z_p^{(1)}$ 1-form symmetries in terms of the pointed graded Witt group $\Witt^\pt(A \oplus \widehat{A}, s)$. In particular we provide the complete group relations and realization of the graded Witt classes in terms of metric groups. 

In Sections \ref{sec: QCAs on Z2 1-form symmetric algebra} and 
\ref{sec:qca_for_root_entangler} we develop the QCAs for the $\Z_2^{(1)}$-form symmetry. We start in Section \ref{sec: QCAs on Z2 1-form symmetric algebra}
 by developing the mathematical framework in which we describe the QCA which implement the analogues of the continuum symmetries. 
We provide the QCA description of the KWW duality and the order 2 SPT entangler, and various relations of these on the lattice. 
The order 4 SPT-entangler is then discussed in Section \ref{sec:qca_for_root_entangler}. 
In Section \ref{sec: QCAs on Zp 1-form symmetric algebra} we extend these QCA results to the $\Z_p^{(1)}$ 1-form symmetry. 

Various more in depth discussions can be found in the appendices: 
Appendix \ref{app:Witty} contains an in depth discussion of the graded pointed Witt groups, and proofs of various theorems that appear in Section \ref{sec:Cats}. Some conventions and details for higher cup products are contained in Appendix \ref{sec: higher cup products}. In Appendix~\ref{sec: Equivalent definition of 3-fermion KWW operator}, we give an equivalent definition of the 3-fermion KWW operator defined in Section~\ref{sec: 3-fermion KWW}, and discuss its relation to the 3-fermion QCA and its adaptive circuit realization. 
Finally, details on the quotient algebras and proofs for statements in Sections \ref{sec:setup_qca_on_Z2_subalg} are provided in Appendix \ref{app:quotient_algebra}.

%%%%%%%%%%%%%%%%%%%%%%%%%%%%%%%%%%%%%%%%%%%%%%%%%%%%%%%%%%%%%%%%%%%
\begin{table*}[t]
\centering
\renewcommand{\arraystretch}{1.4}
\begin{tabular}{@{}l |l| l| l@{}}
\toprule
 & Generators & Relations & Eq. \\
\midrule
Continuum defects
  & $\mc{S}$, $\mc{T}$
  & {$\mc{S}^2 = \mc{C}_0$}, \ $(\mc{ST})^3 = \mc{Z}_\text{semion} \mc{C}_0$, \ $(\mc{ST}^2)^4 = \mc{Z}_\text{3F} \mc{C}_0$
  & \eqref{eq: ST3 continuum defect}, \eqref{eq: DU4 continuum defect} \\
Witt group $\Witt^\pt (\Z_2^2, s)$ 
  & $S$, $T$
  & $S^2=1$, \ $T^4=1$, \ $(ST)^3=Y$, \ $(ST^2)^4 = Y^4$
  & \eqref{eq:STcat}, \eqref{eq: Y4} \\
Lattice
  & $\aKWW$, $\aTW$
  & $\aKWW^2 = \alpha_t$, \ $\aTW^4 = 1$ 
  & \eqref{eq:S2action}, \eqref{eq:T4action}  \\
  &
  & $(\aKWW \circ \aTW^2)^4=\alpha_t^2 \circ \alpha_{\mathrm{framing}} \circ \aTW^2$
  & \eqref{eq: DU4 lattice} \\
\bottomrule
\end{tabular}
\caption{Generators and relations for the $\Z_2^{(1)}$ symmetry case.
As topological defects in a continuum theory, we identify $\mc{S}$ with the Kramers-Wannier-Wegner duality
(gauging the $\Z_2^{(1)}$ symmetry) and $\mc{T}$ with stacking the order-$4$ root $\Z_2^{(1)}$ SPT. $\mc{C}_0$ is the condensation defect for the $\Z_2^{(1)}$ symmetry.
$\mc{Z}_\text{semion}$ is the partition function of the invertible $3+1$d
Crane--Yetter--Walker--Wang TQFT built from the chiral semion, of order $8$, and
$\mc{Z}_\text{3F}$ is the partition function of its fourth power (TQFT based on the three-fermion MTC). In the second line we provide a description in terms of the pointed graded Witt group $\Witt^{\pt} (\Z_2^2, s)$.
On the lattice these become locality-preserving automorphisms of the
$\Z_2^{(1)}$-symmetric algebra: $\aKWW$ realizes $S$ (the KWW operator), $\aTW$
realizes $T$ (the Tsui-Wen SPT entangler), $\alpha_t$ is the lattice translation
QCA, and $\alpha_{\mathrm{framing}}$ is the framing QCA that realizes $Y^4 T^2$, which is equivalent to the three-fermion QCA. The lattice relation matches the continuum $(ST^2)^4=Y^4$
up to the translation $\alpha_t^2$.  We find that the action of the defect $(\mc{ST}^2)^4$ on local operators is trivial, $\opaction{(\mc{ST}^2)^4} = 1$, while on the lattice it mixes with translations, as well as a non-trivial QCA $\alpha_\text{framing}$.
}
\label{tab:p2-relations}
\end{table*}

%%%%%%%%%%%%%%%%%%%%%%%%%%%%%%%%%%%%%%%%%%%%%%%%%%%%%%%%%%%%%%%%%%%
%  p odd prime
%%%%%%%%%%%%%%%%%%%%%%%%%%%%%%%%%%%%%%%%%%%%%%%%%%%%%%%%%%%%%%%%%%%
\begin{table*}[t]
\centering
\renewcommand{\arraystretch}{1.4}
\begin{tabular}{@{}l |l |l| l@{}}
\toprule
 & Generators & Relations & Eq. \\
\midrule
Continuum defects
  & $\mc{S}_p$, $\mc{T}_p$
  & {$\mc{S}_p^2 = \mc{C} \mc{C}_0$}, \ $(\mc{S}_p\mc{T}_p)^3 = \mc{Z}_{\text{SU}(p)_1} \mc{C}_0$
  & \eqref{eq: DTDTDT continuum2} \\
  Witt group $\Witt^\pt (\Z_p^2, s)$ 
  & $S_p$, $T_p$
  & $S_p^2=C$, \ $T_p^{\,p}=1$, \ $(S_pT_p)^3=Y_p$
  & \eqref{eq:SToddrels} \\
Lattice
  & $\aKWWp$, $\aTWp$
  & $(\aKWWp \circ \aTWp)^3=\alpha_t\circ \tilde\beta_p^{(\frac{p-1}{2})} \circ \aTWp$
  & \eqref{eq: DTDTDT minus} \\
\bottomrule
\end{tabular}
\caption{Generators and relations for the $\Z_p^{(1)}$ symmetry case with $p$ an odd prime. In the continuum, $\cS_p$ is the Kramers-Wannier-Wegner duality (gauging the $\Z_p$ $1$-form symmetry) and $\cT_p$ stacks a $\Z_p$ $1$-form SPT. $
\cC_0$ is the condensation defect and 
$\cC$ is charge conjugation.  $Y_p=\mathrm{CYWW}(SU(p)_1)$ is the invertible
$3+1$d TQFT built from $SU(p)_1$, of order $2$ for $p\equiv1$ and $4$ for
$p\equiv3\bmod4$.
{In the second line we provide a description in terms of the pointed graded Witt group $\Witt^{\pt} (\Z_p^2, s)$.}
On the lattice, $\aKWWp$ realizes $S_p$ and $\aTWp$ realizes
$T_p$; $\alpha_t$ is the lattice translation, and $\tilde\beta_p^{(\frac{p-1}{2})}$ the $\Z_p$ Clifford QCA
(realizing $Y_p\mathsf{T}_p^{-1}$). The lattice relation matches the continuum
$(S_pT_p)^3=Y_p$ up to the translation QCA $\alpha_t$. We find that the action of the defect $(\mc{S}_p \mc{T}_p)^3$ on local operators is trivial, $\opaction{(\mc{S}_p \mc{T}_p)^3} = 1$, while on the lattice there is mixing with translations, as well as a non-trivial QCA $\tilde{\beta}_p^{(\frac{p-1}{2})}$.}
\label{tab:podd-relations}
\end{table*}

% \twocolumngrid

% \clearpage
% \newpage

\subsection*{Notations and Conventions}
\label{sec: Notations and conventions}
%In Sections~\ref{sec: QCAs on Z2 1-form symmetric algebra}, \ref{sec: Equivalent definition of 3-fermion KWW operator}, and \ref{sec: QCAs on Zp 1-form symmetric algebra}, we consider lattice models on a 3d cubic lattice with periodic boundary conditions.

Throughout the paper we will use the notation for various operations for the continuum symmetry operators, defects and the lattice operators as outlined in Tables \ref{tab:Z2Not} and \ref{tab:ZpNot}. 

The integer $p$ is assumed to be an odd prime throughout. We will denote by $\Z_p = \Z/p\Z$ both the finite group of order $p$ as well as the field $\mathbb{F}_p$.  
Although the {\bf modular group} is $(P)SL(2,\Z)$, we will refer to the group $(P)SL (2,\Z_p)$ as the modular group over the finite field $\Z_p$.

Regarding {\bf finite depth circuits (FDC)} we will adopt the following terminology: a {\bf symmetric FDC} is a FDC consisting of only symmetric gates. This is to be contrasted with a FDC which we throughout this paper will assume to be symmetric under the 1-form symmetry at play -- however the gates it is comprised of may not necessarily be symmetric (unless we explicitly specify that it is a symmetric FDC). We will never consider FDCs that do not respect the symmetry. 

Furthermore we will require various conventions about the cup products which we summarize in the following.

\vspace*{\baselineskip}
\noindent{\bf Chains and cochains.}
We consider a three-dimensional cubic lattice with periodic boundary conditions.
The sets of cubes, faces, edges, and vertices are denoted by $C$, $F$, $E$, and $V$, respectively.
Elements of these sets are denoted by $c$, $f$, $e$, and $v$.
We will think of $c$, $f$, $e$, and $v$ as a 3-chain, a 2-chain, a 1-chain, and a 0-chain on a cubic lattice.
For these chains, we define the corresponding 3-cochain $\bm{c}$, 2-cochain $\bm{f}$, 1-cochain $\bm{e}$, and 0-cochain $\bm{v}$ by the following equation:
\begin{equation}
\begin{aligned}
\bm{c}(c^{\prime}) = \delta_{c, c^{\prime}}, \quad
\bm{f}(f^{\prime}) = \delta_{f, f^{\prime}}, \\
\bm{e}(e^{\prime}) = \delta_{e, e^{\prime}}, \quad
\bm{v}(v^{\prime}) = \delta_{v, v^{\prime}}.
\end{aligned}
\end{equation}
In general, chains on a cubic lattice are denoted by small letters in the standard font, such as $a$, $b$, $c$, etc.
On the other hand, cochains are written in bold font, such as $\bm{a}$, $\bm{b}$, $\bm{c}$, etc.
Basic operations such as the (co)boundary operation and (higher) cup products on a cubic lattice are reviewed in Appendix~\ref{sec: higher cup products}.

\begin{table}[ht]
\centering
\renewcommand{\arraystretch}{1.4}
\begin{tabular}{l cc cc}
\toprule
 & \multicolumn{2}{c}{\textbf{lattice}} & \multicolumn{2}{c}{\textbf{continuum}} \\
\cmidrule(lr){2-3}\cmidrule(lr){4-5}
 & operator & automorphisms & \makecell{operations\\on theories} & defects \\
% \midrule
%  command & \verb|\sf{}| & & & \verb|\mc{}| \\
\midrule
KWW          & $\sf{S}$ & $\alpha_\KWW$      & $ S $   & $\mathcal{S}$          \\
TW           & $\mathsf{T} $                         & $\alpha_{\TW}$      & $ T $   & $\mathcal{T}$          \\
% TW$^2$       & $\mathsf{T}^2$                       & $\alpha_{\TW}^{2}$  & $ T^2$ & $ \mc{T}^2$ \\
%   &        &  & $ Y $ & 
% \\
Framing QCA  & $\mathsf{U}_{\text{framing}}$       & $\alpha_{\text{framing}}$  & $ Y^4T^2 $ & $\mc{Z}_\text{3F} \cT^2$ 
\\
\bottomrule
\end{tabular}
\caption{Notations for the $\Z_2^{(1)}$ symmetry case. KWW and TW stand for the Kramers-Wannier-Wegner duality and the Tsui-Wen SPT entangler, respectively. The framing QCA is a non-trivial QCA defined in~\cite{Fidkowski:2023dpe} \label{tab:Z2Not}}
\end{table}

\begin{table}[ht]
\centering
\renewcommand{\arraystretch}{1.6}
\begin{tabular}{l cc cc}
\toprule
 & \multicolumn{2}{c}{\textbf{lattice}} & \multicolumn{2}{c}{\textbf{continuum}} \\
\cmidrule(lr){2-3}\cmidrule(lr){4-5}
 & operator & autos & \makecell{operations\\on theories} & defects \\
% \midrule
%  command & \verb|\sf{}| & & & \verb|\mc{}| \\
\midrule
KWW
  & $\mathsf{S}_{p}$
  & $\alpha_{\KWW,p}$
  & ${S}_{p}$
  & $\mathcal{S}_{p}$ \\
TW
  & $\mathsf{T}_{p}$
  & $\alpha_{\TW,p}$
  & ${T}_{p}$
  & $\mathcal{T}_{p}$ \\
Clifford QCA
  & \makecell{}
  & {$\tilde{\beta}_{p}^{(\frac{p-1}{2})}$}
  & ${Y}_{p} T_p^{-1}$
  & $\mc{Z}_{\mathrm{SU}(p)_1}\mathcal{T}_p^{-1}$ \\
\bottomrule
\end{tabular}
\caption{Notations for the $\mathbb{Z}_p^{(1)}$ symmetry case with $p$ an odd prime. The Clifford QCA $\tilde{\beta}_p^{(\frac{p-1}{2})}$ is a non-trivial QCA defined in~\cite{MengSun2026}. \label{tab:ZpNot}}
\end{table}

\vspace*{\baselineskip}
\noindent{\bf Half translation.}
In the main text, we will often consider the translation on a cubic lattice in the $(\frac{1}{2}, \frac{1}{2}, \frac{1}{2})$ direction.
This operation will be called the half translation and denoted by $t^{\frac{1}{2}}$.
We note that the half translation $t^{\frac{1}{2}}$ maps cubes and faces into vertices and edges, and vice versa.
Namely, for each $c \in C$, $f \in F$, $e \in E$, and $v \in V$, we have
\begin{equation}
t^{\frac{1}{2}}(c) \in V, \quad
t^{\frac{1}{2}}(f) \in E, \quad
t^{\frac{1}{2}}(e) \in F, \quad
t^{\frac{1}{2}}(v) \in C.
\end{equation}
One can see that a vertex gets half-translated into the center of a cube. Similarly, the center of an edge gets half-translated into the center of a face.
The half-translations of the corresponding cochains are defined in an obvious way.
Concretely, the half-translations of $\bm{c}$, $\bm{f}$, $\bm{e}$, and $\bm{v}$ are defined by
\begin{equation}
\begin{aligned}
t^{\frac{1}{2}}(\bm{c})(v^{\prime}) = \delta_{t^{\frac{1}{2}}(c), v^{\prime}}, \quad
t^{\frac{1}{2}}(\bm{f})(e^{\prime}) = \delta_{t^{\frac{1}{2}}(f), e^{\prime}}, \\
t^{\frac{1}{2}}(\bm{e})(f^{\prime}) = \delta_{t^{\frac{1}{2}}(e), f^{\prime}}, \quad
t^{\frac{1}{2}}(\bm{v})(c^{\prime}) = \delta_{t^{\frac{1}{2}}(v), c^{\prime}}.
\end{aligned}
\end{equation}
The inverse of $t^{\frac{1}{2}}$ will be denoted by $t^{-\frac{1}{2}}$, which is the half translation in the $(-\frac{1}{2}, -\frac{1}{2}, -\frac{1}{2})$ direction.
Similarly, the square of $t^{\frac{1}{2}}$ will be denoted by $t$, which is the full translation in the $(1, 1, 1)$ direction.

\vspace*{\baselineskip}
\noindent{\bf Poincar\'{e} duality.}
We will also sometimes consider chains and cochains on the dual lattice by using the Poincar\'{e} duality.
The Poincar\'{e} dual of the cochains $\bm{c}$, $\bm{f}$, $\bm{e}$, and $\bm{v}$ on the direct lattice are denoted by $\hat{c}$, $\hat{f}$, $\hat{e}$, and $\hat{v}$, respectively.
We note that $\hat{c}$ is a 0-chain, $\hat{f}$ is a 1-chain, $\hat{e}$ is a 2-chain, and $\hat{v}$ is a 3-chain on the dual lattice.
The corresponding cochains on the dual lattice are denoted by $\hat{\bm{c}}$, $\hat{\bm{f}}$, $\hat{\bm{e}}$, and $\hat{\bm{v}}$, respectively.
Concretely, these cochains are defined by
\begin{equation}
\begin{aligned}
\hat{\bm{c}}(\hat{c}^{\prime}) = \delta_{c, c^{\prime}}, \quad
\hat{\bm{f}}(\hat{f}^{\prime}) = \delta_{f, f^{\prime}}, \\
\hat{\bm{e}}(\hat{e}^{\prime}) = \delta_{e, e^{\prime}}, \quad
\hat{\bm{v}}(\hat{v}^{\prime}) = \delta_{v, v^{\prime}}.
\end{aligned}
\end{equation}
In general, chains and cochains on the dual lattice are denoted by letters with hats.
Specifially, chains are denoted by $\hat{a}$, $\hat{b}$, $\hat{c}$, etc., whereas cochains are denoted by $\hat{\bm{a}}$, $\hat{\bm{b}}$, $\hat{\bm{c}}$, etc.

\vspace*{\baselineskip}
\noindent{\bf Convention for orientations.}
Whenever we draw a 3d cubic lattice in later sections, we will use the following convention for the orientation:
\begin{equation}
\adjincludegraphics[valign=c, trim={10, 10, 10, 10}]{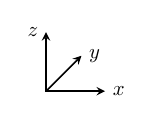}
\end{equation}
This convention will be used when we illustrate various operators written in terms of cochains.

\section{Fusion Rules in the Continuum}
\label{sec: Fusion Rules in the Continuum}

In this section, we will review the continuum fusion rules involving the Kramers-Wannier-Wegner duality defect for both $\Z_2$ and $\Z_p$ 1-form symmetry.
Throughout this section, dynamical gauge fields are denoted by small letters such as $b$, whereas background gauge fields are denoted by capital letters such as $B$.

\subsection{$\Z_2^{(1)}$ Symmetry Case}
\label{sec: continuum fusion rules Z2}
We begin with the case of $\mathbb{Z}_2$ 1-form symmetry.
Following~\cite{Gaiotto:2014kfa, BhardwajLeeTachikawa2020}, we first introduce two operations, called $S$ and $T$ operations, on general 3+1d QFTs with $\mathbb{Z}_2$ 1-form symmetry.
We will construct the  topological defects corresponding to these transformations  and compute their fusion rules, following~\cite{Choi:2021kmx, Choi:2022zal}.
These fusion rules will be compared with their lattice counterparts in later sections.

\subsubsection{The $S$ and $T$ operations}
Let $\mathcal{Q}$ be a 3+1d QFT with a $\mathbb{Z}_2$ 1-form symmetry.
The partition function of $\mathcal{Q}$ on an oriented manifold $M_4$ equipped with a background gauge field $B \in H^2(M_4; \mathbb{Z}_2)$ is denoted by $Z_{\mathcal{Q}}[B]$.
%We consider the $S$ and $T$ operations that map $\mathcal{Q}$ to other QFTs $S\mathcal{Q}$ and $T \mathcal{Q}$, whose partition functions are given by
The $S$ and $T$ operations are defined as maps that map $\mathcal{Q}$ to other $\mathbb{Z}_2$ 1-form symmetric QFTs denoted by $S\mathcal{Q}$ and $T\mathcal{Q}$.
The partition functions of these new QFTs are defined by
\begin{align}
Z_{S\mathcal{Q}}[B] &= \frac{1}{|H^2(M_4; \mathbb{Z}_2)|^{\frac{1}{2}}} \sum_{b \in H^2(M_4; \mathbb{Z}_2)} Z_{\mathcal{Q}}[b] (-1)^{\int_{M_4} b \smile B}, \\
Z_{T \mathcal{Q}}[B] &= Z_{\mathcal{Q}}[B] e^{i \pi \int_{M_4} \frac{1}{2} \mathcal{P}(B)}.
\end{align}
Here, $\mathcal{P}(B)$ is the Pontryagin square of $B$.
We note that the $S$ operation corresponds to gauging the $\mathbb{Z}_2$ 1-form symmetry, while the $T$ operation corresponds to stacking a $\mathbb{Z}_2$ 1-form SPT phase of order 4.
For later use, we also write down the $T^2$ operation explicitly as
\begin{equation}
Z_{T^2\mathcal{Q}}[B] = Z_{\mathcal{Q}}[B] (-1)^{\int_{M_4} B \smile B}.
\end{equation}
This operation corresponds to stacking a $\mathbb{Z}_2$ 1-form SPT phase of order 2. %discussed in Section~\ref{sec: Z2 TW of order 2}:

\vspace*{\baselineskip}
\noindent{\bf Group generated by $S$ and $T$.}
%The $S$ and $T$ operations defined above are invertible and hence generate a group.
It was shown in \cite{Gaiotto:2014kfa, BhardwajLeeTachikawa2020, Choi:2022zal} that $S$ and $T$ obey the following relations:
\begin{equation}
S^2=1, \quad T^4=1, \quad (ST)^3=Y.
\label{eq: ST3}
\end{equation}
Here, $1$ denotes the identity operation, and $Y$ is the stacking of an invertible 3+1d TQFT whose partition function is given by
\begin{equation}
\begin{aligned}
Z_Y(M_4) &= \frac{1}{|H^2(M_4; \mathbb{Z}_2)|^{\frac{1}{2}}} \sum_{b \in H^2(M_4; \mathbb{Z}_2)} e^{i \pi \int_{M_4} \frac{1}{2} \mathcal{P}(b)} \\
&= e^{\frac{2\pi i}{8} \sigma(M_4)},
\end{aligned}
\label{eq: Y partition function}
\end{equation}
where $\sigma(M_4)$ is the signature of $M_4$.
The above TQFT is the Crane-Yetter-Walker-Wang TQFT based on the chiral semion MTC $\mathcal{C}_{\mathrm{semion}}$ (also known as $\mathrm{U}(1)_2$) \cite{Crane:1993if, Crane:1994ji, walker20123+}:\footnote{The chiral semion MTC $\mathcal{C}_{\mathrm{semion}}$ has two simple objects $\{1, s\}$, where the topological spin of the non-trivial object $s$ is $i$.}
\begin{equation}
Y = \mathrm{CYWW}(\mathcal{C}_{\mathrm{semion}}).
\end{equation}
We note that $Y$ has order 8 because $\mathcal{C}_{\mathrm{semion}}$ has order 8 in the Witt group of non-degenerate braided fusion categories~\cite{davydov2013witt} and the Crane-Yetter-Walker-Wang TQFT constructed from a Witt-trivial MTC should be trivial.
Furthermore, $Y$ commutes with $S$ and $T$ because its partition function~\eqref{eq: Y partition function} is independent of the background gauge field $B$.
{In other words, $Y$ is a central element of order 8 in the group generated by $\mathsf{S}$ and $\mathsf{T}$.}

{If $Y$ were 1, the relations in~\eqref{eq: ST3} would define the modular group $PSL(2, \mathbb{Z}_4) = SL(2, \mathbb{Z}_4)/{\{\pm 1\}}$.
Indeed, $PSL(2, \mathbb{Z}_4)$ is isomorphic to $S_4$, which has a presentation $\langle x, y \mid x^2=y^4=(xy)^3=1 \rangle$.
Thus, \eqref{eq: ST3} shows that the group generated by $\mathsf{S}$ and $\mathsf{T}$ is a central extension of $PSL(2, \mathbb{Z}_4)$ by $\mathbb{Z}_8$ generated by $Y$.
The Witt group calculation in Section~\ref{sec:BrAutZ2} (Theorem \ref{thm:modular}) implies that this group is isomorphic to the central product $2O\circ_{\Z_2}\Z_8$, where $2O$ denotes the binary octahedral group. This group is also isomorphic to the single-qubit Clifford group (Lemma~\ref{cor:Cl2}).
}

\vspace*{\baselineskip}
\noindent{\bf Group generated by $S$ and $T^2$.}
The $S$ and $T^2$ operations generate a non-trivial subgroup of the group generated by $S$ and $T$.
To determine the group structure of this subgroup, let us compute the $n$th power of $ST^2$ for every integer $n$.
Using the relations $STS = YT^{-1}ST^{-1}$ and $S^2=1$, one can show that
\begin{equation}
\begin{aligned}
(ST^2)^n &= (STS^2T)^n = (YT^{-1}ST^{-1}ST)^n \\
&= Y^nT^{-1}ST^{-n}ST.
\end{aligned}
\end{equation}
In particular, when $n=4$, we have
\begin{equation}
(ST^2)^4=Y^4.
\label{eq: Y4}
\end{equation}
Here, we used $T^4=1$.
We note that $Y^4$ is a central element of order 2 because $Y$ is a central element of order 8.
Equations~\eqref{eq: ST3} and \eqref{eq: Y4} show that $S$ and $T^2$ obey the following relations:
\begin{equation}
S^2=1, \quad
(T^2)^2=1, \quad
(ST^2)^4 = Y^4.
\label{eq: ST2 4 continuum}
\end{equation}

If $Y^4$ were 1, the above relations would define the dihedral group $D_8$ of order 8.\footnote{Recall that $D_8$ has presentation $\langle r, s \mid s^2=r^4=1, srs=r^{-1} \rangle$. By defining $x \coloneq s$ and $y \coloneq sr$, we find a different presentation with relations $x^2=y^2=(xy)^4=1$.}
Thus,  the subgroup generated by $S$ and $T^2$ is a central extension of $D_8$ by $\mathbb{Z}_2$ generated by $Y^4$.

For later convenience, let us briefly comment on the central element $Y^4$ of order 2.
By definition, $Y^4$ is the stacking of four copies of the Crane-Yetter-Walker-Wang TQFT based on $\mathcal{C}_{\mathrm{semion}}$.
Equivalently, this TQFT can also be regarded as the Crane-Yetter-Walker-Wang TQFT based on the 3-fermion MTC $\mathcal{C}_{3F}$ (also known as $\mathrm{SO}(8)_1$):\footnote{The 3-fermion MTC $\mathcal{C}_{3F}$ has four simple objects $\{1, f_1, f_2, f_3\}$, where the topological spins of $f_1$, $f_2$, and $f_3$ are $-1$.}
\begin{equation}
Y^4 = \mathrm{CYWW}(\mathcal{C}_{\mathrm{semion}}^{\boxtimes 4}) = \mathrm{CYWW}(\mathcal{C}_{3F}).
\end{equation}
The second equality follows from the Witt equivalence between the 3-fermion MTC and the tensor product of four copies of the chiral semion MTC.\footnote{The 3-fermion MTC is obtained by condensing the bound state of all semions in $\mathcal{C}_{\mathrm{semion}}^{\boxtimes 4}$.}
The partition function of $\mathrm{CYWW}(\mathcal{C}_{3F})$  can be written explicitly as
\begin{equation}
Z_{Y^4}(M_4) = \frac{1}{|H^2(M_4; \mathbb{Z}_2)|} \sum_{b_1, b_2} (-1)^{\int b_1 \smile b_1 + b_1 \smile b_2 + b_2 \smile b_2},
\label{eq: 3F WW}
\end{equation}
where the summation is taken over $b_1, b_2 \in H^2(M_4; \mathbb{Z}_4)$, and the integral is as always over the manifold $M_4$.
By a direct computation, one can show that $Z_{Y^4}(M_4)$ in \eqref{eq: 3F WW} is indeed equal to $Z_Y(M_4)^4$.

\subsubsection{Topological defects}
Based on the above definitions of the $S$ and $T$ operations, one can define the corresponding codimension-1 topological defects between $\mathbb{Z}_2$ 1-form symmetric QFTs \cite{Choi:2021kmx, Choi:2022zal}.
In this subsection, we will focus on the topological defects corresponding to $S$, $ST$, and $ST^2$.

A topological defect corresponding to the $S$ operation is defined by gauging the $\mathbb{Z}_2$ 1-form symmetry only in half of the spacetime \cite{Choi:2021kmx, Choi:2022zal}.
At the defect locus, we impose the Dirichlet boundary condition on the dynamical gauge field, which guarantees that the defect is topological.
The topological defect obtained in this way is denoted by $\mathcal{S}$ and is called a duality defect.
See Figure~\ref{fig: D defect} for an illustration of this construction.
\begin{figure}
\centering
\adjincludegraphics[valign=c]{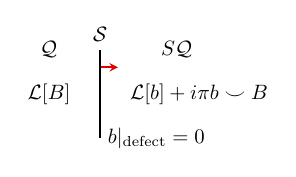}
\caption{The half-space gauging construction of the duality defect $\mathcal{S}$ in the continuum \cite{Choi:2021kmx, Choi:2022zal}. Here, $\mathcal{L}$ denotes the Lagrangian of the QFT $\mathcal{Q}$ before gauging. The dynamical gauge field $b$ satisfies the Dirichlet boundary condition $b=0$ at the defect locus. The red arrow specifies the orientation of the defect.}
\label{fig: D defect}
\end{figure}

Similarly, one can also construct codimension-1 topological defects $\mathcal{ST}$ and $\mathcal{S}\mathcal{T}^2$ corresponding to the $ST$ and $ST^2$ operations by applying these operations only to half of the spacetime.
Again, we impose the Dirichlet boundary condition on the dynamical gauge field at the defect locus.
The constructions of $\mathcal{ST}$ and $\mathcal{S}\mathcal{T}^2$ are illustrated in Figure~\ref{fig: DT defect} and Figure~\ref{fig: DU defect}, respectively.
These defects will be called a triality defect and a quaternity defect, respectively.
\begin{figure}
\centering
\adjincludegraphics[valign=c]{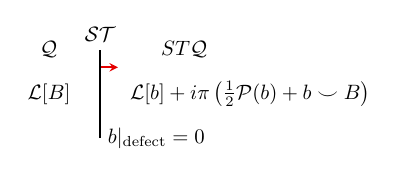}
\caption{The half-space gauging construction of the triality defect $\mathcal{ST}$.}
\label{fig: DT defect}
\end{figure}
\begin{figure}
\centering
\adjincludegraphics[valign=c]{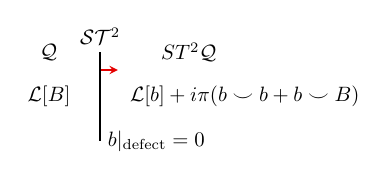}
\caption{The half-space gauging construction of the quaternity defect $\mathcal{S}\mathcal{T}^2$.}
\label{fig: DU defect}
\end{figure}

In later discussions, following \cite{Choi:2022zal}, we will use the following shorthand notations for the defects defined above:
\begin{align}
\mathcal{S}&: \mathcal{L}[b] + i \pi b \smile B, \\
\mathcal{S}\mathcal{T}&: \mathcal{L}[b] + i\pi \left(\frac{1}{2}\mathcal{P}(b) + b \smile B\right), \label{eq: DT lag} \\
\mathcal{S}\mathcal{T}^2&: \mathcal{L}[b] + i \pi (b \smile b + b \smile B) \label{eq: DU lag}.
\end{align}
Namely, we represent a defect by the Lagrangian on the right side of the defect.
Here, $\mathcal{L}$ is the Lagrangian on the left side of the defect.\footnote{The Lagrangian is used only for notational simplicity. The construction of the defects does not rely on the existence of a Lagrangian.}
This notation will be useful when we compute the fusion rules of the defects.

\subsubsection{Fusion rules}
\label{sec: Z2 fusion computation}
Now, we compute the fusion rules of the topological defects following~\cite{Choi:2022zal}.
More specifically, we will compute $(\mathcal{ST})^3$ and $(\mathcal{S}\mathcal{T}^2)^4$ and show that they act trivially on any local operators.
%More specifically, we will compute $(\mathcal{DU})^4$ and compare it with the lattice fusion rule~\eqref{eq: DU4 lattice}.
%Our computation will be parallel to the computation of the fusion rules of the triality defect in \cite{Choi:2022zal}.\footnote{Here, the triality defect, denoted by $\mathcal{D}_3$ in \cite{Choi:2022zal}, is the codimension-1 topological defect corresponding to the $ST$ operation.}

\vspace*{\baselineskip}
\noindent{\bf Fusion rule of $\mathcal{ST}$.}
We first compute $(\mathcal{ST})^3$.
Using the notation in~\eqref{eq: DT lag}, we can represent the fusion of two copies of $\mathcal{ST}$ as
\begin{multline}
 (\mathcal{ST})^2: \mathcal{L}[b_1] + i\pi \bigg( \frac{1}{2} \mathcal{P}(b_1) + b_1 \smile b_2 \\
 + \frac{1}{2} \mathcal{P}(b_2) + b_2 \smile B \bigg) \,.
\end{multline}
By fusion another copy of $\mathcal{ST}$, we obtain
\begin{multline}
 (\mathcal{ST})^3 : \mathcal{L}[b_1] + i\pi \left( \frac{1}{2} \mathcal{P}(b_1) + b_1 \smile b_2 + \frac{1}{2} \mathcal{P}(b_2) \right. \\ 
 \left. + b_2 \smile b_3 + \frac{1}{2} \mathcal{P}(b_3) + b_3 \smile B \right).
\end{multline}
If we define $b \coloneq b_1 + b_2 + b_3$ and rewrite $(\mathcal{ST})^3$ in terms of $b$, $b_1$, and $b_3$, we find
\begin{equation}
(\mathcal{ST})^3: \mathcal{L}[b_1] + i\pi \left( \frac{1}{2} \mathcal{P}(b) + b_1 \smile b_3 + b_3 \smile B \right).
\label{eq: ST3 defect}
\end{equation}
Here, we used the following identity of the Pontryagin square:
\begin{equation}
\frac{1}{2}\mathcal{P}(b_1+b_2) = \frac{1}{2} \mathcal{P}(b_1) + \frac{1}{2} \mathcal{P}(b_2) + b_1 \smile b_2.
\end{equation}
We note that $b$ in \eqref{eq: ST3 defect} is decoupled from the other gauge fields.
Thus, the $\frac{1}{2} \mathcal{P}(b)$ term in \eqref{eq: ST3 defect} represents an independent 3+1d TQFT, whose partition function is given by~\eqref{eq: Y partition function}.
This TQFT hosts a chiral semion topological order on the defect where the gauge field $b$ obeys the Dirichlet boundary condition.
On the other hand, as shown in \cite{Choi:2022zal}, the remaining terms in \eqref{eq: ST3 defect} correspond to the condensation defect for the $\mathbb{Z}_2$ 1-form symmetry, 
which is defined by \cite{Roumpedakis:2022aik, Choi:2022zal}
\begin{equation}
\mathcal{C}_0 \coloneq \frac{1}{|H^0(M_3; \mathbb{Z}_2)|} \sum_{\Sigma \in H_2(X_3; \mathbb{Z}_2)} \eta(\Sigma),
\label{eq: condensation Z2}
\end{equation}
where $M_3$ is the support of the defect and $\eta(\Sigma)$ is the $\mathbb{Z}_2$ 1-form symmetry defect on $\Sigma$. 
Thus, equation~\eqref{eq: ST3 defect} implies that $\mathcal{ST}$ obeys the following fusion rule \cite{Choi:2022zal}:
\begin{equation}
(\mathcal{ST})^3 = \mathcal{Z}_{\mathrm{semion}} \mathcal{C}_0.
\label{eq: ST3 continuum defect}
\end{equation}
Here, the coefficient $\mathcal{Z}_{\mathrm{semion}}$ is the partition function of the 2+1d chiral semion TQFT on the defect.

\vspace*{\baselineskip}
\noindent{\bf Fusion rule of $\mathcal{S}\mathcal{T}^2$.}
Next, we compute $(\mathcal{S}\mathcal{T}^2)^4$.
Using the notation in \eqref{eq: DU lag}, we can represent the fusion of two copies of $\mathcal{S}\mathcal{T}^2$ as
\begin{equation}
(\mathcal{S}\mathcal{T}^2)^2: \mathcal{L}[b_1] + i \pi (b_1 {\smile} b_1 + b_1 {\smile} b_2 + b_2 {\smile} b_2 + b_2 {\smile} B) \,.
\end{equation}
By fusing another copy of $(\mathcal{S}\mathcal{T}^2)^2$, we obtain
\begin{multline}
 (\mathcal{S}\mathcal{T}^2)^4 : \mathcal{L}[b_1] + i \pi (b_1 \smile b_1 + b_1 \smile b_2 + b_2 \smile b_2 \\
 + b_2 \smile b_3 + b_3 \smile b_3 + b_3 \smile b_4 + b_4 \smile b_4 + b_4 \smile B).
\end{multline}
If we define $b \coloneq b_1+b_3$ and $b^{\prime} \coloneq b_2+b_4$, the above expression reduces to
\begin{equation}
\begin{aligned}
(\mathcal{S}\mathcal{T}^2)^4&: i \pi (b \smile b + b \smile b^{\prime} + b^{\prime} \smile b^{\prime}) \\
& \quad + \mathcal{L}[b_1] + i \pi (b_1 \smile b_4 + b_4 \smile B).
\end{aligned}
\label{eq: DU4 continuum lag}
\end{equation}
We note that the gauge fields $b$ and $b^{\prime}$ are decoupled from $b_1$, $b_4$, and $B$.
The first line of \eqref{eq: DU4 continuum lag} is the Lagrangian of the Crane-Yetter-Walker-Wang TQFT based on the 3-fermion MTC, cf.~\eqref{eq: 3F WW}.
In particular, on the defect where the Dirichlet boundary condition is imposed, this TQFT realizes the chiral topological order described by the 3-fermion MTC.
On the other hand, the second line of \eqref{eq: DU4 continuum lag} corresponds to the condensation defect $\mathcal{C}_0$~\cite{Choi:2022zal}.
Therefore, equation~\eqref{eq: DU4 continuum lag} implies the following fusion rule: 
\begin{equation}
(\mathcal{S}\mathcal{T}^2)^4 = \mathcal{Z}_{3F} \mathcal{C}_0.
\label{eq: DU4 continuum defect}
\end{equation}
Here, the coefficient $\mathcal{Z}_{3F}$ is the partition function of the 2+1d 3-fermion TQFT on the defect.

\vspace*{\baselineskip}
\noindent{\bf Action on local operators.}
%To compare the above fusion rules with their lattice counterparts in later sections, we now consider the actions of $(\mathcal{ST})^3$ and $(\mathcal{S}\mathcal{T}^2)^4$ on local operators.
Based on the fusion rules~\eqref{eq: ST3 continuum defect} and \eqref{eq: DU4 continuum defect}, we can now compute the actions of $(\mathcal{ST})^3$ and $(\mathcal{S}\mathcal{T}^2)^4$ on local operators.
These actions will be compared with their lattice counterparts in later sections.

In general, the action of any codimension-1 topological defect $\mathcal{D}$ on a local operator $\mathcal{O}$ is defined by putting $\mathcal{D}$ on a 3-sphere surrounding $\mathcal{O}$.
More specifically, we define the action of $\mathcal{D}$ on $\mathcal{O}$ by
% \begin{equation}
% \alpha_{\mathcal{D}}(\mathcal{O}) \coloneq \mathcal{D}(\mathcal{O}) / \langle \mathcal{D} \rangle_{S^3},
% \label{eq: defect action}
% \end{equation}
\be\label{eq: defect action}
\opaction{\mc{D}}(\mc{O}) \coloneqq \frac{\mc{D}(\mc{O})}{\langle \mc{D} \rangle_{\mathbb{S}^3}}
\ee
where $\mathcal{D}(\mathcal{O})$ is the local operator $\mathcal{O}$ surrounded by $\mathcal{D}$.
Here, we normalized the action of $\mathcal{D}$ by its quantum dimension $\langle \mathcal{D} \rangle_{\mathbb{S}^3}$ on a 3-sphere $\mathbb{S}^3$.
This normalization allows us to directly compare the results in the continuum and those on the lattice in later sections.

On a 3-sphere, the condensation defect~\eqref{eq: condensation Z2} is proportional to the identity operator because there is no non-trivial 2-cycle on a 3-sphere.
Therefore, equations~\eqref{eq: ST3 continuum defect} and \eqref{eq: DU4 continuum defect} imply that $(\mathcal{ST})^3$ and $(\mathcal{S}\mathcal{T}^2)^4$ are both proportional to the identity operator on a 3-sphere.
In particular, these defects can pass through any local operator, meaning that they act trivially on local operators.
That is, we have
\begin{align}
 \opaction{(\mc{ST})^3} &= 1 \,, & \opaction{(\mc{ST}^2)^4} &= 1 \,.
\end{align}
% \begin{equation}
%  \alpha_{\mathcal{ST}}^3 = 1, \quad
%  \alpha_{\mathcal{S}\mathcal{T}^2}^4 = 1.
% (\mathcal{ST})^3(\mathcal{O}) = \mathcal{O}, \quad
% (\mathcal{S}\mathcal{T}^2)^4(\mathcal{O}) = \mathcal{O}
% \end{equation}
%for any local operator $\mathcal{O}$.
This shows that the actions of $\mathcal{ST}$ and $\mathcal{S}\mathcal{T}^2$ on local operators obey the $\mathbb{Z}_3$ and $\mathbb{Z}_4$ fusion rules, respectively.

\subsection{$\Z_p^{(1)}$ Symmetry Case}
\label{sec: continuum fusion rules Zp}
We now generalize the discussion in the previous subsection to the case of $\mathbb{Z}_p$ 1-form symmetry following~\cite{Choi:2022zal}.

\subsubsection{The $S_p$ and $T_p$ Operations}
\label{sec: S and T Zp}
As in the $\mathbb{Z}_2$ case, we first introduce the $S$ and $T$ operations on 3+1d QFTs with $\mathbb{Z}_p$ 1-form symmetry.
For the $\mathbb{Z}_p$ case, we avoid ambiguity and denote the $S$ and $T$ operations by $S_p$ and $T_p$, respectively.
At the level of partition functions, these operations are defined by \cite{Choi:2022zal}
\begin{equation}
\begin{aligned}
Z_{S_p\mathcal{Q}}[B] &= \frac{1}{|H^2(M_4; \mathbb{Z}_p)|^{\frac{1}{2}}} \sum_{b \in H^2(M_4; \mathbb{Z}_p)} Z_{\mathcal{Q}}[b] e^{-\frac{2\pi i}{p} \int b \smile B}, \\
Z_{T_p\mathcal{Q}}[B] &= Z_{\mathcal{Q}}[B] e^{\frac{2\pi i}{p} \int \frac{p-1}{2} B \smile B}.
\end{aligned}
\label{eq: ST Zp}
\end{equation}
Here, $Z_{\mathcal{Q}}[B]$ is the parition function of a $\mathbb{Z}_p$ 1-form symmetric QFT $\mathcal{Q}$ on an oriented 4-manifold $M_4$ equipped with a background $\mathbb{Z}_p$ 2-form gauge field $B$.
We note that the $S_p$ operation corresponds to gauging the $\mathbb{Z}_p$ 1-form symmetry, while the $T_p$ operation corresponds to stacking a $\mathbb{Z}_p$ 1-form SPT phase.\footnote{Our convention for the $S_p$ and $T_p$ operations is slightly different from \cite{Choi:2022zal}. Specifically, the $S_p$ operation in \eqref{eq: ST Zp} is the $S$ operation in \cite{Choi:2022zal} followed by the charge conjugation, and the $T_p$ operation in \eqref{eq: ST Zp} is the inverse of the $T$ operation in~\cite{Choi:2022zal}.}
We also note that since the exponent in the $T_p$ operation is effectively calculated mod $p$, $T_p^{-2} = T_p^{(p-1) \cdot 2}$ is the stacking of a phase $\frac{2 \pi i}{p} \int 1 \cdot B \smile B$, which corresponds a little more clearly to the generating SPT phase we might label as $1 \in \Z_p$. Nevertheless, both $T_p$ and $T_p^{-2}$ are generators of the same $\Z_p$ group of SPTs (actually, for prime $p$, any non-zero element generates the whole group).
%We note that the $S$ operation corresponds to gauging the $\mathbb{Z}_p$ 1-form symmetry, while the $T$ operation corresponds to stacking a $\mathbb{Z}_p$ 1-form SPT phase labeled by $k=\frac{p-1}{2}$.\footnote{Our convention for the $S$ and $T$ operations is slightly different from \cite{Choi:2022zal}. Specifically, the $S$ operation in \eqref{eq: ST Zp} is the $S$ operation in \cite{Choi:2022zal} followed by the charge conjugation, and the $T$ operation in \eqref{eq: ST Zp} is the inverse of the $T$ operation in~\cite{Choi:2022zal}.}

The above $S_p$ and $T_p$ operations obey the following relations \cite{Choi:2022zal}:
\begin{equation}
S_p^2 = C, \qquad
T_p^p=1, \qquad
(S_pT_p)^3 = Y_p.
\label{eq: SL 2 Zp}
\end{equation}
Here, $C$ is the charge conjugation defined by $Z_{C\mathcal{Q}}[B] = Z_{\mathcal{Q}}[-B]$, and $Y_p$ is the stacking of a 3+1d invertible TQFT whose partition function is given by
\begin{equation}
Z_{Y_p} = \frac{1}{|H^2(M_4; \mathbb{Z}_p)|^{\frac{1}{2}}} \sum_{b \in H^2(M_4; \mathbb{Z}_p)} e^{\frac{2\pi i}{p} \int \frac{p-1}{2} b \smile b}.
\label{eq: Y Zp}
\end{equation}
We note that $Z_{Y_p}$ is the partition function of the Crane-Yetter-Walker-Wang TQFT based on the chiral $\mathbb{Z}_p$ MTC whose simple object $a \in \mathbb{Z}_p$ has topological spin
\begin{equation}
\theta(a) = e^{\frac{2\pi i}{p} \frac{p-1}{2} a^2}.
\end{equation}
The chiral $\mathbb{Z}_p$ MTC with the above topological spins is known as $\mathrm{SU}(p)_1$, which describes the topological order realized in $\mathrm{SU}(p)_1$ Chern-Simons theory.

We note that $Y_p$ has order 2 when $p = 1 \bmod 4$ and order $4$ when $p=3 \bmod 4$.
This is because the order of $\mathrm{SU}(p)_1$ in the Witt group of non-degenerate braided fusion categories is 2 or 4 depending on whether $p=1$ mod 4 or $p=3$ mod 4 \cite{davydov2013witt}.
{Furthermore, $Y_p$ is in the center of the group generated by $\mathsf{S}_p$ and $\mathsf{T}_p$ because the partition function~\eqref{eq: Y Zp} is independent of the background gauge field $B$.}

{In general, the relations in \eqref{eq: SL 2 Zp} alone do not determine the group structure generated by $\mathsf{S}_p$ and $\mathsf{T}_p$. Nevertheless, the result in Section~\ref{sec:ZpQCAgroup} strongly suggests that this group is isomorphic to $SL(2, \mathbb{Z}_p) \times \mathbb{Z}_{n_p}$, where $n_p=2$ or $4$ depending on whether $p=1$ or $3$ mod $4$. In particular, the central extension splits as an abstract group, in contrast to the $\mathbb{Z}_2$ case discussed in the previous subsection.}

\subsubsection{Topological Defects}
Based on the above definition of the $S_p$ and $T_p$ operations, we can now define the corresponding topological defects following \cite{Choi:2021kmx, Choi:2022zal}.
In what follows, we will only consider the topological defects corresponding to $S_p$ and $S_pT_p$.
These defects will be denoted by $\mathcal{S}_p$ and $\mathcal{S}_p\mathcal{T}_p$, and are called a duality defect and a triality defect, respectively~\cite{Choi:2022zal}.

The duality defect $\mathcal{S}_p$ is defined by gauging the $\mathbb{Z}_p$ 1-form symmetry only in half of the spacetime \cite{Choi:2021kmx, Choi:2022zal}.
At the defect locus, we impose the Dirichlet boundary condition on the dynamical gauge field.
\begin{figure}
\centering
\adjincludegraphics[valign=c, trim={10, 10, 10, 10}]{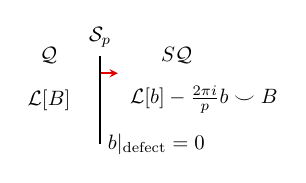}
\caption{The half-space gauging construction of the duality defect $\mathcal{S}_p$. The small red arrow specifies the orientation of the defect. The dynamical gauge field $b$ obeys the Dirichlet boundary condition on the defect.}
\label{fig: duality defect Zp}
\end{figure}
See Figure~\ref{fig: duality defect Zp} for an illustration of this construction.
As in the $\mathbb{Z}_2$ case, the defect defined in this way is denoted by \cite{Choi:2022zal}
\begin{equation}
\mathcal{S}_p: \mathcal{L}[b] - \frac{2\pi i}{p} b \smile B,
\end{equation}
where $\mathcal{L}$ is the Lagrangian of the QFT before gauging.
%This defect is the continuum analogue of the lattice KWW operator $\mathsf{D}_p$.

Similarly, the triality defect $\mathcal{S}_p \mathcal{T}_p$ is defined by applying the $S_pT_p$ operation only in half of the spacetime \cite{Choi:2022zal}.
Namely, we stack a $\mathbb{Z}_p$ 1-form SPT phase and gauge the $\mathbb{Z}_p$ 1-form symmetry only on one side of the defect.
At the defect locus, we again impose the Dirichlet boundary condition on the dynamical gauge field.
See Figure~\ref{fig: triality defect Zp} for an illustration of this construction.
\begin{figure}
\centering
\adjincludegraphics[valign=c, trim={10, 10, 10, 10}]{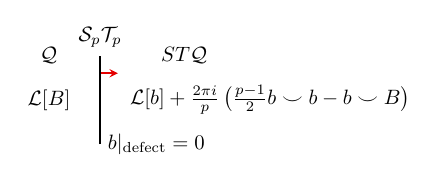}
\caption{The half-space gauging construction of the triality defect $\mathcal{S}_p\mathcal{T}_p$.}
\label{fig: triality defect Zp}
\end{figure}
The defect defined in this way is denoted by
\begin{equation}
\mathcal{S}_p \mathcal{T}_p: \mathcal{L}[b] + \frac{2\pi i}{p} \left( \frac{p-1}{2} b \smile b - b \smile B \right).
\label{eq: DT continuum}
\end{equation}
We will use the above notations when we compute the fusion rules of the defects.
%This defect is the continuum analogue of the composite operator $\mathsf{D}_p T_p^{\left(\frac{p-1}{2}\right)}$ on the lattice.

\subsubsection{Fusion Rules}
We now compute the fusion rule of the triality defect $\mathcal{S}_p\mathcal{T}_p$ following \cite{Choi:2022zal}.
More specifically, we will compute $(\mathcal{S}_p\mathcal{T}_p)^3$ and show that it acts trivially on any local operator.

To compute $(\mathcal{S}_p \mathcal{T}_p)^3$, we first compute the fusion of two copies of $\mathcal{S}_p\mathcal{T}_p$.
Using the notation in~\eqref{eq: DT continuum}, one can represent $(\mathcal{S}_p\mathcal{T}_p)^2$ as
\begin{equation}
\begin{aligned}
(\mathcal{S}_p \mathcal{T}_p)^2 &: \mathcal{L}[b_1] + \frac{2\pi i}{p} \left( \frac{p-1}{2} b_1 \smile b_1 - b_1 \smile b_2 \right. \\
& \quad \left. + \frac{p-1}{2} b_2 \smile b_2 - b_2 \smile B \right).
\end{aligned}
\end{equation}
One can rewrite $(\mathcal{S}_p\mathcal{T}_p)^2$ in terms of $b_1$ and $b \coloneq b_1+b_2$ as follows:
\begin{multline}
 (\mathcal{S}_p\mathcal{T}_p)^2: \mathcal{L}[b_1] + \frac{2\pi i}{p} \left( \frac{p-1}{2} b \smile b \right. \\
 \left. + b_1 \smile B + \frac{p+1}{2} B \smile B \right) \,.
\end{multline}
By fusing another copy of $\mathcal{S}_p\mathcal{T}_p$, we obtain
\begin{equation}
(\mathcal{S}_p\mathcal{T}_p)^3: \mathcal{L}[b_1] + \frac{2\pi i}{p} \left( \frac{p-1}{2} b {\smile} b + b_1 {\smile} b_2 - b_2 {\smile} B \right).
\label{eq: DTDTDT continuum1}
\end{equation}
We note that the gauge field $b$ in \eqref{eq: DTDTDT continuum1} is decoupled from the other gauge fields.
Hence, the $b \smile b$ term in the above expression corresponds to a standalone 3+1d TQFT, whose partition function is given by \eqref{eq: Y Zp}.
In particular, this TQFT is invertible and supports $\mathrm{SU}(p)_1$ Chern-Simons theory on the Dirichlet boundary.
On the other hand, the remaining terms in \eqref{eq: DTDTDT continuum1} correspond to the condensation defect for the $\mathbb{Z}_p$ 1-form symmetry~\cite{Choi:2022zal}, which is defined by \cite{Roumpedakis:2022aik, Choi:2022zal}
\begin{equation}
\mathcal{C}_0 \coloneq \frac{1}{|H^0(M_3; \mathbb{Z}_p)|} \sum_{\Sigma \in H_2(M_3; \mathbb{Z}_p)} \eta(\Sigma),
\end{equation}
where $M_3$ is the worldvolume of the defect and $\eta(\Sigma)$ is the $\mathbb{Z}_p$ 1-form symmetry operator supported on $\Sigma$.
Thus, equation~\eqref{eq: DTDTDT continuum1} implies that $\mathcal{S}_p\mathcal{T}_p$ obeys the following fusion rule:
\begin{equation}
(\mathcal{S}_p\mathcal{T}_p)^3 = \mathcal{Z}_{\mathrm{SU}(p)_1} \mathcal{C}_0.
\label{eq: DTDTDT continuum2}
\end{equation}
Here, $\mathcal{Z}_{\mathrm{SU}(p)_1}$ is the partition function of $\mathrm{SU}(p)_1$ Chern-Simons theory on the defect.

As in the $\mathbb{Z}_2$ case discussed in Section~\ref{sec: Z2 fusion computation}, it immediately follows from the fusion rule~\eqref{eq: DTDTDT continuum2} that the action of $\mathcal{S}_p\mathcal{T}_p$ on local operators satisfies
\be
\opaction{(\mc{S}_p \mc{T}_p)^3} = 1
\ee
% \begin{equation}
% \alpha_{\mathcal{S}_p\mathcal{T}_p}^3=1.
% (\mathcal{S}_p\mathcal{T}_p)^3(\mathcal{O}) = \mathcal{O}.
% \end{equation}
Here, the action of a topological defect on local operators is defined by~\eqref{eq: defect action}.
The above equation shows that the action of $\mathcal{S}_p\mathcal{T}_p$ on local operators obeys the $\mathbb{Z}_3$ fusion rule.

\section{Graded Witt Groups and Automorphisms of the 1-Form SymTFT }
\label{sec:Cats}

A quantum cellular automaton (QCA) is a locality-preserving automorphism of the
operator algebra of a lattice system, considered trivial if it is blend-equivalent to the identity - in particular, if it is a finite depth quantum circuit.  
In $3+1$d the conjectured obstruction to triviality is a $2+1$d abelian anyon theory \cite{Shirley_2022}: 
the basic nontrivial
examples arise as boundary disentanglers of Walker--Wang models \cite{walker20123+},
and, in these known examples, such a QCA is a circuit precisely when the input anyon theory admits a gapped
boundary.  
Since anyon theories related by condensation define equivalent QCAs, what
classifies QCAs is not the anyon theory itself but its {\bf Witt class} --- for Clifford
QCAs modulo Clifford circuits and shifts this is a theorem
\cite{Haah_2021,Haah:2022yyo},\footnote{See also, e.g., \cite{Yang:2025jvn, Czajka:2025mme, Ji:2026fka, Yang:2026lnf} for recent developments on the classification of QCAs.} with the three-fermion theory $SO(8)_1$ (the central
class $\cC$ below) providing the seminal nontrivial example
\cite{HaahFidkowskiHastings2023}.  

Applying the QCA to a trivial short range entangled Hamiltonian only in half of space leads to a Hamiltonian which cannot be trivially gapped with commuting projectors near the boundary of this half-space.  Instead, a typical gapping commuting projector Hamiltonian acting at the boundary will lead to a $2+1$d anyon theory living on the boundary.  Different commuting projector Hamiltonians may possibly be obtained by condensing anyons, but this cannot change the Witt class, reflecting the invariance of the QCA class.  The spacetime picture suggested by this construction is that, like any $0$-form symmetry, each QCA has a corresponding co-dimension $1$ defect, with a $2+1$d anyon theory living on it, and the QCA is classified by the Witt class of the anyon theory.

In this paper the QCAs are in addition required to commute with a $\Z_p^{(1)}$ 1-form symmetry, which endows their co-dimension $1$ defects with additional structure.  More generally, we consider not just QCAs on the full tensor product operator algebra, but also automorphisms of the algebra of $\Z_2$ $1$-form symmetric operators, which we refer to as {\bf symmetric QCAs}.  We note that these more general symmetric QCA can be non-invertible symmetries from the point of view of the tensor product operator algebra. 
As is familiar from ordinary $1+1$d Kramers-Wannier duality, such symmetric QCA can be effectively analyzed from the SymTFT point of view, in this case with a $4+1$d bulk in a $(2,2)$ $\Z_p$ toric code phase, or the $\Z_p$ BF-theory based on 2-form gauge fields.  Namely, the co-dimension $1$ defect can be extended to a co-dimension $1$ interface in the $4+1$d bulk, and this interface can act on the $1+1$d  excitations of the $(2,2)$ toric code: dragging such an electric or magnetic $1+1$d excitation through the interface returns a possibly different $1+1$d excitation; this yields an automorphism of $\Z_p\oplus\widehat{\Z_p}$ preserving the electric-magnetic braiding.  In the case of $\Z_{2n}$ $1$-form symmetries, there are also additional more subtle phases corresponding to dressing the $1+1$d excitations with fermion world-lines \cite{Chen:2023Loops}.

The algebraic structure capturing this invariant is computed in \cite{Bhardwaj:2024xcx}.  One particular presentation of it, which aligns with the above Witt classification of ordinary QCA, is as the graded, syllepsis-twisted Witt group $\Witt^\pt(\Z_p\oplus\widehat{\Z_p},s)$.  An element of this group is a so-called `$s$-invertible $\Z_p\oplus\widehat{\Z_p}$-graded premetric group'.  
Physically, this is just the quasiparticle theory living at the co-dimension $1$ defect associated to the symmetric QCA. 
The additional $\Z_p\oplus\widehat{\Z_p}$ grading structure records which bulk $1+1$d excitation terminates on a particular anyon, and $s$-invertibility encodes a version of braiding non-degeneracy appropriate for this setting.  
Furthermore, only interface anyons invisible to the bulk (grade zero) may be condensed.

The graded, syllepsis-twisted Witt group $\Witt^\pt(\Z_p\oplus\widehat{\Z_p},s)$ contains the ungraded pointed Witt group $\Witt^\pt$ classifying ordinary QCA as a central subgroup.  It is thus natural to conjecture that the Witt classification of ordinary QCA is promoted to a $\Witt^\pt(\Z_p\oplus\widehat{\Z_p},s)$ classification of symmetric QCA.

In this section we will provide evidence for this conjecture by computing the structure of $\Witt^\pt(\Z_p\oplus\widehat{\Z_p},s)$ as a central extension of its quotient by the (standard pointed) Witt group and showing that it matches the field theory defect fusion calculations of Section~\ref{sec: Fusion Rules in the Continuum}  and the spatial lattice QCA calculations of the subsequent sections.  

Specifically we show that for $\Z_2$ $1$-form symmetry, this quotient is ${S}_4 \cong SL(2,\Z_4) / \Z_2$ and has non-trivial mixing with a $\Z_8$ subgroup of the Witt group, inside the full 
$\Witt^\pt(\Z_2\oplus\widehat{\Z_2},s)$. 
This extension is the Clifford group of one qubit, with $S$ and $T$ corresponding to the Hadamard and phase gate respectively, and the ordinary QCA corresponding to $8$'th roots of unity, with the primitive $8$'th root of unity being the semion QCA.  Then, for example, the relation $(ST)^3=Y$ of Theorem~\ref{thm:modular} is reflected in the corresponding algebraic property of $2$ by $2$ unitary matrices.  For odd $p$ no such obstruction survives, and the extension splits, see Theorem~\ref{thm:WittZpMain}.  However, if we insist on identifying $S$ and $T$ with particular pre-metric groups that we view as corresponding to lattice-level symmetric QCA implementing the KWW duality and the $\Z_p$ $1$-form SPT entangler respectively, and do not allow arbitrary re-labelings of these generators, we can still recover the relation $(S_p T_p)^3 = Y_p$, in accordance with the continuum and lattice-level calculations.

\subsection{SymTFT for 1-form Symmetries in 3+1d}

We will study the braided 
automorphism of the SymTFT for the 1-form symmetry, which is a 4+1d $\Z_p$ gauge theory for surface defects. These braided automorphism are given in terms invertible topological interfaces that 
permute and decorate  the electric and magnetic surface defects in the 4+1d SymTFT. 

Consider a $3+1$d theory, which has a $\Z_p^{(1)}$ 1-form symmetry, with $p$ prime. 
Such a theory can be coupled to a background two-form gauge field
$B\in Z^2(M_4,\Z_p)$, or more concretely to a $\Z_p$-valued 2-cocycle modulo
1-form gauge transformations $B\sim B+\delta\lambda$.

The corresponding
SymTFT is a $4+1$d topological theory whose job is to keep track of all
possible boundary conditions, gaugings, and duality interfaces for this
background field.  A convenient cochain model for the universal
$\Z_p^{(1)}$ SymTFT uses two $\Z_p$-valued 2-cochains $B,C\in C^2(M_5,\Z_p)$
with BF action
\be\label{eq:ZpSymTFTBF}
  S_{\rm SymTFT}(B,C)
  ={2\pi i\over p}\int_{M_5} B\smile\delta C .
\ee
The gauge transformations are
\be
  B\mapsto B+\delta\lambda,\qquad C\mapsto C+\delta\mu,
  \qquad \lambda,\mu\in C^1(M_5,\Z_p),
\ee
and the equations of motion impose $\delta B=\delta C=0$ away from charged
insertions.  
On a physical boundary $M_4=\partial M_5$, choosing a
polarization means choosing which linear combination of $B$ and $C$ is held
fixed as the background field for the boundary theory and selects out a gapped boundary condition of the SymTFT. 
For $p$ prime there are two polarizations, given by $B$ or $C$ having Dirichlet boundary conditions, respectively. 
Gauging the 1-form symmetry exchanges the two polarizations, and stacking a 1-form SPT phase shifts
the boundary action by a quadratic term 
\be 
(2\pi i k/p)\int_{M_4}B\smile B
\ee
for odd $p$. We will discuss the $p=2$ case as well. 

\smallskip
\noindent{\bf Surface defects and the electric-magnetic pairing.}
The bulk SymTFT has topological surface defects labelled by pairs
\be
  x=(x_1,x_2)\in \Z_p\oplus\widehat{\Z_p} = \Z_p^2\,.
\ee
For a closed two-cycle $\Sigma\subset M_5$, the corresponding defect can be
represented by a surface 
\be
  W_x(\Sigma)
  =
  \exp\!\left({2\pi i\over p}\int_\Sigma (x_1B+x_2C)\right)\,.
\ee
where the $B$ and $C$ contributions corresponds to the electric and magnetic surface defects, respectively. 
 Because of the BF-coupling in
\eqref{eq:ZpSymTFTBF}, dragging an electric surface through a magnetic surface
produces a phase due to the braiding.
In the ordered convention used below,
the elementary braiding/syllepsis between two surface charges
$x=(x_1,x_2)$ and $y=(y_1,y_2)$ is
\be\label{eq:SymTFTsyllepsis}
  s(x,y)=\zeta_p^{\,x_1y_2}\,,
\ee
where $\zeta_p= e^{2 \pi i/p}$.
The reverse ordered braiding gives $s(y,x)=\zeta_p^{\,y_1x_2}$, so the invariant
commutator, or monodromy, is the alternating electric-magnetic pairing
\be\label{eq:SymTFTAlt}
  \mathsf{Alt}(s)(x,y)
  ={s(x,y)\over s(y,x)}
  =\zeta_p^{\,x_1y_2-x_2y_1}.
\ee
This is the finite-field symplectic form preserved by the duality modular group
$SL(2,\Z_p)$.  Categorically, the same data is encoded
by the sylleptic 2-category
$\bVect^s((\Z_p\oplus\widehat{\Z_p})[0])$: objects are the surface charges,
the tensor product is addition of charges, and the syllepsis is precisely the
phase \eqref{eq:SymTFTsyllepsis}.  Thus the categorical choice of $s$ below is
just the cochain-level electric-magnetic braiding of the SymTFT.  For $p=2$,
$\zeta_2=-1$ and this reduces to the sign convention
$s((x_1,x_2),(y_1,y_2))=(-1)^{x_1y_2}$.

\subsection{Braided Automorphisms and the Graded Witt group}

We focus in this paper on 4+1d SymTFT for
a $A=\Z_p$ 1-form symmetry, with $p$ prime -- though we will have to distinguish $p=2$ from odd primes $p$. 
The bulk theory has two basic types of surface
defects, electric and magnetic, and the automorphisms of interest are the
topological interfaces that act invertibly on this defect system.  The
categorical problem is therefore to compute the group of invertible braided
interfaces preserving the electric-magnetic pairing.  
Computing these interfaces is directly related to the Witt group. 
The Witt group, $\Witt$, is the group of nondegenerate braided fusion categories, modulo those that are Drinfeld centers. The pointed Witt group $\Witt^\pt$ is the pointed part, i.e. every simple object is invertible. These pointed Witt group elements are simply metric groups, which we discuss in Section~\ref{sec:MetricGroup}: $(G, q)$, where 
$G$ is a finite abelian group of invertible simple objects,
and $q: G \to \mathbb{C}^\times$ a quadratic form, where $q(a)$ determines  the spin of the object.  

For the purpose of computing the SymTFT interfaces, we will be interested in a generalization to $s$-twisted graded Witt group $\Witt^\pt (A \oplus \widehat{A}, s)$ \cite{Bhardwaj:2024xcx}, where the
syllepsis $s$ corresponds to the mutual braiding of electric and magnetic surfaces and
the ordinary pointed Witt group encodes the {abelian}  $2+1$d topological orders that can be stacked on the interface.

\smallskip

We will first recall some properties of the SymTFT and the braided automorphisms from \cite{Bhardwaj:2024xcx}.
For a general fusion higher-category $\cC$, the 
interfaces are measured by the Picard  (Pic) group of
the braided Drinfeld center $\cZ(\cC)$, or equivalently by the Brauer-Picard (BrPic) group of the boundary
fusion category.  Mathematically, the Picard group encodes the 
invertible module
categories over the braided category, while the Brauer-Picard group corresponds to the  
invertible bi-module categories, which are identified 
\be 
\Pic(\cZ(\cC))\cong\BrPic(\cC)\,.
\ee
The graded Witt group below is a convenient presentation of this same 
group: it packages the action of the automorphisms on the surface defects in terms of
a permutation, together with possible lower-dimensional invertible phases living on the interface.

For the 1-form symmetry $A^{(1)}$ in $3+1$d, i.e. the fusion 3-category $\cVect A[1]$, the Drinfeld center (which captures the topological defects of the SymTFT) is given by   
\be\label{SurfaceId}
\cZ (\cVect A[1]) = \Sigma \bVect^s (A \oplus \widehat{A} [0] )\,.
\ee
This describes the SymTFT data: 
\begin{itemize}
\item  Surface defects for $A$ and its Pontryagin dual $\widehat{A}$ in the $4+1$d bulk, {which correspond to the Wilson surfaces of $B$ and $C$ in (\ref{eq:ZpSymTFTBF}), respectively.}
\item  Condensation completion $\Sigma$ to a fusion 3-category.
\item  Braiding (syllepsis) $s$ encoding the electric-magnetic pairing.
\end{itemize}
As shown in \cite{Bhardwaj:2024xcx}, the RHS of (\ref{SurfaceId}) provides a direct connection to the Witt group:
\be
\label{WittPic}
\Witt (A \oplus \widehat{A}, s) = \Pic ({\Sigma}\bVect^s (A \oplus \widehat{A}[0]))\,,
\ee
where the LHS is defined as a generalization of the standard Witt group to braided $A \oplus \widehat{A}$-graded, $s$-invertible fusion categories.

\begin{widetext}
{It is also shown in \cite{Bhardwaj:2024xcx}} that this Witt group fits into the exact sequence
\be \label{ShortishSequence}
0\rightarrow \Witt\rightarrow \Witt(A\oplus \widehat{A},s)\rightarrow \text{Aut}^{syp}(A\oplus \widehat{A},\mathsf{Alt}(s))\ltimes H^5((A\oplus \widehat{A})[3];\mathbb{C}^{\times})\rightarrow \Z_2\,,
\ee 
where Aut$^{syp}(A\oplus \widehat{A},\mathsf{Alt}(s))$ is the group of automorphisms of $A\oplus \widehat{A}$ preserving the alternating 2-form $\mathsf{Alt}(s)$.
\end{widetext}

The syllepsis is not an additional choice: it is fixed by the
electric-magnetic pairing of the SymTFT.  For a self-dual group
$A\simeq\widehat{A}$, choosing a basis of electric and magnetic surface defects
identifies $A\oplus\widehat{A}$ with $A\oplus A$ and gives the standard 
syllepsis. Changing this basis conjugates the formulas below by a symplectic
automorphism.  Thus the invariant datum is the non-degenerate alternating form
$\mathsf{Alt}(s)$, whose automorphism group is the corresponding symplectic
group.

\smallskip
\noindent
{\bf ${A \oplus \widehat{A}}$-graded metric groups.}
An important realization of the graded Witt group is in terms of metric groups, see  \cite{Bhardwaj:2024xcx}. First note that 
\be
\Witt ({A \oplus \widehat{A}}, s)/\Witt\cong  
\Witt^{\pt} ({A \oplus \widehat{A}},s)/\Witt^{\pt} \,,
\ee
where $\pt$ indicates the pointed part of the Witt group. 
 The corresponding equality between the full graded
Witt quotient and its pointed part was discussed in \cite{Bhardwaj:2024xcx}. For the purposes of the present paper we only use the pointed version.

$\Witt^{\pt}({A \oplus \widehat{A}}, s)$ is generated by the $s$-invertible pointed braided
${A \oplus \widehat{A}}$-graded fusion 1-categories $\mathcal{B}$: let $G=\text{Inv}(\mathcal{B})$ be the {simple} objects (which in a {pointed} braided category are invertible), the  grading is the
${A \oplus \widehat{A}}$-grading, and the self-braiding $\beta$ trivialised against the syllepsis $s$ (see (5.13) in \cite{Bhardwaj:2024xcx})
\be
 q(B)=\beta(B,B)\cdot s\big(\text{gr}(B),\text{gr}(B)\big) \,.
\ee
Such $s$-invertible pointed braided $\abeliangrp$-graded fusion 1-categories are in 1-1 with $\abeliangrp$-graded pre-metric groups, 
\begin{definition}[$\abeliangrp$-graded Pre-metric Group]\label{def:graded_premetric}
An $\abeliangrp$-graded pre-metric group is given by a triple 
\be 
(f: G\to \abeliangrp,\,q)\,,
\ee
where $G$ is a finite abelian group, $f$ is a homomorphism (the grading), and
$q: G \to \mathbb{C}^\times$ is a {quadratic form}: $q(g)=q(-g)$ and
\be
  \text{Bil}(q)(g,h):=\frac{q(g+h)}{q(g)\,q(h)}
\ee
is bilinear. It is a metric group if $\text{Bil}(q)$ is non-degenerate. Define $G_0=\ker f$ to be the grade-$0$ subgroup.
\end{definition}

The monoidal structure is the $s$-twisted ${A \oplus \widehat{A}}$-graded Deligne product: 
\be
\left(f: G \rightarrow {A \oplus \widehat{A}}, q_G\right) \boxtimes_{A \oplus \widehat{A}}^s\left(k: H \rightarrow {A \oplus \widehat{A}}, q_H\right) \,.
\ee
It takes $G\oplus H$ with grading
$f + k$ and quadratic form
\footnote{Note that eq.~(5.17) of \cite{Bhardwaj:2024xcx} mis-states this formula; the correct twist carries the additional factor $s(k(h),f(g))^2$ as in \eqref{QZ2}.  This is invisible for $p=2$, where $s^2=1$, but essential for odd $p$: it is what makes the anti-diagonal of $X\boxtimes^s_{A \oplus \widehat{A}} X^{s\text{-}\rm op}$ isotropic (cf.\ the coefficient-$3$ discussion in Section~\ref{app:WittZp}), and all computations in this paper use it. We provide a derivation in Appendix \ref{app:framework}.}
\be\label{QZ2}
Q(g , h):=q_G(g) \cdot q_H(h) \cdot s(f(g), k(h)) \cdot s( k(h), f(g))^2 \,.
\ee
Whether such a product is trivial or not can be detected using 
the Lemmas in Section 5.2 of \cite{Bhardwaj:2024xcx}. 
\begin{lemma}[Triviality criterion]\label{lem:criterion}
An $\abeliangrp$-graded metric group $(f: G\to \abeliangrp, q)$ represents $1$ (i.e. is trivial) in $\Witt^{pt}(\abeliangrp, s)/\Witt^{pt}$ if and only if
there is a subgroup $L \subseteq G_0$ that is {isotropic (i.e. bosonic anyons)} ($q|_L \equiv 1$) and whose orthogonal
complement
\be
  L^{\perp}:=\{ g\in G:\ \text{Bil}(q)(g,\ell)=1\ \ \text{for all}\ \ell\in L\,\}
\ee
satisfies $L^{\perp}\subseteq G_0$.
Equivalently (since $\Bil(q)$ is non-degenerate) one may take $L=G_0^{\perp}$ provided
$G_0^{\perp}\subseteq G_0$ and $q|_{G_0^{\perp}} =1$.
\end{lemma}
We will use this description to determine the pointed graded Witt groups for specific 1-form symmetry gradings in the following. 

\noindent{\bf Metric groups as abelian TOs.} 
We will also need to consider the Witt group without any grading. 
A metric group $(G,q)$, is the pointed braided fusion category $\mathcal{B}(G,q)$ whose simple 
objects are the elements of {an abelian group} $G$, with fusion the group law, the spin $\theta_g=q(g)$, and the braiding is 
\be
 M(g,h)=\Bil(q)(g,h) :=\frac{q(g+h)}{q(g)\,q(h)} \,.
\ee
It is modular exactly when $q$ is 
non-degenerate. In the graded setting the homomorphism $f: G\to A\oplus \widehat{A}$ encodes the charge of each
anyon, and $G_0=\ker f$ is the neutral, grade 0, subcategory.

A connected \'etale (i.e.\ commutative separable, i.e. condensable) algebra in $\mathcal{B}(G,q)$ is the same as an isotropic subgroup $L\subseteq G$:
\be\ba 
 & L\ \text{isotropic}\ (q|_L\equiv 1)
 \cr 
 \ \longleftrightarrow\ &
 \text{condensable} \  \cA_L=\bigoplus_{\ell} {\ell}\ \subset \mathcal{B}(G,q)\,,\ea
\ee
i.e. condensable algebras are in 1-1 with isotropic subgroups of $G$. The 
anyons in $L$ have trivial spin, and need to braid trivially with one another. This translates into the condition $q|_L\equiv 1$ on the subgroup
$L$
\be
\Bil(q)(\ell,\ell')=q(\ell+\ell')/q(\ell)q(\ell')=1 \,.
\ee 
So isotropic subgroups are exactly
the condensable algebras.

Condensing the algebra $\cA_L$ produces the category of local $\cA_L$-modules  with \cite{WittGroup}
\be\label{Condis}
  \mathcal{B}(G,q)/{\cA_L}\ \cong\ \mathcal{B}\!\bigl(L^\perp/L,\ \bar q\,\bigr),
  \qquad \bar q\bigl([g]\bigr)=q(g) \,,
\ee
which is well defined because $q|_L\equiv 1$ and $L\subseteq L^\perp$. As usual for anyon condensation: the anyons in $L$ become trivial,  anyons in $G\setminus L^\perp$ braid non-trivially with the condensate and are
confined, whereas anyons in $L^\perp$ are deconfined.  The reduced topological order is $L^\perp/L$.

\subsection{Braided Automorphisms for $\Z_2^{(1)}$}
\label{sec:BrAutZ2}

The goal of this subsection is to extract the group of braided
automorphisms of the $A= \Z_2^{(1)}$ SymTFT from categorical data.  

\subsubsection{Twisted $\Z_2$-graded Witt group}

Let us compute the graded Witt group for $A=\Z_2 \cong \widehat{A}$. Here 
\be
 \text{Aut}^{syp}(A\oplus \widehat{A},\mathsf{Alt}(s)) = SL(2,\Z_2) \,,
\ee
and 
\be\label{H5ClassZ2}
\ba
H^5 (\Z_2\oplus \Z_2 [3], \mathbb{C}^\times)  :=& H^5 (B^3 (\Z_2 \oplus \Z_2), \mathbb{C}^\times)\cr 
=& \Z_2 \oplus \Z_2 \,.
\ea
\ee
The value $\Z_2\oplus\Z_2$ is the mod-$2$ computation of \cite{Bhardwaj:2024xcx}, which we do not rederive here, the analogous statement for odd $p$, $H^5((\Z_p\oplus\Z_p)[3],\mathbb C^\times)=0$, is derived in {Appendix}~\ref{app:H5B3}.

Combining this with the (\ref{ShortishSequence}), we get: 
\be 
\ba
&\Witt(\Z_2\oplus \Z_2,s)/\Witt\cr 
&\cong \text{Ker} \big({SL}(2,\Z_2)\ltimes (\Z_2\oplus \Z_2)\rightarrow \mathbb{Z}_2\big)
\ea
\ee
which, {using the fact that the kernel is all of $SL(2,\Z_2)\ltimes(\Z_2\oplus\Z_2)$ \cite{Bhardwaj:2024xcx},} identifies 
\be\label{S4Witt}
\ba
\boxed{\Witt (\Z_2\oplus \Z_2,s)/\Witt\cong  
SL(2,\Z_2) \ltimes (\Z_2 \oplus \Z_2) \cong 
\Perm_4 }\,.
\ea
\ee
One way to see that the kernel is all of
$SL(2,\Z_2)\ltimes(\Z_2\oplus\Z_2)$, rather than the index-two subgroup $A_4$, is
to use the graded semion representatives below.  They give order-four elements
in the quotient $\Witt(\Z_2\oplus\Z_2,s)/\Witt$, while $A_4$ has no element of
order four.  Hence the map to the final $\Z_2$ in
\eqref{ShortishSequence} is trivial \cite{Bhardwaj:2024xcx}.

This is in fact the {\bf single qubit Clifford group} modulo phases: the $\Z_2\oplus\Z_2$ is the \emph{projective} Pauli group $\cP/\langle \omega_4\bm 1\rangle$ ($\omega_4 = i$) on which $X\equiv(1,0)$ and $Z\equiv(0,1)$ commute, and the Clifford gates are $SL(2,\Z_2)$.  The anticommutation $XZ=-ZX$ is invisible in this quotient: it is remembered by the central extension. The lifts $\ga^2,\gb^2,\gc^2$ of the projective Paulis to the full graded Witt group generate a quaternion group $Q_8$ with $[\ga^2,\gb^2]=\cC$ (Proposition~\ref{prop:Q8}), the central class $\cC$ playing the role of $-\bm 1$.

 \noindent{\bf Relation to QCAs.}
The  physically relevant group is not the quotient \eqref{S4Witt} but the full Picard group  (\ref{WittPic})
itself, i.e. the central extension of $\Perm_4$ by the (pointed) Witt group: this is the group of invertible interfaces, and hence, conjecturally, the natural home of the QCA classification.  Its pointed part is computed in Theorem~\ref{thm:modular} below, and its order-$192$ factor $2O\circ_{\Z_2}\Z_8$ is \emph{exactly} the one-qubit Clifford group $\Cl_2$ (Corollary~\ref{cor:Cl2}), with central kernel the full pointed Witt group $\Witt^\pt=\Z_8\oplus\Z_2\oplus \, \bigoplus_{p\ \mathrm{odd}}W_p$.

\subsubsection{Pre-Metric Group Realization of $\Perm_4$} 
\label{sec:MetricGroup}

We now present the Witt group quotient (\ref{S4Witt}) in terms of (pre-)metric groups.  This presentation also has a natural physical interpretation as the anyon theory living at the spatial interface of the symmetry defect. 
Here will abuse notation a bit and denote  from here on by ${A \oplus \widehat{A}}$ the group of surface charges $\Z_2\oplus\widehat{\Z_2}$ of the SymTFT --- the grading group of the Witt classes --- whereas up to now $A$ denoted the 1-form symmetry group $\Z_2$ itself.

We fix the abelian group (we will remove the Pontrjagin dual from the second $\Z_2$ factor in the following)
\be
{A \oplus \widehat{A}}=\Z_2\oplus\Z_2 
\ee
and let 
\be 
a=(1,0),\quad b=(0,1),\quad a+b=(1,1)\,,
\ee
and the non-trivial syllepsis $s: {A \oplus \widehat{A}} \to\{\pm1\}$ defined on generators by
\be\label{sZ2}
  s(a,b)=-1,\qquad s(a,a)=s(b,b)=s(b,a)=+1 \,.
\ee
This is the standard representative of the electric-magnetic
syllepsis discussed above: $a$ and $b$ are chosen as a symplectic basis of
electric and magnetic surface defects.
This takes the general form on any elements in ${A \oplus \widehat{A}}$
\be\label{eq:s-closed}
s(x,y)=(-1)^{\,x_1 y_2}\,,\qquad x=(x_1,x_2),\ y=(y_1,y_2)\,.
\ee
Its associated alternating $2$-form is the standard symplectic form on $(\Z_2)^2$
\be
  \mathsf{Alt}(s)(x,y)=\frac{s(x,y)}{s(y,x)}=(-1)^{\,x_1y_2+x_2y_1}\,,
\ee
which is non-degenerate. This is the reason the relevant automorphism group is
\be 
\Aut^{\mathrm{syp}}({A \oplus \widehat{A}},\mathsf{Alt}(s)) \cong SL(2,\Z_2)\cong \Perm_3 \,.
\ee

For $A \oplus \widehat{A}= \Z_2 \oplus \Z_2$, the generators of $\Witt^{pt}(\Z_2 \oplus \widehat{\Z_2},s)/\Witt^{pt}= \Perm_4$ are the graded semions: pointed categories on a
single $\Z_2$ whose generator has spin $\pm i$. 
There are three independent gradings in $\Z_2 \oplus \widehat{\Z_2}$:
\begin{align}\label{eq:gens}
 \ga &= \bigl(\Z_2 \xrightarrow{\ a\ } \Z_2 \oplus \widehat{\Z_2};\ q(1)=i\bigr),  & \gr&=a=(1,0), \nonumber \\
 \gb &= \bigl(\,\Z_2 \xrightarrow{\ b\ } \Z_2 \oplus \widehat{\Z_2};\ q(1)=i\,\bigr), &\gr&=b=(0,1), \nonumber \\
 \gc &= \bigl(\,\Z_2 \xrightarrow{\ a+b\ } \Z_2 \oplus \widehat{\Z_2};\ q(1)=1\,\bigr),  &\gr&=a+b=(1,1).
\end{align}
The third representative has $q(1)=1$ and hence does not appear to be a semion theory at first sight.  However, the non-trivial grade $a+b$ and the
syllepsis factor supply the corresponding semionic graded self-statistics.
In fact  $\ga$ and $\gb$ alone generate the whole group $\Perm_4$:
more precisely, Appendix~\ref{app:WittZ2} shows that
$\gc=\ga\gb\ga^{-1}$, and in the quotient one has
\be 
\gc^2=\ga^2\gb^2\,.
\ee 
We can now use the triviality criteria and the quotient by $\Witt^{\pt}$ to determine the full set of relations on this {group of pre-metric groups}:
\begin{itemize}
\item \textbf{{$A \oplus \widehat{A}$}-triviality}: a metric group is {$A \oplus \widehat{A}$}-trivial if $G_0$         contains a Lagrangian
      $L$ (i.e.\ $q|_L\equiv1$ and $L^\perp=L$). These are {$A \oplus \widehat{A}$}-Witt-trivial categories.
\item \textbf{Modulo $\Witt^\pt$}: we further quotient the non-degenerate metric group that is
      concentrated in grade $0$ (a normal ungraded Witt class).
\end{itemize}

\noindent{\bf Powers of $\ga, \gb, \gc$.}
It was shown in \cite{Bhardwaj:2024xcx} that $\ga^2$ is non-trivial but $\ga^4$ is trivial. Similarly $\gb^2$ and $\gc^2$ are non-trivial with $\gb^4= \gc^4=1$, in the quotient
$\Witt^\pt({A \oplus \widehat{A}},s)/\Witt^\pt$. In the full graded
Witt group these fourth powers are instead the central three-fermion class
$\cC=\gu^4$, as used below and proved in Appendix~\ref{app:WittZ2}.
However these are not independent: 
$\gc^2$ with $s(a+b,a+b)=-1$ gives
\be\ba
  \gc^2=&\Bigl((\Z_2)^2\xrightarrow{(a+b,\,a+b)}{A \oplus \widehat{A}};\cr 
  &\quad         q(1,0)=q(0,1)=1,\ q(1,1)=-1\Bigr)\,.\ea
\ee
This is the same as 
\be\label{V4-mult}
 \ga^2\,\gb^2=\gc^2,\qquad \ga^2\,\gb^2=\gb^2\,\ga^2\,.
\ee
Thus the three squares form a Klein-four subgroup $V_4< \Perm_4$ 
\be\label{V4}
  V_4:=\langle \ga^2,\gb^2\rangle=\{1,\ \ga^2,\ \gb^2,\ \gc^2\} \,.
\ee
Next note that 
\be
\Perm_4 = V_4 \rtimes \Perm_3 \,.
\ee
So we only need to identify the $\Perm_3$ and the action of it on $V_4$. 
In \cite{Bhardwaj:2024xcx} the order 3 elements of the Witt group quotient were identified as 
\be
\ga^{\pm 1} \gb^{\pm 1} \,,\quad \gb^{\pm 1} \ga^{\pm 1} \,,
\ee
which all form the 8 order $3$-cycles. 

To show that these act on $V_4$ \footnote{See e.g.  \href{https://groupprops.subwiki.org/wiki/S3_in_S4}{group prop}.} the conjugation by $\ga$ permutes $\gb^2$ and $\gc^2$ and leaves $\ga^2$ fixed. The conjugation action by $\gb$ fixes $\gb^2$ and permutes $\ga^2$  and $\gc^2$. So 
\be
\Perm_4/V_4 = \Perm_3  \,,
\ee
with  $\ga^2 ,\gb^2 \in V_4$ the normal subgroup and $(\ga\gb)^3=1$.

\subsubsection{Full $\Witt^\pt({A \oplus \widehat{A}}, s)$}

We now compute the full graded pointed Witt group $\Witt^\pt ({A \oplus \widehat{A}},s)$.  
In the graded Witt group $\Witt^\pt({A \oplus \widehat{A}},s)$ we condense only anyons that are uncharged under ${A \oplus \widehat{A}}$, $L\subseteq G_0$, so
that the $A$-grading is preserved. The analog of the identification (\ref{Condis}) via condensation in $\Witt^{\pt}({A \oplus \widehat{A}},s)$ is 
\be\label{condenseA}
  \bigl(f: G\to {A \oplus \widehat{A}},\ q\bigr) \equiv \bigl(\bar f : L^\perp/L\to {A \oplus \widehat{A}},\ \bar q\,\bigr)\,.
\ee
More precisely two elements in $\Witt^{\pt}({A \oplus \widehat{A}},s)$ are identified if \cite{Bhardwaj:2024xcx}
\be
\ba
  &X=Y\ \text{in}\ \Witt^{ \text{pt}}({A \oplus \widehat{A}},s) \\
%  \quad\Longleftrightarrow\quad
%  \: \iff \:
  &\iff \:
  X\boxtimes^s_{A \oplus \widehat{A}} Y^{s\text{-}\mathrm{op}}\ \text{is {${A \oplus \widehat{A}}$-trivial}}\,,
\ea
\ee
where an ${A \oplus \widehat{A}}$-graded metric group is ${A \oplus \widehat{A}}$-trivial if and only if it contains a
 Lagrangian subgroup $L\subseteq G_0$, i.e. it is isotropic  $q|_L\equiv 1$ and maximal $L^\perp=L$. 
This is precisely the folded Lagrangian algebra extended to the {$A \oplus \widehat{A}$}-graded setup.
Thus  $\Witt^{ \text{pt}}({A \oplus \widehat{A}},s)$ is the group of abelian
$A \oplus \widehat{A}$-graded anyon theories up to ${A \oplus \widehat{A}}$-grading-preserving gapped interfaces.  
 
An example of a grade-0 metric group that is no longer trivial is 
\be
  \Bigl((\Z_2)^2\xrightarrow{\ 0\ } A \oplus \widehat{A};\ q(1,0)=q(0,1)=q(1,1)=-1\Bigr) \,.
\ee

\noindent{\bf The pointed Witt group.}
By \cite{DGNO} (appendix A.7) the pointed Witt group (ungraded) is explicitly known 
\be\label{Wittpt}
\Witt^\pt = \bigoplus_{p \text{ prime}} W_p \,,
\ee
where $W_p$ is the summand represented by metric groups whose underlying
abelian group is a $p$-group (Notice that $W_p$ itself need not be a
$p$-group as an abstract Witt summand).
Explicitly,
\be\label{Wittp}
   W_2\cong\Z_8\oplus\Z_2,\qquad
  W_p \cong
  \begin{cases}
    \Z_2\oplus\Z_2 \quad & p =1 \mod 4,\\
    \Z_4\quad & p =3 \mod 4 \,.
  \end{cases}
\ee
The two classes that matter for us are $p=2$ generated by the
elementary abelian 2-groups:
\be \ba
  \gu &=\bigl(\Z_2\xrightarrow{\,0\,}{A \oplus \widehat{A}};\ q(1)=i\bigr)
 \cr 
  \cC&=\bigl((\Z_2)^2\xrightarrow{\,0\,}{A \oplus \widehat{A}};\ q\equiv-1\ \text{on}\ G\setminus 0\bigr) \,,
  \ea
\ee
where $\gu$ is the semion and $\cC$ is the three-fermion category $SO(8)_1$. 

{In this section we will label the premetric groups by the same labels, but now treat them as elements in $\Witt^\pt ({A \oplus \widehat{A}},s)$, as opposed to the quotient $\Witt^\pt ({A \oplus \widehat{A}},s)/\Witt^\pt$.  Some relations that these lifted elements satisfy are (as derived in Appendix \ref{app:WittZ2}):}
\be
\gu^4 = \cC \,,\qquad \cC^2 = \gu^8= 1  \,,
\ee
and furthermore 
\be
\ga^4=\gb^4=\gc^4=\cC=\gu^4  \,.
\ee
We note that $\Witt^\pt$ is a central subgroup of $\Witt^\pt (A \oplus \widehat{A}, s)$.
That it is a subgroup is clear; to show that it is central, take a grading $f=0$, then $G_0=G$, and the syllepsis is trivial on these, meaning that the $\boxtimes^s_{A \oplus \widehat{A}}$ is commutative, given by just the product of the two quadratic forms. 
Therefore we have a short exact sequence that is a central extension by (\ref{Wittpt})
\be
1 \to \Witt^\pt \to \Witt^\pt ({A \oplus \widehat{A}}, s) \to \Perm_4 \to 1 \,.
\ee

The upshot of the Witt group calculation is the following theorem relating it to the modular generators, and is proven in Appendix \ref{app:WittZ2}:
\begin{theorem}[Witt group and modular group]
\label{thm:modular}
Define
\be\label{ST}
  T=\gu\,\ga,\qquad S=\gu^{2}\,\ga\gb\ga,\qquad Y=\gu \,.
\ee
Then 
\be 
T^4=S^2=1\,,\quad (ST)^3=Y\,,\quad Y^8=1
\label{eq:STcat}
\ee
with $Y$ central and  $\ord(Y)=8$, and $\langle S,T\rangle$ is all of the order-$192$ factor
$2O\circ_{\Z_2}\Z_8$ of 
\be\label{fullsplit}
\Witt^\pt(\Z_2\oplus \widehat{\Z_2},s) = \big(2O\circ_{\Z_2}\Z_8\big)\,\times\,\Z_2\,\times\!\!\bigoplus_{p\ \mathrm{odd}} W_p \,.
\ee
\end{theorem}
Here $2O$ is the binary octahedral group.  The notation
$\circ_{\Z_2}$ denotes the central product: the central involution
$z\in 2O$ is identified with the order-two class $\cC=\gu^4$ inside the
$\Z_8$ summand of $\Witt^\pt$.

\subsubsection{Relation to the Qubit Clifford Group}

\noindent{\bf Single Qubit Clifford.}
Denote by $\Cl_2$ the Clifford group on a single qubit. It can be presented in terms of Hadamard and phase gates as {(here $[,]$ denotes the group commutator)}:
\be
\ba
\Cl_2 = \big\langle H, P \, \big| \, & H^2=1 \,,\ P^4 = 1 \,, \  (HP)^3= Y\,,\\
& Y^8= 1 \,,\ [Y,H]=[Y,P]=1\big\rangle \,.
\ea
\label{eq:Cl2pres}
\ee
In particular $Y = \omega_8 \bm{1}$, where $\omega_k := e^{2 \pi i/k}$.

Comparing \eqref{eq:Cl2pres} with Theorem~\ref{thm:modular} identifies the graded Witt group with the Clifford group:
\begin{lemma}\label{cor:Cl2}
The map $H\mapsto S$, $P\mapsto T$, $\omega_8\bm 1\mapsto Y$ extends to an isomorphism
\be
\Cl_2 \xrightarrow{\ \sim\ } 2O\circ_{\Z_2}\Z_8 \subset \Witt^\pt(\Z_2\oplus\Z_2,s)\,,
\label{eq:Cl2iso}
\ee
and hence
\be
\Witt^\pt(\Z_2\oplus\Z_2,s) \cong \Cl_2 \times \Z_2\times\!\!\bigoplus_{p\ \mathrm{odd}} W_p \,.
\label{eq:WittisCl2}
\ee
\end{lemma}
% \begin{proof}
% Let $G$ be the group presented by \eqref{eq:Cl2pres}. Modding out $\langle Y\rangle$ gives $\langle H,P\,|\,H^2=P^4=(HP)^3=1\rangle\cong \Perm_4$ of order $24$, and the central subgroup $\langle Y\rangle$ has order at most $8$; hence $|G|\le 192$.  Both $\Cl_2$ (with the concrete gates, where the scalar subgroup is exactly $\langle\omega_8\rangle$ since all determinants lie in $\langle i\rangle$) and $2O\circ_{\Z_2}\Z_8$ (by Theorem~\ref{thm:modular}) are generated by elements satisfying these relations and have order exactly $192$; so both are isomorphic to $G$, and \eqref{eq:Cl2iso} follows. Combining with \eqref{fullsplit} gives \eqref{eq:WittisCl2}.
% \end{proof}
The full pointed graded Witt group --- the group of invertible interfaces of the $\Z_2^{(1)}$ SymTFT --- is the one-qubit Clifford group, up to the decoupled grade-zero factors $\Z_2\oplus\bigoplus_{p\,\mathrm{odd}}W_p$.

Note that this contains the Pauli group
\be
P^2 =Z \,, \
H P^2 H^{-1} = X \,.
\ee
To map to the $SL(2,\Z_2)$ description, note 
\be
|\Cl_2| = 192\,,\qquad |SL(2,\Z_2) \ltimes (\Z_2\oplus \Z_2)| = 24  \,. 
\ee
Quotienting the Clifford group by the $\Z_8$ phase $Y$ results in  
\be\boxed{
\faktor{\Cl_2}{\langle Y\rangle}  \cong SL(2,\Z_2) 
\ltimes \Z_2^2  \cong \faktor{SL(2, \Z_4)}{\pm1}  \cong \Perm_4 \,.
}\ee 
Note that part of this $\Z_8$ quotient can be identified on the Pauli group 
\be
\cP = \langle  X, Z, \omega_4 \bm{1} \rangle \,,
\ee
by taking the quotient of the phases 
\be
\cP / \langle \omega_4 \bm{1} \rangle = \Z_2 \oplus \Z_2 \,.
\ee
So the automorphisms of the SymTFT do not see the full Pauli, but a ``projective version" of it. 

\smallskip

\noindent
{\bf Relation to $SL(2, \Z_4)$.}
First note 
\be
\Cl_2/\Z_8 \cong \Perm_4 \cong SL(2, \Z_4)/{\pm} 1  \cong SL(2, \Z_2) \ltimes (\Z_2 \oplus \Z_2)   \,.
\ee
Note that the gates are mapped as follows 
\be
\begin{array}{c|c|c}
\text{Gate in }\Cl_2 & SL(2, \Z_2) \ltimes \Z_2^2  & SL(2,\Z_4)/ {\pm} 1\cr \hline 
X= \begin{pmatrix}
0 & 1\\
1 & 0
\end{pmatrix} &  (\bm{1},  (1, 0))  &  \begin{pmatrix}
1 & 2\\
0 & 1
\end{pmatrix}\cr 
Z= \begin{pmatrix}
1 & 0 \\
0 & -1
\end{pmatrix} &  (\bm{1},  (0, 1))  &  
\begin{pmatrix}
1 & 0\\
2 & 1
\end{pmatrix}\cr 
H ={1\over \sqrt{2}} \begin{pmatrix}
1& 1 \\
1 & -1
\end{pmatrix}& 
\left(\begin{pmatrix}
0 & 1 \\
1 & 0 
\end{pmatrix} ,\,  (0,0)\right)
&
\begin{pmatrix}
0 & -1 \\
1 & 0
\end{pmatrix}\cr 
P= \begin{pmatrix}
1 & 0\\
0 & i
\end{pmatrix} & \left( \begin{pmatrix}
1 & 0 \\
1 & 1
\end{pmatrix} ,\,  (1,1)\right)
&
\begin{pmatrix}
 1 & 0 \\
 1 & 1
\end{pmatrix}
% \left(\begin{array}{ll}
% 1 & 0 \\
% 1 & 1
% \end{array}\right)\cr 
\end{array}
\ee
The operators in $SL(2,\Z_4)/\pm1$ satisfy the relations (mod 4 and mod $\pm 1$ i.e. modding out the center)
\be\label{relationsinSL4}
X Z =  Z X \,,\qquad P^2 = Z \,,  \ H P^2 H^{-1} = X \,. 
\ee
Similarly in $SL(2,\Z_2) \ltimes (\Z_2 \oplus \Z_2)$ we denote the elements by pairs $(M, v)$ then 
\be
(M, v) \cdot (M' , v') = (M M', v+ M v')
\ee
and then they satisfy the same relations as in (\ref{relationsinSL4}).
Crucially the phase gate has a non-trivial component in the $\Z_2^2$.

\subsection{Braided Automorphisms for $\Z_p^{(1)}$}
\label{sec:BrAutZp}

Let us now consider the $\Z_p$ case for $p$ an odd prime. The previous analysis from \cite{Bhardwaj:2024xcx} can be generalized to 
  $\Z_p^{(1)}$-form symmetry.
 The main difference arises from the fact that for $p>2$ prime the   duality group is $\Z_4$ due to the map 
\be
B \to -C \,, \quad C \to B\,,
\ee
in the SymTFT, which is order 4. This map follows from the anti-symmetry of the pairing in the 4+1d SymTFT. 

Here we will focus on determining the braided automorphisms of the Drinfeld center for the $A=\Z_p$ 1-form symmetry. 
Repeating the analysis for $A=\Z_p$ we find 
the sylleptic automorphisms are
\be
\Aut^{syp}(\Z_p \oplus \widehat{\Z_p}, \mathsf{Alt}(s))=SL(2,\Z_p)  \,,
\ee
and the map to
$\Z_2$ is zero (the modular group $\SL(2,\Z_p)$ has no index-$2$ subgroup), 
we obtain
\be
  \boxed{ \Witt^{\pt}(\Z_p\oplus\Z_p,\,s)/\Witt^{\pt} = SL(2,\Z_p)} \,.
\label{eq:ZpWittQuotient}
\ee
Note that in particular (see Appendix \ref{app:H5B3})\footnote{SSN thanks Matt Yu for discussions on this.}
\be
H^5 ((\Z_p\oplus \Z_p )[3], \mathbb{C}^\times ) = 0\,.
\label{eq:H5ZpZp}
\ee
This has some relation to the Clifford group on $
\Z_p$-qudits. However note that crucially the contribution from 
$H^5 (\Z_p \oplus \Z_p [3], \mathbb{C}^\times)$ vanishes. See Appendix \ref{app:H5B3}.  

\subsubsection{Pointed graded Witt group}

Keeping the central, grade-zero classes, i.e. $\Witt$, gives the full pointed graded Witt group.
The structure is simpler for odd primes than for $p=2$: the central extension
by the ordinary pointed Witt group splits:

\begin{theorem}[Odd-prime pointed graded Witt group]
\label{thm:WittZpMain}
Let $p$ be an odd prime and let
${A \oplus \widehat{A}}=\Z_p\oplus\Z_p$ with $s(x,y)=\zeta_p^{x_1y_2}$.  Then
\be
  \Witt^\pt({A \oplus \widehat{A}},s)/\Witt^\pt \cong SL(2,\Z_p),
\ee
and the full pointed graded Witt group is, as an abstract group,
\be\label{eq:fullWittZp}
\begin{aligned}
  \Witt^\pt({A \oplus \widehat{A}},s)
  &\cong SL(2,\Z_p)\times \Witt^\pt\\
  &=SL(2,\Z_p)\times
    \bigoplus_{\ell\ {\rm prime}} W_\ell\,.
\end{aligned}
\ee
Equivalently, the only central factors in the odd-prime calculation are the
decoupled grade-zero pointed Witt classes.
\end{theorem}

The splitting in (\ref{eq:fullWittZp}) is not canonical:  identifications of specific
modular generators can still obey relations that are shifted by central elements.
% OLD: The point is that for odd $p$ this
% central factor is a coboundary: after a central redefinition one obtains an $SL(2,\Z_p)$ subgroup.
The splitting itself follows from
$H^2(SL(2,\Z_p);\Witt^\pt)=0$, there is no extension class  available for
odd $p$ ({Appendix}~\ref{app:WittZp}).  Moreover the splitting is realized as follows:
$\hat\ga=\kappa^{-1}\ga$, $\hat\gb=\kappa^{-1}\gb$  of exact order
$p$ generate a subgroup isomorphic to $SL(2,\Z_p)$ {(where $\kappa$ is defined in (\ref{eq:kappa_def}))}, so all modular
relations hold for them without central extensions. 
Setting
\be\label{eq:STodd}
  T:=\hat\ga\,,\qquad S:=\hat\ga^{-1}\,\hat\gb^{-4}\,\hat\ga^{-1}\,,
\ee
one has the following counterpart of Theorem~\ref{thm:WittZpMain} for the
modular generators.
\begin{theorem}[Modular relations]\label{thm:SToddMain}
The normalized lifts \eqref{eq:STodd} satisfy, with no central corrections,
\be\label{eq:SToddrels}
  T^p=1\,,\qquad S^4=1\,,\qquad (ST)^3=S^2\,,
\ee
where $S^2=C$ is central and acts on the surface charges by charge conjugation
$x\mapsto-x$, i.e.\ $\pi(C)=-\mathbf 1\in SL(2,\Z_p)$.  In particular
$\langle\hat\ga,\hat\gb\rangle\cong SL(2,\Z_p)$, every modular relation
holding on the nose.
\end{theorem}
\noindent This is proven in Appendix~\ref{app:WittZp}: Proposition~\ref{prop:intrinsicmap}
computes the images $\pi(\hat\ga),\pi(\hat\gb)$ --- the transvections
\eqref{eq:gagbImages} --- whereupon \eqref{eq:SToddrels} reduce to the
$2\times2$ matrix identities of Theorem~\ref{thm:oddprels}.

\subsubsection{Pre-Metric Group Realization}
\label{sec:ZpPremetric}

As in the $\Z_2$ case we can specify the relevant elements of the graded
pointed Witt group in terms of pre-metric groups.
Let $\zeta_p=e^{2\pi i/p}$, and set 
\be
{A \oplus \widehat{A}}=\Z_p\oplus\Z_p,\qquad a=(1,0),\quad b=(0,1)\,,
\ee
and define the odd-prime syllepsis generalizing \eqref{sZ2} by
\be
s(x,y)=\zeta_p^{x_1 y_2}\,,\qquad x=(x_1,x_2)\,,\ y=(y_1,y_2) \,,
\ee
with alternating form
\be
 \mathsf{Alt}(s)(x,y)=\frac{s(x,y)}{s(y,x)}=\zeta_p^{x_1y_2-x_2y_1}
\ee
The two generators will be again defined in the same way as in the case of  $\Z_2\oplus \widehat{\Z_2}$
\be
\ba
  \ga&=\big(\Z_p\xrightarrow{a}{A \oplus \widehat{A}};\,q(1)=\zeta_p\big) \cr 
  \gb&=\big(\Z_p\xrightarrow{b}{A \oplus \widehat{A}};\,q(1)=\zeta_p\big) \,.
\ea
\ee
 In the quotient
$\Witt^\pt({A \oplus \widehat{A}},s)/\Witt^\pt$ these elements have order $p$.  In the full group
their $p$th powers need not vanish, though they land in Witt.  Let
\be \label{eq:kappa_def}
\kappa=\big(\Z_p\xrightarrow{0}{A \oplus \widehat{A}};\ q(1)=\zeta_p\big)\in W_p\subset\Witt^\pt,
\ee
corresponding to a chiral anyon theory with $\Z_p$ fusion rules.  Then
\be\label{eq:gapowercentral}
  \ga^p=\gb^p=\kappa^p
  =
  \begin{cases}
    \kappa,& p=1\!\!\pmod 4,\\
    \kappa^{-1},& p=3\!\!\pmod 4 \,,
  \end{cases}
\ee
where $\kappa$ has order $2$ for $p=1\bmod4$ and order $4$ for
$p=3\bmod4$.  Thus the normalized lifts
\be
  \hat\ga=\kappa^{-1}\ga,\qquad \hat\gb=\kappa^{-1}\gb
\ee
have exact order $p$ and generate a modular group $SL(2,\Z_p)$ factor in
(\ref{eq:fullWittZp}).  The $p$th power is computed as follows:
\be
\ba
\ga^p=
% \Big((\Z_p)^p\xrightarrow{(x_i)\mapsto(\sum_i x_i)a}{A \oplus \widehat{A}};\
% Q(x_1,\cdots,x_p)=\sum_i x_i^2\Big).
\Big((\Z_p)^p\xrightarrow{(x_i)\mapsto(\sum_i x_i)a}{A \oplus \widehat{A}};\
Q({\vec{x}})=\sum_i x_i^2\Big).
\ea
\ee
Here
$G_0=\{(x_1,\cdots,x_p):\sum_i x_i=0\}$ and
$G_0^\perp=\langle(1, \cdots ,1)\rangle\subset G_0$.  This line is isotropic,
since $Q(1, \cdots,1)=p=0$ in $\Z_p$.  Condensing it leaves the same
ordinary Witt class as the $p$-fold grade-zero sum
$\kappa^p=((\Z_p)^p\xrightarrow{0}{A \oplus \widehat{A}};\sum_i x_i^2)$.  Hence
$\ga^p=\kappa^p$, and the same argument gives $\gb^p=\kappa^p$.

\subsubsection{Relation to the Qudit Clifford Group}
\label{sec:qudit}

Let $p$ be an odd prime and write $\omega_p=e^{2\pi i/p}$.  The single-qudit
Pauli group is a group of complex $p \times p$ matrices
\be
\cP_p=\langle X, Z \; \mid \; X^p = Z^p = \bm{1},\; ZX = \omega_p XZ \rangle ,
\ee
where $X$ and $Z$ can be given explicitly as
\be
  X|j\rangle=|j+1\rangle,\qquad
  Z|j\rangle=\omega_p^j|j\rangle .
\ee
The Clifford group is the normalizer of $\cP_p$.  A convenient set of
generators is given by the Fourier and phase gates \cite{Farinholt:2014wul}
\be
\ba
F|j\rangle&={1\over\sqrt p}\sum_{k=0}^{p-1}\omega_p^{jk}|k\rangle,\\
P|j\rangle&=\omega_p^{j(j-1)/2}|j\rangle .
\ea
\ee
They obey
\begin{align}\label{eq:quditCliffordActions}
 F^4 &= \bm{1}, & P^p &= \bm{1}, & (FP)^3 = e^{\frac{\pi i (p-1)}{4p}} \bm{1}, \nonumber \\
 F X F^{-1}&=Z, & F Z F^{-1} &= X^{-1}, \nonumber \\
P X P^{-1}&=X Z, & P Z P^{-1} &= Z .
\end{align}
The scalar in $(FP)^3$ is a property of these particular unitary lifts.  It has
order $2p$ for $p=1\bmod4$ and order $4p$ for $p=3\bmod4$, but it disappears
after projectivizing the Clifford group.

It is useful to separate the projective Clifford group from its Pauli kernel.
Let $\overline{\Cl}_p=\Cl_p/U(1)$, or equivalently quotient by all scalar
phases.  Conjugation on Pauli labels gives an exact sequence
\be\label{eq:projectiveClifford}
1\longrightarrow \cP_p/\langle\omega_p\bm1\rangle
\cong \Z_p^2
\longrightarrow \overline{\Cl}_p
\longrightarrow SL(2,\Z_p)
\longrightarrow 1,
\ee
and hence $|\overline{\Cl}_p|=p^3(p^2-1)$.  The two generators above map to
\be\label{eq:FPmatrices}
  F\longmapsto
  \begin{pmatrix}0&-1\\ 1&0\end{pmatrix},
  \qquad
  P\longmapsto
  \begin{pmatrix}1&0\\ 1&1\end{pmatrix},
\ee
where a Pauli monomial $X^aZ^b$ is labelled by the vector $(a,b)$.
For example, $P^m$ maps to 
$\left(\begin{smallmatrix}1&0\\ m&1\end{smallmatrix}\right)$, while
$F^{-1}P^{-m}F$ is 
$\left(\begin{smallmatrix}1&m\\0&1\end{smallmatrix}\right)$.

The Paulis themselves are already generated by $F$ and $P$:
\be\label{eq:xz_from_fp}
  Z=[F^2,P]=F^2PF^{-2}P^{-1},\qquad
  X=F^{-1}ZF .
\ee
Indeed $F^2|j\rangle=|{-j}\rangle$, so the first commutator acts as
$\omega_p^j$ on $|j\rangle$, and the second identity is just the Fourier
conjugate of $Z$.

\medskip\noindent{\bf Relation to braided automorphisms.}
The braided automorphism computation does not see the Pauli kernel
$\Z_p^2$ in \eqref{eq:projectiveClifford}.  It sees only the quotient
\be
\ba
  \Witt^\pt({A \oplus \widehat{A}},s)/\Witt^\pt
  &\cong
  \Aut^{\rm syp}({A \oplus \widehat{A}},\mathsf{Alt}(s)) \\
  &\cong SL(2,\Z_p),
\ea
\ee
where {$A \oplus \widehat{A} = \Z_p \oplus \Z_p$} is the group of electric/magnetic surface charges.
Thus the correct comparison with the qudit Clifford group is via the projection
\be
  \overline{\Cl}_p\longrightarrow SL(2,\Z_p),
\ee
rather than via the full semidirect product
$\Z_p^2\rtimes SL(2,\Z_p)$.

Under this projection, the duality generator is represented by the Quantum
Fourier Transform $F$,  and
$P^m$ and its Fourier conjugates.
The images of the graded chiral $\Z_p$ anyon theories under
$\Witt^\pt({A \oplus \widehat{A}},s)\to SL(2,\Z_p)$ are
\be\label{eq:gagbImages}
  \hat\ga\mapsto\big(\begin{smallmatrix}1&2^{-1}\\ 0&1\end{smallmatrix}\big)\,,
  \qquad
  \hat\gb\mapsto\big(\begin{smallmatrix}1&0\\ -2^{-1}&1\end{smallmatrix}\big)\,,
\ee
transvections along the two axes, derived from the graded Witt data in
Appendix~\ref{app:WittZp} (Proposition~\ref{prop:intrinsicmap}).
Comparing with \eqref{eq:FPmatrices} (recall $F^{-1}P^{-m}F\mapsto
\big(\begin{smallmatrix}1&m\\ 0&1\end{smallmatrix}\big)$ and
$P^{m}\mapsto\big(\begin{smallmatrix}1&0\\ m&1\end{smallmatrix}\big)$), we can therefore
identify the Witt group elements as
\be\label{eq:gagbClifford}
  \hat\ga\ \leftrightarrow\ F^{-1}P^{\frac{p-1}{2}}F,\qquad
  \hat\gb\ \leftrightarrow\ P^{\frac{p-1}{2}} \,,
\ee
inside the projective Clifford group --- the \emph{same} power $(p-1)/2\equiv-2^{-1}$
of the phase gate for both, one of them Fourier-conjugated.  The half-integer power is
natural: the phase gate carries the half-integer quadratic phase
$\omega_p^{j(j-1)/2}$, while the graded chiral $\Z_p$ anyon theory carries the integer one
$q(j)=\zeta_p^{j^2}$.  
More generally, the graded chiral $\Z_p$ anyon theory of spin
$q(j)=\zeta_p^{t j^2}$ corresponds to $P^{-(2t)^{-1}}$ (Fourier-conjugated for the
$a$-graded one), so other normalizations of the graded chiral $\Z_p$ anyon theories replace $P^{(p-1)/2}$ by
the appropriate power.

In summary, the odd-prime braided automorphisms are embedded in
the qudit Clifford group as the Clifford-mod-Pauli part.
The Pauli translations $X^aZ^b$ form the kernel of the Clifford action on Pauli
labels, and there is no corresponding extra $\Z_p^2$ factor in
$\Witt^\pt(A,s)/\Witt^\pt$.  The central Witt class $Y\in W_p$ discussed above
is also not a Pauli: it is a grade-zero pointed Witt factor, and for odd $p$ it
can be removed from the $SL(2,\Z_p)$ relations by the normalization
$\ga,\gb\mapsto\hat\ga,\hat\gb$.

\section{QCAs on $\mathbb{Z}_2$ 1-form Symmetric Algebra}
\label{sec: QCAs on Z2 1-form symmetric algebra}

In this section, we discuss QCAs on the algebra of $\mathbb{Z}_2$ 1-form symmetric local operators on a 3d cubic lattice. We first introduce the setup in Section~\ref{sec:setup_qca_on_Z2_subalg}, together with a brief reminder of the conventions we will use in the rest of this section. We also comment informally on what the 1-form symmetry means in the context of our lattice system.

In Section~\ref{sec:qca_on_quotient_algebra} we formally define the mathematical framework for talking about $\Z_2^{(1)}$-symmetric QCA.
We then review several examples of such QCAs, namely, the Kramers-Wannier-Wegner operator \cite{Gorantla:2024ocs, Koide:2021zxj} in Section~\ref{sec: Z2 KWW}, the Tsui-Wen SPT entangler {of order 2} \cite{Tsui:2019ykk} in Section~\ref{sec: Z2 TW of order 2},\footnote{{The root SPT entangler of order 4 will be discussed in Section~\ref{sec:qca_for_root_entangler}.}} and the {framing} QCA \cite{Fidkowski:2023dpe} in Section~\ref{sec: 3-fermion QCA 1-qubit version}.
Of those, we only construct KWW in full painstaking detail, since the compatibility with the framework we introduce is a little non-trivial. For the others, the omitted details are essentially trivial and only require one quick computation.

We also construct another example, which we call the 3-fermion Kramers-Wannier-Wegner QCA, in Section~\ref{sec: 3-fermion KWW}.
%Finally, we study the fusion rules of these QCAs in Section~\ref{sec: Symmetry generated by D and U} and discuss their relation to the fusion rules in the continuum in Section~\ref{sec: comparison Z2}.
We then study a particular fusion rule of QCAs in Section~\ref{sec: Symmetry generated by D and U} and discuss the full group generated by the above QCAs in Section~\ref{sec: The group of QCAs generated by D and U}.
Finally, we will compare the fusion rules on the lattice and those in the continuum in Section~\ref{sec: comparison Z2}.
As we will see, the fusion rules on the lattice differ from those in the continuum by a non-trivial QCA and lattice translations.\footnote{To be precise, the fusion rules on the lattice are still consistent with the continuum relation $(ST^2)^4 = Y^4$, as we emphasized in Section~\ref{sec:Intro}. The difference between the lattice fusion rules and the continuum ones comes from the following fact: on the lattice, $Y^4$ gives rise to a non-trivial QCA in the fusion rules, whereas in the continuum, it gives rise to a TQFT coefficient in the fusion rules. We refer the reader to Section~\ref{sec: comparison Z2} for more details on this point.} 
This is the first example of non-trivial QCA-refined fusion rules in 3+1 dimensions.

\subsection{Setup}\label{sec:setup_qca_on_Z2_subalg}

We consider a 3d cubic lattice with qubits on the faces.
The state space on the lattice is given by
\begin{equation}
\mathcal{H}_F \coloneq \bigotimes_{f \in F} \mathbb{C}^2.
\label{eq: qubit Hilb}
\end{equation}
Note that this is only well defined if the lattice is finite. It is also possible to talk about infinite lattices, which is usually done in the language of local operator algebras. For simplicity of the discussion, we will freely switch between an infinite lattice, which is to be interpreted as the thermodynamic limit of the system, and a finite periodic one, where many of the definitions are simpler. This is mostly relevant in the discussion of local symmetric operators: we will ignore line operators stretching all the way across a non-contractible cycle on the periodic lattice. These kinds of operators become non-local in the thermodynamic limit anyway, and they never appear on the infinite lattice in the first place.

The Pauli $X$ and $Z$ operators on face $f$ are denoted by $X_f$ and $Z_f$, respectively.
The actions of these operators in the computational basis are
\begin{equation}
X_f \ket{a}_f = \ket{a + 1 \bmod 2}_f, \quad
Z_f \ket{a}_f = (-1)^{a} \ket{a}_f,
\end{equation}
where $a \in \{0, 1\}$.
The $\mathbb{Z}_2$ 1-form symmetry operator on the lattice is given by
\begin{equation}
\eta(\Sigma) \coloneq \prod_{f \in \Sigma} X_f,
\label{eq: Z2 1-form sym}
\end{equation}
where $\Sigma$ is a closed surface on the direct lattice.
We note that $\eta(\Sigma)$ is not topological on the tensor product Hilbert space~\eqref{eq: qubit Hilb}, nor as an operator on the full algebra of local operators (to be introduced in more detail later).
Local operators symmetric under this non-topological $\mathbb{Z}_2$ 1-form symmetry are generated by
\begin{equation}
X_f, \qquad Z_{\delta \bm{e}} \coloneq \prod_{f^{\prime} \in F} Z_{f^{\prime}}^{\delta \bm{e}(f^{\prime})},
\label{eq: Z2 symmetric op}
\end{equation}
for all faces $f$ and all edges $e$.
Here, $\delta \bm{e}$ is the coboundary of the 1-cochain $\bm{e}$.
We note that $Z_{\delta \bm{e}}$ is the product of four Pauli $Z$ operators around $e$.
See Figure~\ref{fig: Z2 local ops} for an illustration of these generators.
\begin{figure}[t]
\centering
\includegraphics{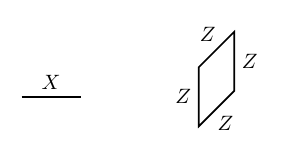}
\caption{The generators of the algebra of $\mathbb{Z}_2$ 1-form symmetric local operators. Here, the solid lines represent edges of the dual lattice.}
\label{fig: Z2 local ops}
\end{figure}

Let us briefly recall the conventions and notation used in this paper (for more details, see Section~\ref{sec: Notations and conventions}). We let $e, f$ denote an edge (a 1-cell), and a face (a 2-cell), respectively. By abuse of notation, they can trivially be identified with the 1 and 2-chains consisting of just that one edge, or that one face. We will denote by $\bm{e}, \bm{f}$ the corresponding cochains, which have value 1 on $e, f$, and value 0 elsewhere. These identifications also extend to chains: if $c_k$ is a k-chain, then it can be expressed as a finite sum of k-cells. Turn each of those k-cells into the corresponding k-cochain, and the resulting sum is $\bm{c_k}$, a k-cochain. Hence, in the same expression we may write $\partial a$ as well as $\delta \bm{a}$, for a k-chain $a$, with the understanding that certain operations only apply to cochains, but also that chains can be converted into (finitely supported) cochains at any time.

The algebra generated by $X_f$ and $Z_{\delta \bm{e}}$ will be denoted by
\begin{equation}
\mathcal{A}_{\mathbb{Z}_2^{(1)}} \coloneq \langle X_f, Z_{\delta \bm{e}} \mid f \in F, e \in E \rangle.
\end{equation}
On the other hand, the algebra of all local operators on $\mathcal{H}_F$ will be denoted by $\mathcal{A}$.

For later use, we also introduce the following operator for each cube $c$:
\begin{equation}
%X_{\partial c} \coloneq \prod_{f^{\prime}} X_{f^{\prime}}^{\bm{f}^{\prime}(\partial c)}.
X_{\partial c} \coloneq \prod_{f \in \partial c} X_f.
\end{equation}
We note that $X_{\partial c}$ is the symmetry operator on the boundary of a cube $c$: $X_{\partial c} = \eta(\partial c)$.
See Figure~\ref{fig: X partial c} for an illustration of this operator.
\begin{figure}
\centering
$X_{\partial c} =$
\adjincludegraphics[valign=c, trim={10, 0, 10, 0}, scale=1]{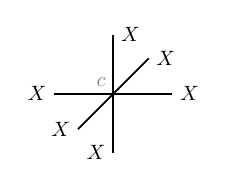}
\caption{The symmetry operator $X_{\partial c}$ around cube $c$. Here, the middle vertex represents the dual of $c$.}
\label{fig: X partial c}
\end{figure}
In the subsequent discussions, we will impose the condition
\begin{equation}
X_{\partial c} = 1
\label{eq: Gauss}
\end{equation}
for all cubes $c$, unless otherwise stated.
Using the language of states (on a finite lattice), this means restricting to the subspace of states for which that condition is true, which is no longer a tensor product of on-site Hilbert spaces. In the language of operator algebras, that is meant as an operator equation. We will see what that means, precisely, later on.
We note that the symmetry operator~\eqref{eq: Z2 1-form sym} becomes topological on the subspace satisfying the above condition. Indeed, in the thermodynamic limit any $\eta$ acts like the identity.

In what follows, we will consider QCAs on the algebra of $\mathbb{Z}_2$ 1-form symmetric operators.
More precisely, we will consider QCAs on the quotient algebra
\be 
\faktor{\mathcal{A}_{\mathbb{Z}_2^{(1)}}}{\mathfrak{I}} \,,
\ee 
where $\mathfrak{I}$ is the ideal generated by $1 - X_{\partial c}$ for all cubes $c$. 
% See Section~\ref{sec:quotient_algebra} for more details.
The quotient algebra $\mathcal{A}_{\mathbb{Z}_2^{(1)}} / \mathfrak{I}$ can be regarded as the algebra of all local operators on the constrained Hilbert space with the condition~\eqref{eq: Gauss} imposed.

\subsection{QCA on the quotient $*$-algebra}\label{sec:qca_on_quotient_algebra}

Let us be more precise about what it means to impose $X_{\partial c} = 1$. The full algebra of local operators is
\be
\mathcal{A} \coloneqq \bigotimes_{f \in F} \mathcal{M}_2(\mathbb{C}) = \bigotimes_{f \in F} \End(\mathbb{C}^2) \,.
\ee
Unlike for Hilbert spaces, this definition can be properly understood also directly in the thermodynamic limit. If the lattice is infinite, $\mc{A}$ is defined as the infinite union of increasingly large, but always finite, algebras of operators. 
This is a $*$-algebra in a natural way, though in contrast to part of the literature, we do not take its metric completion to make it into a $C^*$-algebra, and hence we genuinely talk about local operators, not quasi-local ones.

It can be seen that this algebra is generated by individual Pauli X and Pauli Z operators on each face. For brevity we may write
\be
\mathcal{A} = \langle X_f, Z_f \rangle \,.
\ee

In the presence of a $\Z_2^{(1)}$ symmetry, generated by products of $X$ operators along closed surfaces {as in \eqref{eq: Z2 1-form sym}}, we may talk about the subalgebra of operators which are symmetric (invariant) under this symmetry. This subalgebra is precisely the set of operators in $\mathcal{A}$ which commute with $X_{\partial c}$, for all cubes $c$. It turns out that this is equal to
\be
\mathcal{A}_{\Z_2^{(1)}} = \left\langle X_f, Z_{\delta \bm{e}} \coloneqq \prod_{f' \; \mid \; e \in \partial f'} Z_{f'} \right\rangle \,.
\ee

Lastly, the correct way to consistently impose $X_{\partial c} = 1$ is to first consider the (two-sided) ideal $\mathfrak{I}$ generated by $1 - X_{\partial c}$, for all cubes $c$, and then work in $\mc{A}_{\Z_2^{(1)}} / \mathfrak{I}$.

Actually, saying two-sided in this case is unnecessary, because every element in $\mathcal{A}_{\Z_2^{(1)}}$ commutes with $1 - X_{\partial c}$. Therefore, $\mathfrak{I}$ consists of the elements of the form
\be
I = \sum_{i = 1}^n (1 - X_{\partial c_i}) a_i \,,
\ee
for some cubes $c_i$, and some $a_i \in \mathcal{A}_{\Z_2^{(1)}}$.

The goal of the rest of this subsection is to formalize the concept of QCA on the $*$-algebra $\mathcal{A}_{\Z_2^{(1)}} / \mathfrak{I}$. This does not follow from the standard literature, because $\mathcal{A}_{\Z_2^{(1)}} / \mathfrak{I}$ is neither a tensor product algebra, nor a subalgebra of a tensor product algebra.

A more detailed and in-depth treatment of this quotient algebra will be given in an upcoming paper \cite{WIS}.

\subsubsection{Structure of $\mc{A}_{\Z_2^{(1)}} / \mathfrak{I}$}

In this section, we are going to present a convenient way of writing elements in $\mathcal{A}_{\Z_2^{(1)}} / \mathfrak{I}$, which will also make clear its structure. The problem is that in the quotient many operators get identified, but it is non-trivial, at least initially, to determine which ones end up in the same equivalence class, and which ones do not. For that reason, we will give a natural basis for $\mathcal{A}_{\Z_2^{(1)}} / \mathfrak{I}$, which will also make its relationship with $\mathcal{A}_{\Z_2^{(1)}}$ clear.

It will be useful to introduce a new $*$-algebra, which we denote by $\mathbb{C}^\text{lk}[C_2 \oplus C_1]$. Here, the notation is meant to evoke a group algebra, twisted by the linking number (``lk''). This will be explained in more detail down below. This algebra, in a way, sits even ``before'' $\mc{A}_{\Z_2^{(1)}}$. Schematically
\be
\mathbb{C}^\text{lk}[C_2 \oplus C_1] \twoheadrightarrow \mc{A}_{\Z_2^{(1)}} \twoheadrightarrow \faktor{\mathcal{A}_{\Z_2^{(1)}}}{\mathfrak{I}} \,.
\ee
This is \emph{not} an exact sequence, it is just a way of writing two maps. The one on the right is the quotient map, which is obviously surjective, and the one on the left will turn out to be a quotient map as well.

As a vector space, this new algebra is the complex vector space $\mathbb{C}[C_2 \oplus C_1]$ of finite $\mathbb{C}$-linear combinations of elements of the abelian group $C_2 \oplus C_1$. Here, $C_k$ is the group of $k$-chains of our cubical lattice, with coefficients in $\Z_2$:
\be
C_k \coloneqq C_k(\text{lattice}, \Z_2) \,.
\ee
The multiplication in this algebra is defined to be:
\be\label{eq:twisted_product}
(c_2, c_1) \cdot (c_2', c_1') \coloneqq (-1)^{\delta \bm{c_1} (c_2')} (c_2 + c_2', c_1 + c_1') \,,
\ee
and extended by bilinearity.

Let us remark that $c_1$ is a 1-chain, which gets promoted to a 1-cochain $\bm{c_1}$. Then we can take its coboundary, a 2-cochain, and evaluate it on the 2-chain $c_2'$. We add the boldface for clarity, but it is nevertheless clear that $\delta$ is something which only makes sense for cochains, and the evaluation $\cdot (\cdot)$ also only applies to a cochain-chain pair.

We also define an involution:
\be\label{eq:twisted_involution}
(c_2, c_1)^{*} \coloneqq (-1)^{\delta \bm{c_1}(c_2)} (c_2, c_1) \,,
\ee
and extended by conjugate-linearity to the whole algebra. One can check that those operations are compatible in the correct way to make $\mathbb{C}[C_2 \oplus C_1]$ into a $*$-algebra.

As mentioned earler, we denote this algebra by $\mathbb{C}^{\text{lk}}[C_2 \oplus C_1]$, by analogy with twisted group algebras. It is clear that
\be
(-1)^{\delta \bm{c_1}(c_2)} = (-1)^{\bm{c_1}(\partial c_2)} \,,
\ee
so the sign only depends on $\delta \bm{c_1}$ and $\partial c_2$. Furthermore, it can be shown that
\be
(-1)^{\delta \bm{c_1}(c_2)} = (-1)^{\operatorname{lk}(\widehat{\delta \bm{c_1}}, \partial c_2)} = (-1)^{\operatorname{lk}({\partial \hat{c}_1}, \partial c_2)} \,,
\ee
where $\operatorname{lk}$ denotes the linking number. Indeed $\partial \hat{c}_1$ is a link (in the knot theory sense of the word) on the dual lattice, and $\partial c_2$ is a link on the direct lattice, so they never intersect, and have a well-defined linking number. This is where the superscript in $\mathbb{C}^{\text{lk}}$ comes from.

Now, with those definitions, we have the following natural $*$-homomorphism $\mathbb{C}^{\text{lk}}[C_2 \oplus C_1] \rightarrow \mathcal{A}_{\Z_2^{(1)}}$:
\be
(c_2, c_1) \mapsto \prod_{f \in c_2} X_f \prod_{e \in c_1} Z_{\delta \bm{e}} = \prod_{f \in c_2} X_f \prod_{f' \in \delta \bm{c_1}} Z_{f'} \,.
\ee

This identification of $\mathbb{C}^{\text{lk}}[C_2 \oplus C_1]$ as a ``preimage'' of $\mathcal{A}_{\Z_2^{(1)}}$ also explains where Eq.~\eqref{eq:twisted_product} and Eq.~\eqref{eq:twisted_involution} come from: they are needed to ensure the commutation relations are compatible with the product and involution in $\mathcal{A}_{\Z_2^{(1)}}$ (where the product is just the product of operators, and the involution is the complex adjoint). The sign prefactor, dependent on the linking number, simply counts how many $X$ and $Z$ operators overlap, and adds the corresponding number of signs.

Using that homomorphism, one can show (Proposition~\ref{prop:isomorphisms1}):
\be
\mathcal{A}_{\Z_2^{(1)}} \cong \faktor{\mathbb{C}^{\text{lk}}[C_2 \oplus C_1]}{\sim_\delta} \cong \mathbb{C}^{\text{lk}}[C_2 \oplus B^2_\text{fin.}] \,.
\ee

Here, the identification $\sim_\delta$ is
\be
(c_2, c_1) \sim_\delta (c_2', c_1') \iff c_2 = c_2', \text{ and } \delta \bm{c_1} = \delta \bm{c_1'} \,,
\ee
and $B^2_\text{fin.} < B^2$ is the subgroup of finitely supported coboundaries. It is just a different way of writing $\im \lbbb C_1 \xrightarrow{\delta} C_2 \rbbb$, as always abusing the identification between chains and finitely supported cochains.

Finally, one can also show that (Proposition~\ref{prop:isomorphisms2}):
\be
\faktor{\mathcal{A}_{\Z_2^{(1)}}}{\mathfrak{I}} \cong \faktor{\mathbb{C}^{\text{lk}}[C_2 \oplus C_1]}{\sim_{\delta, \partial}} \cong \mathbb{C}^{\text{lk}}[B_1 \oplus B^2_\text{fin.}] \,,
\label{eq: Q}
\ee
where the new identification $\sim_{\delta, \partial}$ is
\be
(c_2, c_1) \sim_{\delta, \partial} (c_2', c_1') \iff \partial c_2 = \partial c_2', \text{ and } \delta \bm{c_1} = \delta \bm{c_1'} \,,
\ee
which lends it a certain nice symmetry.

Intuitively, $\mathcal{A}_{\Z_2^{(1)}} / \mathfrak{I}$ has as (Hamel) basis pairs of 1-boundary $b_1$ and 2-coboundary $\bm{b^2}$ (finitely supported). The 2-coboundary corresponds to a product of $Z$ operators on its faces, and a choice of preimage $c_1$, such that $\bm{b^2} = \delta \bm{c_1}$, is just a choice of how to decompose
\be
\prod_{f \in \bm{b^2}} Z_f \text{ , into } \prod_{e \in c_1} Z_{\delta \bm{e}} \,.
\ee

On the other hand, a choice of preimage $c_2$, $b_1 = \partial c_2$ corresponds to a choice of product of $X$ operators:
\be
\prod_{f \in c_2} X_f \,,
\ee
and different choices for $c_2$ differ precisely by some product of $X_{\partial c}$ at various cubes.

\subsubsection{Locality}

In this section we define a suitable notion of locality in the context of the quotient algebra $\mathcal{A}_{\Z_2^{(1)}} / \mathfrak{I}$. This is non-trivial, because the usual notion of locality breaks down. Usually, one would say that an operator is localized at the set of sites where it is not proportional to the identity. But now we must make sense of the ``location'' of an entire equivalence class of operators.

For illustrative purposes, let us first consider a tensor product algebra on some set of vertices $v \in V$ (not necessarily a cubic lattice). The support of an operator $a$ can be restated as
\be
\operatorname{Supp}(a) = \curpar{v \in V \; \mid \; \sqpar{X_v, a} \neq 0, \text{ or } \sqpar{Z_v, a} \neq 0} \,,
\ee
since if it is not (proportional to) the identity at some vertex $v$, then it does not commute with at least one of $X_v, Z_v$. The converse is also true.

This motivates the following definitions:
\begin{definition}
 For an operator $a \in \mc{A}_{\Z_2^{(1)}}$, we define its support as
 \be
 \operatorname{Supp}(a) \coloneqq \curpar{e \; \mid \; \sqpar{Z_{\delta \bm{e}}, a} \neq 0} \coprod \curpar{f \; \mid \; \sqpar{X_f, a} \neq 0} \,,
 \ee
 which is a subset of $E \coprod F$.
\end{definition}

\begin{definition}
 For an equivalence class $a + \mathfrak{I} \in \mathcal{A}_{\Z_2^{(1)}} / \mathfrak{I}$, we define its support as
 \be
 \operatorname{Supp}(a + \mathfrak{I}) \coloneqq \bigcap_{I \in \mathfrak{I}} \operatorname{Supp}(a + I) \,.
 \ee
\end{definition}

It can be shown that the isomorphism
\be
\faktor{\mathcal{A}_{\Z_2^{(1)}}}{\mathfrak{I}} \cong \mathbb{C}^{\text{lk}}[B_1 \oplus B^2_\text{fin.}] \,,
\ee
introduced earlier provides a very natural interpretation for the support.
\begin{lemma}
 (Lemma~\ref{lemma:quotient_support}) For any $(b_1, \bm{b^2}) \in B_1 \oplus B^2_\text{fin.}$, choose a $c_2$ such that $b_1 = \partial c_2$. Then:
 \be
 \operatorname{Supp}\paren{\prod_{f \in c_2} X_f \prod_{f' \in b^2} Z_{f'} + \mathfrak{I}} = b_1 \coprod b^2 \,.
 \ee
 In other words, this notion of support can be transported to $\mathcal{A}_{\Z_2^{(1)}} / \mathfrak{I}$ in a way that is very straightforward.
\end{lemma}

Note that $b_1$, being a $\Z_2$-valued 1-chain, can also be seen as a set of edges. Similarly, with $\bm{b^2}$ a $\Z_2$-valued 2-cochain, $b^2$ is a 2-chain, so it can be seen as a set of faces.

Also note that, as discussed earlier, for any other choice $b_1 = \partial c_2'$, the resulting operator differs by an element of $\mathfrak{I}$, so it belongs to the same equivalence class. Hence the previous statement is immediately independent of any particular choice of $c_2$.

A linear combination of basis elements (a formal $\mathbb{C}$-linear combination of elements in $B_1 \oplus B^2_\text{fin.}$) is supported on the union of the supports of the individual basis elements.

This notion of support also satisfies (Lemma~\ref{lemma:support_of_sum_and_prod})
\begin{align}
\operatorname{Supp}(a + b + \mathfrak{I}) &\subseteq \operatorname{Supp}(a + \mathfrak{I}) \cup \operatorname{Supp}(b + \mathfrak{I}) \,, \\
\operatorname{Supp}(ab + \mathfrak{I}) &\subseteq \operatorname{Supp}(a + \mathfrak{I}) \cup \operatorname{Supp}(b + \mathfrak{I}) \,.
\end{align}

Now, let $\mathcal{R} \subseteq E \coprod F$ be some region (the disjoint union of some set of edges and some set of faces). We define
\be\label{eq:localized_subalgebra}
\paren{\mc{A}_{\Z_2^{(1)}} / \mathfrak{I}}_\mathcal{R} \coloneqq \lbbb a + \mathfrak{I} \in \mc{A}_{\Z_2^{(1)}} / \mathfrak{I} \; \mid \; \operatorname{Supp}(a + \mathfrak{I}) \subseteq \mathcal{R} \rbbb \,.
\ee

By the properties stated earlier, it is clear that this is a subalgebra of $\mc{A}_{\Z_2^{(1)}} / \mathfrak{I}$: we call it the subalgebra of operators localized in $\mathcal{R}$. If $\mathcal{R}$ is finite, then $(\mc{A}_{\Z_2^{(1)}} / \mathfrak{I})_\mathcal{R}$ is finite-dimensional.

This is almost an example of a discrete net of algebras \cite{Jones2024DHR,Jones2026QCA}, though we do not complete the $*$-algebras to $C^*$-algebras. This would require a more careful and technical treatment. We also note that our algebra satisfies Haag duality as long as we only consider ``contractible'' regions.

\subsubsection{QCA on the Quotient Algebra}

In the previous sections, we have introduced all the concepts needed to properly formalize the notion of quantum cellular automata on this new $*$-algebra $\mc{A}_{\Z_2^{(1)}} / \mathfrak{I}$.

\begin{definition}
 An algebra homomorphism 
 \be
 \alpha : \mc{A}_{\Z_2^{(1)}} / \mathfrak{I} \rightarrow \mc{A}_{\Z_2^{(1)}} / \mathfrak{I} \,,
 \ee
 is called \emph{locality preserving} if there exists some $r \geq 0$ such that the following is true. For any region $\mc{R} \subseteq E \coprod F$ (finite or infinite),
 \be
 \alpha\paren{\paren{\mc{A}_{\Z_2^{(1)}} / \mathfrak{I}}_\mc{R}} \subseteq \paren{\mc{A}_{\Z_2^{(1)}} / \mathfrak{I}}_{\mc{R}^{+r}} \,.
 \ee
 The ``expanded region'' $\mc{R}^{+r}$ is
 \be
 \mc{R}^{+r} \coloneqq \curpar{\mu \in E \coprod F \ \mid \ \exists \nu \in \mc{R} \text{ s.t. } d(\mu, \nu) \leq r} \,,
 \ee
 with $d(\cdot, \cdot)$ the euclidean distance between centers of edges and faces.
\end{definition}

\begin{definition}
 A \emph{Quantum Cellular Automaton (QCA)} on $\mc{A}_{\Z_2^{(1)}} / \mathfrak{I}$ is a $*$-automorphism $\alpha : \mc{A}_{\Z_2^{(1)}} / \mathfrak{I} \rightarrow \mc{A}_{\Z_2^{(1)}} / \mathfrak{I}$ which is locality preserving.

 The set of all QCA will be denoted $\operatorname{QCA}(\mc{A}_{\Z_2^{(1)}} / \mathfrak{I})$. A constant $r$ associated to a particular $\alpha \in \operatorname{QCA}(\mc{A}_{\Z_2^{(1)}} / \mathfrak{I})$ is an upper bound on the ``spread'' of $\alpha$.
\end{definition}

We recover the following standard result, which should be expected from any reasonable definition of QCA.

\begin{theorem}
 (Theorem~\ref{thm:group_of_qca}) The set $\operatorname{QCA}(\mc{A}_{\Z_2^{(1)}} / \mathfrak{I})$ forms a group under composition.
\end{theorem}

The proofs of this, and the other statements in this section may be found in appendix~\ref{app:quotient_algebra}.

\subsubsection{Extensions of QCAs to the Full Algebra}\label{sec:results_on_QCA_extendability}

In this section, we define the concept of extensions of QCAs on the quotient algebra to the full algebra, and state some results relating them.

\begin{definition}\label{def:QCA_extension}
 Let $\alpha : \mc{A}_{\Z_2^{(1)}} / \mathfrak{I} \rightarrow \mc{A}_{\Z_2^{(1)}} / \mathfrak{I}$ be a QCA. We say that $\tilde{\alpha} : \mc{A} \rightarrow \mc{A}$ is an extension of $\alpha$ to the full algebra if $\tilde{\alpha}$ is a QCA, $\tilde{\alpha}(\mc{A}_{\Z_2^{(1)}}) = \mc{A}_{\Z_2^{(1)}}$, and the following diagram
 \be
 \begin{tikzcd}[nodes={inner sep=5pt}]
  \mc{A}_{\Z_2^{(1)}} \arrow[r, "\tilde{\alpha}"] \arrow[d, "\pi", two heads] & \mc{A}_{\Z_2^{(1)}} \arrow[d, "\pi", two heads] \\
  \mc{A}_{\Z_2^{(1)}} / \mathfrak{I} \arrow[r, "\alpha"]      & \mc{A}_{\Z_2^{(1)}} / \mathfrak{I}
 \end{tikzcd}
 \ee
 commutes.
\end{definition}

Let us note that we call it an extension because it is not enough to lift it from $\mc{A}_{\Z_2^{(1)}} / \mathfrak{I}$ to $\mc{A}_{\Z_2^{(1)}}$, but also one needs to find a consistent definition for the individual $Z_f$ operators.

One can show that the condition $\tilde{\alpha}(\mc{A}_{\Z_2^{(1)}}) = \mc{A}_{\Z_2^{(1)}}$ is equivalent to several others.

\begin{lemma}
 Let $\tilde{\alpha} : \mc{A} \rightarrow \mc{A}$ be an extension of $\alpha : \mc{A}_{\Z_2^{(1)}} / \mathfrak{I} \rightarrow \mc{A}_{\Z_2^{(1)}} / \mathfrak{I}$. Then, the following are equivalent
 \begin{enumerate}
  \item $\tilde{\alpha}(\mc{A}_{\Z_2^{(1)}}) = \mc{A}_{\Z_2^{(1)}}$

  \item $\tilde{\alpha}(\mathfrak{I}) = \mathfrak{I}$

  \item $\mathfrak{I} \subset \tilde{\alpha}(\mc{A}_{\Z_2^{(1)}})$
 \end{enumerate}
\end{lemma}

We can also define the opposite of an extension: the restriction of a QCA from the full algebra to the quotient algebra.

\begin{lemma}
 Let $\tilde{\alpha} : \mc{A} \rightarrow \mc{A}$ be a QCA on the full algebra. If $\tilde{\alpha}(\mc{A}_{\Z_2^{(1)}}) = \mc{A}_{\Z_2^{(1)}}$ and $\tilde{\alpha}(\mathfrak{I}) = \mathfrak{I}$, then $\tilde{\alpha}$ defines a restricted QCA on $\mc{A}_{\Z_2^{(1)}} / \mathfrak{I}$ via
 \be
 \alpha(a + \mathfrak{I}) \coloneqq \tilde{\alpha}(a) + \mathfrak{I} \quad \text{for any } a \in \mc{A}_{\Z_2^{(1)}} \,.
 \ee
\end{lemma}

In Appendix~\ref{app:extensions_of_qca} we show that a certain class of QCA can always be extended from the quotient algebra to the full algebra.

\begin{theorem}\label{thm:QCA_extendability}
 Let $\alpha : \mc{A}_{\Z_2^{(1)}} / \mathfrak{I} \rightarrow \mc{A}_{\Z_2^{(1)}} / \mathfrak{I}$ be a QCA such that $\alpha(Z_{\delta \bm{e}} + \mathfrak{I}) = Z_{\delta \bm{e}} + \mathfrak{I}$. Then, there exists an extension $\tilde{\alpha}$ to the full algebra. Furthermore, we may choose it such that $\tilde{\alpha}(Z_f) = Z_f \; \forall f \in F$.
\end{theorem}

We also prove that a different class of QCA can never be extended.

\begin{proposition}
 Let $\alpha : \mc{A}_{\Z_2^{(1)}} / \mathfrak{I} \rightarrow \mc{A}_{\Z_2^{(1)}} / \mathfrak{I}$ be a QCA such that $\alpha(Z_{\delta \bm{e}} + \mathfrak{I}) = X_{t^\frac{1}{2} e} + \mathfrak{I}$.
 Then, there does not exist an extension of $\alpha$ to the full algebra.
\end{proposition}

\subsection{Kramers-Wannier-Wegner operator}
\label{sec: Z2 KWW}

In this section we will discuss the first example of a QCA on the $\mathbb{Z}_2$ 1-form symmetric algebra: the Kramers-Wannier-Wegner (KWW) duality operator \cite{Gorantla:2024ocs}.\footnote{See \cite{Koide:2021zxj} for the counterpart on a 4d Euclidean lattice and \cite{Choi:2021kmx, Kaidi:2021xfk} for the counterpart in the continuum.} We will begin with an informal introduction of the KWW duality, together with its action on operators. Then, we will provide a formal construction of the KWW duality as an element of $\operatorname{QCA}(\mc{A}_{\Z_2^{(1)}} / \mathfrak{I})$, justifying why the (seemingly) overcomplicated quotient $*$-algebra $\mc{A}_{\Z_2^{(1)}} / \mathfrak{I}$ was necessary.

The KWW operator $\mathsf{S}$\footnote{We use the notation $\sf{S}$ for consistency with the continuum operation $S$, the continuum defect $\mc{S}$, and the fact that $S, T$ are often used as the generators of the family of modular groups $SL(2, \Z_p)$.} is a non-invertible operator that implements the gauging of the $\mathbb{Z}_2$ 1-form symmetry generated by \eqref{eq: Z2 1-form sym}.
Its action on $\mathbb{Z}_2$ 1-form symmetric operators in \eqref{eq: Z2 symmetric op} is given by \cite{Gorantla:2024ocs}
\begin{equation}
X_f \xmapsto{\mathsf{S}} Z_{\delta t^{\frac{1}{2}}(\bm{f})}, \qquad
Z_{\delta \bm{e}} \xmapsto{\mathsf{S}} X_{t^{\frac{1}{2}}(e)}.
\label{eq: Z2 KWW}
\end{equation}
Here, $t^{\frac{1}{2}}$ is the half translation in the $(\frac{1}{2}, \frac{1}{2}, \frac{1}{2})$ direction.\footnote{We note that $t^{\frac{1}{2}}$ maps faces to edges and vice versa.
Accordingly, it swaps 2-(co)chains and 1-(co)chains. For example, $t^{\frac{1}{2}}(\bm{f})$ is a 1-cochain, while $t^{\frac{1}{2}}(e)$ is a 2-chain.}
See Figure~\ref{fig: Z2 KWW} for an illustration of this action.
\begin{figure}[t]
\centering
\adjincludegraphics[valign=c, scale=0.9]{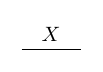}
$\xmapsto{\mathsf{S}}$
\adjincludegraphics[valign=c, scale=0.9]{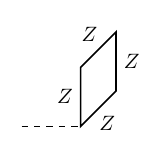}
$\xmapsto{\mathsf{S}}$
\adjincludegraphics[valign=c, scale=0.9]{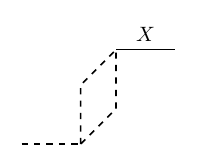}
\caption{The action of the KWW operator $\mathsf{S}$ on $\mathbb{Z}_2$ 1-form symmetric operators illustrated on the dual lattice. The dashed lines represent the edges on which the operators on the left-hand side are supported.}
\label{fig: Z2 KWW}
\end{figure}
A detailed definition of $\mathsf{S}$ as an operator on the Hilbert space will be discussed in Appendix~\ref{sec: KWW revisited}.
We note that $\sf{S}$ maps the symmetry operator $X_{\partial c}$ to the identity operator:
\begin{equation}
X_{\partial c} \xmapsto{\sf{S}} 1.
\end{equation}
The action of $\sf{S}$ can also be written in terms of the cup product as follows:
\begin{equation}
X_f \xmapsto{\sf{S}} \prod_{e^{\prime} \in E} Z_{\delta \bm{e}^{\prime}}^{\int \bm{f} \smile \bm{e}^{\prime}}, \qquad
Z_{\delta \bm{e}} \xmapsto{\sf{S}} \prod_{f^{\prime} \in F} X_{f^{\prime}}^{\int \bm{e} \smile \bm{f}^{\prime}}.
\label{eq: Z2 KWW cup}
\end{equation}
We refer the reader to Appendix~\ref{sec: higher cup products} for our convention on the cup product on a cubic lattice.

From \eqref{eq: Z2 KWW}, it is clear that the KWW operator $\sf{S}$ squares to the lattice translation $t$ when acting on $\mathbb{Z}_2$ 1-form symmetric operators, i.e.,
\begin{align}\label{eq:S2action}
 X_f &\xmapsto{\sf{S}^2} X_{t(f)} & Z_{\delta \bm{e}} &\xmapsto{\sf{S}^2} Z_{\delta t (\bm{e})}
\end{align}
The 1+1d analogue of this phenomenon was studied in detail in~\cite{Seiberg:2023cdc, Seiberg:2024gek}.

More precisely, the KWW duality should instead be regarded as a QCA on the algebra $\mc{A}_{\Z_2^{(1)}} / \mathfrak{I}$ of $\mathbb{Z}_2$ 1-form symmetric local operators.

Let us construct $\alpha_\text{KWW}$ properly. First, we define a ``lifted'' version:
\be
\check{\alpha}_\text{KWW} : \mathbb{C}^{\text{lk}}[C_2 \oplus C_1] \rightarrow \mathbb{C}^{\text{lk}}[C_2 \oplus C_1] \,,
\ee
by
\be
\check{\alpha}_\text{KWW}((c_2, c_1)) \coloneqq (-1)^{\delta \bm{c_1} (c_2)} \left(t^\frac{1}{2} c_1, t^\frac{1}{2} c_2 \right) \,,
\ee
and extended linearly.
Here, we take $t^\frac{1}{2}$ to be the half-translation, applied to both the 1-chain and 2-chain, which makes sense because $t^\frac{1}{2}$ turns edges into faces, and viceversa. The extra sign is needed to ensure compatibility with the $*$-algebra operations.
It is easy to see that this is an automorphism.

Finally, $\alpha_\text{KWW}$ is defined by restricting $\check{\alpha}_\text{KWW}$ to the quotient. This requires verifying that $\check{\alpha}_\text{KWW}$ is compatible with the equivalence $\sim_{\delta, \partial}$. This follows straightforwardly from the properties
\begin{align}
 \partial t^\frac{1}{2} &= t^\frac{1}{2} \delta & \delta t^\frac{1}{2} &= t^\frac{1}{2} \partial \,.
\end{align}

The fact that it is locality preserving is also straightforward, since $t^\frac{1}{2}$ simply moves the support of all operators by $(\frac{1}{2}, \frac{1}{2}, \frac{1}{2})$, which means the spread is bounded by $r = \frac{\sqrt{3}}{2}$.
Thus, we conclude that $\alpha_{\text{KWW}} \in \operatorname{QCA}(\mc{A}_{\Z_2^{(1)}} / \mathfrak{I})$.

Working in the quotient $\mc{A}_{\Z_2^{(1)}} / \mathfrak{I}$ instead of just the symmetric subalgebra $\mc{A}_{\Z_2^{(1)}}$ is necessary for the following reasons:
\begin{align}
 X_{\partial c} = \prod_{f \in \partial c} X_f&\xmapsto{\sf{S}} \prod_{e \; \mid \; t^\frac{1}{2} (c) \in \partial e} Z_{\delta \bm{e}} = 1 \,, \\
 1 = \prod_{e \; \mid \; s \in \partial e} Z_{\delta \bm{e}} &\xmapsto{\sf{S}} X_{\partial t^\frac{1}{2}(s)} \,.
\end{align}
Note that $\prod_{e \; \mid \; t^\frac{1}{2} (c) \in \partial e} Z_{\delta \bm{e}} = 1$ because each face ends up appearing twice, and therefore the two contributions $Z_f$ from different edges $Z_{\delta \bm{e}}, Z_{\delta \bm{e'}}$ end up canceling out.

The first equation shows that a hypothetical QCA on $\mc{A}_{\Z_2^{(1)}}$ would not be injective. But furthermore, the second equation shows that it would not even be well defined, because we also expect $1 \mapsto 1$. This is resolved, naively, by simply saying
\be
X_{\partial c} = 1 \,,
\ee
and the mathematically correct way of doing so is precisely passing to the quotient.

Note that, for brevity, in the next few sections, we will not go into as much detail as in the case of the KWW duality. The mappings of operators presented below can be defined as a QCA on $\mc{A}_{\Z_2^{(1)}} / \mathfrak{I}$ in the intuitive way, the only slight difficulty being verifying that they all properly respect the quotient. To be fully precise, for the QCA on the quotient which come from a QCA $\tilde{\alpha}$ on the full algebra $\mc{A}$, one needs to check that $\tilde{\alpha}(\mc{A}_{\Z_2^{(1)}}) = \mc{A}_{\Z_2^{(1)}}$ and $\tilde{\alpha}(\mathfrak{I}) = \mathfrak{I}$.

\subsection{Tsui-Wen entangler of order 2}
\label{sec: Z2 TW of order 2}

In this section we discuss another example of a QCA on the algebra of $\mathbb{Z}_2$ 1-form symmetric local operators: an entangler for an SPT phase with a $\mathbb{Z}_2$ 1-form symmetry.

In 3+1d, bosonic SPT phases with a $\mathbb{Z}_2$ 1-form symmetry are classified by~\cite{Kapustin:2013uxa}
\begin{equation}
H^4(B^2\mathbb{Z}_2, \mathrm{U}(1)) \cong \mathbb{Z}_4.
\end{equation}
In what follows, we will consider an SPT entangler for the SPT phase generating the $\mathbb{Z}_2$ subgroup of the $\mathbb{Z}_4$ classification. As a unitary, it will be denoted by $\sf{T}^2$ to make it clear that it is not the generator of the full $\Z_4$, and we will call it the Tsui-Wen entangler of order 2 \cite{Tsui:2019ykk}.
{The corresponding QCA on the $\mathbb{Z}_2$ 1-form symmetric algebra is denoted by $\aTW^2$.}
Note that while the construction of $\sf{T}$ itself is much more involved (see Section~\ref{sec:qca_for_root_entangler}), we can define and work with $\sf{T}^2$ directly.

As a unitary acting on the state space $\mathcal{H}_F$, the Tsui-Wen entangler of order 2 is defined by \cite{Tsui:2019ykk, Chen:2020msl}
\begin{equation}
\sf{T}^2 \ket{\hat{\bm{a}}}= (-1)^{\int \hat{\bm{a}} \smile \delta \hat{\bm{a}}} \ket{\hat{\bm{a}}}.
\label{eq: UTW def}
\end{equation}
Here, $\hat{\bm{a}}$ is a $\mathbb{Z}_2$-valued 1-cochain on the dual lattice and $\ket{\hat{\bm{a}}}$ is an arbitrary state in a computational basis, i.e.,
\begin{equation}
\ket{\hat{\bm{a}}} \coloneq \bigotimes_{f \in F} \ket{a_f}_f,
\end{equation}
where $a_f \in \{0, 1\}$ is the value of $\hat{\bm{a}}$ on the 1-chain dual to $\bm{f}$.
We may write $a_f = \hat{\bm{a}}(\hat{f})$ where $\hat{f}$ denotes the dual 1-chain of $\bm{f}$.
By definition, $\sf{T}^2$ has order 2, i.e.,
\begin{equation}
\paren{\sf{T}^2}^2 = 1.
\end{equation}

We note that $\sf{T}$ is diagonal in the computational basis.
Hence, it acts trivially on the Pauli $Z$ operators:
\begin{equation}
Z_f \xmapsto{\sf{T}^2} Z_f.
\label{eq: UZ}
\end{equation}
On the other hand, the action of $\sf{T}^2$ on the Pauli $X$ operators can be computed as
\begin{equation}
\sf{T}^2 X_f \sf{T}^{-2} \ket{\hat{\bm{a}}} = (-1)^{\int \hat{\bm{a}} \smile \delta \hat{\bm{f}} + \hat{\bm{f}} \smile \delta \hat{\bm{a}} + \hat{\bm{f}} \smile \delta \hat{\bm{f}}} \ket{\hat{\bm{a}}+\bm{\hat{\bm{f}}}},
\label{eq: UXU1}
\end{equation}
where $\hat{\bm{f}}$ is the 1-cochain dual to $f$, i.e., a 1-cochain defined by $\hat{\bm{f}}(\hat{f}^{\prime}) = \delta_{f, f^{\prime}}$.
A direct computation shows that the integrals on the right-hand side of~\eqref{eq: UXU1} can be written as
\begin{equation}
\begin{aligned}
\int \hat{\bm{a}} \smile \delta \hat{\bm{f}} &= \hat{\bm{a}}(\partial t^{-\frac{1}{2}}(\hat{f})) \quad \bmod 2, \\
\int \hat{\bm{f}} \smile \delta \hat{\bm{a}} &= \hat{\bm{a}}(\partial t^{\frac{1}{2}}(\hat{f})) \quad \bmod 2, \\
\int \hat{\bm{f}} \smile \delta \hat{\bm{f}} &= 0 \quad \bmod 2.
\end{aligned}
\label{eq: UXU2}
\end{equation}
Here, $t^{-\frac{1}{2}}$ denotes the half translation in the $(-\frac{1}{2}, -\frac{1}{2}, -\frac{1}{2})$ direction.
Equations~\eqref{eq: UXU1} and \eqref{eq: UXU2} imply that $\sf{T}^2$ acts on $X_f$ as
\begin{equation}
X_f \xmapsto{\sf{T}^2} \sf{T}^2 X_f \sf{T}^{-2} = X_f Z_{\delta t^{\frac{1}{2}}(\bm{f})} Z_{\delta t^{-\frac{1}{2}}(\bm{f})}.
\end{equation}
See Figure~\ref{fig: TW of order 2} for an illustration of this action.
\begin{figure}[t]
\centering
\adjincludegraphics[valign=c]{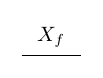}
$\xmapsto{\sf{T}^2}$
\adjincludegraphics[valign=c]{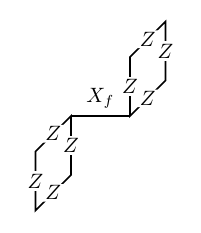}
\caption{The action of the Tsui-Wen entangler $U$ on the Pauli $X$ operator. The figure is illustrated on the dual lattice. The action of $\sf{T}^2$ on the Pauli $Z$ operator is trivial.}
\label{fig: TW of order 2}
\end{figure}
The action of $\sf{T}^2$ on $X_f$ can also be expressed in terms of the cup product as follows:
\begin{equation}
X_f \xmapsto{\sf{T}^2} X_f \prod_{e^{\prime} \in E} Z_{\delta \bm{e}^{\prime}}^{\int \bm{f} \smile \bm{e}^{\prime} + \bm{e}^{\prime} \smile \bm{f}}.
\label{eq: UX}
\end{equation}

\subsection{Framing QCA}
\label{sec: 3-fermion QCA 1-qubit version}
In this subsection, we review a non-trivial QCA constructed in~\cite{Fidkowski:2023dpe}, which is in the same equivalence class as the 3-fermion QCAs in~\cite{HaahFidkowskiHastings2023, Shirley_2022}.
A 3-fermion QCA was originally constructed in \cite{HaahFidkowskiHastings2023} as a disentangler of the Walker-Wang model based on the 3-fermion MTC.
The original 3-fermion QCA in~\cite{HaahFidkowskiHastings2023} and its simplified version constructed in~\cite[Appendix E]{Shirley_2022} are both defined on a cubic lattice with two qubits per edge.
On the other hand, in \cite{Fidkowski:2023dpe}, another QCA in the same equivalence class as the 3-fermion QCAs was constructed on a cubic lattice with a single qubit per edge.
This QCA is an example of a QCA on the algebra of $\mathbb{Z}_2$ 1-form symmetric operators.
We call this QCA the framing QCA and denote it by $\alpha_{\text{framing}}$.
The corresponding unitary operator on the Hilbert space is denoted by $\sf{U}_{\text{framing}}$ \cite{Fidkowski:2023dpe}.

We now review the definition of $\sf{U}_{\text{framing}}$ and compute its action on $\mathbb{Z}_2$ 1-form symmetric operators.
More precisely, we consider $\sf{U}_{\text{framing}}$ of \cite{Fidkowski:2023dpe} on the dual lattice because our qubits are living on faces.

To define the unitary $\sf{U}_{\text{framing}}$, we first define the following fermion hopping operator for each face $f$ \cite{Chen:2018nog, Shirley_2022}:
\begin{equation}
\tilde{u}_f \coloneq Z_f \prod_{f^{\prime}} X_{f^{\prime}}^{\int \bm{f}^{\prime} \smile_1 \bm{f}}.
\label{eq: fermion hopping}
\end{equation}
Here, $\smile_1$ denotes the cup-1 product on a cubic lattice; see Appendix~\ref{sec: higher cup products} for the definition.
The above fermion hopping operator can be written more explicitly as shown in Figure~\ref{fig: fermion hopping}.
\begin{figure}[t]
\centering
\adjincludegraphics[valign=c, trim={10, 0, 0, 0}]{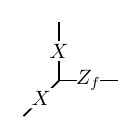}
\adjincludegraphics[valign=c]{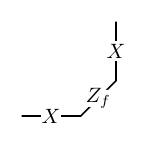}
\adjincludegraphics[valign=c]{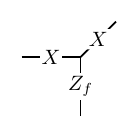}
\caption{The fermion hopping operator $\tilde{u}_f$ illustrated on the dual lattice.}
\label{fig: fermion hopping}
\end{figure}
We then define the fermionic flux operator $\tilde{u}_{\delta \bm{e}}$ for each edge $e$ by taking the product of the fermion hopping operators around $e$:
\begin{equation}
\tilde{u}_{\delta \bm{e}} \coloneq \prod_{f^{\prime}} \tilde{u}_{f^{\prime}}^{\delta \bm{e}(f^{\prime})}.
\label{eq: fermion flux1}
\end{equation}
Here, the order of the product is chosen so that all Pauli $X$ operators act before any Pauli $Z$ operator.
By plugging \eqref{eq: fermion hopping} into \eqref{eq: fermion flux1}, we obtain
\begin{equation}
\tilde{u}_{\delta \bm{e}} = Z_{\delta \bm{e}} \prod_{f^{\prime}} X_{f^{\prime}}^{\int \bm{f}^{\prime} \smile_1 \delta \bm{e}}.
\end{equation}
Furthermore, by using the identity
\begin{equation}
\begin{aligned}
\delta(\bm{f}^{\prime} \smile_1 \bm{e}) &= \delta \bm{f}^{\prime} \smile_1 \bm{e} + \bm{f} \smile_1 \delta \bm{e} \\
& \quad + \bm{f}^{\prime} \smile \bm{e} + \bm{e} \smile \bm{f}^{\prime} \quad \bmod 2,
\end{aligned}
\label{eq: cup-1 Leibniz}
\end{equation}
one can rewrite $\tilde{u}_{\delta \bm{e}}$ as
\begin{align}\label{eq: tilde u}
 \tilde{u}_{\delta \bm{e}}
 &= Z_{\delta \bm{e}} \left(\prod_{f^{\prime}} X_{f^{\prime}}^{\int \bm{f}^{\prime} \smile \bm{e} + \bm{e} \smile \bm{f}^{\prime}}\right) \left(\prod_{c^{\prime}} X_{\partial c^{\prime}}^{\int \bm{c}^{\prime} \smile_1 \bm{e}}\right) \nonumber \\
 &= Z_{\delta \bm{e}} X_{t^{\frac{1}{2}}(e)} X_{t^{-\frac{1}{2}}(e)} \left(\prod_{c^{\prime}} X_{\partial c^{\prime}}^{\int \bm{c}^{\prime} \smile_1 \bm{e}}\right) \,.
\end{align}
This operator can be illustrated as shown in Figure~\ref{fig: fermion flux}.
\begin{figure*}[t]
\centering
\adjincludegraphics[valign=c]{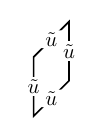} = \adjincludegraphics[valign=c]{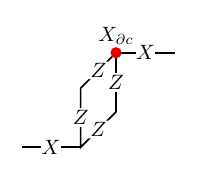},
\adjincludegraphics[valign=c]{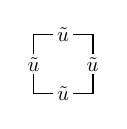} = \adjincludegraphics[valign=c]{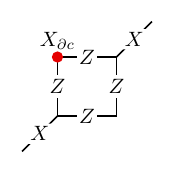},
\adjincludegraphics[valign=c]{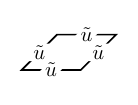} = \adjincludegraphics[valign=c]{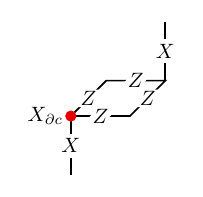}
\caption{The fermion flux operator $\tilde{u}_{\delta \bm{e}}$ illustrated on the dual lattice. The red dot represents the dual of cube $c$ where the operator $X_{\partial c}$ acts. We recall our convention that Pauli $X$'s act before Pauli $Z$'s.}
\label{fig: fermion flux}
\end{figure*}
From the above expression, we can see that the fermionic flux operators commute with each other and square to the identity, i.e.,
\begin{equation}
\tilde{u}_{\delta \bm{e}} \tilde{u}_{\delta \bm{e}^{\prime}} = \tilde{u}_{\delta \bm{e}^{\prime}} \tilde{u}_{\delta \bm{e}}, \quad
\tilde{u}_{\delta \bm{e}}^2 = 1,
\end{equation}
for all edges $e$ and $e^{\prime}$.
Therefore, these operators can be diagonalized simultaneously and have eigenvalues $\pm 1$.
We note that edges with $\tilde{u}_{\delta \bm{e}} = -1$ form a loop, which we call a fermionic flux loop.

Now, we define the unitary $\sf{U}_{\text{framing}}$ by the following action on states \cite{Fidkowski:2023dpe}:
\begin{equation}
\sf{U}_{\text{framing}} \ket{\psi} = (-1)^{\text{framing}(L)} \ket{\psi}.
\end{equation}
Here, $\ket{\psi}$ is a simultaneous eigenstate of all fermionic flux operators $\tilde{u}_{\delta \bm{e}}$, and $L$ is the fermionic flux loop in $\ket{\psi}$.
The exponent on the right-hand side is the framing of $L$, that is, the mod 2 linking number of $L$ and its push-off in the $(\frac{1}{2}, \frac{1}{2}, \frac{1}{2})$ direction.
In particular, when $L$ is homologically trivial, the framing of $L$ can be written as
\begin{equation}
\text{framing}(L) = \int \hat{\bm{b}} \smile \delta \hat{\bm{b}} \quad \bmod 2,
\label{eq: framing}
\end{equation}
where $\hat{\bm{b}}$ is the dual 1-cochain of a 2-chain $\Sigma$ satisfying $\partial \Sigma = L$.
By definition, $\sf{U}_{\text{framing}}$ has order 2, i.e.,
\begin{equation}
\sf{U}_{\text{framing}}^2 = 1.
\end{equation}

Based on the above definition, we can compute the action of $\sf{U}_{\text{framing}}$ on local operators.
In what follows, we will only compute the action on $\mathbb{Z}_2$ 1-form symmetric local operators.
We note that such operators are generated by $X_f$ and $\tilde{u}_{\delta \bm{e}}$ because $Z_{\delta \bm{e}}$ is obtained from $X_f$ and $\tilde{u}_{\delta \bm{e}}$ due to~\eqref{eq: tilde u}.
Thus, it suffices to compute the action of $\sf{U}_{\text{framing}}$ on $X_f$ and $\tilde{u}_{\delta \bm{e}}$.

Since $\sf{U}_{\text{framing}}$ is diagonal in the eigenbasis of $\tilde{u}_{\delta \bm{e}}$, it acts trivially on $\tilde{u}_{\delta \bm{e}}$:
\begin{equation}
\tilde{u}_{\delta \bm{e}} \xmapsto{\sf{U}_{\text{framing}}} \tilde{u}_{\delta \bm{e}}.
\label{eq: U framing u}
\end{equation}
On the other hand, $\sf{U}_{\text{framing}}$ acts non-trivially on $X_f$ because $X_f$ changes the configuration of the fermionic flux loop.
Specifically, we can compute the action of $\sf{U}_{\text{framing}}$ on $X_f$ as
\begin{multline}\label{eq: UXU framing1}
 \sf{U}_{\text{framing}} X_f \sf{U}_{\text{framing}}^{-1} \ket{\psi} \\
= (-1)^{\mathrm{framing}(L+\partial f) - \mathrm{framing}(L)} X_f \ket{\psi} \,.
\end{multline}
We note that the sign on the right-hand side depends only on $L$ near face $f$.
In particular, it does not depend on global data of $L$ such as its homology class.
Therefore, to compute $\sf{U}_{\text{framing}} X_f \sf{U}_{\text{framing}}^{-1}$, we can assume without loss of generality that $L$ is homologically trivial.
In this case, due to \eqref{eq: framing}, the sign on the right-hand side of~\eqref{eq: UXU framing1} can be computed as
\begin{align}\label{eq: UXU framing2}
 &(-1)^{\mathrm{framing}(L+\partial f) - \mathrm{framing}(L)} \nonumber \\
 &\quad =(-1)^{\int \hat{\bm{b}} \smile \delta \hat{\bm{f}} + \hat{\bm{f}} \smile \delta \hat{\bm{b}} + \hat{\bm{f}} \smile \delta \hat{\bm{f}}} \nonumber \\
 &\quad = (-1)^{\hat{\bm{b}}(\partial t^{-\frac{1}{2}}(\hat{f}) + \partial t^{\frac{1}{2}}(\hat{f}))} \,,
\end{align}
where the second equality follows from \eqref{eq: UXU2}.
Equations~\eqref{eq: UXU framing1} and \eqref{eq: UXU framing2} imply that $\sf{U}_{\text{framing}}$ acts on $X_f$ as follows:
\begin{equation}
X_f \xmapsto{\sf{U}_{\text{framing}}} X_f \tilde{u}_{\delta t^{\frac{1}{2}}(\bm{f})} \tilde{u}_{\delta t^{-\frac{1}{2}}(\bm{f})}.
\label{eq: U framing X}
\end{equation}
See Figure~\ref{fig: U framing} for an illustration of this action.
\begin{figure}[t]
\centering
\adjincludegraphics[valign=c]{fig/Xf.pdf}
$\xmapsto{\sf{U}_{\text{framing}}}$
\adjincludegraphics[valign=c]{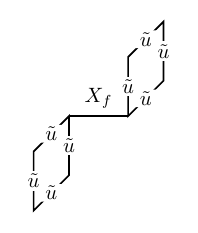}
\caption{The action of $\sf{U}_{\text{framing}}$ on the Pauli $X$ operator. The figure is illustrated on the dual lattice. The action of $\sf{U}_{\text{framing}}$ on the fermionic flux operator $\tilde{u}_{\delta \bm{e}}$ is trivial.}
\label{fig: U framing}
\end{figure}
The action of $\sf{U}_{\text{framing}}$ on $X_f$ can also be expressed in terms of the cup product as
\begin{equation}
X_f \xmapsto{\sf{U}_{\text{framing}}} X_f \prod_{e^{\prime}} \tilde{u}_{\delta \bm{e}^{\prime}}^{\int \bm{f} \smile \bm{e}^{\prime} + \bm{e}^{\prime} \smile \bm{f}}.
\label{eq: U framing X cup}
\end{equation}
We note that the action of $\sf{U}_{\text{framing}}$ on $Z_{\delta \bm{e}}$ is automatically determined by \eqref{eq: U framing u} and \eqref{eq: U framing X} because $Z_{\delta \bm{e}}$ can be written as a product of $\tilde{u}_{\delta \bm{e}}$ and $X_f$ due to \eqref{eq: tilde u}.

\subsection{3-fermion Kramers-Wannier-Wegner QCA}
\label{sec: 3-fermion KWW}
Using the KWW operator $\mathsf{S}$ and the Tsui-Wen entangler $\sf{T}^2$, we can define the following operator:
\begin{equation}
\mathsf{D}_\text{3F} \coloneq \mathsf{S} \sf{T}^2 \mathsf{S} \sf{T}^2 \sf{S}^{\dagger}
= \mathsf{S} \sf{T}^2 \mathsf{S} \sf{T}^2 \sf{S} t^{-1} \,.
\end{equation}
We refer to this operator as the 3-fermion Kramers-Wannier-Wegner operator. We note that this same definition works for the corresponding QCA on $\mc{A}_{\Z_2^{(1)}} / \mathfrak{I}$: $\aKWW$ and $\aTW^2$. Composing them in the same order yields the 3-fermion KWW QCA $\athreeF$, i.e.,
\begin{equation}
\begin{aligned}
\athreeF
&\coloneq (\aKWW \circ \aTW^2)^2 \circ \aKWW^{-1} \\
&= (\aKWW \circ \aTW^2)^2 \circ \aKWW \circ \alpha_t^{-1}.
\end{aligned}
\end{equation}
The relation to the 3-fermion QCA in~\cite[Appendix E]{Shirley_2022} will be discussed in detail in Appendix~\ref{sec: Equivalent definition of 3-fermion KWW operator}.

Based on the above definition, we can compute the ``action'' of $\mathsf{D}_\text{3F}$ on $\mathbb{Z}_2$ 1-form symmetric local operators. This can be understood in two ways. For example, for the KWW duality, we can write it in terms of the duality operator (not unitary), or in terms of the QCA
\be
\sf{S} X_f = Z_{\delta t^{\frac{1}{2}} \bm{f}} \sf{S} \,, \quad \aKWW(X_f + \mathfrak{I}) = Z_{\delta t^{\frac{1}{2}} \bm{f}} + \mathfrak{I} \,.
\ee
Both are a little cumbersome and excessively verbose, so either version may be abbreviated as
\be
X_f \xmapsto{\sf{S}} Z_{\delta t^{\frac{1}{2}} \bm{f}} \,.
\ee

To compute the action of $\sf{D}_\text{3F}$, we first compute the action of $\sf{S}\sf{T}^2$ on $X_f$ and $Z_{\delta \bm{e}}$ as follows:
\begin{equation}
\begin{aligned}
X_f &\xmapsto{\sf{S}\sf{T}^2} Z_{\delta t^{\frac{1}{2}}(\bm{f})} X_{t(f)} X_f \, , \\
Z_{\delta \bm{e}} &\xmapsto{\sf{S}\sf{T}^2} X_{t^{\frac{1}{2}}(e)} \,.
\end{aligned}
\end{equation}
By applying the above equation repeatedly, we find
\begin{equation}
\begin{aligned}
X_f &\xmapsto{\sf{S}\sf{T}^2\sf{S}\sf{T}^2} Z_{\delta t^{\frac{3}{2}}(\bm{f})} Z_{\delta t^{\frac{1}{2}}(\bm{f})} X_{t^2(f)} X_{t(f)} X_{f}, \\
Z_{\delta \bm{e}} &\xmapsto{\sf{S}\sf{T}^2\sf{S}\sf{T}^2} Z_{\delta t(\bm{e})} X_{t^{\frac{3}{2}}(e)} X_{t^{\frac{1}{2}}(e)}.
\end{aligned}
\end{equation}
Then, by acting with $\sf{S}\sf{T}^2\sf{S}\sf{T}^2$ on $X_f \xmapsto{\sf{S}^\dag} Z_{\delta t^{-\frac{1}{2}} (\bm{f})}$ and $Z_{\delta \bm{e}} \xmapsto{\sf{S}^\dag} X_{t^{-\frac{1}{2}}(\bm{e})}$, we obtain
\begin{equation}
\begin{aligned}
X_f &\xmapsto{\sf{D}_\text{3F}} Z_{\delta t^{\frac{1}{2}}(\bm{f})} X_{t(f)} X_f, \\
Z_{\delta \bm{e}} &\xmapsto{\sf{D}_\text{3F}} Z_{\delta t(\bm{e})} Z_{\delta \bm{e}} X_{t^{\frac{3}{2}}(e)} X_{t^{\frac{1}{2}}(e)} X_{t^{-\frac{1}{2}}(e)}.
\end{aligned}
\end{equation}
The action of $\sf{D}_\text{3F}$ can also be expressed more concisely by using the fermionic flux operators as follows:
\begin{equation}
X_f \xmapsto{\sf{D}_\text{3F}} \tilde{u}_{\delta t^{\frac{1}{2}}(\bm{f})}, \quad
\tilde{u}_{\delta \bm{e}} \xmapsto{\sf{D}_\text{3F}} X_{t^{\frac{1}{2}}(e)}.
\label{eq: 3F KWW action}
\end{equation}
Here, we used \eqref{eq: tilde u} and $X_{\partial c} = 1$.
The above action is illustrated in Figure~\ref{fig: 3F KWW}.
\begin{figure}
\centering
\adjincludegraphics[valign=c, scale=0.9]{fig/Z_2_KWW1.pdf}
$\xmapsto{\sf{D}_\text{3F}}$
\adjincludegraphics[valign=c, scale=0.9]{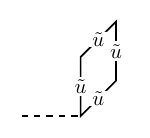}
$\xmapsto{\sf{D}_\text{3F}}$
\adjincludegraphics[valign=c, scale=0.9]{fig/Z_2_KWW3.pdf}
\caption{The action of the 3-fermion Kramers-Wannier-Wegner QCA on $\mathbb{Z}_2$ 1-form symmetric operators. The figure is illustrated on the dual lattice.}
\label{fig: 3F KWW}
\end{figure}
We note the similarity between the action of $\mathsf{D}_\text{3F}$ in \eqref{eq: 3F KWW action} and that of $\mathsf{S}$ in \eqref{eq: Z2 KWW}.
The only difference is that $\mathsf{D}_\text{3F}$ exchanges $X_f$ and $\tilde{u}_{\delta \bm{e}}$, whereas $\mathsf{S}$ exchanges $X_f$ and $Z_{\delta \bm{e}}$.
As in the case of the KWW operator, $\athreeF^2 = \alpha_t$ where $\alpha_t$ is the QCA corresponding to lattice translation.

Before proceeding, we remark that there is no $\mathbb{Z}_2$ 1-form symmetric QCA $\sf{V}$ (understood as a locality preserving unitary) on the full algebra $\mathcal{A} = \langle X_f, Z_f \mid f \in F \rangle$ such that $\mathsf{D}_\text{3F} = \sf{V} \mathsf{S}$ or $\mathsf{D}_\text{3F} = \mathsf{S} \sf{V}$.
We will show this by contradiction. 
To this end, we suppose that the 3-fermion KWW operator can be written as $\mathsf{D}_\text{3F} = \sf{V} \mathsf{S}$, where $\sf{V}$ is a $\mathbb{Z}_2$ 1-form symmetric unitary.
Since $\mathsf{S}$ and $\sf{D}_\text{3F}$ map the Pauli $X$ operator into $Z_{\delta \bm{e}}$ and $\tilde{u}_{\delta \bm{e}}$ respectively, the unitary $\sf{V}$ should map $Z_{\delta \bm{e}}$ to $\tilde{u}_{\delta \bm{e}}$ modulo the symmetry operators $X_{\partial c}$.
On the other hand, since $\sf{V}$ is supposed to be $\mathbb{Z}_2$ 1-form symmetric, it preserves $X_{\partial c}$ by assumption.
Therefore, $\sf{V}$ should map the states stabilized by $\{Z_{\delta \bm{e}}, X_{\partial c}\}$ to those stabilized by $\{\tilde{u}_{\delta \bm{e}}, X_{\partial c}\}$.
In other words, $\sf{V}$ maps the ground states of the following commuting projector Hamiltonians into each other:
\begin{align}
H &= -\sum_{c} X_{\partial c} - \sum_{e} Z_{\delta \bm{e}},
\label{eq: bosonic TC}
\\
H^{\prime} &= -\sum_{c} X_{\partial c} - \sum_{e} \tilde{u}_{\delta \bm{e}}.
\label{eq: fermionic TC}
\end{align}
However, since these Hamiltonians realize different topological orders,\footnote{Equation~\eqref{eq: bosonic TC} is the toric code Hamiltonian, whereas equation~\eqref{eq: fermionic TC} is the Hamiltonian of the ``fermionic" toric code, which is obtained by the bosonization of a trivial fermionic Hamiltonian \cite{Chen:2018nog, Chen_2023}.} their ground states cannot be connected by locality-preserving unitaries.
Therefore, there does not exist a $\mathbb{Z}_2$ 1-form symmetric unitary $\sf{V}$ such that $\sf{D}_\text{3F} = \sf{V} \mathsf{S}$.
Similarly, one can also show that there is no $\mathbb{Z}_2$ 1-form symmetric unitary $\sf{V}$ such that $\sf{D}_\text{3F} = \mathsf{S} \sf{V}$.

The above result implies that $\mathsf{D}_\text{3F} \mathsf{S}^{\dagger}$ and $\mathsf{S}^{\dagger} \sf{D}_\text{3F}$ cannot be extended to QCAs on the full algebra $\mathcal{A}$.
Nevertheless, as we will see in the next subsection, $\sf{D}_\text{3F} (\mathsf{S}\sf{T}^2)^{\dagger}$ can be extended to a QCA on the full algebra.

\subsection{Framing QCA is generated by $\mathsf{S}$ and $\sf{T}^2$}
\label{sec: Symmetry generated by D and U}
In this subsection, we will show that the KWW operator $\mathsf{S}$ and the Tsui-Wen entangler $\sf{T}^2$ generate the framing QCA $\sf{U}_{\text{framing}}$.
More specifically, we will show the following equality of QCAs on the $\mathbb{Z}_2$ 1-form symmetric algebra $\mathcal{A}_{\mathbb{Z}_2^{(1)}} / \mathfrak{I}$:
\be\label{eq: DUDUDUD}
\paren{\aKWW \circ \aTW^2}^2 \circ \aKWW^{-1} \circ \aTW^2 \circ \aKWW^{-1} = \alpha_\text{framing}
\ee
where $\alpha_\text{framing}$ is the framing QCA as defined earlier.
Equivalently, we have
\begin{equation}
\athreeF = \alpha_{\text{framing}} \circ \aKWW \circ \aTW^2 \,.
\end{equation}
This implies that $\athreeF \circ (\aKWW \circ \aTW^2)^{-1}$, or equivalently the QCA defined by the operators $\mathsf{D}_\text{3F}(\mathsf{S} \sf{T}^2)^{\dagger}$, can be extended to the framing QCA on the full algebra $\mathcal{A}$. Equation~\eqref{eq: DUDUDUD} is also equivalent to
\begin{equation}\label{eq: DU4 lattice}
 \paren{\aKWW \circ \aTW^2}^4 = \alpha_t^2 \circ \alpha_\text{framing} \circ \aTW^2
\end{equation}
where we used $\aKWW^{-1} = \aKWW \circ \alpha_t^{-1}$ and the translation invariance of $\aKWW$ and $\aTW^2$.
The relation to the corresponding fusion rules in the continuum will be discussed in Section~\ref{sec: comparison Z2}.

To show \eqref{eq: DUDUDUD}, we first rewrite the left-hand side of \eqref{eq: DUDUDUD} as
\be
\paren{\aKWW \aTW^2 \aKWW^{-1}} \aTW^2 \paren{\aKWW \aTW^2 \aKWW^{-1}}.
\ee
We can also reformulate it in terms of the action of $(\sf{S} \sf{T}^2 \sf{S}^\dag) \sf{T}^2 (\sf{S} \sf{T}^2 \sf{S}^\dag)$.
We note that $\sf{S} \sf{T}^2 \sf{S}^\dag$ leaves the Pauli $X$ operator invariant, i.e.,
\begin{equation}
X_f \xmapsto{\sf{S} \sf{T}^2 \sf{S}^\dag} X_f,
\label{eq: DUD X}
\end{equation}
while it exchanges the bosonic and fermionic flux operators:
\begin{equation}
Z_{\delta \bm{e}} \xmapsto{\sf{S} \sf{T}^2 \sf{S}^\dag} \tilde{u}_{\delta \bm{e}}, \quad
\tilde{u}_{\delta \bm{e}} \xmapsto{\sf{S} \sf{T}^2 \sf{S}^\dag} Z_{\delta \bm{e}}.
\label{eq: DUD Z}
\end{equation}
In particular, $\sf{S} \sf{T}^2 \sf{S}^\dag$ exchanges fermionic flux loops and bosonic ones.
On the other hand, the Tsui-Wen entangler $\sf{T}^2$ defined by \eqref{eq: UTW def} measures the framing of the bosonic flux loop.
Thus, the composite operator $(\sf{S} \sf{T}^2 \sf{S}^\dag) \sf{T}^2 (\sf{S} \sf{T}^2 \sf{S}^\dag)$ measures the framing of the fermionic flux loop, which is exactly what $\sf{U}_{\text{framing}}$ does.
This suggests that equation~\eqref{eq: DUDUDUD} holds.
Indeed, by a direct computation, we can show that $(\sf{S} \sf{T}^2 \sf{S}^\dag) \sf{T}^2 (\sf{S} \sf{T}^2 \sf{S}^\dag)$ acts on $X_f$ and $\tilde{u}_{\delta \bm{e}}$ as
\begin{multline}
 X_f \xmapsto{\sf{S} \sf{T}^2 \sf{S}^\dag} X_f \xmapsto{\sf{T}^2} X_f Z_{\delta t^{\frac{1}{2}}(\bm{f})} Z_{\delta t^{-\frac{1}{2}}(\bm{f})} \\
 \xmapsto{\sf{S} \sf{T}^2 \sf{S}^\dag} X_f \tilde{u}_{\delta t^{\frac{1}{2}}(\bm{f})} \tilde{u}_{\delta t^{-\frac{1}{2}}(\bm{f})} \,,
\end{multline}
and
\begin{equation}\label{eq: DUDUDUD X u}
 \tilde{u}_{\delta \bm{e}} \xmapsto{\sf{S} \sf{T}^2 \sf{S}^\dag} Z_{\delta \bm{e}} \xmapsto{\sf{T}^2} Z_{\delta \bm{e}} \xmapsto{\sf{S} \sf{T}^2 \sf{S}^\dag} \tilde{u}_{\delta \bm{e}}
\end{equation}
Here, we used \eqref{eq: DUD X}, \eqref{eq: DUD Z}, \eqref{eq: UZ}, and \eqref{eq: UX}.
The above equation agrees with the action of $\sf{U}_{\text{framing}}$ discussed in Section~\ref{sec: 3-fermion QCA 1-qubit version}.
Thus, we find that \eqref{eq: DUDUDUD} holds.

\subsection{Group of QCAs generated by $\mathsf{S}$ and $\sf{T}^2$}
\label{sec: The group of QCAs generated by D and U}
In this subsection, we discuss the group of QCAs generated by $\mathsf{S}$ and $\sf{T}^2$ modulo lattice translations and finite-depth circuits built of symmetric local gates.
%We refer the reader to Section~\ref{sec: continuum fusion rules Z2} for the relation to the corresponding group in continuum field theories.

We first clarify the definition of the group of QCAs we will discuss below. Recall that we already defined $\operatorname{QCA}(\mc{A}_{\Z_2^{(1)}} / \mathfrak{I})$. However, this time we are going to restrict to the translation invariant subgroup, denoted by $\mathsf{QCA}_{\mathbb{Z}_2^{(1)}}$.
The group $\mathsf{QCA}_{\mathbb{Z}_2^{(1)}}$ contains a subgroup consisting of all translation invariant finite-depth circuits with $\mathbb{Z}_2$ 1-form symmetric local gates. Each layer is to be understood as conjugation by some collection of unitaries $u_i + \mathfrak{I}$ (in the quotient, it is a unitary if $u_i^{*} u_i + \mathfrak{I} = 1 + \mathfrak{I}$), supported on disjoint homologically trivial regions of the lattice\footnote{Essentially, solid balls, or deformed balls, but not solid tori, for example.}.
We denote this subgroup by $\mathsf{FDC}_{\mathbb{Z}_2^{(1)}}$.
We note that $\mathsf{FDC}_{\mathbb{Z}_2^{(1)}}$ is a normal subgroup of $\mathsf{QCA}_{\mathbb{Z}_2^{(1)}}$ because any element of $\mathsf{QCA}_{\mathbb{Z}_2^{(1)}}$ maps a symmetric local gate to another symmetric local gate.
Futhermore, $\mathsf{QCA}_{\mathbb{Z}_2^{(1)}}$ also contains a subgroup generated by the lattice translations.
We denote this subgroup by $\mathsf{Shift}$.
By definition, $\mathsf{Shift}$ is a central subgroup of $\mathsf{QCA}_{\mathbb{Z}_2^{(1)}}$.
Since $\mathsf{FDC}_{\mathbb{Z}_2^{(1)}} \times \mathsf{Shift}$ is a normal subgroup of $\mathsf{QCA}_{\mathbb{Z}_2^{(1)}}$, the quotient
\begin{equation}
\faktor{\mathsf{QCA}_{\mathbb{Z}_2^{(1)}}}{(\mathsf{FDC}_{\mathbb{Z}_2^{(1)}} \times \mathsf{Shift})}
\end{equation}
is also a group.
An element of this group will be denoted by $[\alpha]$, where $\alpha$ is a translation invariant QCA on $\mathcal{A}_{\mathbb{Z}_2^{(1)}} / \mathfrak{I}$.
This group contains a subgroup generated by $[\aKWW]$ and $[\aTW^2]$, which we denote by
\begin{equation}
\langle [\aKWW], [\aTW^2] \rangle \,.
\end{equation}
This subgroup is our focus in this subsection.

Due to the results in the previous subsections, the generators $[\aKWW]$ and $[\aTW^2]$ obey the following relations:
\begin{equation}
\begin{aligned}
[\aKWW]^2 = 1 \,, \qquad 
[\aTW^2]^2 = 1 \,, \\
[\aKWW \circ \aTW^2]^4 = [\alpha_{\text{framing}} \circ \aTW^2] \,.
\end{aligned}
\label{eq: Z2 ext D8}
\end{equation}
In what follows, we will argue that $[\alpha_{\text{framing}} \circ \aTW^2]$ is a central element of order $2$ in  $\langle [\aKWW], [\aTW^2] \rangle$.
This result shows that the group generated by $\aKWW$ and $\aTW^2$ is consistent with the group generated by $S$ and $T^2$ in the continuum (cf. Section~\ref{sec: continuum fusion rules Z2}).

In essence, our argument will rely on one assumption which we believe to be physically reasonable: that the classification of $\Z_2^{(1)}$-SPT phases is complete, and given by the entanglers $\sf{T}^k$ for $k \in \Z_4$.\footnote{In fact, this assumption will also be crucial in Section~\ref{sec:derivation_of_STSTST_Y}.}

%Thus, the relations in Equation~\eqref{eq: Z2 ext D8} imply that the group $\langle \aKWW, \aTW^2 \rangle$ is isomorphic to a central extension of $D_8$ by $\mathbb{Z}_2$, where $D_8$ is the dihedral group of order 8.\footnote{Recall that $D_8$ has presentation $\langle r, s \mid s^2=r^4=1, srs=r^{-1} \rangle$. By defining $x \coloneq s$ and $y \coloneq sr$, we find a different presentation with relations $x^2=y^2=(xy)^4=1$, which makes it much clearer that $x \sim \aKWW$ and $y \sim \aTW^2$, and the central extension occurs in the relation $(xy)^4 = c$ with $c$ central and $c^2 = 1$.}
%This group is isomorphic to the group generated by the operations $S$ and $T^2$ in the continuum, which we discussed in Section~\ref{sec: continuum fusion rules Z2}.

\subsubsection{Order of $[\alpha_{\text{framing}} \circ \aTW^2]$}

We first argue that $[\alpha_{\text{framing}} \circ \aTW^2]$ has order 2, namely, $(\alpha_{\text{framing}} \circ \aTW)^2$ is an element of $\mathsf{FDC}_{\mathbb{Z}_2^{(1)}} \times \mathsf{Shift}$. We begin by noting that $\sf{T}^2$ is actually a finite depth circuit, though its gates are \emph{not} $\Z_2^{(1)}$-symmetric.
Then
\begin{equation}
\sf{U}_{\text{framing}} \sf{T}^2 \sf{U}_{\text{framing}} = \sf{U}_{\text{framing}} \sf{T}^2 \sf{U}_{\text{framing}}^{-1} \,,
\end{equation}
is also a finite-depth circuit because it is the conjugate of a finite-depth circuit by a locality-preserving unitary. Again, the local gates of the above finite-depth circuit \emph{cannot} be $\mathbb{Z}_2^{(1)}$-symmetric: if they were, we could undo the previous mapping and find symmetric gates for $\sf{T}^2$.
Furthermore, the above FDC has order 2, i.e.,
\begin{equation}
(\sf{U}_{\text{framing}} \sf{T}^2 \sf{U}_{\text{framing}})^2 = 1 \,.
\end{equation}
This implies that $\sf{U}_{\text{framing}} \sf{T}^2 \sf{U}_{\text{framing}}$ is an entangler for a non-trivial SPT phase of order $2$.
Thus, under the assumption that the classification of SPT phases is complete, its action by conjugation must be equivalent up to translations and an FDC (with \emph{symmetric} gates) to the action by conjugation of $\sf{T}^2$, or in QCA language,
\begin{equation}
\alpha_\text{framing} \circ \aTW^2 \circ \alpha_\text{framing} \circ \aTW^{-2} \in \mathsf{FDC}_{\mathbb{Z}_2^{(1)}} \times \mathsf{Shift} \,.
\label{eq: order of Uframing U}
\end{equation}
The above equation implies
\begin{equation}
[\alpha_{\text{framing}} \circ \aTW^2]^2 = 1.
\end{equation}

We emphasize that \eqref{eq: order of Uframing U} relies on the conjectural classification of SPT phases.

\subsubsection{Centrality of $[\alpha_{\text{framing}} \circ \aTW^2]$}
Next, we argue that $[\beta] \coloneqq [\alpha_{\text{framing}} \circ \aTW^2]$ is in the center of $\langle [\aKWW], [\aTW^2] \rangle$.
It suffices to show that $[\beta]$ commutes with the generators, $[\aKWW]$ and $[\aTW^2]$.
Based on \eqref{eq: Z2 ext D8}, one can compute the commutation relations as follows:
\begin{equation}\label{eq: Uframing U centrality}
 \begin{aligned}
  [\aKWW] [\beta] [\aKWW]^{-1} &= [\aTW^2 \circ \aKWW]^4 \\
  &= [\aTW^2] [\beta] [\aTW^2] \\
  &= [\aTW^2 \circ \alpha_\text{framing}] \\
  [\aTW^2] [\beta] [\aTW^2]^{-1} &= [\aTW^2 \circ \alpha_\text{framing}] \,.
 \end{aligned}
\end{equation}
On the other hand, equation~\eqref{eq: order of Uframing U} implies
\be
[\aTW^2 \circ \alpha_\text{framing}] = [\alpha_{\text{framing}} \circ \aTW^2] = [\beta]
\ee
By plugging this into \eqref{eq: Uframing U centrality}, we find
\begin{equation}
\begin{aligned}
[\aKWW] [\beta] [\aKWW]^{-1} &= [\beta], \\
[\aTW^2] [\beta] [\aTW^2]^{-1} &= [\beta].
\end{aligned}
\end{equation}
This shows that $[\beta] = [\alpha_{\text{framing}} \circ \aTW^2]$ is central in $\langle [\aKWW], [\aTW^2] \rangle$.

We remark again that this argument is based on the previous result \eqref{eq: order of Uframing U}, which itself is based on the conjectural complete classification of $\Z_2^{(1)}$ SPT phases.

\subsection{Comparison between lattice and continuum}
\label{sec: comparison Z2}
Having computed the fusion rules both on the lattice and in the continuum, we can now compare them and discuss their relation.

As discussed in Section~\ref{sec: Symmetry generated by D and U}, the KWW QCA $\alpha_{\mathrm{KWW}}$ and the Tsui-Wen entangler $\alpha_{\mathrm{TW}}^2$ obey the following fusion rule:
\begin{equation}
(\alpha_{\mathrm{KWW}} \circ \alpha_{\mathrm{TW}}^2)^4 = \alpha_t^2 \circ \alpha_{\mathrm{framing}} \circ \alpha_{\mathrm{TW}}^2.
\label{eq: comparison lat Z2}
\end{equation}
In the continuum, the QCA $\alpha_{\mathrm{KWW}} \circ \alpha_{\mathrm{TW}}^2$ corresponds to $\opaction{\mc{ST}^2}$, which is the (normalized) action of the quaternity defect $\mathcal{S}\mathcal{T}^2$ on local operators.
As computed in Section~\ref{sec: continuum fusion rules Z2}, it obeys the $\mathbb{Z}_4$ fusion rule:
\be\label{eq: comparison cont Z2}
\opaction{(\mc{ST}^2)^4} = 1
\ee
% \begin{equation}
% (\alpha_{\mathcal{S}\mathcal{T}^2})^4 = 1.
% \label{eq: comparison cont Z2}
% \end{equation}
By comparing \eqref{eq: comparison lat Z2} and \eqref{eq: comparison cont Z2}, we find that the fusion rule on the lattice differs from its counterpart in the continuum.
Specifically, the lattice fusion rule involves a non-trivial QCA $\alpha_t^2 \circ \alpha_{\mathrm{framing}} \circ \alpha_{\mathrm{TW}}^2$, while the continuum fusion rule does not.
Thus, equation~\eqref{eq: comparison lat Z2} is an example of a lattice fusion rule mixing with a non-trivial QCA.
This result shows that the non-invertible symmetry generated by the corresponding operator $\mathsf{S}\mathsf{T}^2$ mixes with a non-trivial QCA.

We emphasize that the appearance of the 3-fermion framing QCA $\alpha_{\mathrm{framing}}$ on the lattice is not a coincidence.
It can be traced back to the relation
\begin{equation}
(ST^2)^4 = Y^4,
\end{equation}
where $Y^4$ is the stacking of the Crane-Yetter-Walker-Wang TQFT based on the 3-fermion MTC.
On the lattice, the above relation implies that $(\alpha_{\mathrm{KWW}} \circ \alpha_{\mathrm{TW}}^2)^4$ is a non-trivial QCA in the 3-fermion equivalence class,\footnote{More precisely, the QCA $(\alpha_{\mathrm{KWW}} \circ \alpha_{\mathrm{TW}}^2)^4$ on the $\mathbb{Z}_2$ 1-form symmetric algebra can be extended to a non-trivial QCA on the full algebra, which is in the 3-fermion equivalence class.} which entangles the ground state of the above TQFT from the trivial product state.
On the other hand, in the continuum, the above relation gives rise to the fusion rule~\eqref{eq: DU4 continuum defect}, which involves the 3-fermion TQFT as a fusion coefficient.
This TQFT coefficient disappears in~\eqref{eq: comparison cont Z2} because the decoupled 2+1d TQFT act trivially on any local operator.

\section{QCA for $\Z_2^{(1)}$ SPT Entangler}\label{sec:qca_for_root_entangler}

In this section we are going to write down an SPT entangler for the $\Z_2^{(1)}$ symmetry in 3+1d, in its unitary form we call it $\sf{T}$, precisely because the order 2 entangler introduced in the previous section was $\sf{T}^2$.

\subsection{Preliminaries}

Let us introduce some notation and definitions.

We define $C^n_k \coloneqq C^n(M_3; \Z_k)$ the space of cochains of a triangulated/cellulated 3-manifold $M$, and similarly for $B^\bullet, Z^\bullet$. We omit the space manifold, because it will always be the same one. Let us also write $C^n \coloneqq C^n(M_3; \Z)$.

From the reduction modulo $k, \Z \rightarrow \Z_k$, we have the homomorphism
\be
\lbb - \rbb_k : C^n \rightarrow C^n_k \,.
\ee

% In order to be as explicit as possible at every step, we are also going to introduce the notation $\delta_k : C^n_k \rightarrow C^{n+1}_k$, so that it is always clear which coefficients we are working with. Similarly, $\delta : C^n \rightarrow C^{n+1}$, so no subscript means the coefficients are integers. It can be checked that the following is true:
We write $\delta$ for the coboundary. One can check that the following is true:
\be
\delta \lbb \bm{c} \rbb_k = \lbb \delta \bm{c} \rbb_k \,, \qquad \forall \bm{c} \in C^n_k \,.
\ee

Next, we need to introduce a partial ``inverse'' to the reduction mod $k$, the integer lift. We define
\be
\lbbb - \rbbb_k : C^n_k \rightarrow C^n \,,
\ee
with the aid of the following identification: $\Z_k \equiv \lbbb 0, 1, ..., k - 1 \rbbb \subset \Z$. Let $\bm{c} \in C^n_k$ be a cochain. Firstly, note that $C^n \coloneqq \operatorname{Hom}(C_n, \Z)$, where $C_n$ is a free abelian group (here $C_n$ denotes the chain complex, not coefficients mod $n$. The notation is ambiguous, but fortunately we will not need to reference $C_n$ anymore afterwards, so this remark applies only to this small section; this and the next few paragraphs). Therefore, it is enough to specify the image $\lbbb \bm{c} \rbbb$ on the generators, and furthermore, by freeness, any such assignment is valid and gives a homomorphism.

Hence, let $\sigma^n$ be an $n$-cell or $n$-simplex (such as what appears in the definition $C_n \coloneqq \bigoplus_{\sigma^n} \Z$). Define
\be
\lbbb \bm{c} \rbbb_k (\sigma^n) \coloneqq \lbbb \bm{c}(\sigma^n) \rbbb_k \,,
\ee
where, on the right hand side, we use the identification $\Z_k \subset \Z$, as explained earlier.

This is a lifting map, in the sense that
\be
\lbb \lbbb \bm{c} \rbbb_k \rbb_k = \bm{c} \,.
\ee

But unfortunately it is much less well behaved than the reduction mod $k$. It is not a homomorphism, and it does not play well with coboundaries. In fact, let us codify the way in which those properties fail. Let $\bm{a}, \bm{b} \in C^n_k$ be cochains. $\lbbb - \rbbb_k$ is not a homomorphism, and thus
\be
\lbbb \bm{a} + \bm{b} \rbbb_k - \lbbb \bm{a} \rbbb_k - \lbbb \bm{b} \rbbb_k \,,
\ee

is not 0. Reduce the previous expression mod $k$:
\begin{multline}
 \lbb \lbbb \bm{a} + \bm{b} \rbbb_k - \lbbb \bm{a} \rbbb_k - \lbbb \bm{b} \rbbb_k \rbb_k \\
 = \lbb \lbbb \bm{a} + \bm{b} \rbbb_k \rbb_k - \lbb \lbbb \bm{a} \rbbb_k \rbb_k - \lbb \lbbb \bm{b} \rbbb_k \rbb_k \\
 = \bm{a} + \bm{b} - \bm{a} - \bm{b} = 0 \,.
\end{multline}

Therefore, the following is well defined:
\be
\Sigma(\bm{a}, \bm{b}) \coloneqq \frac{\lbbb \bm{a} + \bm{b} \rbbb_k - \lbbb \bm{a} \rbbb_k - \lbbb \bm{b} \rbbb_k}{k} \,,
\ee

as a map $C^\bullet_k \times C^\bullet_k \rightarrow C^\bullet$. It is the \emph{sum fixing map}.

Similarly, define
\be
\Delta(\bm{a}) \coloneqq \frac{\lbbb \delta \bm{a} \rbbb_k - \delta \lbbb \bm{a} \rbbb_k}{k} \,,
\ee

as a map $C^\bullet_k \rightarrow C^\bullet$, \emph{the coboundary fixing map}.

\subsection{Defining the SPT Entangler}

In this section we are going to define the SPT entangler as a unitary. Its action on operators will be computed later.

Like before, suppose we have $M_3$, a cellulated 3-manifold without boundary, and for now we assume the cellulation is finite. This will not work in the case of $\mathbb{R}^3$, which is what we are interested in, but we will see that nevertheless we can define an action on the local operators, which will give us a QCA on $\mc{A}$. We will also see that since the QCA is $\Z_2^{(1)}$-symmetric, this QCA induces a QCA on $\mc{A}_{\Z_2^{(1)}} / \mathfrak{I}$.

By analogy with the constructions from the previous section, we place a qubit on each face of the cellulation. The Hilbert space has basis
\be
\ket{\bm{b}} \coloneqq \bigotimes_{f \text{ face}} \ket{\bm{b}(f)} \,, \quad \text{for } \bm{b} \in C^2_2
\ee

The general form of the entangler is the following unitary
\be
\sf{T}^m \coloneqq \sum_{\bm{b} \in C^2_2} e^{\frac{2 \pi i m}{4} \int_{M_3} \phi_3[\bm{b}]} \ket{\bm{b}} \bra{\bm{b}} \,,
\ee

where $\phi_3$ is some map $C^2_2 \rightarrow C^3_4$. It is also immediately clear that
\be\label{eq:T4action}
\sf{T}^4 = 1 \,.
\ee

As it turns out, it is much easier to write this map when working on the dual cellulation; that which has qubits on the edges instead of on the faces.

Based on the results in \cite{Tsui:2019ykk}, we propose
\be\label{eq:phi3}
\phi_3[\bm{\hat{a}}] = \lbbb \bm{\hat{a}} \rbbb_2 \smile \delta \lbbb \bm{\hat{a}} \rbbb_2 \pm \delta \lbbb \bm{\hat{a}} \rbbb_2 \smile_1 \lbbb \delta \bm{\hat{a}} \rbbb_2 \,,
\ee
where $\bm{\hat{a}}$ is a 1-cochain on the dual cellulation. Trivially due to the duality, 2-cochain $\bm{b}$ is in bijection with dual 1-cochains $\bm{\hat{a}}$, which is why we can work on the dual cellulation. Unfortunately it is unclear whether this construction could be translated to the direct lattice; in particular, the cup products are not easy to translate because all the degrees would end up mismatched.

In order to ensure that this unitary is $\Z_2^{(1)}$-symmetric, we need to check that $\phi_3$ is gauge invariant (unchanged under $\bm{\hat{a}} \rightarrow \bm{\hat{a}} + \delta \bm{\hat{s}}$), or rather, that the resulting phase is invariant. That means that $\phi_3[\bm{\hat{a}} + \delta \bm{\hat{s}}] - \phi_3 [\bm{\hat{a}}]$ should be a coboundary, mod 4, or $\phi_3[\bm{\hat{a}} + \delta \bm{\hat{s}}] - \phi_3 [\bm{\hat{a}}] \in B^3_4$.

Let us remark that the sign in the definition does not actually matter. Indeed, the difference between the two signs is:
\begin{align}
 2 \delta \lrb{\bm{\hat{a}}} \smile_1 \lrb{\delta \bm{\hat{a}}} &= 2 \delta \lrb{\bm{\hat{a}}} \smile_1 (\delta \lrb{\bm{\hat{a}}} + 2 \Delta(\bm{\hat{a}})) \nonumber \\
 &= 2 \delta \lrb{\bm{\hat{a}}} \smile_1 \delta \lrb{\bm{\hat{a}}} + 4 (\dots) \,.
\end{align}

We remark that from now on we omit the subscript $\{\}_2$ when there is no ambiguity. Now, note that by the Leibniz rule we have
\begin{align}
 \delta (\delta \lrb{\bm{\hat{a}}} \smile_2 \delta \lrb{\bm{\hat{a}}}) &= (-1)^{2 + 2 - 2} \delta \lrb{\bm{\hat{a}}} \smile_1 \delta \lrb{\bm{\hat{a}}} \nonumber \\
 &\hspace{2em} + (-1)^{2 \cdot 2 + 2 + 2} \delta \lrb{\bm{\hat{a}}} \smile_1 \delta \lrb{\bm{\hat{a}}} \nonumber\\
 &= 2 \delta \lrb{\bm{\hat{a}}} \smile_1 \delta \lrb{\bm{\hat{a}}} \,,
\end{align}
from which we conclude
\be
2 \delta \lrb{\bm{\hat{a}}} \smile_1 \lrb{\delta \bm{\hat{a}}} = \delta (\dots) + 4(\dots) \,.
\ee

\subsection{Gauge invariance of \texorpdfstring{$\phi_3[\bm{\hat{a}}]$}{$\phi_3[\hat{a}]$}}

In this section we are going to verify that Eq.~\eqref{eq:phi3} is gauge invariant (still working on the dual lattice). That is, that $\phi_3$ is unchanged under $\bm{\hat{a}} \to \bm{\hat{a}} + \delta \bm{\hat{s}}$, at least up to coboundaries and multiples of 4. This is what ensures that the phase defined by $\frac{\pi i}{2} \int \phi_3[\bm{\hat{a}}]$ is $\Z_2^{(1)}$-symmetric.

Let $\bm{\hat{a}} \in C^1_2$, $\bm{\hat{s}} \in C^0_2$ be cochains. Define $A \coloneqq \lbbb \hat{a} \rbbb_2, S \coloneqq \lbbb \bm{\hat{s}} \rbbb_2$. We use uppercase letters to distinguish integer cochains from the ones with $\Z_2$ coefficients. However, despite the fact that they are not bold, and do not have a hat $\hat{}$, they are still cochains on the dual lattice. Then, we note that
\be
\lrb{\bm{\hat{a}} + \delta \bm{\hat{s}}} = A + \delta S + 2 (\Delta(\bm{\hat{s}}) + \Sigma(\bm{\hat{a}}, \delta \bm{\hat{s}})) \,,
\ee
and for convenience we define $P \coloneqq \Delta(\bm{\hat{s}}) + \Sigma(\bm{\hat{a}}, \delta \bm{\hat{s}})$.

For the first term:
\begin{multline}
 \lbbb \bm{\hat{a}} + \delta \bm{\hat{s}} \rbbb \smile \delta \lbbb \bm{\hat{a}} + \delta \bm{\hat{s}} \rbbb \\
 = \lb A + \delta S + 2 P \rb \smile \lb \delta A + 2 \delta P \rb
\end{multline}

For the second term:
\be
\delta \lrb{\bm{\hat{a}} + \delta \bm{\hat{s}}} \smile_1 \lrb{\delta \bm{\hat{a}}} \\
= \lb \delta A + 2 \delta P \rb \smile_1 \lrb{\delta \bm{\hat{a}}}
\ee

By expanding everything in $\phi_3[\bm{\hat{a}} + \delta \bm{\hat{s}}]$ (and taking the positive sign), we get
\begin{multline}
 \phi_3[\bm{\hat{a}} + \delta \bm{\hat{s}}] = A \smile \delta A + \delta S \smile \delta A \\
 + 2 (A \smile \delta P + \delta S \smile \delta P + P \smile \delta A) \\
 + \delta A \smile_1 \lrb{\delta \bm{\hat{a}}} + 2 \delta P \smile_1 \lrb{\delta \bm{\hat{a}}}
\end{multline}

We observe that
\be
\delta (S \smile \delta A) = \delta S \smile \delta A \,,
\ee
and similarly $\delta S \smile \delta P$ is also a coboundary. If we subtract $\phi_3[\bm{\hat{a}}]$, we get
\begin{multline}
 \phi_3[\bm{\hat{a}} + \delta \bm{\hat{s}}] - \phi_3[\bm{\hat{a}}] = 2 (A \smile \delta P + P \smile \delta A \\
 + \delta P \smile_1 \lrb{\delta \bm{\hat{a}}}) + \delta (\dots)
\end{multline}

Since everything that is left is multiplied by 2, and the calculations are being done mod 4, we no longer need to choose lifts. We also use the Leibniz rule again, this time mod 2:
\be
\delta (P \smile_1 \delta A) \equiv \delta P \smile_1 \delta A + P \smile \delta A + \delta A \smile P
\ee

Therefore
\begin{multline}
 \phi_3[\bm{\hat{a}} + \delta \bm{\hat{s}}] - \phi_3[\bm{\hat{a}}] = 2 (A \smile \delta P + \delta A \smile P) + \delta,4 (\dots)
\end{multline}

Finally, we apply the Leibniz rule one last time:
\be
\delta (A \smile P) = \delta A \smile P - A \smile \delta P \,,
\ee
to show that what remains is a coboundary, up to multiples of 4. That is, $\phi_3$ is gauge invariant, up to coboundaries and multiples of 4.

\subsection{Action on Operators}

For now, we keep the simplified notation of uppercase letters being the integer lifts of cochains. We also still drop the subscript in $\lrb{}_2$, because the lift will always be mod 2, at least until Section \ref{sec: QCAs on Zp 1-form symmetric algebra}. All the calculations will be implicitly done mod 4, and up to coboundaries, so we omit those terms as well. Lastly, we will often omit $\smile$, especially in the middle of long calculations, and simply write $\alpha \beta$ for $\alpha \smile \beta$.

As shown earlier, the unitary which describes the generating SPT entangler is
\be
\sf{T} \coloneqq \sum_{\bm{\hat{a}}} e^{\frac{\pi i}{2} \int \phi_3[\bm{\hat{a}}]} \ket{\bm{\hat{a}}} \bra{\bm{\hat{a}}}
\ee
It is generating, in the sense that $\sf{T}^k (-) (\sf{T}^\dag)^k$ for $k = 0, \cdots, 3$ are the SPT entanglers in the $\Z_4$ classification, $\sf{T}^4 = 1$, and $\sf{T}^2$ is the one constructed explicitly in \cite{Tsui:2019ykk}.

It is clear that, since $\sf{T}$ is diagonal in the computational ($Z$) basis,
\be
\sf{T} Z_f \sf{T}^\dag = Z_f
\ee
for any face $f$. Therefore, on the $\Z_2^{(1)}$-symmetric subalgebra, it leaves the $Z_{\delta \bm{e}} = \prod_{f \, \mid \, e \in \partial f } Z_f$ generators invariant.

On the other hand, upon conjugating $X_f$, we see that
\begin{align}
\sf{T} X_f \sf{T}^\dag &= \sum_{\bm{\hat{a}},\bm{\hat{b}}} e^{\frac{\pi i}{2} \int \phi_3[\bm{\hat{b}}] - \phi_3[\bm{\hat{a}}]} \ket{\bm{\hat{b}}} \bra{\bm{\hat{b}}} X_f \ket{\bm{\hat{a}}} \bra{\bm{\hat{a}}} \\
&= \sum_{\bm{\hat{a}}} e^{\frac{\pi i}{2} \int \phi_3[\bm{\hat{a}} + \bm{\hat{f}}] - \phi_3[\bm{\hat{a}}]} \ket{\bm{\hat{a}} + \bm{\hat{f}}} \bra{\bm{\hat{a}}} \\
&= X_f \sum_{\bm{\hat{a}}} e^{\frac{\pi i}{2} \int \phi_3[\bm{\hat{a}} + \bm{\hat{f}}] - \phi_3[\bm{\hat{a}}]} \ket{\bm{\hat{a}}} \bra{\bm{\hat{a}}}
\end{align}

Here, $\bm{\hat{f}}$ is to be understood as a dual 1-cochain ($\Z_2$-valued) which takes the value 1 on the dual edge $\bm{\hat{f}}$, which is dual to the face $f$, and 0 elsewhere. In the second line, we used
\be
\bra{\bm{\hat{b}}} X_f \ket{\bm{\hat{a}}} = \langle \bm{\hat{b}} \mid \bm{\hat{a}} + \bm{\hat{f}} \rangle = \delta_{\bm{\hat{b}}, \bm{\hat{a}}+\bm{\hat{f}}} \,,
\ee
which after summing over $\bm{\hat{b}}$, simply sets $\bm{\hat{b}} = \bm{\hat{a}} + \bm{\hat{f}}$. In the third line, we again used $\ket{\bm{\hat{a}} + \bm{\hat{f}}} = X_{\bm{\hat{f}}} \ket{\bm{\hat{a}}}$, and then pulled out that operator in front of the summation.

Now, the next step is to rewrite $\phi_3[\bm{\hat{a}} + \bm{\hat{f}}] - \phi_3[\bm{\hat{a}}]$, up to multiples of 4, and up to coboundaries, as some local expression. This is necessary for $\sf{T} (-) \sf{T}^\dag$ to be a QCA, because otherwise it would not have bounded spread.

\begin{claim}
The phase acquired by $X_f$ under conjugation by $T$ can be written as
\begin{multline}
 \phi_3[\bm{\hat{a}}+\bm{\hat{f}}]-\phi_3[\bm{\hat{a}}] 
 \equiv \delta{\{\bm{\hat{f}}}\}\smile_1 \lrb{\delta \bm{\hat{a}}} + 2 \bm{\hat{f}} \smile \delta \bm{\hat{a}} \\
 + 2 \left( \delta{\bm{\hat{a}}}+\delta{\bm{\hat{f}}} \right) \smile_1 \left( \Delta(\bm{\hat{f}})+\Sigma(\delta{\bm{\hat{a}}},\delta{\bm{\hat{f}}}) \right)
\end{multline}

This is not strictly an equality of cochains, but they are equal up to multiples of 4, and up to coboundaries, so when integrated over a manifold without boundaries as $\frac{\pi i}{2} \int \dots$, they give the same phase (same integer multiple of $\frac{\pi i}{2}$).
\end{claim}
This is a localised phase in the sense that it only depends on the values of $\bm{\hat{a}}$ near $f$. To verify this, note that in each cup product, at least one of the cochains is non-zero only nearby $f$. Hence, the cup product itself is also non-zero only near $f$. For example, $\delta \{ \bm{\hat{f}} \}$ is straightforward because it is supported only on the four dual faces around a dual edge\footnote{which actually correspond to the dual of $\partial f$. Indeed, one can check that $\widehat{\partial a} = \delta \hat{\bm{a}}$, and $\widehat{\delta \bm{a}} = \partial \hat{a}$, abusing the usual identification between chains and cochains.}, and the corresponding cup-1 product is non-zero only on the 4 dual cubes surrounding the dual edge. Similarly, one can use the definition of $\Sigma$ to verify that $\Sigma(\delta \bm{\hat{a}}, \delta \bm{\hat{f}})$ is 0 everywhere outside those same 4 dual faces, and so on.

This is why we can conclude that the image of $X_f$ under this QCA is again $X_f$, times some localized, diagonal operator.

{\bf Proof:}
Recall that
\begin{align}
\phi_3[\bm{\hat{a}}] &= \lrb{\bm{\hat{a}}} \smile \delta \lrb{\bm{\hat{a}}} + \delta \lrb{\bm{\hat{a}}} \smile_1 \lrb{\delta \bm{\hat{a}}} \,, \\
\lrb{\delta \bm{\hat{a}}} &= \delta \lrb{\bm{\hat{a}}} + 2 \Delta(\bm{\hat{a}}) \,.
\end{align}

Like before, we are still working with cup products on the dual lattice. Let us define
\begin{align}
 A &\coloneqq \lrb{\bm{\hat{a}}} \,, & F &\coloneqq \lrb{\bm{\hat{f}}} \,. 
\end{align}

We have the following:
\begin{align}
 \lrb{\bm{\hat{a}} + \bm{\hat{f}}} &= A + F + 2 \Sigma(\bm{\hat{a}}, \bm{\hat{f}}) \,, \\
 \lrb{\delta \bm{\hat{a}} + \delta \bm{\hat{f}}} &= \lrb{\delta \bm{\hat{a}}} + F + 2\Delta(\bm{\hat{f}}) + 2 \Sigma(\delta \bm{\hat{a}}, \delta \bm{\hat{f}}) \,.
\end{align}

For brevity, we also define $d A \coloneqq \lrb{\delta \bm{\hat{a}}}$, and $P \coloneqq \Delta(\bm{\hat{f}}) + \Sigma(\delta \bm{\hat{a}}, \delta \bm{\hat{f}})$. Then we may define unambiguously $\Sigma \coloneqq \Sigma(\bm{\hat{a}}, \bm{\hat{f}})$, because the other appearance of $\Sigma(\delta \bm{\hat{a}}, \delta \bm{\hat{f}})$ has been absorbed into $P$.
\begin{align}
&\phi_3[\bm{\hat{a}} + \bm{\hat{f}}]  \nonumber \\
&= \lbbb \bm{\hat{a}} + \bm{\hat{f}}\rbbb \delta \lbbb \bm{\hat{a}} +\bm{\hat{f}} \rbbb + \delta \lbbb \bm{\hat{a}}+\bm{\hat{f}} \rbbb \smile_1 \lbbb \delta \bm{\hat{a}} + \delta \bm{\hat{f}} \rbbb \nonumber \\
&= (A + F + 2 \Sigma)(\delta A + \delta F + 2 \delta \Sigma) \nonumber \\
&\hspace{2em}+ (\delta A + \delta F + 2 \delta \Sigma) \smile_1 (dA + \delta F + 2P) \,. % \nonumber \\
%&= (\lbbb a \rbbb + \lbbb l \rbbb + 2\Sigma(a,l))\cup(\delta \lrb{a} + \delta \lrb{l} + 2\delta\Sigma(a,l)) \nonumber \\
%&+ (\delta \lrb{a} + \delta \lrb{l} + 2\delta\Sigma(a,l)) \cup_1 \nonumber \\
%&\,\,\,\,\,\,\,\,(\delta\lrb{a} + \delta{\lrb{l}} +2\Delta(a) + 2\Delta(l)+ 2\Sigma(\delta a,\delta l)) 
\end{align}

We now expand out the above. Let us first focus on the odd terms, i.e. those without a coefficient of $2$ in front of them:
\begin{multline}
A \delta A + A \delta F + F \delta A + F \delta F + \delta A \smile_1 dA \\
+ \delta A \smile_1 \delta F + \delta F \smile_1 dA + \delta F \smile_1 \delta F \,.
\end{multline}

We now use the Leibniz rules for $\smile_1$ and $\smile$, as follows
\begin{align}
\delta (A \smile_1 \delta F) &= \delta A \smile_1 \delta F + A \delta F - (\delta F) A \,, \\
\delta (F A) &= (\delta F) A - F \delta A \,.
\end{align}

We use them to show that
\begin{multline}\label{eq:intermediate_cochain_result}
 A \delta F + F \delta A + \delta A \smile_1 \delta F \\
 = F \delta A + (\delta F) A + \delta (A \smile_1 \delta F) \\
 = 2 F \delta A + \delta (A \smile_1 \delta F + FA) \,.
\end{multline}

We use this result to rewrite the odd terms:
\begin{multline}
 \phi_3[\bm{\hat{a}}] + \delta F \smile_1 dA + 2 F \delta A \\
 + \paren{F \delta F + \delta F \smile_1 \delta F} + \delta (\dots) \,.
\end{multline}

Now, we move on to the even terms, where we drop the factor of 2 and simply do all calculations modulo 2:
\begin{multline}
 A \delta \Sigma + F \delta \Sigma + \Sigma \delta A + \Sigma \delta F + \delta A \smile_1 P \\
 + \delta F \smile_1 P + \delta \Sigma \smile_1 dA + \delta \Sigma \smile_1 \delta F
\end{multline}

Firstly, let us remark that $dA = \delta A + 2 (\dots)$, so in this calculation we can simply replace one with the other. Secondly, notice how we already proved the following equalities previously (see Eq.~\eqref{eq:intermediate_cochain_result}):
\begin{align}
 A \delta \Sigma + \Sigma \delta A + \delta \Sigma \smile_1 \delta A &= 2 \Sigma \delta A + \delta (\dots) \\
 F \delta \Sigma + \Sigma \delta F + \delta \Sigma \smile_1 \delta F &= 2 \Sigma \delta F + \delta (\dots)
\end{align}

Therefore, all these terms cancel out, up to multiples of 2 and up to coboundaries. The even terms become:
\be
(\delta A + \delta F) \smile_1 P + 2(\dots) + \delta(\dots) \,.
\ee

Finally, we can conclude that
\begin{align}
 \phi_3&[\bm{\hat{a}} + \bm{\hat{f}}] - \phi_3[\bm{\hat{a}}] \nonumber \\
 &= \delta F \smile_1 dA + \paren{F \delta F + \delta F \smile_1 \delta F} \nonumber \\
 &\; + 2 \paren{F \delta A + \paren{\delta A + \delta F} \smile_1 P} + 4 (\dots) + \delta (\dots) \,,
\end{align}
which is essentially the result we wanted to prove. The missing step is to show that $F \smile \delta F + \delta F \smile_1 \delta F$ vanishes identically. Like before, using the Leibniz rule, one can reduce this to showing that $\delta F \smile F = 0$. This can be seen by direct computation (either on the cubic lattice, or on any triangulation), by using the fact that $F$ is (the integer lift of) a single dual edge.

Any of the explicit formulae will give $\delta F \smile F = 0$ identically for any 3-cell (cube, or tetrahedron).

We have thus shown that, after undoing all the definitions that were made for the sake of brevity and simplicity of the notation:
\begin{multline}
 \phi_3[\bm{\hat{a}}+\bm{\hat{f}}]-\phi_3[\bm{\hat{a}}] 
 \equiv \delta{\{\bm{\hat{f}}}\}\smile_1 \lrb{\delta \bm{\hat{a}}} + 2 \bm{\hat{f}} \smile \delta \bm{\hat{a}} \\
 + 2 \left( \delta{\bm{\hat{a}}}+\delta{\bm{\hat{f}}} \right) \smile_1 \left( \Delta(\bm{\hat{f}})+\Sigma(\delta{\bm{\hat{a}}},\delta{\bm{\hat{f}}}) \right)
\end{multline}

Let us add some additional remarks:
\begin{itemize}
 \item in any term a multiple of 2, the integer lift $\lrb{}$ can be neglected, because effectively everything ends up getting reduced mod 2 anyway,

 \item as noted earlier, this expression is clearly local, in that it depends on the values of $\bm{\hat{a}}$ only near $f$,

 \item the action of the SPT entangler is trivial on the Pauli $Z$ operators, and maps
 \be
 X_f \mapsto X_f \cdot (\text{local diagonal operator})
 \ee
\end{itemize}

\subsection{Group Structure: KWW and SPT Entangler}

In this section, we want to determine some of the group relations that the QCA constructed previously (Kramers-Wannier-Wegner, and the Tsui-Wen SPT entangler) satisfy.
In particular, mirroring the relations in the $\Witt^\pt (A,s)$ derived in Section \ref{sec:Cats}, we expect
\be
(ST)^3 = Y \,,
\ee
with $S$ symbolizing the KWW duality, $T$ the SPT entangler, and where $Y$ is the semion QCA.

Due to this analogy, we will write $S \coloneqq \alpha_\KWW$ and $T \coloneqq \alpha_\TW$, even though they are QCA on the quotient algebra.

Using some physically motivated assumptions, we will show that $(S \circ T)^3$ is equivalent to the one in  \cite{Shirley_2022}, which entangles the commuting projector for the semion Walker-Wang model, and has order 8: 
\be
Y^8=1 \,.
\ee
For lack of exact equality, we will have to argue it in a roundabout way.

\subsubsection{Separators and Flippers}

The arguments in this section will rely on the concepts of separators and flippers. Let us briefly introduce those notions, where for simplicity we restrict to separators in the full tensor product algebra $\mathcal{A}$.

\begin{definition}
A \emph{set of separators} is a set $\mathcal{S}$ of operators such that:
 \begin{itemize}
  \item For any $A \in \mathcal{S}$, $A^2 = 1$

  \item Any pair $A_1, A_2 \in \mathcal{S}$ commute: $[A_1, A_2] = 0$

  \item For any collection of $\{ \epsilon_A \}_{A \in \mathcal{S}}$, each sign $\epsilon_A \in \{ \pm 1 \}$, the space of states $\ket{\psi}$ satisfying $A \ket{\psi} = \epsilon_A \ket{\psi}$ is one-dimensional.\footnote{This last condition is not quite well defined on an infinite lattice, due to the lack of a Hilbert space. It is possible that it could be made precise by working with states as positive linear functionals, but we will ignore this subtlety, especially since it would also require working in the $C^*$-algebra $\overline{\mc{A}}$.}
 \end{itemize}

A set of separators is called \emph{local} if all the separators have uniformly bounded support. That is, there is some constant, such that the supports of all the operators in $\mathcal{S}$ are smaller than that constant.
\end{definition}

In a sense, separators parametrize the set of states of a system. In particular, the corresponding commuting projector Hamiltonian
\be
H \coloneqq - \sum_{A \in \mathcal{S}} \frac{1 - A}{2} \,,
\ee
has a unique ground state. But furthermore, each state can be uniquely specified according to its sequence of eigenvalues $\epsilon_A$ (uniquely, up to complex phases).

\begin{definition}
 A set of local separators $\mc{S}$ is called \emph{flippable} if there exists a set of flippers $\{F_A\}_{A \in \mc{S}}$, parametrized by $\mc{S}$, which satisfy:
 \begin{itemize}
  \item $F_A A = - A F_A$ for any $A \in \mc{S}$

  \item For any $A_1 \neq A_2 \in \mc{S}$, $[F_{A_1}, A_2] = 0$
 \end{itemize}

 It is called \emph{locally flippable} if there exists a set of flippers whose supports are uniformly bounded.
\end{definition}

Flippers are called that because they flip the sign of some eigenvalue. That is, suppose $\ket{\{\epsilon_A\}}$ is the (unique up to a phase) state such that
\be
A' \ket{\{\epsilon_A\}} = \epsilon_{A'} \ket{\{\epsilon_A\}} \,, \quad \forall A' \text{ separator.}
\ee

Then we also have
\be
A' \paren{F_{A'} \ket{\{\epsilon_A\}}} = - F_{A'} A' \ket{\{\epsilon_A\}} = - \epsilon_{A'} \paren{F_{A'} \ket{\{\epsilon_A\}}}
\ee

A prototypical example of a locally flippable set of local separators is the set of $X$ operators, one for each qubit (assuming the system only has qubits). Trivially, the unique state with some set of eigenvalues is a tensor product of $\ket{\pm}$. The local flippers are the $Z$ operators. For a given vertex $v$, clearly $X_v Z_v = - Z_v X_v$. Furthermore, for any two different vertices, they commute.

This standard example also satisfies an additional property: the flippers commute among themselves, though this is not necessary from the definition alone.

However, it has been shown that given a locally flippable set of local separators, one can always find a set of local flippers which also commute between each other \cite{HaahFidkowskiHastings2023}.

This is a fairly abstract way of constructing and thinking about QCA. For simplicity, assume we have a system of qubits on some vertices $v$. Suppose we are given some locally flippable set of local separators $\mc{S} = \{A_v\}_v$, where the separator labeled by a vertex $v$ is supported nearby $v$. By the theorem, we know there exists a set of mutually commuting local flippers $\{\tilde{F}_v \}_v$, each supported near its own vertex. Then the following mapping defines a QCA:
\begin{align}
 X_v &\longmapsto A_v \\
 Z_v &\longmapsto \tilde{F}_v
\end{align}

Intuitively speaking, having a separator ensures that there are enough operators on the right-hand-side for this kind of mapping to be bijective. Meanwhile, the existence of flippers ensures that this mapping can be extended to the whole algebra. The locality of both the separators and the flippers ensures that this is an actual QCA.

Crucially for the discussion later on, one does not actually need the mutually commuting flippers to be able to say things about the resulting QCA. Finding any (local) flippers can be sufficient, and those can potentially be much simpler than their commuting counterparts.

\subsubsection{Separators for $STSTS$}

One part of this computation actually simplifies, because
\be
(S \circ T \circ S) (X_f + \mathfrak{I}) = X_{t(f)} + \mathfrak{I} \,.
\ee
This is because $S$ maps $X$ operators to $Z$ operators, $T$ leaves $Z$ invariant, and $S$ maps it again to $X$ operators, now translated by one full site along the diagonal. Therefore:
\be
(S \circ T \circ S \circ T \circ S) (X_f) = S \big(T (X_{t(f)}) \big) \,,
\ee
which, up to translations, amounts to applying $S$ to
\be
X_f \sum_{\bm{\hat{a}}} e^{\frac{\pi i}{2} \int \phi_3[\bm{\hat{a}} + \bm{\hat{f}}] - \phi_3[\bm{\hat{a}}]} \ket{\bm{\hat{a}}} \bra{\bm{\hat{a}}} \,.
\ee

Using the terminology introduced earlier, we can see that finding the separators is going to be much easier than finding the flippers.

% Note that, using Eq.~\eqref{eq:phase_in_TXTdag} this expression is written entirely in terms of $\delta{a}$ (importantly, there are no more terms containing $\delta{\lrb{a}}$).  This fact is important for applying the $S$-modular transformation (i.e. KWW), after applying the $T$ SPT-entangler,  because the $S$ transformation tells you how $Z Z Z^\dagger Z^\dagger$ around a plaquette transforms, not individual $Z$'s, which is precisely $\delta{a}$.  Note that everything really just depends on $\delta{a}$ mod $2$; i.e. $\delta$ is taken with mod $2$ coefficients and then lifts are taken, not the other way around.

We will now show the following: applying $S$ results in replacing $\delta \bm{\hat{a}}$ by a dual 2-cochain $\bm{\hat{b}}$ as follows
\begin{claim}
\begin{align}\label{STXST}
\overline{X}'_f& \coloneqq S \big(T (X_f) \big)  \nonumber \\
% =&\left(\prod_{l'\in \partial \dot{l}} Z_{l'}\right) \exp\left(\frac{\pi i}{2}\int \delta \lrb{l} \cup_1 \lrb{b}\right) \cdot  \\
=&Z_{\delta t^{\frac{1}{2}}(\bm{f})} \sum_{\bm{\hat{b}}} e^{\frac{\pi i}{2}\int \delta \{ \bm{\hat{f}} \} \smile_1 \{\bm{\hat{b}}\} } e^{\pi i \int \bm{\hat{f}} \smile \bm{\hat{b}}} \cdot  \nonumber \\
%&\cdot \exp \left(\pi i \left(l\cup b + (b+\delta l)\cup_1 (\Delta(l)+\Sigma(b,\delta l))\right)\right)\,,
&\cdot e^{\pi i \int (\bm{\hat{b}} + \delta \bm{\hat{f}}) \smile_1 \paren{\Delta(\bm{\hat{f}}) +\Sigma(\bm{\hat{b}}, \delta \bm{\hat{f}})} } \ket{(-1)^{\bm{\hat{b}}}} \bra{(-1)^{\bm{\hat{b}}}} \,,
\end{align}
where $\bm{\hat{b}}$ is a dual $2$-cochain. The notation $\ket{(-1)^{\bm{\hat{b}}}}$ is meant to indicate that the state is written in the $\ket{\pm}$ basis:
\be
\ket{(-1)^{\bm{\hat{b}}}} \coloneqq \bigotimes_{f \in F} \ket{(-1)^{\bm{\hat{b}}(t^{-\frac{1}{2}}(\hat{f}))}} \,.
\ee

This is in contrast with
\be
\ket{\bm{\hat{a}}} \coloneqq \bigotimes_{f \in F} \ket{\bm{\hat{a}}(\hat{f})} \,.
\ee
\end{claim}

\noindent
{\bf Proof.} This is actually a more general result, in the sense that applying the Kramers-Wannier-Wegner duality to anything written purely in terms of $\delta \bm{\hat{a}}$ (without any dependence on $\bm{\hat{a}}$ itself) is straightforward.

The trick is to notice that $\bm{\hat{a}}$ can be identified with an operator as follows:
\be
\bm{\hat{a}}(\hat{f}) \longleftrightarrow \frac{1}{2}(1 - Z_f) \,,
\ee
for any face $f \in F$. Actually, in the most literal way, we have
\be\label{eq:equivalence_of_1cochain_and_op}
\lrb{\bm{\hat{a}}(\hat{f})} \ket{\bm{\hat{a}}} = \frac{1}{2} (1 - Z_f) \ket{\bm{\hat{a}}} \,,
\ee

The integer lift is only needed to ensure we identify $1 \in \Z_2$ with $1 \in \Z$. This kind of equality also extends to any expression involving $\bm{\hat{a}}$. Cup products become products of operators, integrals over cochains become sums of operators, etc.

In particular, we have
\be
\lrb{ \delta \bm{\hat{a}} (\hat{e}) } = \frac{1}{2} \paren{1 - \prod_{\hat{f} \in \partial \hat{e}} Z_f} \,.
\ee

Again, this equality can be checked by comparing the action of the operator on a state $\ket{\bm{\hat{a}}}$. Briefly, the LHS calculates the sum mod 2 of 4 dual edges ($\hat{e}$ is a square face on the dual lattice, which has 4 dual edges as its boundary). Meanwhile, on the RHS, the product of $Z$ operators takes the product of 4 $(\pm 1)$ terms, which ends up being effectively a mod 2 calculation.

The point now is that we know where KWW maps such a product of $Z$ operators:
\be
\prod_{\hat{f} \in \partial \hat{e}} Z_f = \prod_{f \in \delta \bm{e}} Z_f \rightsquigarrow X_{t^\frac{1}{2} (e)}
\ee

Now we must find an interpretation for an operator of the form
\be
\frac{1}{2} \paren{1 - X_{t^\frac{1}{2}(e)}} \,.
\ee

By analogy with Eq.~\eqref{eq:equivalence_of_1cochain_and_op}, we define a 2-cochain $\bm{\hat{b}}$ which will be in correspondence with the above operator:
\be
\bm{\hat{b}}(\hat{e}) \longleftrightarrow \frac{1}{2} \paren{1 - X_{t^\frac{1}{2}(e)}} \,.
\ee

This analogy can only work if we now work in the basis where $X$ is diagonal, the $\ket{\pm}$ basis. Define
\be
\ket{(-1)^{\bm{\hat{b}}}} \coloneqq \bigotimes_{f \in F} \ket{(-1)^{\bm{\hat{b}}(t^{-\frac{1}{2}}(\hat{f}))}} \,.
\ee

Observe that this implies that
\begin{align}
\frac{1}{2} \paren{1 - X_{t^\frac{1}{2}(e)}} \ket{(-1)^{\bm{\hat{b}}}}
&= \frac{1}{2} \paren{1 - (-1)^{\bm{\hat{b}}(\hat{e})}} \ket{(-1)^{\bm{\hat{b}}}} \nonumber \\
&= \lrb{\bm{\hat{b}}(\hat{e})} \ket{(-1)^{\bm{\hat{b}}}} \,,
\end{align}
which is what we wanted to show.

The end result of this, is that we can say:
\be
\sum_{\bm{\hat{a}}} K(\lrb{\delta \bm{\hat{a}}}) \ket{\bm{\hat{a}}} \bra{\bm{\hat{a}}} \overset{\text{KWW}}{\rightsquigarrow}
\sum_{\bm{\hat{b}}} K(\{\bm{\hat{b}}\}) \ket{(-1)^{\bm{\hat{b}}}} \bra{(-1)^{\bm{\hat{b}}}} \,,
\ee
for any cochain function $K$.

This is exactly what we see in the claim, and furthermore we have the simplification that in the terms which are being calculated mod 2 (those with a prefactor $\pi i$, not $\frac{\pi i}{2}$), the integer lifts end up being unnecessary.

Note that we have a choice for how to lift this to an operator on the full tensor product operator algebra: we can multiply by any function of the symmetry operators $\eta(\Sigma)$, and get the same operator in the quotient.

Another way of looking at it is that, by the previous rule, we should be allowed to perform the substitution
\be
\delta \delta \bm{\hat{a}} \rightsquigarrow \delta \bm{\hat{b}} \,.
\ee

But the LHS is identically 0, so we conclude that we can add any function of $\delta b$ inside the exponentials above. Actually, it is fairly easy to see that
\be
\lrb{\delta \bm{\hat{b}}(\hat{s})} \ket{(-1)^{\bm{\hat{b}}}} = \frac{1}{2} \paren{1 - X_{\partial t^{\frac{1}{2}}(s)}} \ket{(-1)^{\bm{\hat{b}}}} \,.
\ee

Clearly, the operator on the RHS is in the ideal $\mathfrak{I}$, so any product is still in the ideal, and its exponential is in $1 + \mathfrak{I}$. This is why the operator in the quotient remains unchanged, even though in the full algebra it is different.

With this in mind, we define:
% \begin{align}
% \overline{X}_l &\coloneqq ST X_l(ST)^\dagger = \left(\prod_{l'\in \partial \dot{l}} Z_{l'}\right)\cdot \\ % &\exp\left(\frac{\pi i}{2}\int \left(\delta \lrb{l} \cup_1 \lrb{b} + \delta \lrb{l} \cup_2 \lrb{\delta b}\right)\right) \\ &\exp \left(\pi i \left(l\cup b + (b+\delta l)\cup_1 (\Delta(l)+\Sigma(b,\delta l))\right)\right)
% \end{align}
\begin{align}
\overline{X}_f \coloneqq& S \big( T (X_f) \big) \nonumber \\ 
=& Z_{\delta t^\frac{1}{2}(f)} \sum_{\bm{\hat{b}}} e^{\frac{\pi i}{2} \int \delta \{\bm{\hat{f}}\} \smile_1 \bm{\hat{b}} + \delta \{\bm{\hat{f}}\} \smile_2 \{\delta \bm{\hat{b}}\}} \cdot \nonumber \\
&\cdot e^{\pi i \int \bm{\hat{f}} \smile \bm{\hat{b}} + (\bm{\hat{b}}+ \delta \bm{\hat{f}}) \smile_1 (\Delta(\bm{\hat{f}})+\Sigma(\bm{\hat{b}},\delta \bm{\hat{f}}))} \ket{\bm{\hat{b}}} \bra{\bm{\hat{b}}} \,.
\end{align}

We will show that this modification allows us to find a set of local flippers for the local separators $\overline{X}_f$.

\subsubsection{Flippers for $STSTS$}

In this section, we are going to derive a set of separators and associated flippers.

The separators and flippers can be written as
\begin{align} \label{eq:Xbar}
\overline{X}_f &\coloneqq
Z_{\delta t^\frac{1}{2}(\bm{f})} X_f \sum_{\bm{\hat{b}}} e^{\frac{\pi i}{2} \int g(f, \bm{\hat{b}})} \ket{(-1)^{\bm{\hat{b}}}} \bra{(-1)^{\bm{\hat{b}}}} \,, \\ 
\overline{Z}_p &\coloneqq Z_p \sum_{\bm{\hat{c}}} e^{\frac{\pi i}{2} \int h(p, \bm{\hat{c}})} \ket{(-1)^{\bm{\hat{c}}}} \bra{(-1)^{\bm{\hat{c}}}} \,.
\end{align}

Note that, while $g$ comes from the expression we derived in the Claim \ref{STXST}, there are some differences. Firstly, one can check explicitly that
\be
\int \bm{\hat{f}} \smile \bm{\hat{b}} \,,
\ee
is the term which contributes a single $X_f$. Secondly, as argued earlier, we can modify it by adding terms which depend on $\delta \bm{\hat{b}}$. To be precise:
\begin{multline}
 g(f, \bm{\hat{b}}) \coloneqq \delta \{ \bm{\hat{f}} \} \smile_1 \{ \bm{\hat{b}} \} + \delta \{\bm{\hat{f}} \} \smile_2 \{ \delta \bm{\hat{b}} \} \\
 + 2 (\bm{\hat{b}} + \delta \bm{\hat{f}}) \smile_1 \paren{\Delta(\bm{\hat{f}}) + \Sigma(\bm{\hat{b}}, \delta \bm{\hat{f}})}
\end{multline}

Note that only the $\smile_2$ term is new, but we have also removed $\bm{\hat{f}} \smile \bm{\hat{b}}$ because a single $X_f$ is already included in the front of the definition earlier.

\begin{claim}
 We claim that, if we denote $\hat{e} \coloneqq t^{-\frac{1}{2}}(\hat{p})$ (recall that the dual of an edge $e$ is a face on the dual lattice), then
 \begin{multline}
  h(p, \bm{\hat{c}}) \coloneqq - \{ \bm{\hat{e}} \} \smile_1 \{ \bm{\hat{c}} + \bm{\hat{e}} \} \\
  + 2 \sqpar{\bm{\hat{c}} \smile_1 \Sigma(\bm{\hat{c}}, \bm{\hat{e}}) + \bm{\hat{c}} \smile_2 \paren{\Delta(\bm{\hat{e}}) + \Sigma(\delta \bm{\hat{c}}, \delta \bm{\hat{e}})}}
 \end{multline}
 gives a flipper for the separators defined earlier.

 Since it only depends on $e$, for simplicity we will write $h(e, \bm{\hat{c}})$.
\end{claim}

\noindent
{\bf Proof.} Due to their definitions, $\overline{X}_f$ and $\overline{Z}_p$ already have a single $X_f$, $Z_p$ each, which provide the correct (anti)commutation relations: they commute for $f \neq p$, and anticommute for $f = p$. Therefore, what we must check is that everything else commutes.

If we write $\overline{X}_f \overline{Z}_p$ explicitly, we can see that in order to bring out the $Z_p$ to the front, we must modify $g$ to $g(f, \bm{\hat{b}} + \bm{\hat{e}})$. Similarly, if we write $\overline{Z}_p \overline{X}_f$, we need to pull out $Z_{\delta t^\frac{1}{2} (\bm{f})}$ to the front, which makes us modify $h$ to $h(p, \bm{\hat{b}} + \delta \bm{\hat{f}})$.

What remains is establishing the equivalence
\be
g(f, \bm{\hat{b}} + \bm{\hat{e}}) + h(e, \bm{\hat{b}}) \equiv g(f, \bm{\hat{b}}) + h(e, \bm{\hat{b}} + \delta \bm{\hat{f}})
\ee

The equivalence is to be understood, as always, up to coboundaries and up to multiples of 4. For this reason, at every step of the computation we are allowed to eliminate multiples of 4, and we often do so implicitly, without mentioning it.

For didactic purposes, let us suppose that we do not know the expressions for $g$ and $h$, and derive them from scratch.

Firstly, before the modification by $\delta \bm{\hat{b}}$, we have:
\begin{multline}
g'(f, \bm{\hat{b}}) \coloneqq \delta \{\bm{\hat{f}}\} \smile_1 \{ \bm{\hat{b}} \} + 2 \Big[ \bm{\hat{b}} \smile_1 \Delta(\bm{\hat{f}}) + \delta \bm{\hat{f}} \smile_1 \Delta(\bm{\hat{f}}) \\
+ \bm{\hat{b}} \smile_1 \Sigma(\bm{\hat{b}}, \delta \bm{\hat{f}}) + \delta \bm{\hat{f}} \smile_1 \Sigma(\bm{\hat{b}}, \delta \bm{\hat{f}}) \Big] \,.
\end{multline}

Let us define $B \coloneqq \{\bm{\hat{b}}\}$, $F \coloneqq \{\bm{\hat{f}}\}$, $E \coloneqq \{\bm{\hat{e}}\}$. Then
\be
\{\bm{\hat{b}} + \bm{\hat{e}}\} = B + E + 2 \Sigma(\bm{\hat{b}}, \bm{\hat{e}}) \,.
\ee

Substituting this change of $\bm{\hat{b}} + \bm{\hat{e}}$ into $g'$
\begin{multline}
 g'(f, \bm{\hat{b}} + \bm{\hat{e}}) = \delta F \smile_1 B + \delta F \smile_1 E \\
 + 2 \Big[ \delta F \smile_1 \Sigma(\bm{\hat{b}}, \bm{\hat{e}})
 + B \smile_1 \Delta(\bm{\hat{f}}) + E \smile_1 \Delta(\bm{\hat{f}}) \\
 + \delta F \smile_1 \Delta(\bm{\hat{f}}) + B \smile_1 \Sigma(\bm{\hat{b}} + \bm{\hat{e}}, \delta \bm{\hat{f}}) \\
 + E \smile_1 \Sigma(\bm{\hat{b}} + \bm{\hat{e}}, \delta \bm{\hat{f}}) + \delta F \smile_1 \Sigma(\bm{\hat{b}} + \bm{\hat{e}}, \delta \bm{\hat{f}}) \Big] \,.
\end{multline}

The difference between those two is
\begin{multline}
 \Delta g' \coloneqq g'(f, \bm{\hat{b}} + \bm{\hat{e}}) - g'(f, \bm{\hat{b}}) \\
 = \delta F \smile_1 E + 2 \Big[ B \smile_1 \paren{\Sigma(\bm{\hat{b}}, \delta \bm{\hat{f}}) + \Sigma(\bm{\hat{b}} + \bm{\hat{e}}, \delta \bm{\hat{f}})} \\
 + \delta F \smile_1 \paren{\Sigma(\bm{\hat{b}} + \bm{\hat{e}}, \delta \bm{\hat{f}}) + \Sigma(\bm{\hat{b}}, \delta \bm{\hat{f}}) + \Sigma(\bm{\hat{b}}, \bm{\hat{e}})} \\
 + E \smile_1 \paren{\Delta(\bm{\hat{f}}) + \Sigma(\bm{\hat{b}} + \bm{\hat{e}}, \delta \bm{\hat{f}})} \Big] \,.
\end{multline}

Now, let us suppose that $h(e, \bm{\hat{b}})$ has a term of the form $2 B \smile_1 \Sigma(\bm{\hat{b}}, \bm{\hat{e}})$. Under the modification $\bm{\hat{b}} \to \bm{\hat{b}} + \delta \bm{\hat{f}}$, it would contribute
\begin{multline}
 2(B + \delta F) \smile_1 \Sigma(\bm{\hat{b}} + \delta \bm{\hat{f}}, \bm{\hat{e}}) - 2 B \smile_1 \Sigma(\bm{\hat{b}}, \bm{\hat{e}}) \\
 = 2 \bigg[ \delta F \smile_1 \Sigma(\bm{\hat{b}} + \delta \bm{\hat{f}}, \bm{\hat{e}}) \\
 + B \smile_1 \paren{\Sigma(\bm{\hat{b}}, \bm{\hat{e}}) + \Sigma(\bm{\hat{b}} + \delta \bm{\hat{f}}, \bm{\hat{e}})} \bigg]
\end{multline}

Those terms happen to be the exact same as the ones which appear in $\Delta g'$. Indeed, the term with $\delta F \smile_1 ...$:
\begin{align}
 2 &\sqpar{ \Sigma(\bm{\hat{b}} + \bm{\hat{e}}, \delta \bm{\hat{f}}) - \Sigma(\bm{\hat{b}}, \delta \bm{\hat{f}}) + \Sigma(\bm{\hat{b}}, \bm{\hat{e}})} \nonumber \\
 &= \{\bm{\hat{b}} + \bm{\hat{e}} + \delta \bm{\hat{f}} \} - \{ \bm{\hat{b}} + \bm{\hat{e}} \} - \{ \delta \bm{\hat{f}} \} - \{ \bm{\hat{b}} + \delta \bm{\hat{f}} \} \nonumber \\
 &\hspace{2em} + \{ \bm{\hat{b}} \} + \{ \delta \bm{\hat{f}} \} + \{ \bm{\hat{b}} + \bm{\hat{e}} \} - \{ \bm{\hat{b}} \} - \{ \bm{\hat{e}} \} \nonumber \\
 &= \{\bm{\hat{b}} + \bm{\hat{e}} + \delta \bm{\hat{f}} \} - \{ \bm{\hat{b}} + \delta \bm{\hat{f}} \} - \{ \bm{\hat{e}} \} \nonumber \\
 &= 2 \Sigma(\bm{\hat{b}} + \delta \bm{\hat{f}}, \bm{\hat{e}})
\end{align}

We have used the definition of $\Sigma$, as well as the fact that, mod 4, $2 = -2$, to change a sign in the first line.

This same equation, rearranged, gives us equality also of the terms $B \smile_1 \dots$. Therefore
\be
\Delta g' = \delta F \smile_1 E + 2 E \smile_1 \paren{\Delta(\bm{\hat{f}}) + \Sigma(\bm{\hat{b}} + \bm{\hat{e}}, \delta \bm{\hat{f}})} + (\dots) \,,
\ee
where $(\dots)$ includes multiples of 4, coboundaries, and all the terms which we know how to cancel via $h$.

Next, consider the following Leibniz rule:
\be
\delta (\delta F \smile_2 E) = \delta F \smile_2 \delta E + \delta F \smile_1 E + E \smile_1 \delta F \,,
\ee
which is exact also in integer coefficients (all the signs just happen to be $+1$). Therefore:
\begin{multline}
 \Delta g' = - E \smile_1 \delta F - \delta F \smile_2 \delta E \\
 + 2 E \smile_1 \paren{\Delta(\bm{\hat{f}}) + \Sigma(\bm{\hat{b}} + \bm{\hat{e}}, \delta \bm{\hat{f}})} + (\dots) \,.
\end{multline}

Now, let us suppose that $h$ has a term of the form $\{ \bm{\hat{e}} \} \smile_1 \{ \bm{\hat{b}} + \bm{\hat{e}} \}$. Its contribution to $\Delta h$ is
\begin{multline}
 E \smile_1 \paren{\{ \bm{\hat{b}} + \bm{\hat{e}} \} + \delta F + 2 \Delta(\bm{\hat{f}}) + 2 \Sigma(\bm{\hat{b}} + \bm{\hat{e}}, \delta \bm{\hat{f}}) } \\
 - E \smile_1 \{\bm{\hat{b}} + \bm{\hat{e}}\} \\
 = E \smile_1 \delta F + 2 E \smile_1 \sqpar{\Delta(\bm{\hat{f}}) + \Sigma(\bm{\hat{b}} + \bm{\hat{e}}, \delta \bm{\hat{f}})} \,.
\end{multline}

We conclude that this term must appear with a $-1$ sign in $h$. At this point, we have
\be
\Delta g' = - \delta F \smile_2 \delta E + (\dots) \,,
\ee
where, once again, $(\dots)$ contains multiples of 4, coboundaries, and terms which we know how to absorb via $h$.

Next, we add a term to $g'$, which we can do as long as it only depends on $\delta \bm{\hat{b}}$: $\delta F \smile_2 \{\delta \bm{\hat{b}}\}$. Its contribution is
\begin{multline}
 \delta F \smile_2 \{ \delta \bm{\hat{b}} + \delta \bm{\hat{e}} \} - \delta F \smile_2 \{ \delta \bm{\hat{b}} \} \\
 = \delta F \smile_2 \paren{ \{ \delta \bm{\hat{e}} \} + 2 \Sigma(\delta \bm{\hat{b}}, \delta \bm{\hat{e}}) } \\
 = \delta F \smile_2 \delta E + 2 \delta F \smile_2 \sqpar{\Delta(\bm{\hat{e}}) + \Sigma(\delta \bm{\hat{b}}, \delta \bm{\hat{e}})} \,.
\end{multline}

If we let $g \coloneqq g' + \delta F \smile_2 \{\delta \bm{\hat{b}}\}$, then
\be
\Delta g = 2 \delta F \smile_2 \sqpar{\Delta(\bm{\hat{e}}) + \Sigma(\delta \bm{\hat{b}}, \delta \bm{\hat{e}})} \,.
\ee

Finally, if $h$ had a term $2 B \smile_2 \sqpar{\Delta(\bm{\hat{e}}) + \Sigma(\delta \bm{\hat{b}}, \delta \bm{\hat{e}})}$, then its contribution would be
\begin{multline}
 2 (B + \delta F) \smile_2 \sqpar{\Delta(\bm{\hat{e}}) + \Sigma(\delta \bm{\hat{b}} + \delta^2 \bm{\hat{f}}, \delta \bm{\hat{e}}) } \\
 - 2 B \smile_2 \sqpar{\Delta(\bm{\hat{e}}) + \Sigma(\delta \bm{\hat{b}}, \delta \bm{\hat{e}}) } \\
 = 2 \delta F \smile_2 \sqpar{\Delta(\bm{\hat{e}}) + \Sigma(\delta \bm{\hat{b}}, \delta \bm{\hat{e}}) } \,.
\end{multline}

Putting everything together, and undoing some of the definitions, we find that if we define
\begin{align}
 g(f, \bm{\hat{b}}) &\coloneqq \delta \{\bm{\hat{f}}\} \smile_1 \{ \bm{\hat{b}} \} + 2 \Big[ \bm{\hat{b}} \smile_1 \Delta(\bm{\hat{f}}) + \delta \bm{\hat{f}} \smile_1 \Delta(\bm{\hat{f}}) \nonumber \\
+ \bm{\hat{b}} &\smile_1 \Sigma(\bm{\hat{b}}, \delta \bm{\hat{f}}) + \delta \bm{\hat{f}} \smile_1 \Sigma(\bm{\hat{b}}, \delta \bm{\hat{f}}) \Big] + \delta \{\bm{\hat{f}}\} \smile_2 \{\delta \bm{\hat{b}}\} \,, \\
h(e, \bm{\hat{b}}) &\coloneqq 2 \bm{\hat{b}} \smile_1 \Sigma(\bm{\hat{b}}, \bm{\hat{e}}) - \{\bm{\hat{e}}\} \smile_1 \{ \bm{\hat{b}} + \bm{\hat{e}} \} \nonumber \\
&+ 2 \bm{\hat{b}} \smile_2 \sqpar{\Delta(\bm{\hat{e}}) + \Sigma(\delta \bm{\hat{b}}, \delta \bm{\hat{e}}) } \,.
\end{align}

Then the following is true
\be
g(f, \bm{\hat{b}} + \bm{\hat{e}}) - g(f, \bm{\hat{b}}) \equiv h(e, \bm{\hat{b}} + \delta \bm{\hat{f}}) - h(e, \bm{\hat{b}}) \,,
\ee
with equivalence up to multiples of 4, and up to coboundaries.

\subsubsection{Derivation of $(ST)^3=Y$}
\label{sec:derivation_of_STSTST_Y}

Our goal now is to show that $(\alpha_\KWW \circ \alpha_\TW)^3$ extends to a QCA on the full tensor product algebra, the Witt class of this extension is well-defined, and this Witt class is that of the semion QCA \cite{Shirley_2022}. It is important to note that we make two crucial assumptions: one is physically motivated, and we believe it to be very reasonable, namely, that the classification of $\Z_2^{(1)}$ SPT phases on the lattice is complete, and is given by the entanglers $\sf{T}^j$ for $j = 0, \cdots, 3$. The other is perhaps a little stronger: that any two entanglers of the same phase are related by a finite depth circuit (with each gate being symmetric). 

More precisely, our assumption can be spelled out as follows: 
%\begin{itemize}
%\item[] {\bf Assumption:}
%If $\alpha$ is a QCA which is extendable to a QCA $\tilde{\alpha}$ on the full algebra, such that $\tilde{\alpha}$ is a finite depth circuit (composed of not necessarily symmetric gates), then $\alpha$ is equivalent, up to a symmetric finite depth circuit (with symmetric gates) to $\aTW^k$, for some $k = 0, \cdots, 3$.
%\end{itemize}

{\noindent{\bf Assumption.} \textit{If $\alpha$ is a QCA which is extendable to a QCA $\tilde{\alpha}$ on the full algebra, such that $\tilde{\alpha}$ is a finite depth circuit (composed of not necessarily symmetric gates), then $\alpha$ is equivalent, up to a symmetric finite depth circuit (with symmetric gates) to $\aTW^k$, for some $k = 0, \cdots, 3$.}}

Note that, for brevity of the notation, and for ease of comparison with the continuum and categorical results, we are going to write $S \coloneqq \alpha_\KWW$ and $T \coloneqq \alpha_\TW$ in this section.

Since $T$ is a Witt trivial QCA that is well defined on the full tensor product algebra, it suffices to demonstrate this for $S \circ T \circ S \circ T \circ S$. Recall from the previous subsection that 
\begin{align}
 (S \circ T \circ S \circ T \circ S)(X_f + \mathfrak{I}) = \overline{X}_f + \mathfrak{I} \,,
\end{align}
with ${\overline{X}}_f$ given in \eqref{eq:Xbar} (the equality is technically up to some translation).  Now \eqref{eq:Xbar} also defines a set of flippers ${\overline{Z}}_p$ for ${\overline{X}}_f$.  These are not commuting, but by a general result \cite{HaahFidkowskiHastings2023} this implies the existence of commuting flippers ${\overline{Z}}'_p$.  The assignment
\begin{align}
X_f &\rightarrow {\overline{X}}_f \\
Z_p &\rightarrow {\overline{Z}}'_p
\end{align}
then defines a QCA $\tilde{Q} : \mc{A} \rightarrow \mc{A}$ on the full tensor product algebra.  $\tilde{Q}$ commutes with the symmetry operators, since these are just determined by the images of the $X_l$, and so there is an induced QCA $Q : \mc{A}_{\Z_2^{(1)}} / \mathfrak{I} \rightarrow \mc{A}_{\Z_2^{(1)}} / \mathfrak{I}$ on the $\Z_2$ 1-form symmetric algebra.  Then, since
\be
S \circ T \circ S \circ T \circ S = Q \circ \paren{Q^{-1} \circ S \circ T \circ S \circ T \circ S }
\ee
the claim about $STSTS$ above follows if we can prove 
%\begin{itemize}
%\item[]{\bf Claim 1.} $\tilde{Q}$ is a non-trivial QCA in the equivalence class of the semion QCA 
%\item[]{\bf Claim 2.} $R:=Q^{-1} \circ STSTS$ extends to a Witt trivial QCA on the full tensor product algebra.
%\end{itemize}

\noindent{\bf Claim 1.} \textit{$\tilde{Q}$ is a non-trivial QCA in the equivalence class of the semion QCA.}

\noindent{\bf Claim 2.} \textit{$R:=Q^{-1} \circ STSTS$ extends to a Witt trivial QCA on the full tensor product algebra.}

Claim 1 would follow if we could show that the stabilizers ${\overline{X}}_f$ match those of \cite{Shirley_2022} exactly, because then it would follow that $\tilde{Q}$ maps the trivial Hamiltonian $-\sum_f X_f$ to precisely the semion Walker-Wang Hamiltonian, which has a commuting projector semion model boundary.  Alternatively, it would follow if we could explicitly find a commuting projector boundary Hamiltonian with semion topological order which commutes with the bulk Hamiltonian $-\sum_f {\overline{X}}_f$.  We do not currently have a proof of either of these statements.  However, we do know that $-\sum_f {\overline{X}}_f$ is a strictly local commuting projector parent Hamiltonian for the semion Walker-Wang ground state, because the semion Walker-Wang ground state is the $1$-form gauging of the Tsui-Wen $\Z_2$ one-form SPT.  This just follows from the fact that, by construction of the ${\overline{X}}_f$, the ground state of $-\sum_f {\overline{X}}_f$ is $ST$ applied to the ground state of $-\sum_f X_f$: $T$ creates the Tsui-Wen $1$-form SPT, and $S$ gauges the $1$-form symmetry.  It would be very surprising if it were not in the same equivalence class of local Hamiltonians as the Hamiltonian of \cite{Shirley_2022}.

For the rest of this section, we will focus on proving Claim $2$: that $R \coloneqq Q^{-1} \circ STSTS$ extends to the full tensor product algebra, and its extension is uniquely defined and trivial, i.e. it is a finite depth circuit on the full tensor product algebra. We note that $R$, by construction, maps $X_l$ to $X_l$ for each $l$ (or rather, their equivalence classes map like this).

Theorem~\ref{thm:QCA_extendability} implies that $R' \coloneqq S \circ R \circ S$ is extendable to the full algebra. That is, there is some $\tilde{R}' : \mc{A} \rightarrow \mc{A}$ such that
\be
R'(a + \mathfrak{I}) = \tilde{R}'(a) + \mathfrak{I} \,, \quad \forall a \in \mc{A}_{\Z_2^{(1)}} \,.
\ee
Furthermore, $\tilde{R}'(Z_f) = Z_f$, and $\tilde{R}'(X_{\partial c}) = X_{\partial c}$.

We would now like to argue -- under the Assumption stated above -- that $\tilde{R}' \sim \tilde{T}^{-j}$, for some $j$, with the equivalence being up to a symmetric finite depth circuit (whose gates are $\Z_2^{(1)}$-symmetric).  

Let us now imagine a finite but large system, say on a torus. We can then talk about the unitary operator corresponding to $\tilde{R}'$, $\sf{U}_{\tilde{R}'}$ (in a finite dimensional Hilbert space, every automorphism of its operators algebra is inner, that is, is given by conjugation by a unitary). Since $\sf{U}_{\tilde{R}'}$ preserves all of the $Z_f$, this is just multiplication by some phase $\phi(\{Z\}_f)$ in the computational basis ($\ket{0, 1}$). Since $\tilde{R}'$ is a QCA, $-\sum_f \tilde{R}'(X_f)$ is a parent Hamiltonian for the state $\sf{U}_{\tilde{R}'}\ket{+}$, which in the computational basis just has amplitudes $\phi(\{Z\}_f)$. This must be some $\Z_2^{(1)}$ one-form SPT ground state.

In fact, we claim that, by virtue of being diagonal in the computational basis, $\tilde{R}'$ is a finite depth circuit, with gates that may not necessarily be symmetric.  We now provide an argument for this, which is a special case of Lemma II.3 in  \cite{Fidkowski:2019nju}.  First, let us take a cellulation of space, with cells of size larger than the Lieb-Robinson length of $\tilde{R}'$.  The dual graph of this cellulation is $n$-colorable for sufficiently large $n$ ($n$ is some finite constant).  Let us take such a coloring, and let $R_k$ ($k \in 0, \ldots, n$), be the union of all cells whose dual vertices have colors $\leq k$.  Thus $R_0$ is the empty set, and $R_n$ is the whole space.  Now, given any spin configuration $\{Z\}_f$ in the computational basis, let us define, for each $k \in 0, \ldots, n$, a corresponding configuration $\{Z^{(k)}\}_f$ by setting $Z^{(k)}_f = Z_f$ if $f \in R_k$, and $Z^{(k)}_f = 1$ otherwise.  Thus $\{Z^{(0)}\}_f$ is always the trivial all spin up configuration, and $\{Z^{(n)}\}_f = \{Z\}_f$.  Finally, let us define, for each $k \in 1, \ldots n$, the unitary operator $V_k$ as follows.  $V_k$ is diagonal in the computational basis, and multiplies by the phase $\phi(Z^{(k)}_f)\cdot \phi^{-1}(Z^{(k-1)}_f)$.

Clearly we have that $\tilde{R}' = V_1 \ldots V_n$.  Thus to show that $\tilde{R}'$ is an FDC, it suffices to show that each $V_k$ is a disjoint product of local gates.  It is illustrative to first consider $k=1$.  $V_1$ is a product of local gates, each of which act within a  $\tilde{R}'$-Lieb-Robinson length thickening of a cell $C$ which has color $1$.  By virtue of the fact that we have an $n$-coloring, and the fact that the cells are all larger than the Lieb-Robinson length of $\tilde{R}'$, none of these thickened neighborhoods overlap.  In each neighborhood of such a cell $C$, $V_1$ simply acts by multiplying, in the computational basis, by the phase $\phi(\{Z^{(1,C)}\}_f)$, where $\{Z^{(1,C)}\}_f$ is the configuration that is equal to $Z_f$ within $C$ and $1$ elsewhere.  By virtue $\tilde{R}'$ being a QCA, this kind of multiplication is conditioned only on the spin configuration within the thickened neighborhood of $C$.
A similar argument shows that each $V_k$ is a product of disjoint local gates, and hence $\tilde{R}'$ is a finite depth circuit.  We therefore have, by our Assumption, the existence of a $j=0, \cdots, 3$ such that 
\be
R = S \circ R' \circ S = Q^{-1} \circ STSTS \sim (STS)^{-j} \,,
\ee
as QCA on $\mc{A}_{\Z_2^{(1)}} / \mathfrak{I}$.

We also find that $Q^{-1} \circ STSTS \circ (STS)^{j}$ is an FDC, and is therefore extendable to the full algebra. Now note that, $Q$ being extendable to the full algebra means that there is some $\tilde{W} \in \Witt^\pt$ such that $\tilde{Q} \circ \tilde{W}^{-1}$ is an FDC, {with possibly non-symmetric gates}. Such $W$ could be written as $Y^k Y'^l$ for $Y, Y'$ the generators of $W_2 = \Z_8 \oplus \Z_2$. {In fact, from the construction of $\tilde{Q}$ we know that it entangles the semion Walker-Wang model.  Making the standard assumption that the bulk Hamiltonian it entangles allows for a commuting projector semion boundary, the only possibility is $k=1, l'=0$.  We note in passing that the other possibility, namely $(ST)^3 = Y Y'$, would be realized in the case of $\Z_4^{(1)}$ symmetry, since in that case the relevant anyon theory is precisely $U(1)_4$, which corresponds to the $Y Y'$ generator of $\Witt^\pt$.} - see e.g. \cite{Shirley_2022}, Sec. IV A. 

Now we use our Assumption, namely the classification of SPT states again, which tells us that $\tilde{Q} \circ \tilde{W}^{-1}$ is equivalent to $\tilde{T}^m$, up to \emph{symmetric} FDCs.

Now we look at the equivalence $Q^{-1} \circ STSTS \circ (STS)^{j} \sim 1$ but in the quotient (see Section \ref{sec:BrAutZ2})
\be
\frac{\Witt^\pt(\Z_2^2, s)}{\Witt^\pt} \cong \Perm_4 \cong \langle \tau, \sigma \; \mid \; \sigma^2 = \tau^4 = (\sigma \tau)^3 = 1 \rangle \,.
\ee
From the previous discussion, we have
\be
\tau^{-m} \tau (\sigma \tau \sigma)^j = 1
\ee
It can be checked, by brute force (using for example $\sigma = (1, 2), \tau = (1, 2, 3, 4)$ in cycle notation for symmetric groups) that the only way for this equality to be satisfied is if $m = 1$ and $j = 0$.

Therefore, we find that $Q^{-1} \circ STSTS$ is an FDC (because $j = 0$). This also means $STSTS$ can be extended to the full algebra, because any FDC can be extended, and by definition of $\tilde{Q}$ and $Q$, $Q$ can be extended as well. This concludes the reasoning in support of Claim 2.

The claim about the uniqueness of the Witt class of an extension comes down to showing that any extension of the identity automorphism of the $\Z_2$ 1-form symmetric algebra has trivial Witt class.  This is trivial to show, since any such extension maps the $X_f$ to themselves up to symmetry operators, and hence the image of the Hamiltonian $-\sum X_f$ has a trivial gapped boundary.

\noindent{\bf Comments.}
{Let us conclude with some remarks about the group structure of these QCA (up to translations and finite-depth circuits with symmetric gates). Our results strongly suggest a match between the Witt group calculations and the group generated by $\aKWW$ and $\aTW$. Indeed, $\aKWW$ is order 2, and $\aTW$ is order 4.}

{We also find, under certain physically motivated assumptions, that $(\aKWW \circ \aTW)^3$ is liftable to the full algebra and equivalent to the semion QCA \cite{Shirley_2022}. Finally, since the semion QCA is central, we recover all the same relations which indicate the group of QCA is exactly $2O \circ_{\Z_2} \Z_8$ (compare with Theorem~\ref{thm:modular}, and Equation~\eqref{eq:STcat} there). Note, however, we do not generate in this way the additional $\Z_2$ factor in grade $p=2$ from $\Witt^\pt(\Z_2 \oplus \widehat{\Z_2}, s)$. This together with the remaining $W_p$ classes for $p\not=2$ factorize,  consistently with the general structure in Theorem \ref{thm:modular}. }

\section{QCAs on $\mathbb{Z}_p^{(1)}$-symmetric Algebra for odd prime $p$}
\label{sec: QCAs on Zp 1-form symmetric algebra}

In this section, we discuss QCAs on the algebra of $\mathbb{Z}_p$ 1-form symmetric local operators for odd prime $p$.
We will consider several examples of such QCAs and study their fusion rules.
As the first example, we will consider the $\mathbb{Z}_p$ Kramers-Wannier-Wegner operator in Section~\ref{sec: Zp Kramers-Wannier-Wegner operator}, which generalizes the KWW operator for the $\mathbb{Z}_2$ case.
We will then review some other examples of $\mathbb{Z}_p$ 1-form symmetric QCAs, such as the SPT entanglers of Tsui and Wen \cite{Tsui:2019ykk} in Section~\ref{sec: Zp Tsui-Wen entanglers} and the non-trivial Clifford QCAs of Sun et al.~\cite{MengSun2026} in Section~\ref{sec: Zpk QCA}.
We then compute particular fusion rules of these QCAs in Section~\ref{sec: Zpk QCA generated by D and T} and discuss the full group structure in Section~\ref{sec:ZpQCAgroup}.
Finally, we will compare the group of QCAs with the Witt group calculation (Section~\ref{sec: Comparison with the Graded Pointed Witt Group}) and the continuum calculation (Section~\ref{sec: comparison Zp}).
As in the $\mathbb{Z}_2$ case, it turns out that the fusion rules on the lattice and those in the continuum differ by non-trivial QCAs and lattice translations.

%\vspace*{\baselineskip}
%\noindent{\bf Setup.}
\subsection{Setup}
Throughout this section, we will consider a cubic lattice with a single $\mathbb{Z}_p$ qudit on each face.
The state space is thus given by
\begin{equation}
\mathcal{H}_F = \bigotimes_{f \in F} \mathbb{C}^p.
\label{eq: Zp Hilb}
\end{equation}
We recall that $F$ denotes the set of all faces; see Section~\ref{sec: Notations and conventions} for the notations and conventions.
The Pauli operators $X_f$ and $Z_f$ on each face $f$ are defined by
\begin{equation}
X_f \ket{a}_f = \ket{a+1 \bmod p}_f, \quad
Z_f \ket{a}_f = e^{\frac{2\pi i a}{p}} \ket{a}_f,
\end{equation}
where $a \in \{0, 1, \cdots, p-1\}$.

The $\mathbb{Z}_p$ 1-form symmetry operators are defined by
\begin{equation}
\eta(\Sigma) \coloneq \prod_{f \in F} X_f^{\bm{f}(\Sigma)},
\label{eq: Zp sym op}
\end{equation}
where $\Sigma$ is a 2-cycle with $\mathbb{Z}_p$ coefficients.
For example, the symmetry operator on the boundary of cube $c$ is given by
\begin{equation}
X_{\partial c} \coloneq \prod_{f \in F} X_f^{\bm{f}(\partial c)}.
\label{eq: symmetry bubble Zp}
\end{equation}
We emphasize that one needs to be careful about the signs in the boundary $\partial c$ (defined by \eqref{eq: partial c def}) because $X_f$ is not of order 2.
See Figure~\ref{fig: symmetry bubble Zp} for an illustration of $X_{\partial c}$.
\begin{figure}[t]
\centering
$X_{\partial c} =$
\adjincludegraphics[valign=c, trim={10, 0, 10, 0}]{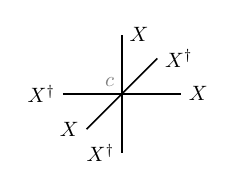}
\caption{The $\mathbb{Z}_p$ 1-form symmetry operator $X_{\partial c}$ supported on the boundary of a cube $c$. Here, the middle vertex is the dual of $c$, and the edges attached to it represent the dual of the boundary of $c$.}
\label{fig: symmetry bubble Zp}
\end{figure}

The local operators invariant under the above $\mathbb{Z}_p$ 1-form symmetry are generated by
\begin{equation}
X_f, \quad
Z_{\delta \bm{e}} = \prod_{f \in F} Z_f^{\delta \bm{e}(f)}
\label{eq: XZ Zp}
\end{equation}
for all faces $f \in F$ and all edges $e \in E$.
We again emphasize that one needs to be careful about the signs in the coboundary $\delta \bm{e}$ (defined by \eqref{eq: delta e def}) because $Z_f$ is not of order 2.
The operator $Z_{\delta \bm{e}}$ is illustrated in Figure~\ref{fig: Zp flux}.
\begin{figure}
\centering
$Z_{\delta \bm{e}} =$
\adjincludegraphics[valign=c, trim={10, 0, 0, 0}]{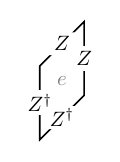},
\adjincludegraphics[valign=c]{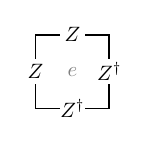},
\adjincludegraphics[valign=c]{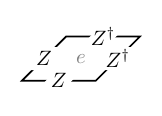}
\caption{The flux operator $Z_{\delta \bm{e}}$ for each edge $e$. Here, the plaquette in each diagram represents the dual of $e$.}
\label{fig: Zp flux}
\end{figure}
The algebra of local operators generated by $X_f$'s and $Z_{\delta \bm{e}}$'s is denoted by
\begin{equation}
\mathcal{A}_{\mathbb{Z}_p^{(1)}} \coloneq \langle X_f, Z_{\delta \bm{e}} \mid f \in F, e \in E \rangle.
\end{equation}

In what follows, we will often impose the condition that the symmetry operator~\eqref{eq: symmetry bubble Zp} around each cube acts as the identity operator.
In other words, we consider the subspace of $\mathcal{H}_F$ such that
\begin{equation}
X_{\partial c} = 1
\end{equation}
for all cubes $c \in C$.
We note that the symmetry operator~\eqref{eq: Zp sym op} is topological on this subspace.
The algebra of local operators on this subspace is given by the quotient $\mathcal{A}_{\mathbb{Z}_p^{(1)}} / \mathfrak{I}_p$, where $\mathfrak{I}_p$ is the ideal generated by $X_{\partial c} -1$ for all cubes $c$.

\subsection{$\mathbb{Z}_p$ Kramers-Wannier-Wegner operator}
\label{sec: Zp Kramers-Wannier-Wegner operator}
In this subsection, we define the $\mathbb{Z}_p$ Kramers-Wannier-Wegner operator $\mathsf{S}_p$ and compute its action on $\mathbb{Z}_p$ 1-form symmetric local operators~\eqref{eq: XZ Zp}.
Specifically, we will show that the $\mathbb{Z}_p$ KWW operator $\mathsf{S}_p$ acts on local operators $X_f$ and $Z_{\delta \bm{e}}$ as follows:
\begin{equation}
X_f \xmapsto{\sf{S}_p} Z_{\delta t^{\frac{1}{2}}(\bm{f})}^{-s_f}\,, \quad
Z_{\delta \bm{e}} \xmapsto{\sf{S_p}} X_{t^{\frac{1}{2}}(e)}^{s_e} \,.
\label{eq: Dp action}
\end{equation}
Here, $s_f$ and $s_e$ are the signs defined by
\begin{align}
s_f &=
\begin{cases}
+1 \quad &\text{if $f$ is a $yz$- or an $xy$-plaquette},\\
-1 \quad &\text{if $f$ is an $xz$-plaquette},
\end{cases}
\label{eq: sf}
\\
s_e &=
\begin{cases}
+1 \quad &\text{if $e$ is an $x$- or a $z$-link}, \\
-1 \quad &\text{if $e$ is a $y$-link}.
\end{cases}
\label{eq: se}
\end{align}
We note that $s_f$ and $s_e$ are invariant under the half translation, that is,
\begin{equation}
s_f = s_{t^{\frac{1}{2}}(f)}, \quad
s_e = s_{t^{\frac{1}{2}}(e)}.
\label{eq: s inv}
\end{equation}
From \eqref{eq: Dp action} and \eqref{eq: s inv}, it immediately follows that the action of $\mathsf{S}_p$ squares to the lattice translation followed by the charge conjugation:
\begin{equation}
X_f \xmapsto{\sf{S}_p^2} X_{t(f)}^{\dagger}\,, \quad
Z_{\delta \bm{e}} \xmapsto{\sf{S}_p^2} Z_{\delta t(\bm{e})}^{\dagger} \,.
\end{equation}
The action~\eqref{eq: Dp action} can also be expressed in terms of the cup product as
\begin{equation}
X_f \xmapsto{\sf{S}_p} \prod_{e \in E} Z_{\delta \bm{e}}^{-\int \bm{f} \smile \bm{e}} \,, \quad
Z_{\delta \bm{e}} \xmapsto{\sf{S}_p} \prod_{f \in F} X_f^{\int \bm{e} \smile \bm{f}} \,.
\end{equation}

As in the $\mathbb{Z}_2$ case, $\mathsf{S}_p$ maps the symmetry operator $X_{\partial c}$ to the identity operator.
Thus, $\mathsf{S}_p$ is not technically a QCA on $\mathcal{A}_{\mathbb{Z}_p^{(1)}}$, but defines a QCA on the quotient algebra $\mathcal{A}_{\mathbb{Z}_p^{(1)}} / \mathfrak{I}_p$.
This QCA is denoted by $\aKWWp$.

We remark that like in the $\Z_2$ case, we use the notation $\sf{S}_p$ instead of the more conventional $\sf{D}_p$ (for ``duality'') for consistency with the continuum notation.

In what follows, we will define the $\mathbb{Z}_p$ KWW operator $\mathsf{S}_p$ acting on the tensor product Hilbert space~\eqref{eq: Zp Hilb} and derive its action~\eqref{eq: Dp action} on local symmetric operators.

\subsubsection{Gauging operator for $\mathbb{Z}_p$ 1-form symmetry}
We first define the gauging operator for $\mathbb{Z}_p$ 1-form symmetry, which we will use later to define the $\mathbb{Z}_p$ KWW operator.
To this end, we begin by recalling the gauging procedure of $\mathbb{Z}_p$ 1-form symmetry on the lattice.

To gauge the $\mathbb{Z}_p$ 1-form symmetry, we first add a $\mathbb{Z}_p$ qudit on each edge.
The qudits on the edges are called the gauge field.
On the other hand, the original qudits on the faces are called the matter field.
The total Hilbert space is given by 
\begin{equation}
\mathcal{H} = \mathcal{H}_F \otimes \mathcal{H}_E,
\end{equation}
where $\mathcal{H}_F = \bigotimes_{f \in F} \mathbb{C}^{p}$ and $\mathcal{H}_E = \bigotimes_{e \in E} \mathbb{C}^{p}$ are the state spaces of the matter field and gauge field, respectively.
The gauge transformation around face $f$ is implemented by applying the Gauss law operator
\begin{equation}
G_f \coloneq X_f X_{\partial f},
\label{eq: Gauss law op Zp}
\end{equation}
where $X_{\partial f} \coloneq \prod_{e \in E} X_e^{\bm{e}(\partial f)}$.
See Figure~\ref{fig: Gauss Zp} for an illustration of the Gauss law operator $G_f$.
\begin{figure}[t]
\centering
$G_f =$
\adjincludegraphics[valign=c, trim={10, 0, 10, 0}, scale=1]{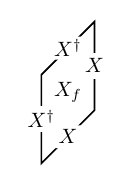},
\adjincludegraphics[valign=c, trim={0, 0, 10, 0}, scale=1]{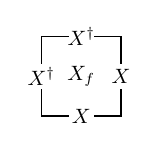},
\adjincludegraphics[valign=c, trim={0, 0, 10, 0},  scale=1]{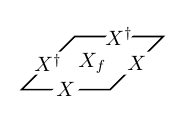}
\caption{The Gauss law operator $G_f$ illustrated on the direct lattice. Here, the middle plaquette in each diagram represents face $f$.}
\label{fig: Gauss Zp}
\end{figure}
The state space of the gauged model is then given by the subspace of $\mathcal{H}$ on which the Gauss law operator $G_f$ acts trivially for all faces $f \in F$.
Namely, the state space of the gauged model is
\begin{equation}
\mathcal{H}_{\mathrm{gauged}} = \{\ket{\psi} \in \mathcal{H} \mid G_f \ket{\psi} = \ket{\psi}, ~ \forall f \in F\}.
\end{equation}
This state space is not a tensor product Hilbert space.
Nevertheless, we can make it into a tensor product Hilbert space by fixing the gauge so that the matter degrees of freedom are frozen.
The state space after the gauge fixing is
\begin{equation}
\mathcal{H}_{\mathrm{gauged}}^{\prime} = \bigotimes_{e \in E} \mathbb{C}^2 = \mathcal{H}_E.
\end{equation}
Namely, the state space after the gauge fixing consists only of the gauge field on the edges.

The gauging operator $\mathsf{D}_{\mathrm{gauge}}: \mathcal{H}_F \to \mathcal{H}_E$ is defined as a linear map that maps the matter field to the gauge field in the gauge-equivalent configuration.
More specifically, the explicit action of $\mathsf{D}_{\mathrm{gauge}}$ in the computational basis is given by
\begin{equation}
\mathsf{D}_{\mathrm{gauge}} \ket{\bm{a}}_F = \frac{1}{p^{N/2}} \left( \prod_{f \in F} X_{\partial f}^{- \bm{a}(f)} \right) \ket{\bm{0}}_E.
\label{eq: D gauge Zp}
\end{equation}
Here, $\bm{a}$ is a $\mathbb{Z}_p$-valued 2-cochain, which represents a configuration of the matter field.
The corresponding state $\ket{\bm{a}}_F$ is defined by
\begin{equation}
\ket{\bm{a}}_F \coloneq \bigotimes_{f \in F} \ket{\bm{a}(f)}_f.
\end{equation}
Similarly, $\bm{0}$ on the right-hand side of \eqref{eq: D gauge Zp} is the trivial $\mathbb{Z}_p$-valued 1-cochain, which represents the trivial configuration of the gauge field.
The corresponding state $\ket{\bm{0}}_E$ is the tensor product of the $\ket{0}$ states on all edges.
The integer $N$ in the overall factor of \eqref{eq: D gauge Zp} is the number of cubes.
We note that $\mathsf{D}_{\mathrm{gauge}}$ is normalized so that it satisfies
\begin{equation}
\mathsf{D}_{\mathrm{gauge}}^{\dagger} \mathsf{D}_{\mathrm{gauge}} = \frac{1}{p^N} \sum_{\Sigma: \text{2-cycles}} \eta(\Sigma),
\end{equation}
where the right-hand side is the condensation operator on the lattice.
The normalization of $\mathsf{D}_{\mathrm{gauge}}$ is not important in later discussions.

Equation~\eqref{eq: D gauge Zp} can also be written as
\begin{equation}
\mathsf{D}_{\mathrm{gauge}} \ket{\bm{a}}_F = \frac{1}{p^{N/2}} \ket{\bm{a}^{\prime}}_E,
\label{eq: D gauge Zp2}
\end{equation}
where $\bm{a}^{\prime}$ is a $\mathbb{Z}_p$-valued 1-cochain defined by
\begin{equation}
\begin{aligned}
\bm{a}^{\prime}(e) &= -\sum_{f \in F} \bm{a}(f) \bm{e}(\partial f) \quad \bmod p \\
&= -\sum_{f \in F} \bm{a}(f) \delta \bm{e}(f) \quad \bmod p.
\end{aligned}
\label{eq: a prime Zp}
\end{equation}
We will use this expression later when we compute the action of $\mathsf{D}_{\mathrm{gauge}}$ on local operators.

%\vspace*{\baselineskip}
%\noindent{\bf Dual symmetry.}
For later use, we also briefly discuss the symmetry of the gauged model.
The gauged model has a dual $\mathbb{Z}_p$ 1-form symmetry generated by the Wilson surface operators~\cite{Gaiotto:2014kfa}.
The symmetry operators on the lattice are given by
\begin{equation}
Z_{\bm{b}} = \prod_{e \in E} Z_e^{\bm{b}(e)},
\end{equation}
where $\bm{b}$ is an arbitrary 1-cocycle on the direct lattice.
We can think of $Z_{\bm{b}}$ as an operator supported on a closed surface on the dual lattice.
For instance, the symmetry operator supported on a small sphere surrounding a vertex $v$ is given by
\begin{equation}
Z_{\delta \bm{v}} = \prod_{e \in E} Z_e^{\delta \bm{v}(e)}.
\end{equation}
The operator $Z_{\delta \bm{v}}$ is the product of Pauli $Z$'s (and their Hermitian conjugate) on the edges connected to vertex $v$.
See Figure~\ref{fig: Z delta v} for an illustration of this symmetry operator.
\begin{figure}
\centering
$Z_{\delta \bm{v}} =$
\adjincludegraphics[valign=c, trim={10, 0, 10, 0}]{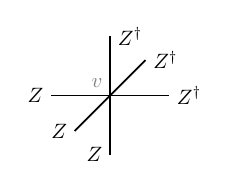}
\caption{The symmetry operator $Z_{\delta \bm{v}}$ of the gauged model. The figure is illustrated on the direct lattice, where the middle vertex is $v$.}
\label{fig: Z delta v}
\end{figure}

%\vspace*{\baselineskip}
%\noindent{\bf $\mathbb{Z}_p$ Kramers-Wannier-Wegner operator.}
\subsubsection{Definition of the $\mathbb{Z}_p$ KWW operator}
We now define the $\mathbb{Z}_p$ KWW operator $\mathsf{S}_p$ acting on the tensor product Hilbert space $\mathcal{H}_F$.
The $\mathbb{Z}_p$ KWW operator $\mathsf{S}_p$ is defined by the gauging operator $\mathsf{D}_{\mathrm{gauge}}$ followed by a basis transformation that maps the symmetry operators of the gauged model to those of the original model.
More concretely, we define
\begin{equation}
\mathsf{S}_p \coloneq t^{\frac{1}{2}} \mathsf{C}_y \mathsf{H} \mathsf{D}_{\mathrm{gauge}}.
\end{equation}
Here, $t^{\frac{1}{2}}$ is the half translation from edges to faces, $\mathsf{H}$ is the tensor product of the Hadamard gates on all edges, and $\mathsf{C}_y$ is the tensor product of the charge conjugation operators on the $y$-links:
\begin{equation}
\mathsf{H} \coloneq \bigotimes_{e \in E} H_e, \qquad
\mathsf{C}_y \coloneq \bigotimes_{e: \text{$y$-links}} C_e.
\end{equation}
We recall that the Hadamard gate $H$ and the charge conjugation operator $C$ for a $\mathbb{Z}_p$ qudit are defined by\footnote{The charge conjugation $C$ should not be confused with the set of cubes also denoted by $C$.}
\begin{align}
H \ket{a} &= \frac{1}{\sqrt{p}} \sum_{b = 0}^{p-1} e^{\frac{2\pi i ab}{p}} \ket{b}, \\
C \ket{a} &= \ket{-a \bmod p}.
\end{align}
The Hadamard gate $H$ and the charge conjugation $C$ obey the following commutation relations with $X$ and $Z$:
\begin{align}
HX = ZH, \quad
HZ = X^{\dagger} H,
\label{eq: H Zp}
\\
CX = X^{\dagger} C, \quad
CZ = Z^{\dagger} C.
\label{eq: C Zp}
\end{align}
We note that $H$ squares to the charge conjugation, i.e.,
\begin{equation}
H^2 = C.
\end{equation}
Due to \eqref{eq: H Zp} and \eqref{eq: C Zp}, it follows that $t^{\frac{1}{2}}\mathsf{C}_y\mathsf{H}$ maps the symmetry operators of the gauged model to those of the original model.
In particular, the symmetry operator $Z_{\delta \bm{v}}$ around vertex $v$ is mapped to
\begin{equation}
(t^{\frac{1}{2}} \mathsf{C}_y \mathsf{H}) Z_{\delta \bm{v}} (t^{\frac{1}{2}} \mathsf{C}_y \mathsf{H})^{-1} = X_{\partial t^{\frac{1}{2}}(v)},
\end{equation}
which is the symmetry operator on the boundary of cube $t^{\frac{1}{2}}(v)$.

%\vspace*{\baselineskip}
%\noindent{\bf Action on local symmetric operators.}
\subsubsection{Action on local symmetric operators}
We now compute the action of the $\mathbb{Z}_p$ KWW operator $\mathsf{S}_p$ on local symmetric operators.
To this end, we first compute the action of the gauging operator $\mathsf{D}_{\mathrm{gauge}}$.
Using the definition~\eqref{eq: D gauge Zp}, one can compute the action of $\mathsf{D}_{\mathrm{gauge}}$ on $X_f$ as follows:
\begin{equation}
\begin{aligned}
\mathsf{D}_{\mathrm{gauge}} X_f \ket{\bm{a}}_F
&= \frac{1}{p^{N/2}} \left( \prod_{f^{\prime} \in F} X_{\partial f^{\prime}}^{-\bm{a}(f^{\prime}) - \delta_{f, f^{\prime}}} \right) \ket{\bm{0}}_E \\
&= X_{\partial f}^{\dagger} \mathsf{D}_{\mathrm{gauge}} \ket{\bm{a}}_F.
\end{aligned}
\label{eq: D gauge Zp X}
\end{equation}
Similarly, one can also compute the action of $\mathsf{D}_{\mathrm{gauge}}$ on $Z_{\delta \bm{e}}$ as
\begin{equation}
\begin{aligned}
\mathsf{D}_{\mathrm{gauge}} Z_{\delta \bm{e}} \ket{\bm{a}}_F
&= \left( \prod_{f \in F} e^{\frac{2\pi i}{p} \bm{a}(f) \delta \bm{e}(f)} \right) \mathsf{D}_{\mathrm{gauge}} \ket{\bm{a}}_F \\
&= Z_e^{\dagger} \mathsf{D}_{\mathrm{gauge}} \ket{\bm{a}}_F.
\end{aligned}
\label{eq: D gauge Zp Z}
\end{equation}
Here, the second equality follows from~\eqref{eq: D gauge Zp2} and \eqref{eq: a prime Zp}.
Equations~\eqref{eq: D gauge Zp X} and \eqref{eq: D gauge Zp Z} show that the action of $\mathsf{D}_{\mathrm{gauge}}$ on $X_f$ and $Z_{\delta \bm{e}}$ is given by
\begin{equation}
X_f \xmapsto{\sf{D}_{\mathrm{gauge}}} X_{\partial f}^{\dagger} \,, \quad
Z_{\delta \bm{e}} \xmapsto{\sf{D}_{\mathrm{gauge}}} Z_e^{\dagger} \,.
\label{eq: D gauge Zp action}
\end{equation}
On the other hand, using \eqref{eq: H Zp} and \eqref{eq: C Zp}, one can also compute the action of $t^{\frac{1}{2}}\mathsf{C}_y\mathsf{H}$ on $X_{\partial f}$ and $Z_e$ as
\begin{equation}
X_{\partial f} \xmapsto{t^{\frac{1}{2}} \mathsf{C}_y \mathsf{H}} Z_{\delta t^{\frac{1}{2}}(\bm{f})}^{s_f} \,, \quad
Z_e \xmapsto{t^{\frac{1}{2}} \mathsf{C}_y \mathsf{H}} X_{t^{\frac{1}{2}}(e)}^{-s_e} \,,
\label{eq: H Zp action}
\end{equation}
where $s_f$ and $s_e$ are the signs defined by \eqref{eq: sf} and \eqref{eq: se}, respectively.
By combining \eqref{eq: D gauge Zp action} and \eqref{eq: H Zp action}, we find that the $\mathbb{Z}_p$ KWW operator $\mathsf{S}_p$ acts on local symmetric operators $X_f$ and $Z_{\delta \bm{e}}$ as
\begin{equation}
X_f \xmapsto{\sf{S}_p} Z_{\delta t^{\frac{1}{2}}(\bm{f})}^{-s_f} \,, \quad
Z_{\delta \bm{e}} \xmapsto{\sf{S}_p} X_{t^{\frac{1}{2}}(e)}^{s_e} \,.
\end{equation}
This shows \eqref{eq: Dp action}.

\subsection{Tsui-Wen entanglers of order $p$}
\label{sec: Zp Tsui-Wen entanglers}

In this subsection, we review $\mathbb{Z}_p$ 1-form SPT entanglers of Tsui and Wen \cite{Tsui:2019ykk} and compute their actions on local operators.
In 3+1d, bosonic SPT phases with $\mathbb{Z}_p$ 1-form symmetry for odd prime $p$ are classified by~\cite{Kapustin:2013uxa}
\begin{equation}
H^4(B^2\mathbb{Z}_p, \mathrm{U}(1)) \cong \mathbb{Z}_p.
\end{equation}
Thus, these SPT phases are labeled by $k \in \{0, 1, \cdots, p-1\}$.
The generator for the $\mathbb{Z}_p$ 1-form SPT phase entanglers will be denoted by $\sf{T}_p$, and we will call $\sf{T}_p^k$ for $k \in \Z_p$ the Tsui-Wen entanglers.

In what follows, we will write down the Tsui-Wen entangler $\sf{T}_p$ on a cubic lattice and show that it acts on local operators $X_f$ and $Z_f$ as
\begin{align}\label{eq: Tpk action}
 X_f &\xmapsto{\sf{T}_p} X_f Z_{\delta t^{\frac{1}{2}}(\bm{f})}^{\frac{p - 1}{2} s_f} Z_{\delta t^{-\frac{1}{2}}(\bm{f})}^{\frac{p - 1}{2} s_f} \,, \nonumber \\
 Z_f &\xmapsto{\sf{T}_p} Z_f
\end{align}
where $s_f$ is the sign defined by \eqref{eq: sf}. We note that this definition of $\sf{T}_p$, with $Z$ to the power of $\frac{p - 1}{2}$, has been chosen in order to match the continuum definition: the continuum operation $T_p$, and the continuum defect $\mc{T}_p$, correspond to the lattice unitary $\sf{T}_p$. One can check that $\sf{T}_p^{-2}$ corresponds to the more intuitive generator, which maps $X_f$ to $X_f Z_{\delta t^{\frac{1}{2}}(\bm{f})}^{1 \cdot s_f} Z_{\delta t^{-\frac{1}{2}}(\bm{f})}^{1 \cdot s_f}$.

The action on $X_f$ can also be written in terms of the cup product as
\begin{equation}
X_f \xmapsto{\sf{T}_p} X_f \prod_{e^{\prime} \in E} Z_{\delta \bm{e}^{\prime}}^{\frac{p-1}{2} \int \bm{f} \smile \bm{e}^{\prime} + \bm{e}^{\prime} \smile \bm{f}}.
\end{equation}
See Figure~\ref{fig: Tpk} for an illustration of this action.
\begin{figure*}
\centering
$X_f \xmapsto{\sf{T}_p^{-2k}}$
\adjincludegraphics[valign=c, scale=0.95]{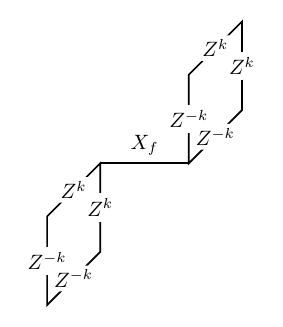},
\adjincludegraphics[valign=c, scale=0.95]{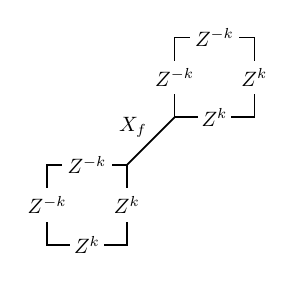},
\adjincludegraphics[valign=c, scale=0.95]{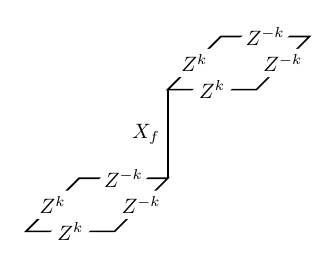}
\caption{The action of the powers of the Tsui-Wen entangler generator $\sf{T}_p$ on the Pauli $X$ operator illustrated on the dual lattice. To recover $\sf{T}_p$ itself, set $k = \frac{p - 1}{2}$ which is the inverse of $-2$ in $\Z_p$, for $p$ an odd prime.}
\label{fig: Tpk}
\end{figure*}

The unitary $\mathsf{T}_p$ preserves the $\mathbb{Z}_p$ 1-form symmetry.
Therefore, it defines a QCA on the $\mathbb{Z}_p$ 1-form symmetric algebra $\mathcal{A}_{\mathbb{Z}_p^{(1)}}/\mathfrak{I}_p$.
This QCA will be denoted by $\aTWp$.

\vspace*{\baselineskip}
\noindent{\bf Tsui-Wen entangler $\sf{T}_p$.}
We first define the Tsui-Wen entangler $\sf{T}_p$ as a unitary operator acting on the state space $\mathcal{H}_F$.
The action of $\sf{T}_p$ on a state in the computational basis is defined by \cite{Tsui:2019ykk}
\begin{equation}
\sf{T}_p \ket{\bm{a}}_F = e^{\frac{\pi i (p - 1)}{p} \int \tilde{\bm{a}} \smile \delta \tilde{\bm{a}}} \ket{\bm{a}}_F.
\label{eq: Tpk def}
\end{equation}
Here, $\tilde{\bm{a}}$ is a $\mathbb{Z}_p$-valued 1-cochain defined by
\begin{equation}
\tilde{\bm{a}}(e) \coloneq \sum_{f \in F} \bm{a}(f) \int \bm{f} \smile \bm{e}.
\end{equation}
The 1-cochain $\tilde{\bm{a}}$ can also be written more explicitly as
\begin{equation}
\tilde{\bm{a}}(e) = \bm{a}(t^{-\frac{1}{2}}(e)) s_e.
\label{eq: a tilde}
\end{equation}
Namely, $\tilde{\bm{a}}$ is the half-translation of $\bm{a}$ multiplied by the sign depending on the links.\footnote{The sign is multiplied so that the unitary defined by~\eqref{eq: Tpk def} commutes with the $\mathbb{Z}_p$ 1-form symmetry operators; cf.~\eqref{eq: T pk sym}.}
We note that both $\bm{a}$ and $\tilde{\bm{a}}$ are cochains on the direct lattice

Clearly, the Tsui-Wen entangler $\sf{T}_p$ has order $p$, i.e.,
\begin{equation}
\left( \sf{T}_p \right)^p = 1 \,,
\end{equation}
because $p - 1$ is an even integer.

Furthermore, $\sf{T}_p$ commutes with the $\mathbb{Z}_p$ 1-form symmetry operators~\eqref{eq: symmetry bubble Zp}.
Indeed, the action of the symmetry operator $X_{\partial c}$ on $\ket{\bm{a}}_F$ changes the 1-cochain $\tilde{\bm{a}}$ by a coboundary
\begin{equation}
\tilde{\bm{a}} \mapsto \tilde{\bm{a}} -\delta t^{\frac{1}{2}}(\bm{c}),
\end{equation}
which does not change the phase $e^{\frac{\pi i (p - 1)}{p} \int \tilde{\bm{a}} \smile \delta \tilde{\bm{a}}}$ in \eqref{eq: Tpk def}.
Thus, we have
\begin{equation}
\sf{T}_p X_{\partial c} = X_{\partial c} \sf{T}_p
\label{eq: T pk sym}
\end{equation}
for all cubes $c$.
More generally, $\sf{T}_p$ also commutes with the symmetry operator $\eta(\Sigma)$ supported on any 2-cycle $\Sigma$.

\vspace*{\baselineskip}
\noindent{\bf Action on local operators.}
Now, we compute the action of $\sf{T}_p$ on local operators.
Since $\sf{T}_p$ is diagonal in the computational basis, it acts trivially on the Pauli $Z$ operators, i.e., 
\begin{equation}
\sf{T}_p Z_f \sf{T}_p^{-1} = Z_f.
\label{eq: Tpk Z}
\end{equation}
On the other hand, the action of $\sf{T}_p$ on $X_f$ can be computed as
\begin{equation}
\sf{T}_p X_f \sf{T}_p^{-1} \ket{\bm{a}}_F = e^{\frac{\pi i (p - 1)}{p} \int \tilde{\bm{a}} \smile \delta \tilde{\bm{f}} + \tilde{\bm{f}} \smile \delta \tilde{\bm{a}} + \tilde{\bm{f}} \smile \delta \tilde{\bm{f}}} X_f \ket{\bm{a}}_F.
\label{eq: Tpk X1}
\end{equation}
Here, $\tilde{\bm{f}}$ is defined in the same way as we defined $\tilde{\bm{a}}$ from $\bm{a}$.
Namely, $\tilde{\bm{f}}$ is a 1-cochain defined by
\begin{equation}
\tilde{\bm{f}}(e) = \bm{f}(t^{-\frac{1}{2}}(e)) s_e = \delta_{e, t^{\frac{1}{2}}(f)} s_e.
\end{equation}
Equivalently, we have
\begin{equation}
\tilde{\bm{f}} = t^{\frac{1}{2}}(\bm{f}) s_f.
\end{equation}
Based on the above definition, one can compute the integral on the right-hand side of \eqref{eq: Tpk X1} as
\begin{equation}
\begin{aligned}
\int \tilde{\bm{a}} \smile \delta \tilde{\bm{f}} &= \tilde{\bm{a}}(\partial f) \quad \bmod p, \\
\int \tilde{\bm{f}} \smile \delta \tilde{\bm{a}} &= \tilde{\bm{a}}(\partial t(f)) \quad \bmod p, \\
\int \tilde{\bm{f}} \smile \delta \tilde{\bm{f}} &= 0 \quad \bmod p.
\end{aligned}
\label{eq: Tpk X2}
\end{equation}
Furthermore, by a direct computation using \eqref{eq: a tilde}, one can show that
\begin{equation}
e^{\frac{\pi i (p-1)}{p} \tilde{\bm{a}}(\partial f)} \ket{\bm{a}}_F = Z_{\delta t^{-\frac{1}{2}}(\bm{f})}^{ \frac{p - 1}{2} s_f} \ket{\bm{a}}_F.
\label{eq: Tpk X3}
\end{equation}
By plugging \eqref{eq: Tpk X2} and \eqref{eq: Tpk X3} into \eqref{eq: Tpk X1}, we find
\begin{equation}
\sf{T}_p X_f \sf{T}_p^{-1} = X_f Z_{\delta t^{\frac{1}{2}}(\bm{f})}^{ \frac{p - 1}{2} s_f} Z_{\delta t^{-\frac{1}{2}}(\bm{f})}^{ \frac{p - 1}{2} s_f} \,.
\label{eq: Tpk X}
\end{equation}
Equations~\eqref{eq: Tpk X} and \eqref{eq: Tpk Z} show that the action of $\sf{T}_p$ on $X_f$ and $Z_f$ is given by \eqref{eq: Tpk action}.

\subsection{$\mathbb{Z}_p$-QCA for odd prime $p$}
\label{sec: Zpk QCA}
In this subsection, we review the non-trivial Clifford QCAs constructed in \cite{MengSun2026} for odd prime $p$.
These QCAs are defined as entanglers of the Walker-Wang ground states based on chiral $\mathbb{Z}_p$ MTCs.
See also \cite{Haah_2021} for an earlier construction of non-trivial QCAs for odd prime $p$.
These QCAs together with the 3-fermion QCA exhaust all equivalence classes of Clifford QCAs \cite{Haah:2022yyo} in 3+1d.

We consider a 3d cubic lattice with a single $\mathbb{Z}_p$ qudit on each face.
For each $k \in \{1, \cdots, p-1\}$, one can define a non-trivial Clifford QCA $\beta_p^{(k)}$ by the following action on local operators \cite{MengSun2026}:
\begin{align}
\beta_p^{(k)}(Z_f) &= Z_f \prod_{e^{\prime} \in E} \left. G_{e^{\prime}}^{(-k)} \right.^{\frac{1}{2k} \int \bm{f} \smile \bm{e}^{\prime}},
\label{eq: beta pk Z}
\\
\beta_p^{(k)}(X_f) &= U_f^{(-k)} \prod_{f^{\prime} \in F} \left[\beta_p^{(k)}(Z_{f^{\prime}})\right]^{k \int \bm{f}^{\prime} \smile_1 \bm{f}}.
\label{eq: beta pk X}
\end{align}
Here, $\frac{1}{2k}$ is the multiplicative inverse of $2k$ in $\mathbb{Z}_p$, which exists for any $k \in \{ 1, 2, \ldots, p - 1 \}$, and $G_e^{(-k)}$ and $U_f^{(-k)}$ are defined by \cite{MengSun2026}
\begin{align}
G_e^{(-k)} &\coloneq X_{\delta \bm{e}} \prod_{f^{\prime} \in F} Z_{f^{\prime}}^{-k \int \delta \bm{e} \smile_1 \bm{f}^{\prime}}, 
\label{eq: Gek}
\\
U_f^{(-k)} &\coloneq X_f \prod_{f^{\prime} \in F} Z_{f^{\prime}}^{-k \int \bm{f}^{\prime} \smile_1 \bm{f}}.
\label{eq: Ufk}
\end{align}
The QCA $\beta_p^{(k)}$ is called a $\mathbb{Z}_p$ QCA in \cite{MengSun2026}.

The QCA defined above leaves $Z_{\partial c}$ invariant for all cubes $c$ \cite{MengSun2026}, that is,
\begin{equation}
\beta_{p}^{(k)}(Z_{\partial c}) = Z_{\partial c}.
\label{eq: beta pk Z partial c}
\end{equation}
This implies that the following QCA preserves the $\mathbb{Z}_p$ 1-form symmetry operators $X_{\partial c}$ for any $c$:
\begin{equation}
\tilde{\beta}_p^{(k)} \coloneq h^{-1} \circ \beta_p^{(k)} \circ h.
\label{eq: tilde beta pk def}
\end{equation}
Here, $h$ is the QCA defined by the conjugation action of the Hadamard gates on all faces.
Specifically, the action of $h$ on $X_f$ and $Z_f$ are given by
\begin{equation}
h(X_f) = Z_f, \quad
h(Z_f) = X_f^{\dagger}.
\label{eq: h QCA}
\end{equation}
It immediately follows from \eqref{eq: beta pk Z partial c} and \eqref{eq: h QCA} that
\begin{equation}
\tilde{\beta}_p^{(k)}(X_{\partial c}) = X_{\partial c}.
\label{eq: beta X partial c}
\end{equation}
Thus, $\tilde{\beta}_p^{(k)}$ is a $\mathbb{Z}_p$ 1-form symmetric QCA.
In what follows, we will compute the action of $\tilde{\beta}_p^{(k)}$ on local symmetric operators $X_f$ and $Z_{\delta \bm{e}}$.

\subsubsection{Action on $X_f$}
We first compute the action of $\tilde{\beta}_p^{(k)}$ on $X_f$.
Due to~\eqref{eq: beta pk Z} and \eqref{eq: h QCA}, we have
\begin{equation}
\tilde{\beta}_p^{(k)}(X_f) = X_f \prod_{e^{\prime} \in E} \left. \widetilde{G}_{e^{\prime}}^{(-k)} \right.^{\frac{1}{2k} \int \bm{f} \smile \bm{e}^{\prime}},
\label{eq: beta tilde}
\end{equation}
where $\widetilde{G}_e^{(-k)}$ is defined by
\begin{equation}
\begin{aligned}
\widetilde{G}_e^{(-k)}
&\coloneq h^{-1}(G_e^{(-k)}) \\
&= Z_{\delta \bm{e}}^{\dagger} \prod_{f^{\prime} \in F} X_{f^{\prime}}^{-k \int \delta \bm{e} \smile_1 \bm{f}^{\prime}}.
\end{aligned}
\end{equation}
Since we have
\begin{equation}
\int \bm{f} \smile \bm{e}^{\prime} = \delta_{e^{\prime}, t^{\frac{1}{2}}(f)} s_f,
\end{equation}
equation~\eqref{eq: beta tilde} can also be written as
\begin{equation}
\tilde{\beta}_p^{(k)}(X_f) = X_f \left.\widetilde{G}_{t^{\frac{1}{2}}(f)}^{(-k)}\right.^{\frac{1}{2k}s_f}.
\label{eq: beta tilde2}
\end{equation}
Furthermore, by using the identity
\begin{equation}
\delta(\bm{e} \smile_1 \bm{f}^{\prime}) = \delta \bm{e} \smile_1 \bm{f}^{\prime} - \bm{e} \smile_1 \delta \bm{f}^{\prime} + \bm{e} \smile \bm{f}^{\prime} - \bm{f}^{\prime} \smile \bm{e} \,,
\end{equation}
one can rewrite $\widetilde{G}_e^{(-k)}$ as
\begin{equation}
\begin{aligned}
\widetilde{G}_e^{(-k)} &= Z_{\delta \bm{e}}^{\dagger} \prod_{f^{\prime} \in F} X_{f^{\prime}}^{-k \int \bm{e} \smile_1 \delta \bm{f}^{\prime} - \bm{e} \smile \bm{f}^{\prime} + \bm{f}^{\prime} \smile \bm{e}} \\
&= Z_{\delta \bm{e}}^{\dagger} X_{t^{\frac{1}{2}}(e)}^{ks_e} X_{t^{-\frac{1}{2}}(e)}^{-ks_e} \prod_{c^{\prime} \in C} X_{\partial c^{\prime}}^{-k \int \bm{e} \smile_1 \bm{c}^{\prime}}.
\end{aligned}
\end{equation}
By plugging this into \eqref{eq: beta tilde2}, we find
\begin{equation}
\tilde{\beta}_{p}^{(k)}(X_f) = Z_{\delta t^{\frac{1}{2}}(\bm{f})}^{-\frac{1}{2k}s_f} X_{t(f)}^{\frac{1}{2}} X_f^{\frac{1}{2}} \prod_{c^{\prime} \in C} X_{\partial c^{\prime}}^{-\frac{1}{2}s_f \int t^{\frac{1}{2}}(\bm{f}) \smile_1 \bm{c}^{\prime}}.
\label{eq: beta tilde X}
\end{equation}
This action can be illustrated as shown in Figure~\ref{fig: beta tilde X}.
\begin{figure*}[t]
\centering
$\tilde{\beta}_p^{(k)}(X_f) =$
\adjincludegraphics[valign=c, trim={0, 10, 0, 10}]{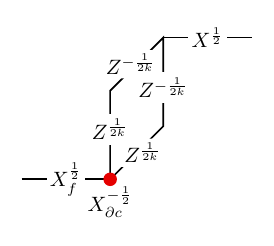},
\adjincludegraphics[valign=c, trim={0, 10, 0, 10}]{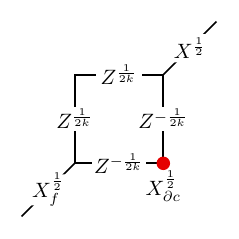},
\adjincludegraphics[valign=c, trim={0, 10, 0, 10}]{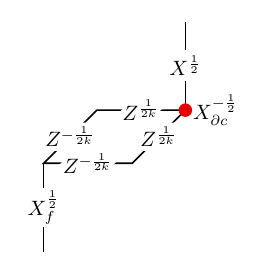}
\caption{The action of $\tilde{\beta}_p^{(k)}$ on the Pauli $X$ operator. The figure is drawn on the dual lattice.}
\label{fig: beta tilde X}
\end{figure*}

\subsubsection{\texorpdfstring{Action on $Z_{\delta \bm{e}}$}{Action on Z delta e}}
Next, we compute the action of $\tilde{\beta}_p^{(k)}$ on $Z_{\delta \bm{e}}$.
To this end, we first compute the action of $\beta_p^{(k)}$ on $X_{\delta \bm{e}}$.
Equation~\eqref{eq: beta pk X} implies that $\beta_p^{(k)}$ acts on $X_{\delta \bm{e}}$ as
\begin{equation}
\beta_p^{(k)}(X_{\delta \bm{e}}) = U_{\delta \bm{e}}^{(-k)} \prod_{f^{\prime} \in F} \left[ \beta_p^{(k)}(Z_{f^{\prime}}) \right]^{k \int \bm{f}^{\prime} \smile_1 \delta \bm{e}},
\label{eq: beta pk X delta e}
\end{equation}
where $U_{\delta \bm{e}}^{(-k)}$ is defined by
\begin{equation}
U_{\delta \bm{e}}^{(-k)} = X_{\delta \bm{e}} \prod_{f^{\prime} \in F} Z_{f^{\prime}}^{-k \int \bm{f}^{\prime} \smile_1 \delta \bm{e}}.
\label{eq: U delta e Zp}
\end{equation}
Here, we have a technical comment on the derivation of~\eqref{eq: beta pk X delta e}.
In general, the action of $\beta_p^{(k)}$ on $X_{\delta \bm{e}}$ can be written as
\begin{equation}
\beta_p^{(k)}(X_{\delta \bm{e}}) = \prod_{f^{\prime} \in F} \left[\beta_p^{(k)}(X_{f^{\prime}})\right]^{\delta \bm{e}(f^{\prime})},
\label{eq: beta pk X delta e2}
\end{equation}
where the order of the product is arbitrary.
The right-hand side of the above equation is a product of $U_f^{(-k)}$'s and $\beta_p^{(k)}(Z_f)$'s due to \eqref{eq: beta pk X}.
Equation~\eqref{eq: beta pk X delta e} is obtained by arranging the order of the product in~\eqref{eq: beta pk X delta e2} so that all $\beta_p^{(k)}(Z_f)$'s act before $U_f^{(-k)}$'s.
We note that $U_f^{(-k)}$'s in the product are automatically arranged in a way that all Pauli $Z$ operators act before any Pauli $X$ operator as in \eqref{eq: U delta e Zp}.

Due to~\eqref{eq: beta pk X delta e} and \eqref{eq: h QCA}, we can compute the action of $\tilde{\beta}_p^{(k)}$ on $Z_{\delta \bm{e}}$ as
\begin{equation}
\tilde{\beta}_p^{(k)}(Z_{\delta \bm{e}}) = \left.\widetilde{U}_{\delta \bm{e}}^{(-k)}\right.^{\dagger} \prod_{f^{\prime} \in F} \left[\tilde{\beta}_p^{(k)}(X_{f^{\prime}})\right]^{-k \int \bm{f}^{\prime} \smile_1 \delta \bm{e}},
\label{eq: tilde beta pk Z}
\end{equation}
where $\widetilde{U}_{\delta \bm{e}}^{(-k)}$ is defined by
\begin{equation}
\begin{aligned}
\widetilde{U}_{\delta \bm{e}}^{(-k)}
&\coloneq h^{-1}(U_{\delta \bm{e}}^{(-k)}) \\
&= Z_{\delta \bm{e}}^{\dagger} \prod_{f^{\prime} \in F} X_{f^{\prime}}^{-k \int \bm{f}^{\prime} \smile_1 \delta \bm{e}}.
\end{aligned}
\end{equation}
By using the identity
\begin{equation}
\delta (\bm{f}^{\prime} \smile_1 \bm{e}) = \delta \bm{f}^{\prime} \smile_1 \bm{e} + \bm{f}^{\prime} \smile_1 \delta \bm{e} + \bm{f}^{\prime} \smile \bm{e} - \bm{e} \smile \bm{f}^{\prime} \,,
\label{eq: delta f cup1 e}
\end{equation}
one can rewrite $\widetilde{U}_{\delta \bm{e}}^{(-k)}$ as
\begin{equation}
\begin{aligned}
\widetilde{U}_{\delta \bm{e}}^{(-k)}
&= Z_{\delta \bm{e}}^{\dagger} \prod_{f^{\prime} \in F} X_{f^{\prime}}^{k \int \delta \bm{f}^{\prime} \smile_1 \bm{e} + \bm{f}^{\prime} \smile \bm{e} - \bm{e} \smile \bm{f}^{\prime}} \\
&= Z_{\delta \bm{e}}^{\dagger} X_{t^{\frac{1}{2}}(e)}^{-k s_e} X_{t^{-\frac{1}{2}}(e)}^{k s_e} \prod_{c^{\prime} \in C} X_{\partial c^{\prime}}^{k \int \bm{c}^{\prime} \smile_1 \bm{e}}.
\end{aligned}
\label{eq: U tilde delta e}
\end{equation}
Furthermore, by using the same identity~\eqref{eq: delta f cup1 e}, one can also rewrite $\tilde{\beta}_p^{(k)}(Z_{\delta \bm{e}})$ in \eqref{eq: tilde beta pk Z} as
\begin{widetext}
\begin{equation}
\begin{aligned}
\tilde{\beta}_p^{(k)}(Z_{\delta \bm{e}})
&= \left.\widetilde{U}_{\delta \bm{e}}^{(-k)}\right.^{\dagger} \left[\tilde{\beta}_p^{(k)}(X_{t^{\frac{1}{2}}(e)})\right]^{-ks_e} \left[\tilde{\beta}_p^{(k)}(X_{t^{-\frac{1}{2}}(e)})\right]^{ks_e} \prod_{c^{\prime} \in C} \left[\tilde{\beta}_p^{(k)}(X_{\partial c^{\prime}})\right]^{k \int \bm{c}^{\prime} \smile_1 \bm{e}} \\
&= \left.\widetilde{U}_{\delta \bm{e}}^{(-k)}\right.^{\dagger} \left[\tilde{\beta}_p^{(k)}(X_{t^{\frac{1}{2}}(e)})\right]^{-ks_e} \left[\tilde{\beta}_p^{(k)}(X_{t^{-\frac{1}{2}}(e)})\right]^{ks_e} \prod_{c^{\prime} \in C} X_{\partial c^{\prime}}^{k \int \bm{c}^{\prime} \smile_1 \bm{e}}.
\end{aligned}
\end{equation}
Here, the second equality follows from \eqref{eq: beta X partial c}.
By plugging \eqref{eq: beta tilde X} and \eqref{eq: U tilde delta e} into the above equation, we find
\begin{equation}
\tilde{\beta}_p^{(k)}(Z_{\delta \bm{e}}) = Z_{\delta t(\bm{e})}^{\frac{1}{2}} Z_{\delta \bm{e}}^{\frac{1}{2}} X_{t^{\frac{3}{2}}(e)}^{-\frac{1}{2}ks_e} X_{t^{\frac{1}{2}}(e)}^{ks_e} X_{t^{-\frac{1}{2}}(e)}^{-\frac{1}{2}ks_e} \prod_{c^{\prime} \in C} X_{\partial c^{\prime}}^{\frac{1}{2}k \int t(\bm{e}) \smile_1 \bm{c}^{\prime} - \bm{e} \smile_1 \bm{c}^{\prime}}.
\label{eq: beta tilde Z}
\end{equation}
\end{widetext}
See Figure~\ref{fig: beta tilde Z} for an illustration of this action.
\begin{figure}
\centering
$\tilde{\beta}_p^{(k)}(Z_{\delta \bm{e}}) =$
\adjincludegraphics[valign=c, trim={10, 10, 10, 10}]{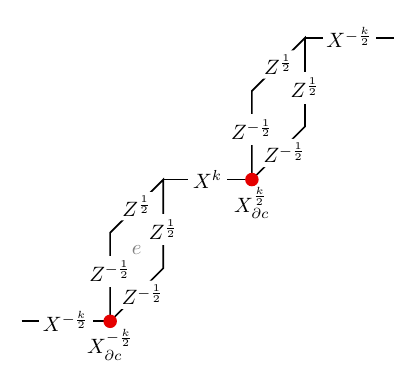}
\caption{The action of $\tilde{\beta}_p^{(k)}$ on the flux operator $Z_{\delta \bm{e}}$ for an $x$-link $e$. Here, the bottom-left plaquette is the dual of $e$. The action on $Z_{\delta \bm{e}}$ for $y$-links and $z$-links can also be illustrated similarly.}
\label{fig: beta tilde Z}
\end{figure}

\subsection{$\mathbb{Z}_p$ QCA is generated by $\mathsf{S}_p$ and $\sf{T}_p$}
\label{sec: Zpk QCA generated by D and T}

In this subsection, we show that the non-trivial Clifford QCA $\tilde{\beta}_p^{(k)}$ for any $k$ is generated by the $\mathbb{Z}_p$ KWW operator $\mathsf{S}_p$ and the Tsui-Wen entangler $\sf{T}_p$.
More specifically, we will show the following equality of QCAs on the $\mathbb{Z}_p$ 1-form symmetric algebra $\mathcal{A}_{\mathbb{Z}_p^{(1)}} / \mathfrak{I}_p$:
\begin{equation}
\ba
&~\aKWWp \circ \aTWp^{2k} \circ \aKWWp \circ \aTWp^{\frac{1}{2k}} \circ \aKWWp \cr 
=&~ \alpha_t \circ \tilde{\beta}_p^{(-k)} \circ \alpha_{\text{C}}^{(2k)}.
\label{eq: beta mixing}
\ea \end{equation}
Here, $\alpha_{\text{C}}^{(2k)}$ is a generalized charge conjugation QCA, which is the restriction of the following QCA on the full tensor-product algebra:
\begin{equation}
\check{\alpha}_{\text{C}}^{(2k)}(X_f) = X_f^{2k}, \qquad
\check{\alpha}_{\text{C}}^{(2k)}(Z_f) = Z_f^{\frac{1}{2k}}.
\end{equation}
We note that $\check{\alpha}_{\text{C}}^{(2k)}$ does not preserve the symmetry operator $X_{\partial c}$ because $\check{\alpha}_{\text{C}}^{(2k)}(X_{\partial c}) = X_{\partial c}^{2k}$.
Nevertheless, it does preserve the ideal $\mathfrak{I}_p$ generated by $1-X_{\partial c}$.
Hence, $\alpha_{\text{C}}^{(2k)}$ is well-defined as a QCA on the quotient algebra $\mathcal{A}_{\mathbb{Z}_p^{(1)}}/\mathfrak{I}_p$.

Equation~\eqref{eq: beta mixing} is our second example of lattice fusion rules that mix with non-trivial QCAs.
This result shows that $\mathsf{S}_p$ and $\mathsf{T}_p$ generate all Witt non-trivial Clifford QCAs; $\alpha_\text{C}^{(2k)}$ is just an on-site unitary (and furthermore it is Clifford) so it is Witt trivial.

When $2k=1$ mod $p$ (i.e., when $k = \frac{p+1}{2}$), $\alpha_{\text{C}}^{(2k)}$ becomes the identity QCA.
Thus, the fusion rule~\eqref{eq: beta mixing} in this case reduces to
\begin{equation}
(\aKWWp \circ \aTWp)^2 \circ \aKWWp = \alpha_t \circ \tilde{\beta}_p^{\left(\frac{p-1}{2}\right)} \,.
\label{eq: DTDTD minus}
\end{equation}
On the other hand, when $2k=-1$ mod $p$ (i.e., when $k = \frac{p-1}{2}$), $\alpha_{\text{C}}^{(2k)}$ becomes the ordinary charge conjugation QCA.
Thus, the fusion rule~\eqref{eq: beta mixing} reduces to
\begin{equation}
(\aKWWp \circ \aTWp^{-1})^2 \circ \aKWWp = \alpha_t \circ \alpha_\text{C} \circ \tilde{\beta}_p^{\left(\frac{p+1}{2}\right)} \,.
\label{eq: DTDTD plus}
\end{equation}
Here, $\alpha_{\text{C}} \coloneq \alpha_{\text{C}}^{(-1)}$ denotes the ordinary charge conjugation.
The above two equations can be written more concisely as
\begin{align}
 (\aKWWp \circ \aTWp)^3 &= \alpha_t \circ \tilde{\beta}_p^{\left(\frac{p-1}{2}\right)} \circ \aTWp \,, \label{eq: DTDTDT minus} \\
 (\aKWWp \circ \aTWp^{-1})^3 &= \alpha_t \circ \alpha_\text{C} \circ \tilde{\beta}_p^{\left(\frac{p+1}{2}\right)} \circ \aTWp^{-1} \label{eq: DTDTDT plus} \,, 
\end{align}
The relation to the fusion rules in the continuum will be discussed in Section~\ref{sec: comparison Zp}.

In what follows, we will show \eqref{eq: beta mixing}.
More specifically, we will compute the action of $\mathsf{S}_p\mathsf{T}_p^{2k}\mathsf{S}_p\mathsf{T}_p^{\frac{1}{2k}}\mathsf{S}_p$ and compare it with $\alpha_t \circ \tilde{\beta}_p^{(-k)} \circ \alpha_{\text{C}}^{(2k)}$.

Let us first compute the action on $X_f$.
Since $\mathsf{S}_p$ maps $X_f$ to $Z_{\delta e}$ (or its Hermitian conjugate) and $\mathsf{T}_p$ leaves $Z_{\delta \bm{e}}$ invariant, it follows that
\begin{equation}
X_f \xmapsto{\sf{S}_p \sf{T}_p^{\frac{1}{2k}} \sf{S}_p} X_{t(f)}^{\dagger} \,.
\end{equation}
Thus, the calculation reduces to finding the action of $\sf{S}_p \sf{T}_p^{2k}$ on $X_{t(f)}^\dag$.

Recalling that $\mathsf{S}_p$ and $\sf{T}_p$ act on local operators as in~\eqref{eq: Dp action} and \eqref{eq: Tpk action}, one can check that
\begin{equation}
\begin{aligned}
X_{t(f)}^\dag
&\xmapsto{\sf{T}_p^{2k}} X_{t(f)}^\dag Z_{\delta t^{\frac{3}{2}}(\bm{f})}^{ k s_f} Z_{\delta t^{\frac{1}{2}}(\bm{f})}^{ k s_f} \\
&\xmapsto{\sf{S}_p} Z_{\delta t^{\frac{3}{2}}(\bm{f})}^{s_f} X_{t^2(f)}^{k} X_{t(f)}^{k} \,.
\end{aligned}
\label{eq: DTDTD X}
\end{equation}
This shows that
\begin{equation}
 X_f \xmapsto{\mathsf{S}_p \sf{T}_p^{2k} \mathsf{S}_p \sf{T}_p^{\frac{1}{2k}} \mathsf{S}_p} Z_{\delta t^{\frac{3}{2}}(\bm{f})}^{s_f} X_{t^2(f)}^{k} X_{t(f)}^{k}
\end{equation}
for any $k = 1, 2, \cdots, p-1$.

Similarly, one can perform the same kind of calculation for $Z_{\delta \bm{e}}$.
A direct computation shows that $\mathsf{S}_p\mathsf{T}_p^{2k}\mathsf{S}_p\mathsf{T}_p^{\frac{1}{2k}}\mathsf{S}_p$ acts on $Z_{\delta e}$ as
\be
Z_{\delta \bm{e}} \mapsto Z_{\delta t^2(e)}^{\frac{1}{4k}} Z_{\delta t(e)}^{\frac{1}{4k}}X_{t^{\frac{5}{2}}(e)}^{\frac{1}{4}s_e} X_{t^{\frac{3}{2}}(e)}^{-\frac{1}{2}s_e} X_{t^{\frac{1}{2}}(e)}^{\frac{1}{4}s_e}
\ee
for any $k = 1, 2, \cdots, p-1$

When comparing these expressions to the action of $\tilde{\beta}_p^{(-k)}$, we are allowed to ignore any products of $X_{\partial c}$, because the comparison is being done in $\mc{A}_{\Z_p^{(1)}} / \mathfrak{I}_p$. With this caveat, we can see that the action looks structurally very similar.
More specifically, on the quotient algebra $\mc{A}_{\Z_p^{(1)}} / \mathfrak{I}_p$, the action of the composite QCA $\tilde{\beta}_p^{(-k)} \circ \alpha_{\text{C}}^{(2k)}$ can be written as
\begin{equation}
\begin{aligned}
&X_f \xmapsto{\alpha_{\text{C}}^{(2k)}} X_f^{2k} \xmapsto{\tilde{\beta}_p^{(-k)}} Z_{\delta t^{\frac{1}{2}}(\bm{f})}^{s_f} X_{t(f)}^{k} X_{f}^{k}, \\
&Z_{\delta e} \xmapsto{\alpha_{\text{C}}^{(2k)}} Z_{\delta e}^{\frac{1}{2k}} \xmapsto{\tilde{\beta}_p^{(-k)}} Z_{\delta t(e)}^{\frac{1}{4k}} Z_{\delta e}^{\frac{1}{4k}}X_{t^{\frac{3}{2}}(e)}^{\frac{1}{4}s_e} X_{t^{\frac{1}{2}}(e)}^{-\frac{1}{2}s_e} X_{t^{-\frac{1}{2}}(e)}^{\frac{1}{4}s_e}.
\end{aligned}
\label{eq: beta alpha}
\end{equation}
This computation shows that $\tilde{\beta}_p^{(-k)} \circ \alpha_{\text{C}}^{(2k)}$ agrees with the action of $\mathsf{S}_p \sf{T}_p^{2k} \mathsf{S}_p \sf{T}_p^{\frac{1}{2k}} \mathsf{S}_p$ up to lattice translation $\alpha_t$.
Thus, we find that \eqref{eq: beta mixing} holds as an equality of QCAs on $\mathcal{A}_{\mathbb{Z}_p^{(1)}}/\mathfrak{I}_p$.

\subsection{Group of QCAs Generated by $\mathsf{S}_p$, $\mathsf{T}_p$, and Generalized Charge Conjugation}
\label{sec:ZpQCAgroup}

In this section we are going to outline three assumptions under which the group structure of the QCA generated by $\aKWWp, \aTWp, \alpha_C^{(k)}, \tilde{\beta}_p^{(k)}$ can be fully determined. We will find that this group is isomorphic to $SL(2, \Z_p) \times W_p$, where $W_p$ is either $\Z_2 \times \Z_2$ or $\Z_4$, depending on whether $p \equiv 1$ or $p \equiv 3 \mod 4$.\footnote{This group contains a subgroup generated by $\aKWWp$ and $\aTWp$, which is briefly discussed in Section~\ref{sec: comparison Zp}. This subgroup is isomorphic to $SL(2, \mathbb{Z}_p) \times \mathbb{Z}_{n_p}$ with $n_p = 2$ or $4$ depending on whether $p \equiv 1$ or $3$ mod $4$, so it is a strict subgroup only if $p \equiv 1 \mod 4$.}

We remark that the first two assumptions are physically reasonable, and one has already been used multiple times before. The third assumption is almost obvious, but we mention it nonetheless, to make the remaining arguments formal.

To be more precise, we consider a subgroup of the group
\begin{equation}\label{eq:translation_inv_qca_group}
\faktor{\mathsf{QCA}_{\mathbb{Z}_p^{(1)}}}{(\mathsf{FDC}_{\mathbb{Z}_p^{(1)}} \times \mathsf{Shift})},
\end{equation}
where $\mathsf{QCA}_{\mathbb{Z}_p^{(1)}}$ is the group of all translation invariant QCAs on $\mathcal{A}_{\mathbb{Z}_p^{(1)}} / \mathfrak{I}_p$ and $\mathsf{FDC}_{\mathbb{Z}_p^{(1)}}$ is its subgroup consisting of finite-depth circuits with $\mathbb{Z}_p$ 1-form symmetric local gates.
We recall that $\mathsf{Shift}$ denotes the group of lattice translations.
We denote the equivalence class of $\alpha$ as $[\alpha]$.

\subsubsection{Assumptions}

\noindent{\bf Assumption 1.} \textit{{Motivated by the assumed completeness of the classification of SPT phases, we assume a related but a priori slightly stronger statement about operators, formalized as follows}: let $\alpha$ be a QCA on the quotient, which is extendable to a QCA $\tilde{\alpha}$ on the full algebra. If $\tilde{\alpha}$ is a finite depth circuit (with possibly non-symmetric gates), then
\be
[\alpha] = [\alpha_{\text{TW}, p}^j] \quad \text{for some } j \in \{0, \cdots, p - 1 \}
\ee}

Let us also define the Legendre symbol, which will help us state Assumption 2 more concisely. For any $a \in \Z_p^\times = \{1, \cdots, p - 1 \}$
\be
\paren{\frac{a}{p}} \coloneqq \begin{cases}
 1 & \text{if } \exists x \text{ s.t. } a \equiv x^2 \mod p \,, \\
 -1 & \text{if } \nexists x \text{ s.t. } a \equiv x^2 \mod p \,.
\end{cases}
\ee

In our context we ignore the situation where $p \mid a$, which simplifies the definition a little. The sign of the Legendre symbol determines whether $a$ is a \emph{quadratic residue} mod $p$. {In turn, the Legendre symbol $\paren{\frac{k}{p}}$ determines the Witt class of the Walker-Wang model disentangled by $\tilde{\beta}_p^{(k)}$.} 

\noindent{\bf Assumption 2.} \textit{Any two disentanglers of {two Walker-Wang model Hamiltonians} are related by a finite-depth circuit (with possibly non-symmetric gates) and translations if, and only if, the Modular Tensor Categories the WW-models are based on are Witt-equivalent. To be more specific, we apply this to the classification of $\tilde{\beta}_p^{(k)}$. Let $k,l \in \Z_p^\times$, then:
\be
\paren{\frac{k}{p}} = \paren{\frac{l}{p}} \quad \iff \quad \tilde{\beta}_p^{(k)} \sim \tilde{\beta}_p^{(l)} \,,
\ee
where we use $\sim$ to denote the equivalence up to FDC with possibly non-symmetric gates and translations.}

By Assumption 1, this can be strengthened to
\be\label{eq:beta_equiv_up_to_T}
\paren{\frac{k}{p}} = \paren{\frac{l}{p}} \quad \iff \quad \exists m \text{ s.t. } [\tilde{\beta}_p^{(k)}] = [\tilde{\beta}_p^{(l)}] [\alpha_{\text{TW}, p}^m] \,.
\ee

\noindent{\bf Assumption 3.} \textit{We assume that the relations stated below, in Eqs.~\eqref{eq:sty_relations}, \eqref{eq:s2c_relations}, \eqref{eq:c_k_relations}, \eqref{eq:additional_identities}, \eqref{eq:CS_CT},
\eqref{ApplyPi}, \eqref{eq:betaT_equalities}, \eqref{PecanPi}, \eqref{WalnutPi} and \eqref{YuzuPi}, are all the independent relations that these QCA satisfy (that is to say, any other relations can be deduced from those).}\footnote{To be more precise, we assume that those equations are sufficient to give the correct ``set of relations'' $R$, in the sense of a finitely presented group on a set of letters $X$, $\langle X \mid R \rangle$. In that context, the group of QCA generated by $S, T, \beta^{(k)}$ and $C_g$ is isomorphic to the free group on $X$ quotiented by the normalizer of $R$.} 

{We note that all the relations which are given as lemmas or corollaries (see the following sections) require Assumptions 1 and 2 in their proofs.} 

Let us now recall the QCA we have constructed, and the relations they satisfy. We have the Kramers-Wannier-Wegner duality $\aKWWp$, and the Tsui-Wen SPT entangler $\aTWp$. We define abbreviated notation for their equivalence classes because that will significantly simplify the notation
\begin{align}
 S &\coloneqq [\aKWWp] \,, & T &\coloneqq [\aTWp] \,.
\end{align}

We also define $\beta^{(k)} \coloneqq [\tilde{\beta}_p^{(k)}]$. These satisfy the relations
\begin{align}\label{eq:sty_relations}
 S^4 &= 1 \,, & T^p &= 1 \,, & (ST)^3 = Y \coloneqq \beta^{(\frac{p-1}{2})} T
\end{align}
where we obviously denote by $1$ the equivalence class of the identity QCA, and we define $Y$ for convenience, because it will appear often in later expressions.

Lastly, we recall that there is a ``generalized charge conjugation'', $\alpha_C^{(k)}$. It is well known that the multiplicative group $\Z_p^\times$ is cyclic of order $p - 1$. We pick a generator $g \in \Z_p^\times$ for it, i.e. some $g$ such that $g^{p-1} = 1$, but $g^n \neq 1$ for any $1 \leq n < p - 1$. Then $C_g \coloneqq [\alpha_C^{(g)}]$ generates all the charge conjugations. In particular
\be\label{eq:s2c_relations}
S^2 = C_g^{\frac{p-1}{2}} = [\alpha_C^{(-1)}]
\ee
is the usual charge conjugation ($X \mapsto X^\dag, Z \mapsto Z^\dag$). This is independent of $g$, because $\frac{p - 1}{2}$ is the only order 2 element in $\Z_{p-1}$, just like $-1$ is the only order 2 element in $\Z_p^\times$.

We may also write $C_{k} \coloneqq [\alpha_C^{(k)}]$ for any other $k$, with the remark that
\begin{align}\label{eq:c_k_relations}
C_k C_l &= C_{k l} \,, & C_g^m C_g^n &= C_{g^m} C_{g^n} = C_{g^{m+n}} = C_g^{m + n} \,.
\end{align}
We use both notations because there is no general formula to express $k$ as a power of $g$.

This generalized charge conjugation allows us to write another family of relations satisfied by those QCA
\be\label{eq:additional_identities}
S T^{-2k} S T^{-\frac{1}{2k}} S = \beta^{(k)} C_{-2k} \quad \text{for } k \in \Z_p^\times \,,
\ee
mirroring \eqref{eq: beta mixing} (but without the translation), as well as
\begin{align}\label{eq:CS_CT}
 C_g S &= S C_g^{-1} & C_g T &= T^{g^{-2}} C_g \,.
\end{align}

These last equalities can be checked explicitly: they are exact at the level of the QCA themselves.

\subsubsection{General properties}

In this section we derive some general identities, some exact at the level of QCA, some utilizing the assumptions, that work for any odd prime $p$. Later we will specialize to the two cases $p \equiv 1,3 \mod 4$,

By direct computation we derive the identities
\begin{align}\label{eq:Cbeta}
 C_g \beta^{(k)} &= C_g S T^{-2k} S T^{-\frac{1}{2k}} S C_{- \frac{1}{2k}} \nonumber \\
 &= S C_g^{-1} T^{-2k} S T^{-\frac{1}{2k}} S C_{- \frac{1}{2k}} \nonumber \\
 &= ST^{-2g^2 k} C_g^{-1} S T^{-\frac{1}{2k}} S C_{- \frac{1}{2k}} = \cdots \nonumber \\
 &= S T^{-2 g^2 k} S T^{-\frac{1}{2 g^2 k}} S C_{-\frac{1}{2gk}} = \beta^{(g^2 k)} C_g \,,
\end{align}
and
\begin{align}
 \paren{\beta^{(k)}}^{-1} &= C_{-2k} S^{-1} T^\frac{1}{2k} S^{-1} T^{2k} S^{-1} \nonumber \\
 &= C_{-2k} S T^\frac{1}{2k} S T^{2k} S C_{-1} \nonumber \\
 &= \cdots = ST^{2k} S T^{\frac{1}{2k}} S C_{\frac{1}{2k}} = \beta^{(-k)}.
\end{align}
Here we use $S^{-1} = S C_{-1}$ with $C_{-1} = C_g^{\frac{p-1}{2}}$ central.

We also note that we may rewrite the additional conditions \eqref{eq:additional_identities} as
\be\label{eq:additional_identities_rewrite}
S T^{-2k} S T^{-\frac{1}{2k}} S T^{-2k} C_{-\frac{1}{2k}} = \beta^{(k)} T^{-\frac{1}{2k}}.
\ee

\begin{lemma}\label{lem:beta_T_commutation}
 For any $k \in \Z_p^\times$, $\beta^{(k)}$ commutes with $T$:
 \be\label{ApplyPi}
\beta^{(k)} T = T \beta^{(k)} \,.
 \ee
\end{lemma}

\begin{proof}
 We are going to check that $\beta^{(k)} T^{-\frac{1}{2k}} = T^{-\frac{1}{2k}} \beta^{(k)}$. If so, $-\frac{1}{2k}$ is obviously non-zero, and so $\beta^{(k)}$ also commutes with $T$ itself. First, we see that
 \be
 \beta^{(k)} T^{-\frac{1}{2k}} \paren{\beta^{(k)}}^{-1} T^{\frac{1}{2k}} = \beta^{(k)} T^{-\frac{1}{2k}} \beta^{(-k)} T^{\frac{1}{2k}} \,,
 \ee
 where the left hand side is extendable to an FDC (with possibly non-symmetric gates) on the full algebra. This is because $T$ is extendable to an FDC, and $\beta^{(k)} T^{-\frac{1}{2k}} \paren{\beta^{(k)}}^{-1}$ is the conjugate of an FDC by a QCA, and hence an FDC. Therefore, by Assumption 1, there is some $j$ such that
 \be\label{eq:betaTbetaT}
 \beta^{(k)} T^{-\frac{1}{2k}} \beta^{(-k)} T^{\frac{1}{2k}} = T^j.
 \ee

 Now conjugate $\beta^{(k)} T^{-\frac{1}{2k}}$ by $S$
 \begin{align}
  S^{-1} \beta^{(k)} T^{-\frac{1}{2k}} S &= T^{-2k} S T^{-\frac{1}{2k}} S C_{-\frac{1}{2k}} T^{-\frac{1}{2k}} S \nonumber \\
  &= T^{-2k} S T^{-\frac{1}{2k}} S T^{-2k} S C_{-2k} \nonumber \\
  &= T^{-2k} \beta^{(\frac{1}{4k})} \,. \label{eq: betaTconj}
 \end{align}

 Applying this to both sides of Eq.~\eqref{eq:betaTbetaT} yields
 \begin{align}
  S^{-1} T^j S &= \paren{ S^{-1} \beta^{(k)} T^{-\frac{1}{2k}} S} \paren{S^{-1} \beta^{(-k)} T^{\frac{1}{2k}} S} \nonumber \\
  &= T^{-2k} \beta^{(\frac{1}{4k})} T^{2k} \beta^{(-\frac{1}{4k})} \,.
 \end{align}

 By the same argument as before, the right hand side is extendable to the full algebra, and therefore $j = 0$,\footnote{When $j \neq 0$, the operator $\mathsf{S}_p^{\dagger} \mathsf{T}_p^j \mathsf{S}_p$ maps the $\mathbb{Z}_p$ toric code ground state to an invertible Walker-Wang ground state. Thus, the corresponding QCA cannot be extended to the full algebra. Alternatively, one may use an argument similar to the one in the proof of Proposition~\ref{prop:qca_nonextendability}. \label{fn: Tp centrality}} which concludes the proof.
\end{proof}

\begin{lemma}\label{lem:betaT_from_legendre_symb}
 For any $k, l \in \Z_p^\times$, if $\paren{\frac{k}{p}} = \paren{\frac{l}{p}}$, then
 \be\label{eq:betaT_equalities}
 \beta^{(k)} T^{-\frac{1}{2k}} = \beta^{(l)} T^{-\frac{1}{2l}} \,.
 \ee
\end{lemma}

\begin{proof}
 We begin by conjugating $\beta^{(k)} T^{-\frac{1}{2k}}$ by $S$:
 \be
 S^{-1} \beta^{(k)} T^{-\frac{1}{2k}} S = T^{-2k} \beta^{(\frac{1}{4k})} \,.
 \ee

 Now, if $\paren{\frac{k}{p}} = \paren{\frac{l}{p}}$, by Assumption 2, there is some $j$ such that
 \be
 \beta^{(k)} T^{-\frac{1}{2k}} \beta^{(-l)} T^{\frac{1}{2l}} = T^j \,.
 \ee
 This is because $\beta^{(k)}$ is $\beta^{(l)}$ up to a power of $T$ (Eq.~\eqref{eq:beta_equiv_up_to_T}), and $\beta^{(l)}, \beta^{(-l)}$ cancel out. Conjugate the whole expression by $S$ to find:
 \begin{align}
  S^{-1} T^j S &= S^{-1} \beta^{(k)} T^{-\frac{1}{2k}} S S^{-1} \beta^{(-l)} T^{\frac{1}{2l}} S \nonumber \\
  &= T^{-2k} \beta^{(\frac{1}{4k})} T^{2l} \beta^{(-\frac{1}{4l})} \,. \label{eq: STjS}
 \end{align}
 The right-hand side is extendable to the full algebra, hence the left hand side has to be, too. But this forces $j = 0$, which gives us the result we wanted.
\end{proof}

\begin{corollary}\label{cor:betaT_commutes_with_S}
 $\beta^{(k)} T^{-\frac{1}{2k}}$ commutes with $S$
 \be\label{PecanPi}
 S \beta^{(k)} T^{-\frac{1}{2k}}  = \beta^{(k)} T^{-\frac{1}{2k}} S \,.
 \ee
\end{corollary}

\begin{proof}
 We reuse the first part of the previous proof
 \be
 S^{-1} \beta^{(k)} T^{-\frac{1}{2k}} S = T^{-2k} \beta^{(\frac{1}{4k})} \,,
 \ee
 and use the previous lemma
 \be
 T^{-2k} \beta^{(\frac{1}{4k})} = T^{-2k} T^{\frac{1}{2 \frac{1}{4k}}} T^{-\frac{1}{2k}} \beta^{(k)} \,,
 \ee
 to conclude
 \be
 S^{-1} \beta^{(k)} T^{-\frac{1}{2k}} S = T^{-\frac{1}{2k}} \beta^{(k)} = \beta^{(k)} T^{-\frac{1}{2k}} \,.
 \ee

 We can use the previous lemma because
 \be
 \paren{\frac{\frac{1}{4k}}{p}} = \paren{\frac{\paren{\frac{1}{2k}}^2 k}{p}} = \paren{\frac{k}{p}} \,,
 \ee
 since $k$ and $\frac{1}{4k}$ are related by a perfect square.
\end{proof}

\begin{lemma}\label{lem:betaT_commutes_with_Cg}
 $\beta^{(k)} T^{-\frac{1}{2k}}$ commutes with $C_g$.
 \be\label{WalnutPi}
 C_g \beta^{(k)} T^{-\frac{1}{2k}}  = \beta^{(k)} T^{-\frac{1}{2k}} C_g \,.
 \ee
\end{lemma}

\begin{proof}
 Using the properties derived earlier (Eqs.~\eqref{eq:CS_CT}, \eqref{eq:Cbeta})
 \be
 C_g \beta^{(k)} T^{-\frac{1}{2k}} C_g^{-1} = \beta^{(g^2 k)} T^{-\frac{1}{2 g^2 k}} = \beta^{(k)} T^{-\frac{1}{2k}}.
 \ee
 For the last equality we apply Eq.~\eqref{eq:betaT_equalities}, because $g^2 k$ has the same Legendre symbol as $k$.
\end{proof}

\begin{lemma}\label{lem:order_of_betaT}
 $\beta^{(k)} T^{-\frac{1}{2k}}$ is of order $n_p$, where $n_p$ is 2 if $p \equiv 1 \mod 4$, and 4 if $p \equiv 3 \mod 4$
 \be\label{YuzuPi}
 \ord \paren{\beta^{(k)} T^{-\frac{1}{2k}}} = n_p \,.
 \ee
\end{lemma}

\begin{proof}
 By the properties deduced so far,
 \be
 \paren{\beta^{(k)} T^{-\frac{1}{2k}}}^n = \paren{\beta^{(k)}}^n T^{-\frac{n}{2k}} \,.
 \ee

 It has been shown \cite{MengSun2026} that $\tilde{\beta}_p^{(k)}$ has order $n_p$ (for now, the exact value of $n_p$ is irrelevant), up to translations \emph{and circuits which might have non-symmetric gates}. We emphasize this because the quotient we are taking here is with respect to circuits whose gates are all $\Z_p^{(1)}$-symmetric. But in any case,
 \be
 \paren{\beta^{(k)} T^{-\frac{1}{2k}}}^{n_p} = \paren{\beta^{(k)}}^{n_p} T^{-\frac{n_p}{2k}} \,,
 \ee
 where on the right hand side everything is a finite depth circuit \emph{with possibly non-symmetric gates}. Indeed, $\aTWp^m$ is an FDC for any $m$. By Assumption 1, this means that for some $j \in \Z_p$,
 \be
 \paren{\beta^{(k)} T^{-\frac{1}{2k}}}^{n_p} = T^j \,.
 \ee

 But we have shown that the left hand side is central in the group generated by $S, T$ and $C_g$, so the right hand side must be central too. This only happens for $j = 0$, cf.~footnote~\ref{fn: Tp centrality}. In other words, we find that the order of $\beta^{(k)} T^{-\frac{1}{2k}}$ divides $n_p$.

 By a similar argument, if we had $\paren{\beta^{(k)} T^{-\frac{1}{2k}}}^n = 1$ for some $0 < n < n_p$, then it would follow that
 \be
 1 = \paren{\beta^{(k)}}^n T^{-\frac{n}{2k}} \,.
 \ee
 This would mean that $(\tilde{\beta}_p^{(k)})^n$ (the QCA itself, not just its equivalence class) could be written as a finite depth circuit with possibly non-symmetric gates, modulo lattice translations. %in contradiction with what Ref.~\cite{MengSun2026} found.
 This is in contradiction with the fact that $\tilde{\beta}_p^{(k)}$ has order $n_p$ (modulo lattice translations and finite-depth circuits with possibly non-symmetric gates).

 Hence, the order of $\beta^{(k)} T^{-\frac{1}{2k}}$ is exactly $n_p$, which is given by \cite{MengSun2026}\footnote{When $p=3 \bmod 4$, Ref.~\cite{MengSun2026} showed that $\tilde{\beta}_p^{(\frac{p-1}{2})}$ has order 4 under the stable equivalence. That is, $\tilde{\beta}_p^{(\frac{p-1}{2})} \otimes \mathrm{id}_{\mathrm{ancilla}}$ has order 4 modulo finite-depth circuits and lattice translations. Here, $\mathrm{id}_{\mathrm{ancilla}}$ is the identity QCA on an ancillary system, which in this case is another copy of the physical system~\cite{MengSun2026}. This result, combined with the ancilla removal in~\cite{Freedman:2019ucy}, implies that $\tilde{\beta}_p^{(\frac{p-1}{2})}$ has order 4 modulo finite-depth circuits and lattice translations even without adding ancillas.}
 \be\label{npDef}
 n_p = \begin{cases}
  2 & \text{when $p=1$ mod 4} \\
  4 & \text{when $p=3$ mod 4}
 \end{cases}
 \,.
 \ee
\end{proof}

\subsubsection{$\mathbf{p \equiv 3 \text{ mod } 4}$}

We begin the analysis with the case where $p$ is an odd prime with $p \equiv 3 \mod 4$. Let us focus on the family of relations \eqref{eq:additional_identities_rewrite}. By Eq.~\eqref{eq:betaT_equalities}, the right hand side is equal to either $Y$ or $Y^{-1}$, depending on the Legendre symbols.\footnote{We recall that $Y = \beta^{(\frac{p-1}{2})}T$ and $Y^{-1} = \beta^{(\frac{p+1}{2})}T^{-1}$. When $p=3$ mod 4, $\frac{p-1}{2}$ and $\frac{p+1}{2}$ have opposite Legendre symbols. Hence, any $k$ has the same Legendre symbol as one or the other.} Using the generator $g$, we have
\be
\paren{\frac{g^m}{p}} = (-1)^m \,.
\ee

This is obvious if $m$ is even, since $g^m = \paren{g^{\frac{m}{2}}}^2$ is clearly a quadratic residue. The odd ones have to be quadratic non-residues because there are $\frac{p-1}{2}$ of each in $\Z_p^\times$. Then, we can rewrite the additional relations \eqref{eq:additional_identities_rewrite} as
\be\label{eq:additional_identities_rewrite2}
S T^{g^m} S T^{g^{-m}} S T^{g^m} C_{g^{-m}} = Y^{(-1)^m} \,,
\ee
with $m = 0$ being the base case $(ST)^3 = Y$. However, those relations are not independent.

\begin{lemma}
 All the relations in Eq.~\eqref{eq:additional_identities_rewrite2} follow from $(ST)^3 = Y$ and the other identities.
\end{lemma}

\begin{proof}
 Using the fact that $Y$ commutes with $C_g$ (Lemma~\ref{lem:betaT_commutes_with_Cg})
 \begin{align}
  Y &= C_g Y C_g^{-1} = C_g STSTST C_g^{-1} \nonumber \\
  &= S C_g^{-1} T STST C_g^{-1} = \cdots \nonumber \\
  &= ST^{g^2} S T^{g^{-2}} S T^{g^2} C_{g^{-2}} \,.
 \end{align}

 This allows us to generate all the identities in Eq.~\eqref{eq:additional_identities_rewrite2} for $m$ even. For $m$ odd, we instead begin with
 \begin{align}
 Y^{-1} &= T (Y^{-1}) T^{-1} = S^{-1} T^{-1} S^{-1} T^{-1} S^{-1} T^{-1} \nonumber \\
 &= (ST^{-1})^3 C_{-1} \,.
 \end{align}

 Conjugating by $C_g$, we find all the identities with $m$ odd.
\end{proof}

Let us summarize the results so far. In the group generated by $S, T, Y, C_g$, the following hold
\begin{align}\label{eq:STYCg_relations}
 S^2 &= C_g^{\frac{p-1}{2}}, & T^p &= 1, & (ST)^3 &= Y, & C_g^{p-1} &= 1, \nonumber \\
 C_gS &= SC_g^{-1}, & Y^4 &= 1 & C_gT &= T^{g^{-2}} C_g ,
\end{align}
and $Y, C_g^{\frac{p-1}{2}}$ are central. We know that $Y$ has order $4$ because of Lemma~\ref{lem:order_of_betaT}.

\begin{theorem}\label{thm:SL2pZ4}
 Let $G$ be the finitely presented group with generators $s, t, y, c_g$ and relations given by the analogues of the relations above (and $y, c_g^{\frac{p-1}{2}}$ central as well).
 
 Then $G \cong SL(2, \Z_p) \times \Z_4$.
\end{theorem}

Then, under Assumption 3, the group of QCA (modulo FDC and translations), contains a subgroup isomorphic to $SL(2, \Z_p) \times \Z_4$.

\begin{proof}
 Let us begin by defining a map from $G$ to $SL(2, \Z_p) \times \Z_4$. We will first show that it is a group homomorphism, and then that it is an isomorphism.

 Define
 \begin{align}
  s &\longmapsto \paren{ \begin{pmatrix} 0 & 1 \\ -1 & 0 \end{pmatrix}, 3 } \nonumber \\
  t &\longmapsto \paren{ \begin{pmatrix} 1 & 0 \\ -1 & 1 \end{pmatrix}, 0 } \nonumber \\
  y &\longmapsto \paren{ \begin{pmatrix} 1 & 0 \\ 0 & 1 \end{pmatrix}, 1 } \nonumber \\
  c_g &\longmapsto \paren{ \begin{pmatrix} g & 0 \\ 0 & g^{-1} \end{pmatrix}, 2 }\,.
 \end{align}

 Showing that it is a group homomorphism requires verifying that all the relations defining $G$ are satisfied. This is a straightforward, if tedious, calculation.

 The much more interesting question is if this is an isomorphism. It is clear that it is surjective. $y$ generates the $\Z_4$ subgroup, and $sy, t$ generate $SL(2, \Z_p)$. Hence $\abs{G} \geq 4 p (p^2 - 1)$ because the cardinality of $SL(2, \Z_p)$ is known.
 
 Consider the following set of words on the generators $s, t, y, c_g$:
 \be
 \Sigma \coloneqq \curpar{y^{a_1} t^{a_2} c_g^{a_3}, y^{b_1} t^{b_2} s t^{b_3} c_g^{b_4}}
 \ee
 for any $0 \leq a_1, b_1 \leq 3$, $0 \leq a_2, b_2, b_3 \leq p - 1$, and $0 \leq a_3, b_4 \leq p - 2$. Clearly, this set contains at most
 \be
 4 \cdot p \cdot (p - 1) + 4 \cdot p \cdot p \cdot (p-1) = 4p (p - 1)(p + 1) = 4 p (p^2 - 1) \,,
 \ee
 elements of $G$ (at most, because some of them may appear multiple times).

 Note that for any $w \in \Sigma$, $wy \in \Sigma$ because $y$ commutes with everything and increases the exponents $a_1, b_1$, and we use $y^4 = 1$ to bring them down to the range $0, \cdots, 3$. Similarly, $w c_g \in \Sigma$, because it simply increases $a_3, b_4$. For $wt$, we use the rule $c_g t = t^{g^{-2}} c_g$ to bring $t$ to the left, and thus $wt \in \Sigma$. Finally, $ws$ is similar, but first we use $c_g s = s c_g^{-1}$, and then additionally we use
 \be
 sts = y t^{-1} s t^{-1} c_g^{\frac{p-1}{2}}
 \ee
 to convert the double $s$ in a word of the second type ($y^{b_1} t^{b_2} s t^{b_3} c_g^{b_4} \cdot s$) back to a form in $\Sigma$. Thus, we conclude $\Sigma \cdot G \subseteq \Sigma$.

 But in particular $1 \in \Sigma$, so $G \subseteq \Sigma$. Therefore $\abs{G} \leq 4 p (p^2 - 1)$.

 From the two inequalities, we find $\abs{G} = 4p (p^2 - 1)$, and therefore the homomorphism constructed earlier, which is surjective, must be an isomorphism.
\end{proof}

\subsubsection{$\mathbf{ p \equiv 1 \text{ mod } 4}$}

Now we proceed with the other case. Many things will be similar, and the type of arguments we use will be almost identical.

The main difference is the set of conditions
\be
S T^{-2k} S T^{-\frac{1}{2k}} S T^{-2k} C_{-\frac{1}{2k}} = \beta^{(k)} T^{-\frac{1}{2k}} \,,
\ee
which as we will see now separate into two, depending on the Legendre symbol of $k$.

Lemma~\ref{lem:beta_T_commutation} and Corollary~\ref{cor:betaT_commutes_with_S} still hold, so the right hand side commutes with $S$ and $T$.

We can again parametrize those conditions using a generator $g \in \Z_p^\times$
\be\label{eq:additional_identities_rewrite3}
S T^{g^m} S T^{g^{-m}} S T^{g^m} C_{g^{-m}} = \beta^{(- \frac{1}{2} g^m)} T^{g^{-m}}
\ee

By Eq.~\eqref{eq:betaT_equalities}, the equivalence class on the right hand side only depends on the Legendre symbol of $-\frac{1}{2} g^m$, and that in turn depends on the parity of $m$. All $m$ even have the same Legendre symbol, and all $m$ odd have the opposite. We cannot say which is which, because that depends on whether $-2^{-1}$ is a quadratic residue or not, and that in turn depends on the prime\footnote{One can check that for $p = 5$, $-2^{-1} = \frac{5 - 1}{2} = 2$ is a non-residue, while for $p = 17, \ -2^{-1} = \frac{17 - 1}{2} = 8$ is $5^2 = 8 + 17$, hence a quadratic residue.}, but the distinction between $m$ odd and even remains.

\begin{lemma}
 All relations in Eq.~\eqref{eq:additional_identities_rewrite3} follow from any one with odd $m$, and any one with even $m$, and from the other identities.
\end{lemma}

\begin{proof}
 The idea is the same as before. Choose $m_1$ odd $m_2$ even, and define
 \begin{align}
  Y_1 &\coloneqq \beta^{(- \frac{1}{2} g^{m_1})} T^{g^{-m_1}} = S T^{g^{m_1}} S T^{g^{-m_1}} S T^{g^{m_1}} C_{g^{-m_1}} \,, \label{eq:Y1} \\
  Y_2 &\coloneqq \beta^{(- \frac{1}{2} g^{m_2})} T^{g^{-m_2}} = S T^{g^{m_2}} S T^{g^{-m_2}} S T^{g^{m_2}} C_{g^{-m_2}} \,. \label{eq:Y2}
 \end{align}

 By Lemma~\ref{lem:betaT_commutes_with_Cg}, we can conjugate both sides by $C_g$ to obtain $Y_i$ on the left hand side, and $S T^{g^{m_i + 2}} S T^{g^{-m_i - 2}} S T^{g^{m_i + 2}} C_{g^{-m_i - 2}}$ on the right hand side. This is exactly the same identity, with an exponent of the same parity. Using this method we can only increase or decrease the exponents by 2.
\end{proof}

At this point it is worth remarking what the crucial difference is between this case ($p \equiv 1 \mod 4$) and the previous case. Indeed, previously we argued that starting with $ST^{-1} ST^{-1} ST^{-1} C_{-1} = Y^{-1}$ we can generate all the identities with odd powers of $g$, while this seemingly cannot happen in the case of $p \equiv 1 \mod 4$. Let us see what the difference is:
\begin{align}
 C_g (S T^{-1})^3 C_{-1} C_g^{-1} &= S C_g^{-1} T^{-1} S T^{-1} S T^{-1} C_{-g^{-1}} = \cdots \nonumber \\
 &= S T^{-g^2} S T^{-g^{-2}} S T^{-g^2} C_{- g^{-2}} \,.
\end{align}

Now we rewrite $-g^2$ as a power of $g$ by recalling that $(-1) = g^{\frac{p-1}{2}}$, so
\be
-g^2 = g^{\frac{p-1}{2} + 2} = g^{\frac{p+3}{2}} \,.
\ee

Finally, we see that if $p \equiv 3 \mod 4$, then $p + 3 \equiv 2 \mod 4$ and so $\frac{p + 3}{2}$ is odd. But if $p \equiv 1 \mod 4$, then $\frac{p + 3}{2}$ is even, and we cannot use this trick to switch parity.

In this case, Eqs.~\eqref{eq:Y1} and \eqref{eq:Y2} really are two independent conditions.

The last step is very similar to the previous case. The relations satisfied by $S, T, Y_1, Y_2, C_g$ are similar to Eq.~\eqref{eq:STYCg_relations}, except the ones involving $Y$, which must be changed to
\begin{align}
 Y_1 &= S T^{g} S T^{g^{-1}} S T^{g} C_{g^{-1}}, & Y_1^2 &= 1,  \\
  Y_2 &= (ST)^3, & Y_2^2 &= 1,
\end{align}
where we have chosen $m_1 = 1$ odd and $m_2 = 0$ even for simplicity. $Y_1, Y_2$ and $C_g^{\frac{p-1}{2}}$ are all central.

\begin{theorem}\label{thm:SL2pZ2Z2}
 Let $G$ be the finitely presented group with generators $s, t, y_1, y_2, c_g$ and relations given by the analogues of the relations above (and $y_1, y_2, c_g^{\frac{p-1}{2}}$ central as well).
 
 Then $G \cong SL(2, \Z_p) \times \Z_2 \times \Z_2$.
\end{theorem}

The conclusion is that, under Assumption 3, the group of QCA contains a subgroup isomorphic to $SL(2, \Z_p) \times \Z_2 \times \Z_2$.

\begin{proof}
 The proof is essentially identical. First we construct a homomorphism, and then we show that it is an isomorphism. Define
 \begin{align}
  s &\longmapsto \paren{ \begin{pmatrix} 0 & 1 \\ -1 & 0 \end{pmatrix}, \begin{pmatrix} 1 \\ 0 \end{pmatrix} } \nonumber \\
  t &\longmapsto \paren{ \begin{pmatrix} 1 & 0 \\ -1 & 1 \end{pmatrix}, \begin{pmatrix} 0 \\ 0 \end{pmatrix} } \nonumber \\
  y_1 &\longmapsto \paren{ \begin{pmatrix} 1 & 0 \\ 0 & 1 \end{pmatrix}, \begin{pmatrix} 0 \\ 1 \end{pmatrix} } \nonumber \\
  y_2 &\longmapsto \paren{ \begin{pmatrix} 1 & 0 \\ 0 & 1 \end{pmatrix}, \begin{pmatrix} 1 \\ 0 \end{pmatrix} } \nonumber \\
  c_g &\longmapsto \paren{ \begin{pmatrix} g & 0 \\ 0 & g^{-1} \end{pmatrix}, \begin{pmatrix} 1 \\ 1 \end{pmatrix} }\,.
 \end{align}
 Once again, showing that it is a homomorphism is a simple, if tedious, calculation.

 It is clearly surjective, because $y_1, y_2$ generate the $\Z_2 \times \Z_2$ subgroup, while $s y_2, t$ generate $SL(2, \Z_p)$. This gives us $\abs{G} \geq 4 p (p^2 - 1)$.

 The proof that it is an isomorphism is also very similar. Consider the set $\Sigma$ of words of the form $y_1 y_2 t c_g$ and $y_1 y_2 t s t c_g$ (with all allowed exponents for each letter). There are at most
 \be
 2 \cdot 2 \cdot p \cdot (p-1) + 2 \cdot 2 \cdot p \cdot p \cdot (p-1) = 4p (p^2 - 1)
 \ee
 distinct elements of $G$ in it. This set is closed under multiplication by $s, t, y_1, y_2, c_g$ from the right, because we can use the relations in $G$ to manipulate any such word into a form which is in $\Sigma$, and so $\Sigma \cdot G \subseteq \Sigma$.

 But $\Sigma$ contains the trivial word $1$, so in particular $G \subseteq \Sigma$. Finally, we conclude $\abs{G} = 4 p (p^2 - 1)$, and the surjective homomorphism from before is an isomorphism.
\end{proof}

\subsection{Comparison with the Graded Pointed Witt Group}
\label{sec: Comparison with the Graded Pointed Witt Group}

In Section \ref{sec:Cats} we derived the structure of the graded pointed Witt groups for $\Z_p^{(1)}$. There we showed that there is a -- non-canonical -- isomorphism in terms of groups for $p$ an odd prime 
\be\label{WittpIso}
\Witt^\pt (\Z_p \oplus \widehat{\Z_p}, s) \cong SL(2,\Z_p) \times  W_p \times \bigoplus_{l (\neq p)} W_l
\ee
where $W_p =\Z_2\times \Z_2$ for $p=1\mod 4$ and $W_p = \Z_4$ for $p=3\mod 4$.

Theorems \ref{thm:SL2pZ4} and \ref{thm:SL2pZ2Z2} show that, under some assumptions, the group of translation invariant QCA contains a subgroup isomorphic to $SL(2,\Z_p) \times  W_p$. However, we remark that the image under this isomorphism in $W_p$ cannot be interpreted as the  Witt class {of QCA in the standard sense}.

To see this, let us restrict to a subgroup of those QCA which are extendable to the full algebra: that is, the subgroup generated by $T, \beta^{(k)}, C_g$. Using the isomorphisms in the proofs of the Theorems, one can check that this subgroup is isomorphic to $\text{Borel}(SL(2, \Z_p)) \times W_p$ where the Borel subgroup is
\be\label{eq:borel_subgroup}
\text{Borel}(SL(2, \Z_p)) = \curpar{ \begin{pmatrix} a &  0 \\ b & a^{-1}  \end{pmatrix}, a \in \Z_p^\times, b \in \Z_p } \,.
\ee

However, this same isomorphism assigns a non-trivial Witt class to $C_g$ ($2 \in \Z_4$, or $(1, 1) \in \Z_2 \times \Z_2$, depending on the prime $p$), which is actually just an on-site unitary (though it is not $\Z_p^{(1)}$-symmetric), so it should be Witt trivial {in the usual sense}.

Let us see what happens more generally, when we try to assign a  ``Witt class'' to any QCA, that agrees with the Witt class when restricted to QCAs that are extendable to the full algebra, including $S$ (it is sufficient to add $S$ as a generator to all the ones which are extendable, to recover all of $SL(2, \Z_p) \times W_p$).

If $p \equiv 3 \mod 4$, then we may focus on $Y \coloneqq \beta^{(\frac{p-1}{2})} T$ which is obviously extendable. $T$ must have trivial Witt class because it is a finite depth circuit (with some non-symmetric gates). By construction, then, $Y$ must have a non-trivial Witt class of order $4$, and without loss of generality, we choose $1 \in \Z_4$. If we now try to assign a Witt class to everything (including $S$), by using the relation $(ST)^3 = Y$ we find that $S$ has ``Witt class'' $3 \in \Z_4$. Finally, using one of the other relations
\be\label{eq:g_odd_relation_3mod4}
S T^g S T^{g^{-1}} S T^g C_g^{-1} = Y^{-1} \,,
\ee
we conclude that $C_g$ must have ``Witt class'' $2$, which is the contradiction mentioned earlier.

Something similar happens if $p \equiv 1 \mod 4$. By construction, $Y_1, Y_2$ must have two different non-trivial Witt classes. Without loss of generality, we choose $(0, 1), (1, 0) \in \Z_2 \times \Z_2$. From $(ST)^3 = Y_2$, we find that $S$ has ``Witt class'' $(1, 0)$. But from the other independent relation
\be\label{eq:g_odd_relation_1mod4}
S T^g S T^{g^{-1}} S T^g C_g^{-1} = Y_2 \,,
\ee
we conclude that $C_g$ must have ``Witt class'' $(1, 1)$.

This is a general feature: the group structure itself (once we include $S$) automatically forces either $C_g$, or $Y$ ($Y_1, Y_2$), to have the  Witt class: if we assume $Y$ ($Y_1, Y_2$) have the correct Witt class, $1 \in \Z_4$ ($(0, 1), (1, 0) \in \Z_2 \times \Z_2$), then we show that $C_g$ must have the non-trivial Witt class $2 \in \Z_4$ ($(1, 1) \in \Z_2 \times \Z_2$), which is in contradiction with $C_g$ being an on-site unitary.

Without including $S$, we can actually construct a different isomorphism from $\langle T, \beta^{(k)}, C_g \rangle$ to $\text{Borel}(SL(2, \Z_p)) \times W_p$, via the map
\begin{align}\label{eq:witt_correct_iso}
 T &\longmapsto \paren{\begin{pmatrix} 1 & 0 \\ -1 & 1 \end{pmatrix}, \begin{matrix} 0 & \text{if } p \equiv 3 \mod 4 \\ (0, 0) & \text{if } p \equiv 1 \mod 4 \end{matrix}} \,, \nonumber \\
 \beta^{(k)} &\longmapsto \paren{\begin{pmatrix} 1 & 0 \\ -\frac{1}{2k} & 1 \end{pmatrix}, \begin{matrix} \pm 1 & \text{if } p \equiv 3 \mod 4 \\ (0, 1), (1, 0) & \text{if } p \equiv 1 \mod 4 \end{matrix}} \,, \nonumber \\
 C_g &\longmapsto \paren{\begin{pmatrix} g & 0 \\ 0 & g^{-1} \end{pmatrix}, \begin{matrix} 0 & \text{if } p \equiv 3 \mod 4 \\ (0, 0) & \text{if } p \equiv 1 \mod 4 \end{matrix}} \,,
\end{align}
where the choice of $\pm 1$ or $(0, 1), (1, 0)$ depends on the Lagrange symbol of $k$. This map correctly identifies the Witt class, since $T$ is a finite depth circuit (with non-symmetric gates), and $C_g$ is an on-site unitary, while $\beta^{(k)}$ is a disentangler for a Walker-Wang model, based on an MTC with a non-trivial Witt class.

In summary, we have two isomorphisms
\begin{align}
 \langle S, T, \beta^{(k)}, C_g \rangle &\xrightarrow{\cong} SL(2, \Z_p) \times W_p \label{eq:STCiso} \\
 \langle T, \beta^{(k)}, C_g \rangle &\xrightarrow{\cong} \text{Borel}\paren{SL(2, \Z_p)} \times W_p \label{eq:TCbetaiso}
\end{align}
the first constructed in Theorems \ref{thm:SL2pZ4} and \ref{thm:SL2pZ2Z2}, and the second one in Eq.~\eqref{eq:witt_correct_iso}.

Although both map into $W_p$, only the second one, \eqref{eq:witt_correct_iso} correctly identifies the physical Witt class.
On the other hand, while we know that the isomorphism \eqref{eq:STCiso} exists, there is no choice for this isomorphism which would agree with \eqref{eq:TCbetaiso} on the entire subgroup. By the previous arguments, if it agrees on $Y$ ($Y_1, Y_2$), then it must disagree on $C_g$.

\subsection{Comparison between Lattice and Continuum}
\label{sec: comparison Zp}
We now compare the fusion rules of QCAs on the lattice and their counterparts in the continuum. Here we focus on the subgroup generated by $\aKWWp$ and $\aTWp$ only, because we have not discussed the generalized charge conjugation in the continuum.

By using the isomorphisms in the proofs of Theorems \ref{thm:SL2pZ4} and \ref{thm:SL2pZ2Z2}, one can check that this subgroup is isomorphic to $SL(2, \Z_p) \times \Z_{n_p}$, with $n_p$ as in Eq.~\eqref{npDef}. This is the entire group if $p \equiv 3 \mod 4$, but it misses one $\Z_2$ subgroup if $p \equiv 1 \mod 4$.

As discussed in Section~\ref{sec: Zpk QCA generated by D and T}, the $\mathbb{Z}_p$ KWW QCA and a Tsui-Wen entangler obey the following fusion rule:
\begin{equation}
\left(\aKWWp \circ \aTWp\right)^3 = \alpha_t \circ \tilde{\beta}_p^{\left( \frac{p-1}{2} \right)} \circ \aTWp \,.
\label{eq: comparison lat Zp}
\end{equation}
The continuum counterpart of $\aKWWp \circ \aTWp$ is $\opaction{\mathcal{ST}}$, which is the (normalized) action of the triality defect $\mathcal{ST}$ on local operators.
As reviewed in Section~\ref{sec: continuum fusion rules Zp}, the action $\opaction{\mathcal{ST}}$ obeys the $\mathbb{Z}_3$ fusion rule:
\be\label{eq: comparison cont Zp}
\opaction{(\mc{ST})^3} = 1
\ee
% \begin{equation}
% \alpha_{\mathcal{ST}}^3 = 1.
% \label{eq: comparison cont Zp}
% \end{equation}
From the above equations, we find that the fusion rule on the lattice differs from its continuum counterpart.
Specifically, the lattice fusion rule~\eqref{eq: comparison lat Zp} involves a non-trivial QCA $\alpha_t \circ \tilde{\beta}_p^{\left( \frac{p-1}{2} \right)} \circ \aTWp$, whereas the continuum fusion rule~\eqref{eq: comparison cont Zp} does not.
Thus, equation~\eqref{eq: comparison lat Zp} is an example of a lattice fusion rule that mixes with a non-trivial QCA.
Correspondingly, the non-invertible symmetry generated by $\mathsf{ST}$ on the lattice mixes with a non-trivial QCA.

We emphasize that the appearance of the QCA $\tilde{\beta}_p^{\left( \frac{p-1}{2} \right)}$ in the lattice fusion rule is not a coincidence.
It can be traced back to the relation
\begin{equation}
(S_pT_p)^3 = Y_p,
\end{equation}
where $Y_p$ is the stacking of the Crane-Yetter-Walker-Wang TQFT based on the MTC $\mathrm{SU}(p)_1$.
On the lattice, the above relation implies that $\left(\aKWWp \circ \aTWp\right)^3$ is a non-trivial QCA in the same equivalence class as $\tilde{\beta}_p^{\left(\frac{p-1}{2}\right)}$, which entangles the ground state of the $\mathrm{SU}(p)_1$ Crane-Yetter-Walker-Wang TQFT from the trivial product state.
On the other hand, in the continuum, the above relation gives rise to the fusion rule~\eqref{eq: DTDTDT continuum2}, which involves $\mathrm{SU}(p)_1$ Chern-Simons theory as a fusion coefficient.
This TQFT coefficient disappears in~\eqref{eq: comparison cont Zp} because the decoupled 2+1d TQFT acts trivially on any local operator.

\subsection*{Acknowledgments}
We are grateful to Andrea Antinucci, Thibault D\'{e}coppet, 
Yuhan Gai, Po-Shen Hsin, Marvin Qi, Nat Tantivasadakarn,  Rui Wen,
 Carolyn Zhang, in particular  Alison Warman and Matt Yu  for discussions. 
Some of these computations have been done with help of  ChatGPT/Codex 5.5 and Claude Code Opus 4.8 and Fable. 
The work of SSN is supported by the UKRI Frontier Research Grant, underwriting the ERC Advanced Grant ``Generalized Symmetries in Quantum Field Theory and Quantum Gravity''.
KI and SSN are supported in part by the EPSRC Open Fellowship EP/X01276X/1 (Schafer-Nameki).
KI is also supported by the Leverhulme-Peierls Fellowship funded by the Leverhulme Trust.
We thank the authors of \cite{Zhang:2026kjf} and \cite{Sun:2026toj} for coordinating submission of related works.

\appendix

\section{Pointed Graded Witt Groups $\Witt^\pt (\Z_p^2, s)$}
\label{app:Witty}

We discuss various properties of the pointed graded Witt group in this appendix. 
The two cases that are separately discussed are $p=2$ (Appendix~\ref{app:WittZ2}) and odd $p$
(Appendix~\ref{app:WittZp}). We first set out some general principles underlying these Witt groups. We will abbreviate sometimes $\Z_p\oplus \widehat{\Z_p}$ with $\Z_p^2$.

\subsection{Graded Metric Groups and  $s$-twisted Product}
\label{app:framework}

Let $A\oplus\widehat A$ be the group of surface charges with syllepsis $s$.  A
{\bf graded metric group} is a pair $(f:G\to A\oplus\widehat A,\,q)$ with $G$ a
finite abelian group and $q=e^{2\pi i\,Q}$ a non-degenerate quadratic form.  The
spin exponent $Q$ and the syllepsis exponent $\sigma$, defined by
$s(x,y)=e^{2\pi i\,\sigma(x,y)}$, take values in a common cyclic group: $\Z_4$
for $p=2$ (so that $q$ may equal the semion spin $i$) and $\Z_p$ for odd
$p$.  The polarization
\be\label{eq:polgen}
  B(x,y)=Q(x+y)-Q(x)-Q(y)
\ee
is a symmetric bicharacter, and when $2$ is invertible (odd $p$) it determines
$Q$ by $Q(x)=2^{-1}B(x,x)$.  Fixing a basis $e_1,\dots,e_n$ of $G$, an element
$g=\sum_{i\in I}e_i$ has $Q(g)=\sum_{i\in I}Q(e_i)+\sum_{\{i<j\}\subseteq I}B(e_i,e_j)$.
The $s$-twisted product and $s$-opposite are
\be\label{eq:arith-gen}
\ba
  Q_{X\otimes Y}(g,h)&=Q_X(g)+Q_Y(h)\\
  &\quad+\sigma(fg,kh)+2\,\sigma(kh,fg)\,,\\
  \widetilde Q(g)&=-Q(g)+3\,\sigma(fg,fg)\,,
\ea
\ee
with $f,k$ the gradings of $X,Y$ and cross-term
$\tau(x,y):=\sigma(x,y)+2\,\sigma(y,x)$.  Two classes agree, $X=Y$, iff
$X\otimes Y^{\rm op}$ is trivial, i.e.\ admits a Lagrangian $L\subseteq G_0=\ker f$
($Q|_L=0$ and $L^\perp=L$ with respect to $B$).  The coefficient $3$ in
$\widetilde Q$ is fixed by requiring the anti-diagonal $(g,-g)\in X\otimes X^{\rm op}$
to be isotropic: with the ansatz $\widetilde Q=-Q+c\,\sigma\circ f$ and $x=f(g)$,
its value is $Q(g)+\widetilde Q(g)+\tau(x,-x)=(c-3)\,\sigma(x,x)$, which vanishes
for all $g$ iff $c=3$.  At $p=2$ the exponent $\sigma$ takes values in
$\{0,2\}\subset\Z_4$, so $2\sigma\equiv0$ and $3\sigma\equiv\sigma$, and
\eqref{eq:arith-gen} collapses to the single-factor forms $Q_X+Q_Y+\sigma(fg,kh)$
and $-Q+\sigma\circ f$.

\begin{lemma}[Condensation]\label{lem:cond}
If $L\subseteq G_0$ is isotropic, then $(f:G\to A\oplus\widehat A,q)$ and
$(\bar f:L^\perp/L\to A\oplus\widehat A,\bar q)$, with $\bar q([g])=q(g)$ and
$\bar f([g])=f(g)$, define the same class in $\Witt^\pt(A\oplus\widehat A,s)$.
\end{lemma}
\begin{proof}
Since $Q|_L=0$ we have $L\subseteq L^\perp$ and $\bar q$ is well-defined.  In
$(G,q)\otimes(L^\perp/L,\bar q)^{\rm op}$ the anti-diagonal
$\Delta=\{(g,-[g]):g\in L^\perp\}$ is grade zero, and by the $c=3$ computation
above $Q|_\Delta=0$; as $|\Delta|^2=|L^\perp|^2=|G|\,|L^\perp/L|$, it is a
Lagrangian, so the two classes agree.
\end{proof}

\subsection{Pointed Graded Witt Group for $\Z_2^{(1)}$}
\label{app:WittZ2}

In this appendix we prove Theorem~\ref{thm:modular} for $\Z_2$, specialising the
framework of Appendix~\ref{app:framework}.  Here $A=\Z_2$, the charges are
$A\oplus\widehat A=\Z_2^2$ with $a=(1,0)$, $b=(0,1)$, and the syllepsis is
$s(x,y)=(-1)^{x_1y_2}$ with additive exponent\footnote{The exponent $\sigma$ defined by \eqref{eq:sigmaZ2} is four times $\sigma$ defined in Appendix~\ref{app:framework}. Similarly, $Q$ defined below is four times $Q$ defined in Appendix~\ref{app:framework}.}
\be
  \sigma(x,y):=2\,x_1y_2\in\{0,2\}\subset\Z_4\,,\qquad s(x,y)=i^{\,\sigma(x,y)}\,.
\label{eq:sigmaZ2}
\ee
A $\Z_2$ anyon may carry the semion spin $\pm i$, so here $q$ takes fourth roots,
$q=i^{\,Q}$ with $Q:G\to\Z_4$; this is why $\ga,\gb,\gu$ below have order divisible
by $4$, and the polarization takes values $B\in\{0,2\}\subset\Z_4$.  With
$\sigma\in\{0,2\}$ the arithmetic \eqref{eq:arith-gen} reduces to
$Q_{X\otimes Y}=Q_X+Q_Y+\sigma(f g,k h)$ and $\widetilde Q=-Q+\sigma\circ f$.  In
this notation the classes we use, with $Q(1)=1$ (i.e.\ $q=i$), are
\be\label{gens}
\ba
\ga&=(\Z_2\xrightarrow{a}A\oplus\widehat A;\,1)\cr
\gb&=(\Z_2\xrightarrow{b}A\oplus\widehat A;\,1)\cr
\gc&=(\Z_2\xrightarrow{a+b}A\oplus\widehat A;\,0) \cr
\gu&=(\Z_2\xrightarrow{0}A\oplus\widehat A;\,1)\cr
\cC&=\big((\Z_2)^2\xrightarrow{0}A\oplus\widehat A;\,Q\equiv2\text{ on }G\setminus0\big)\,.
\ea\ee

We now show the relations that result in the proof of Theorem \ref{thm:modular}: 
\begin{lemma}\label{lem:basic}
$\gu$ is central, and $\cc$ is central with $\cC\neq 1$ and $\cC^2=1$.
\end{lemma}
\begin{proof}
Any grade-$0$ class meets every $\sigma$-factor in \eqref{eq:arith-gen} as $\sigma(0,\cdot)=\sigma(\cdot,0)=0$,
so it commutes with everything. And $\gu,\cC$ are grade $0$ and thus central. For $\cC$: all nonzero $Q=2$, so the only
isotropic subgroup is $0$, with $0^\perp=G\neq0$. It is not  Lagrangian, hence $\cC\neq 1$. On the other hand, 
$\cC\otimes\cC$ (grade $0$ on $(\Z_2)^4$, $Q(e_i)=2$, $B(e_1,e_2)=B(e_3,e_4)=2$, else $0$) has the
Lagrangian $L=\{0,\,e_1{+}e_3,\,e_2{+}e_4,\,e_1{+}e_2{+}e_3{+}e_4\}$ with $Q|_L\equiv0$, $|L|=4$,
$L^\perp=L$, so $\cC^2=1$.
\end{proof}

\begin{proposition}\label{prop:ord8}
$\ga^4=\gb^4=\gu^4=\cC$, and $\ord\ga=\ord\gb=\ord\gu=8$.
\end{proposition}

\begin{proof}
$\ga^4$ is the form on $(\Z_2)^4$ with all grades $a$, $Q(e_i)=1$, and $B(e_i,e_j)={\sigma(a,a)=0}$
for $i\neq j$. The line $L=\langle(1,1,1,1)\rangle\subseteq G_0$ is isotropic
($Q=4\cdot1\equiv0$), $L^\perp=\{(n_1,\dots,n_4)\in(\Z_2)^4:\textstyle\sum_i n_i=0\}$ (the even-weight vectors) and $L^\perp/L$ has two grade-$0$
generators with $Q=2$ and mutual $B=2$: this is $\cC$. Lemma~\ref{lem:cond} gives $\ga^4=\cC$.
The computation is identical for $\gb$ ($\sigma(b,b)=0$) and $\gu$ (grade $0$).
Finally, for $x\in\{\ga,\gb,\gu\}$ we have $x^4=\cC\neq1$ and $x^8=\cC^2=1$
(Lemma~\ref{lem:basic}), so $\ord x$ divides $8$ but not $4$, i.e.\ $\ord x=8$.
\end{proof}

\begin{proposition}\label{prop:c}
$\gc=\ga\gb\ga^{-1}$, $\gc^4=\cC$ and $\ord\gc=8$.
\end{proposition}

\begin{proof}
$\ga\gb\ga^{-1}$ is the form on $(\Z_2)^3$ with grades $(a,b,a)$, $Q=(1,1,3)$ (the
third generator is that of $\ga^{-1}=\ga^{\rm op}$, so
$Q(e_3)=\widetilde Q_{\ga}(1)=-Q_{\ga}(1)+\sigma(a,a)=-1+0=3$), and
$B(e_1,e_2)=2$, 
$B(e_2,e_3)=0$, 
$B(e_1,e_3)=0$.
Condensing $L=\langle e_1{+}e_3\rangle$ (grade $0$, $Q=1+3\equiv0$) leaves the rank-$1$ class
$\langle[e_1{+}e_2]\rangle$, where $[\,\cdot\,]$ denotes the class in $L^\perp/L$, of grade $a+b$ with $Q=1+1+2\equiv0$, i.e.\ $\gc$. Then
$\gc^4=\ga\gb^4\ga^{-1}=\ga\,\cC\,\ga^{-1}=\cC$ ($\cC$ central), and $\ord\gc=\ord\ga=8$ by
conjugacy.
\end{proof}

\begin{proposition}\label{prop:Q8}
$\ga^2\gb^2=\gc^2$,\quad $\gb^2\ga^2=\gc^{-2}$,\quad $[\ga^2,\gb^2]=\cC$. Hence
$\{e,\cC,\ga^{\pm2},\gb^{\pm2},\gc^{\pm2}\}\cong Q_8$, the quaternion group, via $\{\ga^2,\gb^2,\gc^2\}\mapsto \{i,j,k\}$.
\end{proposition}

\begin{proof}
$\ga^2\gb^2$ is the form on $(\Z_2)^4$ with grades $(a,a,b,b)$, $Q(e_i)=1$, $B(e_i,e_j)=2$ for
$i\in\{1,2\}$, $j\in\{3,4\}$ (one $a$- and one $b$-generator) and $0$ within each block.
Condensing $\langle(1,1,1,1)\rangle$ (grade $0$, $Q=4+4\cdot2\equiv0$) yields the rank-$2$ class of
grade $a+b$ with $Q(1,0)=Q(0,1)=0$, $Q(1,1)=2$, which is $\gc^2$. Thus $\ga^2\gb^2=\gc^2$.
Using $\gb^4=\ga^4=\cC$ and $\cC^2=e$,
\be
  (\ga^2\gb^2)(\gb^2\ga^2)=\ga^2\gb^4\ga^2=\ga^2\cC\,\ga^2=\cC\,\ga^4=\cC^2=e,
\ee
so $\gb^2\ga^2=(\ga^2\gb^2)^{-1}=\gc^{-2}$; and, since $\ga^{-2}\gb^{-2}=(\gb^2\ga^2)^{-1}=\gc^2$,
\be
  [\ga^2,\gb^2]=\ga^2\gb^2\,\ga^{-2}\gb^{-2}=\gc^2\cdot\gc^2=\gc^4=\cC .
\ee
These are exactly the relations $ij=k$, $ji=-k$, $i^2=j^2=k^2=-1$ of $Q_8$.
\end{proof}

\begin{proposition}\label{prop:ord3}
$(\ga\gb)^3=\cC$, and more generally any product of two non-commuting graded semions
\be 
(\ga^{\pm1}\gb^{\pm1}\,,\ \gb^{\pm1}\ga^{\pm1})
\ee
cubes to $\cC$,  so $\ord(\ga\gb)=6$.
\end{proposition}

\begin{proof}
$(\ga\gb)^3$ is the form on $(\Z_2)^6$ with grades $(a,b,a,b,a,b)$ and $Q(e_i)=1$; here
$B(e_i,e_j)=2$ when the earlier index is $a$-graded and the later one $b$-graded
($B(e_i,e_j)=\sigma(a,b)=2$ for $i<j$ with $e_i\!\sim\!a$, $e_j\!\sim\!b$), and
$0$ otherwise. Condensing the $2$-dimensional isotropic subgroup
$L=\langle(0,1,1,0,1,1),\,(1,0,1,1,0,1)\rangle\subseteq G_0$ --- both generators
have grade $0$ and $Q\equiv0$ and are mutually $B$-orthogonal, so $Q|_L\equiv0$
--- reduces $(\ga\gb)^3$ directly to the rank-$2$ class $\cC$
(Lemma~\ref{lem:cond}).
The grade-$0$ part of $\Witt(A\oplus\widehat A,s)$ is central with quotient $\Perm_4$, in which
$\ga\gb$ maps to a $3$-cycle; thus $3\mid\ord(\ga\gb)$, while $(\ga\gb)^3=\cC\neq e$ and
$(\ga\gb)^6=\cC^2=e$, giving $\ord(\ga\gb)=6$. The other sign choices are the same computation.
\end{proof}

\begin{proposition}\label{prop:transp}
$(\ga\gb\ga)^2=\cC$, and $\ord(\ga\gb\ga)=4$.
\end{proposition}

\begin{proof}
$(\ga\gb\ga)^2$ is the form on $(\Z_2)^6$ with grades $(a,b,a,a,b,a)$ and $Q(e_i)=1$. Condensing
$\langle(0,1,1,1,1,0)\rangle$ (grade $0$, $Q\equiv0$) reduces it to a rank-$4$ form with all grades
$a$, $Q(e_i)=1$ and diagonal $B$ --- i.e.\ the same shape as $\ga^4$ in
Proposition~\ref{prop:ord8} --- and condensing $\langle(1,1,1,1)\rangle$ there gives $\cC$. Since
$\cC\neq e$, $\ord(\ga\gb\ga)=4$.
\end{proof}

\begin{theorem}[Pointed Graded Witt for $A=\Z_2^{(1)}$]\label{thm:2O}
The metric groups $\langle\ga,\gb\rangle$ generate the binary octahedral group 
\be 
2O=\langle r,s\mid r^3=s^4=(rs)^2\rangle
\ee 
of order $48$, via 
\be 
s=\ga\,,\ r=\ga\gb\,,\ 
rs=\ga\gb\ga\,,
\ee
with central involution
\be z=r^3=s^4=(rs)^2=\cC\,.
\ee
The grade-$0$ classes form the central, ungraded pointed Witt group
$\Witt^\pt$ (metric groups with $f=0$), and
\be\label{full}
  \Witt^\pt(A\oplus\widehat A,s)\;\cong\;2O\circ_{\Z_2}\Witt^\pt,\qquad z=\cC=\gu^4,
\ee
with the identification 
\be 
z\in2O \ \longleftrightarrow \ \cC\in\Witt^\pt \,.
\ee
As $\cC=\gu^4$ lies in that $\Z_8$, this extends as
\be\label{fullsplitapp}
\Witt^\pt(A\oplus\widehat A,s) \cong
\big(2O\circ_{\Z_2}\Z_8\big)\,\times\,\Z_2\,\times\!\!\bigoplus_{p\ \mathrm{odd}} W_p\,.
\ee
The five classes $\ga,\gb,\gc,\gu,\cC$ generate exactly the factor $2O\circ_{\Z_2}\Z_8$
(order $192$).  The remaining $\Z_2$ and odd-$p$ summands come from other metric groups in $\Witt^\pt$.
\end{theorem}

\begin{proof}
By Propositions~\ref{prop:ord8}, \ref{prop:ord3}, \ref{prop:transp},
$s^4=\ga^4=\cC$, $r^3=(\ga\gb)^3=\cC$ and $(rs)^2=(\ga\gb\ga)^2=\cC$ coincide in the central
involution $z=\cC$ (Lemma~\ref{lem:basic}); these are the defining relations of $2O$.
The abstract group $\langle r,s\mid r^3=s^4=(rs)^2\rangle$ is the binary octahedral group of order $48$
(in it, $z=r^3=s^4=(rs)^2$ is automatically central with $z^2=1$) and we have a surjection $2O\twoheadrightarrow\langle\ga,\gb\rangle$, so
$|\langle\ga,\gb\rangle|\le 48$.  On the other hand the image of
$\langle\ga,\gb\rangle$ in the quotient $\Witt^\pt(A\oplus\widehat A,s)/\Witt^\pt\cong \Perm_4$ is all of
$\Perm_4$, and the central class $\cC\neq1$ lies over the identity; hence
$|\langle\ga,\gb\rangle|\ge 2\cdot 24=48$.  Therefore the surjection is an isomorphism,
$\langle\ga,\gb\rangle\cong 2O$, and in particular the fiber over $1\in \Perm_4$ inside
$\langle\ga,\gb\rangle$ is exactly $\{1,\cC\}$, i.e.\
$\langle\ga,\gb\rangle\cap\Witt^\pt=\{1,\cC\}$.
(As a consistency remark, one can see the double cover must be $2O$ rather than
the other Schur cover $GL(2,\mathbb F_3)$ of $\Perm_4$ directly from the computed lifts: all
lifts of transpositions have order $4$ by Proposition~\ref{prop:transp}, so $\cC$ is the
unique involution, forcing the Sylow $2$-subgroup to be generalized quaternion as in $2O$;
in $GL(2,\mathbb F_3)$ transpositions admit order-$2$ lifts.)
For the full group: the grade-$0$ classes are central (Lemma~\ref{lem:basic}) and form $\Witt^\pt$
(with $f=0$ the $s$-twist is trivial), and the quotient $\Witt^\pt(A\oplus\widehat A,s)/\Witt^\pt\cong \Perm_4$. Since $\gu$ is central of order $8$
(Proposition~\ref{prop:ord8}) with $\gu^4=\cC=z$ and $\langle\ga,\gb\rangle\cap\Witt^\pt=\{1,\cC\}$, the
multiplication map $2O\times\Witt^\pt\to\Witt^\pt(A\oplus\widehat A,s)$ has kernel $\langle(z,\cC)\rangle$ and is onto,
giving \eqref{full}. Then \eqref{fullsplitapp} follows
because $\cC=\gu^4$ lies in the $\Z_8$ summand of $\Witt^\pt$.
\end{proof}

We can now prove Theorem \ref{thm:modular}:
\begin{proof}
We set 
\be 
T=\gu\,\ga,\qquad S=\gu^{2}\,\ga\gb\ga,\qquad Y=\gu \,.
\ee
We want to show: $T^4=S^2=1$, $(ST)^3=Y$, $Y^8=1$. 
As $\gu$ is central with $\gu^4=\cC$, so $Y$ is
central with $Y^8=\cC^2=1$. Furthermore by Propositions~\ref{prop:ord8} and \ref{prop:transp}
\be
  T^4=\gu^4\ga^4=\cC^2=1,\qquad
  S^2=\gu^4(\ga\gb\ga)^2=\cC^2=1 \,.
\ee
Finally, 
\be
  ST=\gu^{3}\,\ga\gb\ga^2,\qquad (ST)^3=\gu^{9}(\ga\gb\ga^2)^3=\gu\,(\ga\gb\ga^2)^3 .
\ee
Conjugating by $\ga^{-1}$ and using $\ga^3=\ga^{-1}\ga^4=\ga^{-1} \cC$,
\be\ba 
  &\ga^{-1}(\ga\gb\ga^2)\ga=\gb\ga^{3}=\cC\,\gb\ga^{-1}
 \cr  
 \Rightarrow\ &
{\ga^{-1}} (\ga\gb\ga^2)^3 \ga
 =(\cC\,\gb\ga^{-1})^3 \cr 
  = &\cC^{3}(\gb\ga^{-1})^3
 =\cC\cdot\cC=1 \,,\cr 
\ea\ee
where $(\gb\ga^{-1})^3=\cC$ by Proposition~\ref{prop:ord3}. Hence $(ST)^3=\gu=Y$, and
$\ord(ST)=24$: since $Y=(ST)^3$ has order $8$, $\ord(ST)\in\{8,24\}$; if it were $8$
then $\langle ST\rangle=\langle Y\rangle$ (as $\gcd(3,8)=1$) would consist of central grade-$0$ classes, contradicting the nontrivial image of $ST$ in $\Perm_4$. Finally $(ST)^3=\gu$ puts $\gu\in\langle S,T\rangle$, so $\ga=\gu^{-1}T$ and
$\gb=\gu^{-2}\ga^{-1}S\ga^{-1}$ lie in $\langle S,T\rangle$. Thus
\be 
\langle S,T\rangle=\langle\ga,\gb,\gu\rangle=2O\circ_{\Z_2}\Z_8\,.
\ee
\end{proof}

\subsection{Pointed Graded Witt Group for $\Z_p^{(1)}$, odd prime $p$}
\label{app:WittZp}

We write the odd-prime calculation in a form parallel to Appendix~\ref{app:WittZ2}.
Let $p$ be odd and $A=\Z_p$, so the surface charges form the self-dual group
$A\oplus\widehat A=\Z_p^2$ as in Section~\ref{sec:BrAutZp}, with basis $a=(1,0)$,
$b=(0,1)$ and
\be
  s(x,y)=\zeta_p^{x_1y_2}\,,\qquad
  \sigma(x,y)=x_1y_2\in\Z_p \,.
\ee
Then $\mathsf{Alt}(s)(x,y)=\zeta_p^{x_1y_2-x_2y_1}$ is the standard symplectic
form on $A\oplus\widehat A\cong\Z_p^2$, so
\be
  \Aut^{\rm syp}(A\oplus\widehat A,\mathsf{Alt}(s))
  =Sp(2,\Z_p)=SL(2,\Z_p) \,.
\ee
The result to prove is
\be\label{eq:WittZp-goal}
  \Witt^\pt(A\oplus\widehat A,s)\cong SL(2,\Z_p)\times\Witt^\pt \,.
\ee

\bigskip\noindent{\bf Odd-prime specialisation.}
Here $q$ takes $p$-th roots, $q=\zeta_p^Q$ with $Q\in\Z_p$, and since $2$
is invertible the form is $Q(x)=2^{-1}B(x,x)$; it suffices to take
$G=(\Z_p)^n$. The $s$-twisted arithmetic \eqref{eq:arith-gen} and the
Condensation Lemma~\ref{lem:cond} of Appendix~\ref{app:framework} apply verbatim,
now with $\sigma\in\Z_p$ and both cross-terms present.

\medskip\noindent{\bf The pointed Witt group.}
The grade-zero metric groups form the usual pointed Witt group $\Witt^\pt$ (\ref{Wittpt}) and (\ref{Wittp}), and
they are central because all $s$-factors in \eqref{eq:arith-gen} are trivial when
one grading is zero.  Hence there is a central extension
\be\label{eq:oddp-exact}
  1\longrightarrow \Witt^\pt
  \longrightarrow \Witt^\pt(A\oplus\widehat A,s)
  \longrightarrow SL(2,\Z_p)
  \longrightarrow 1 \,.
\ee
Here, we used $\Witt^{\pt}(A\oplus\widehat A, s)/\Witt^{\pt} \cong SL(2, \Z_p)$, which was shown in~\eqref{eq:ZpWittQuotient} based on the result in Appendix~\ref{app:H5B3}.
In $\Witt^\pt$ only the $p$-primary summand can mix with the non-zero grading, since any
homomorphism from an $\ell$-group to $A\oplus\widehat A$ is zero for $\ell\neq p$.

\medskip\noindent{\bf Generators and graded anyons.}
Write the {\it graded $\Z_p$ anyon} of axis $v\in A\oplus\widehat A$ ($v\neq0$)
and level $k\in\Z_p^\times$ as
\be\label{eq:semionvk}
  \mathfrak s_{v,k}=\big(\Z_p\xrightarrow{\ v\ }A\oplus\widehat A;\ q(1)=\zeta_p^{\,k}\big)\,,
\ee
the graded chiral $\Z_p$ anyon theory of spin $\zeta_p^{\,k}$ (Section~\ref{sec:qudit}).
Its grade line $\langle v\rangle$ is the {\it axis}: since
$\mathfrak s_{\lambda v,k}\cong\mathfrak s_{v,k/\lambda^2}$, the axis depends only
on the line $\langle v\rangle\in\mathbb P^1(\Z_p)$, and two anyons share
an axis when their grades are proportional, and lie on {\it independent axes}
when their grades span $A\oplus\widehat A$.  The two generators, together with the
ungraded class $\kappa$, are
\be\label{eq:gensp}
\ba
  \ga&=\mathfrak s_{a,1}\,,\qquad \gb=\mathfrak s_{b,1}\,,\\
  \kappa&=\big(\Z_p\xrightarrow{0}A\oplus\widehat A;\ q(1)=\zeta_p\big)\,,
\ea
\ee
where $\kappa$ is the ordinary ungraded generator $\langle1\rangle$ in $W_p$, of
order $2$ for $p\equiv1\bmod4$ and $4$ for $p\equiv3\bmod4$.
Note that for $p\equiv1\bmod4$ there is a second class generating a $\Z_2$ class in $W_p=\Z_2 \oplus\Z_2$ defined through 
\be 
\kappa_a=\big(\Z_p\xrightarrow{0}A\oplus\widehat A;\ q(1)=\zeta_p^{a}\big)\,,
\ee
for $a$ quandratic non-residue mod $p$, and the $\kappa =\kappa_1$ in that case. 

\begin{proposition}\label{prop:apowerp}
In the full group
\be
  \ga^p=\gb^p=\kappa^p\,.
\ee
Consequently the normalized lifts
\be
  \hat\ga=\kappa^{-1}\ga,\qquad
  \hat\gb=\kappa^{-1}\gb
\ee
have exact order $p$.
\end{proposition}
\begin{proof}
For $\ga^p$ the underlying group is $(\Z_p)^p$, every basis vector has
grade $a$, and the cross terms vanish because $\sigma(a,a)=0$.  Thus
\be
  Q(x_1,\ldots,x_p)=\sum_i x_i^2,\qquad
  f(x_1,\ldots,x_p)=\Big(\sum_i x_i\Big)a \,.
\ee
The grade-zero subgroup is $G_0=\{\sum_i x_i=0\}$, and
$G_0^\perp=\langle(1,\ldots,1)\rangle$.  This line is contained in $G_0$ and is
isotropic, since $Q(1,\ldots,1)=p=0$.  By Lemma~\ref{lem:cond}, condensing it
leaves the grade-zero metric group
$G_0/\langle(1,\ldots,1)\rangle$ with quadratic form $\sum_i x_i^2$.  The same
condensation applied to the ordinary grade-zero form
$\kappa^p=((\Z_p)^p\xrightarrow{0}A\oplus\widehat A;\sum_i x_i^2)$ gives
exactly this reduced metric group, so $\ga^p=\kappa^p$.  The proof for $\gb$ is
identical.

Since $\ord(\kappa)\in\{2,4\}$ is coprime to $p$, $\kappa^p$ generates
$\langle\kappa\rangle$, {a central subgroup}.  Therefore $(\kappa^{-1}\ga)^p=1$, and the image of
$\hat{\ga}$ in the quotient $SL(2,\Z_p)$ is a nontrivial element of order
$p$.  Hence $\hat\ga$ has exact order $p$, and similarly for $\hat\gb$.
\end{proof}
Explicitly, since $\ord(\kappa)\in\{2,4\}$,
\be
  \kappa^p=
  \begin{cases}
    \kappa,&p\equiv1\pmod4,\\
    \kappa^{-1},&p\equiv3\pmod4\,,
  \end{cases}
\ee
matching \eqref{eq:gapowercentral}.

\begin{theorem}[Full odd-prime pointed graded Witt group]\label{thm:WittZp}
For every odd prime $p$,
\be
  \Witt^\pt(\Z_p\oplus\Z_p,s)
  \cong SL(2,\Z_p)\times\Witt^\pt .
\ee
Equivalently, on the $p$-primary part,
\be
  \Witt^\pt({\Z_p \oplus \Z_p},s)_{p{\rm -primary}}
  \cong SL(2,\Z_p)\times W_p\,,
  \ee
  where $W_p$ is defined in (\ref{Wittp}).
\end{theorem}
\begin{proof}
The exact sequence (\ref{eq:oddp-exact}) identifies that the quotient is correct,  and the computation in Appendix~\ref{app:H5B3} shows that there is no
additional 
{$H^5(B^3(\Z_p \oplus \Z_p), \mathbb{C}^\times)$} contribution for odd $p$.

It remains to identify the extension in (\ref{eq:oddp-exact}).
Since \eqref{eq:oddp-exact} is a central extension, its class lives in
$H^2(SL(2,\Z_p);\Witt^\pt)$ with trivial action, and by the universal
coefficient theorem (integral homology throughout)
\be
\ba
H^2\big(SL(2,&\Z_p);\Witt^\pt\big)\\
\cong\ & \mathrm{Ext}^1\big(H_1(SL(2,\Z_p);\Z),\Witt^\pt\big)\\
&\oplus \mathrm{Hom}\big(H_2(SL(2,\Z_p);\Z),\Witt^\pt\big) \,.
\ea
\label{eq:H2UCT}
\ee
Both summands vanish.  The Schur multiplier $H_2(SL(2,\Z_p);\Z)$ is trivial
for every prime $p\ge3$.  For
$p\ge5$ the group $SL(2,\Z_p)$ is perfect, so $H_1(SL(2,\Z_p);\Z)=0$.
For $p=3$ one has $H_1(SL(2,3);\Z)=\Z_3$, and $\mathrm{Ext}^1(\Z_3,\Witt^\pt)=0$
because $\Witt^\pt$ is an abelian $2$-group (of exponent $8$, see \eqref{Wittp}).
Hence $H^2(SL(2,\Z_p);\Witt^\pt)=0$: the central extension splits, and
since the kernel is central the split extension is the direct product
\eqref{eq:WittZp-goal}.
The $p$-primary part follows by inserting the known summand $W_p$ of
the ordinary pointed Witt group.  All $\ell\neq p$ summands are purely
grade-zero and remain decoupled.
\end{proof}

\bigskip\noindent{\bf Proof of Theorem \ref{thm:SToddMain} (Modular relation)}
By \eqref{eq:oddp-exact} the quotient is $SL(2,\Z_p)$, realized
concretely as the action on the Witt clases by 
\be
\ba
  \pi:\ \Witt^\pt(A\oplus\widehat A,s)/\Witt^\pt
  \ \xrightarrow \ & \Aut^{\rm syp}(A\oplus\widehat A,\mathsf{Alt}(s))\\
  &=\ SL(2,\Z_p)\,,
\ea
\label{eq:pimapp}
\ee
sending a class to the symplectic automorphism it induces on the charges.  We
write $X\equiv Y$ for $\pi(X)=\pi(Y)$.  Since $\kappa$ is grade zero,
$\hat\ga\equiv\ga$ and $\hat\gb\equiv\gb$, so quotient computations may use the
undressed generators \eqref{eq:gensp}.  With
$\langle x,y\rangle:=\mathsf{Alt}(s)(x,y)=x_1y_2-x_2y_1$, the transvection of axis
$v\neq0$ and parameter $c\in\Z_p$,
\be
  t_{v,c}(x)=x+c\,\langle x,v\rangle\,v\,,
\label{eq:transvection}
\ee
fixes $\langle v\rangle$, lies in $SL(2,\Z_p)$, and satisfies
$t_{v,c}\,t_{v,c'}=t_{v,c+c'}$, $t_{\mu v,c}=t_{v,\mu^2c}$, and
$M\,t_{v,c}\,M^{-1}=t_{Mv,c}$ for $M\in SL(2,\Z_p)$.  As
$SL(2,\Z_p)=\Aut^{\rm syp}$ acts on the classes by relabeling the grading,
$f\mapsto M\circ f$, and $\pi$ records the induced charge-action, $\pi$ is
equivariant,
\be
  M_*\mathfrak s_{v,k}=\mathfrak s_{Mv,k}\,,\qquad \pi(M_*X)=M\,\pi(X)\,M^{-1}\,.
\label{eq:equivariance}
\ee
We abbreviate a rank-$2$ representative (a metric group with $G=\Z_p^2$) by
\be
  [\,g_1,g_2;\ d_1,d_2;\ B_{12}\,]:=\big(f:\Z_p^2\to A\oplus\widehat A;\ q=\zeta_p^{Q}\big)\,,
\label{eq:rank2abbrev}
\ee
with $f(x,y)=x\,g_1+y\,g_2$ and $Q(x,y)=d_1x^2+d_2y^2+B_{12}\,xy$, so the basis
grades are $g_1,g_2$, the diagonal values $d_i=Q(e_i)$ and the polarization
$B_{12}=B(e_1,e_2)$.  Furthermore note that the values of $\tau$ in \eqref{eq:arith-gen} are 
\be\label{eq:tauvalues}
  \tau(a,b)=1\,,\qquad \tau(b,a)=2\,,\qquad \tau(a,a)=\tau(b,b)=0\,.
\ee
Every step below divides only by $2$ and $4$, so it holds for {\it every} odd
prime $p$.  The sole $p$-dependent instance is $\tfrac32\equiv0$ at $p=3$, which
we flag where it arises.

\begin{lemma}[Reduction and triviality]\label{lem:trivp}
Let $X=(f:G\to A\oplus\widehat A,\,q)$ be $s$-invertible and $G_0=\ker f$.
\begin{itemize}
\item[(i)] If $N\subseteq G_0$ has $B|_N$ non-degenerate, then
$X\equiv(f|:N^{\perp}\to A\oplus\widehat A,\,q|)$ with $q|$ the restriction: a
non-degenerate grade-zero subgroup splits off.
\item[(ii)] If $L\subseteq G_0$ is isotropic with $L^{\perp}\subseteq G_0$,
then $X\equiv 1$.
\item[(iii)] If $f$ is injective and $G\neq 0$, then $X\not\equiv 1$.
\end{itemize}
\end{lemma}
\begin{proof}
(i) $B|_N$ non-degenerate gives the orthogonal splitting $G=N\perp N^{\perp}$.
This is a splitting of \emph{graded} premetric groups because every $\tau$-term
against the zero grade of $N$ vanishes, and the ungraded factor $N$ lies in
$\Witt^\pt$.
(ii) Condensing $L$ (Lemma~\ref{lem:cond}) represents $X$ by
$(\bar f:L^{\perp}/L\to A\oplus\widehat A,\,\bar q)$, and $L^{\perp}\subseteq G_0$
forces $\bar f=0$, so the class is grade zero.
(iii) If $X\equiv1$, then $X\otimes w^{\rm op}$ is $A \oplus \widehat{A}$-trivial for a grade-zero
$w=(0:H\to A\oplus\widehat A,\,q_H)$, with a Lagrangian
$\Lambda\subseteq(G\oplus H)_0=0\oplus H$.  Block-diagonality of the
polarization gives $|\Lambda|^2\le|H|$, while $\Lambda^{\perp}=\Lambda$ gives
$|\Lambda|^2\ge|G|\,|H|$, so $|G|=1$.  (No non-degeneracy of $B$ on $G$ is
used, which matters at $p=3$ below.)
\end{proof}

\begin{lemma}[No central corrections]\label{lem:nocentralp}
$\langle \hat\ga,\hat\gb\rangle\cap \Witt^\pt=1$, hence $\pi$ is injective on
$\langle \hat\ga,\hat\gb\rangle$, and any word $\omega$ in
$\hat\ga^{\pm1},\hat\gb^{\pm1}$ with $\omega\equiv 1$ satisfies $\omega=1$.
\end{lemma}
\begin{proof}
The extension \eqref{eq:oddp-exact} splits (Theorem~\ref{thm:WittZp}).  Fix
$\Witt^\pt(A\oplus\widehat A,s)=\Witt^\pt\times \mathcal{S}$ and let $w$ be the
projection to $\Witt^\pt$.  By Proposition~\ref{prop:apowerp}, $\hat\ga^{\,p}=1$,
so $w(\hat\ga)^p=1$.  As $\Witt^\pt$ has exponent $8$ (see \eqref{Wittp}) and
$\gcd(p,8)=1$, this gives $w(\hat\ga)=1$, and likewise $w(\hat\gb)=1$.  Thus
$\langle\hat\ga,\hat\gb\rangle\subseteq\mathcal{S}=\ker w$, which meets
$\Witt^\pt$ only in $1$.
\end{proof}

\begin{lemma}[Products of graded anyons]\label{lem:axissemion}
For $v\in\{a,b\}$ and $k,l\neq 0$, along a fixed axis $v$
\be
  \mathfrak{s}_{v,k}^{\,\rm op}=\mathfrak{s}_{v,-k}\,,\qquad
  \mathfrak{s}_{v,k}\otimes \mathfrak{s}_{v,l}\equiv
  \begin{cases}
    \mathfrak{s}_{v,\,kl/(k+l)}\,, & k+l\neq 0\,,\\[1mm]
    1\,, & k+l=0\,,
  \end{cases}
\label{eq:semionprod}
\ee
so $\hat\ga^{\,n}\equiv \mathfrak{s}_{a,1/n}$, $\hat\gb^{\,n}\equiv
\mathfrak{s}_{b,1/n}$ for $p\nmid n$.  Along independent axes,
$\mathfrak{s}_{u,j}\otimes\mathfrak{s}_{w,k}=[\,u,w;\ j,k;\ \tau(u,w)\,]$, with
no reduction.
\end{lemma}
\begin{proof}
Since $\sigma(v,v)=0$, the opposite in \eqref{eq:arith-gen} is
$\widetilde Q(x)=-kx^2=\mathfrak{s}_{v,-k}$.  The same-axis product is
$Q(x,y)=kx^2+ly^2$, with grade-zero line $L=\langle(1,-1)\rangle$ on which
$Q=(k+l)t^2$.  If $k+l=0$ then $L$ is isotropic with $L^\perp=L$, so
$\equiv1$ by Lemma~\ref{lem:trivp}(ii).  If $k+l\neq0$ then splitting off this
non-degenerate grade-zero line (Lemma~\ref{lem:trivp}(i)) leaves
$\langle(l,k)\rangle$, of grade $(k+l)v$ and $Q=kl(k+l)$, i.e.\
$\mathfrak{s}_{v,\,kl/(k+l)}$ after rescaling, and induction gives the power
formula.  An independent-axis product is injectively graded, hence already
reduced.
\end{proof}

\begin{proposition}\label{prop:intrinsicmap}
For every $k\neq0$,
\be
\ba
  \pi(\mathfrak{s}_{a,k})&=t_{a,\,-1/2k}=\left(\begin{smallmatrix}1&\tfrac1{2k}\\[1mm]0&1\end{smallmatrix}\right),\\
  \pi(\mathfrak{s}_{b,k})&=t_{b,\,-1/2k}=\left(\begin{smallmatrix}1&0\\[1mm]-\tfrac1{2k}&1\end{smallmatrix}\right),
\ea
\label{eq:semionmatrices}
\ee
in particular $\pi(\hat\ga)=\left(\begin{smallmatrix}1&\tfrac12\\0&1\end{smallmatrix}\right)$ and $\pi(\hat\gb)=\left(\begin{smallmatrix}1&0\\-\tfrac12&1\end{smallmatrix}\right)$, the images \eqref{eq:gagbImages} identified with qudit Clifford gates in Section~\ref{sec:qudit}.
\end{proposition}
\begin{proof}
\emph{Transvection.}  Each $\mathfrak{s}_{v,k}$ has order $p$ in the quotient
(the computation of Proposition~\ref{prop:apowerp}, which uses only
$\sigma(v,v)=0$), so $\pi(\mathfrak{s}_{v,k})$ is a non-trivial unipotent, i.e.\
a transvection.

\emph{Center $=$ grade.}  Along a fixed axis the anyons commute (\eqref{eq:semionprod} is symmetric in $k,l$), so $\{\pi(\mathfrak{s}_{a,k})\}_k$
is a commuting family of transvections and shares a single center $v_a$,
independent of $k$.  For $M$ in the stabiliser $B_a$ of $\langle a\rangle$ with
$Ma=\lambda a$, rescaling the  generator gives
$M_*\mathfrak{s}_{a,k}=\mathfrak{s}_{a,\,k/\lambda^2}$, so by
\eqref{eq:equivariance} $M\pi(\mathfrak{s}_{a,k})M^{-1}=\pi(\mathfrak{s}_{a,k/\lambda^2})$
again has center $v_a$, hence $Mv_a=v_a$ for every $M\in B_a$.  The only
$B_a$-fixed line of $\mathbb P^1(\Z_p)$ is $\langle a\rangle$, so
$v_a=\langle a\rangle$, and likewise $v_b=\langle b\rangle$.

Writing $\pi(\mathfrak{s}_{a,k})=t_{a,c(k)}$, the map
$k\mapsto c(k)$ is a homomorphism to $(\Z_p,+)$ and \eqref{eq:semionprod}
gives $c(kl/(k+l))=c(k)+c(l)$.  As $1/(kl/(k+l))=1/k+1/l$, this forces
$c(k)=\gamma_a/k$ for a constant $\gamma_a$, and similarly $c(k)=\gamma_b/k$ on
the $b$-axis.

The swap $S_0=\left(\begin{smallmatrix}0&-1\\1&0\end{smallmatrix}\right)$ has $S_0a=b$, so
\eqref{eq:equivariance} gives
$\pi(\mathfrak{s}_{b,k})=S_0\,\pi(\mathfrak{s}_{a,k})\,S_0^{-1}=t_{b,\gamma_a/k}$,
i.e.\ $\gamma_b=\gamma_a=:\gamma$.  Thus $\pi(\hat\ga)=t_{a,\gamma}$,
$\pi(\hat\gb)=t_{b,\gamma}$, and
\be
\ba
  \pi(S)&=\left(\begin{smallmatrix}1-4\gamma^2&\gamma(2-4\gamma^2)\\[1mm]-4\gamma&1-4\gamma^2\end{smallmatrix}\right),\\
  \mathrm{tr}\,\pi(S)&=2(1-4\gamma^2)\,.
\ea
\label{eq:piS}
\ee
It remains to fix $\gamma$, and for this we use that $S^2$ is a non-trivial
involution.  By Lemma~\ref{lem:axissemion},
\be
\ba
  S\ &\equiv\ \Sigma:=[\,a-2b,\ a+2b;\ -4,2;\ 0\,]\,,\\
  S^2\ &\equiv\ \Gamma:=[\,a,b;\ 0,0;\ \tfrac32\,]\,,
\ea
\label{eq:SGamma}
\ee
where $\Sigma$ is the rank-$3$ product
$\mathfrak{s}_{a,-1}\mathfrak{s}_{b,-1/4}\mathfrak{s}_{a,-1}$ reduced by splitting
off its grade-zero line $\langle(1,0,-1)\rangle$ ($Q=-2\neq0$,
Lemma~\ref{lem:trivp}(i)), and $\Gamma$ is $\Sigma^2$ reduced by splitting off
its non-degenerate grade-zero plane $\mathrm{diag}(-4,-4)$.  From
$\widetilde Q=-Q+3\,\sigma\circ f$ one checks $\Gamma^{\rm op}=\Gamma$, so
$\Gamma^2\equiv1$, while $\Gamma\not\equiv1$ by Lemma~\ref{lem:trivp}(iii)
(valid at $p=3$ too, where $\tfrac32\equiv0$).  Hence $\pi(S)^2=-\mathbf1$ and
$\mathrm{tr}\,\pi(S)=0$, so \eqref{eq:piS} gives $4\gamma^2=1$.  We fix the sign
$\gamma=-\tfrac12$ to match the qudit Clifford identification \eqref{eq:gagbImages}
of Section~\ref{sec:qudit}, giving
\eqref{eq:semionmatrices}.\footnote{The other choice $\gamma\mapsto-\gamma$ is
conjugation by $\mathrm{diag}(1,-1)$ and leaves \eqref{eq:oddrelsexact}
unchanged.}
\end{proof}

\begin{theorem}[Modular relations]\label{thm:oddprels}
In $\Witt^\pt(A\oplus\widehat A,s)$ the words $T=\hat\ga$, $S=\hat\ga^{-1}\hat\gb^{-4}\hat\ga^{-1}$
of \eqref{eq:STodd} satisfy \eqref{eq:SToddrels} exactly,
\be
  T^p=1\,,\qquad S^4=1\,,\qquad (ST)^3=S^2\,,
\label{eq:oddrelsexact}
\ee
with no central corrections. Moreover $S^2=C$ is central and acts on the
surface charges by charge conjugation $x\mapsto-x$.
\end{theorem}
\begin{proof}
By Lemma~\ref{lem:nocentralp} it suffices to verify the relations after $\pi$,
where by Proposition~\ref{prop:intrinsicmap} they are matrix identities over
$\mathbb Z[\tfrac12]$, hence hold for every odd prime $p$:
\be
\ba
  &\pi(T)^p=\left(\begin{smallmatrix}1&p/2\\0&1\end{smallmatrix}\right)=\mathbf1\,,\quad
   \pi(S)=\left(\begin{smallmatrix}0&-\tfrac12\\2&0\end{smallmatrix}\right),\\[1mm]
  &\pi(S)^2=-\mathbf1\,,\qquad \pi(S)^4=\mathbf1\,,\\[1mm]
  &\pi(ST)=\left(\begin{smallmatrix}0&-\tfrac12\\2&1\end{smallmatrix}\right),\qquad
   \pi\big((ST)^3\big)=-\mathbf1=\pi(S^2)\,.
\ea
\label{eq:matrixrels}
\ee
Lemma~\ref{lem:nocentralp} lifts each $\equiv$ to an equality in
$\Witt^\pt(A\oplus\widehat A,s)$.  Finally $\pi(S^2)=-\mathbf1$ is central in
$SL(2,\Z_p)$ and acts by $x\mapsto-x$.  With
$\langle\hat\ga,\hat\gb\rangle\cap\Witt^\pt=1$ and the splitting of
Lemma~\ref{lem:nocentralp}, $S^2=C$ is central in $\Witt^\pt(A\oplus\widehat A,s)$.
\end{proof}

The normalized generators realize the split form of the group: 
\begin{corollary}\label{cor:exactlifts}
$\langle\hat\ga,\hat\gb\rangle\cong SL(2,\Z_p)$, mapping isomorphically onto the
quotient in \eqref{eq:oddp-exact}: every relation of $SL(2,\Z_p)$ holds for
$\hat\ga,\hat\gb$ with no central corrections.
\end{corollary}

\begin{proof}[Proof]
By Proposition~\ref{prop:intrinsicmap}, $\pi(\hat\ga)=t_{a,1/2}$ and
$\pi(\hat\gb)=t_{b,1/2}$ are transvections with distinct centers
$\langle a\rangle\neq\langle b\rangle$.  Their non-zero powers are the standard
elementary transvections, which generate $SL(2,\Z_p)$.  Thus
$\pi(\langle\hat\ga,\hat\gb\rangle)=SL(2,\Z_p)$, and by
Lemma~\ref{lem:nocentralp} $\pi$ restricts to an isomorphism
$\langle\hat\ga,\hat\gb\rangle\cong SL(2,\Z_p)$: every relation of
$SL(2,\Z_p)$ holds for $\hat\ga,\hat\gb$.
\end{proof}

Thus the entire content of \eqref{eq:SToddrels} sits in the single scalar
$\gamma^2=\tfrac14$ of Proposition~\ref{prop:intrinsicmap}: the transvection
form, the grade as center, and the level law
$\mathfrak{s}_{v,k}\mapsto t_{v,-1/2k}$ are forced by equivariance, after which
the relations are $2\times2$ matrix algebra.

\subsection{Vanishing of $H^5(B^3 (\Z_p \oplus \Z_p), \mathbb{C}^\times)$}
\label{app:H5B3}

We prove that for $p$ odd (prime)
\be
 H^5( (\Z_p \oplus \Z_p)[3], \mathbb{C}^\times) =0 \,.
\ee
First note that the LHS is defined to be 
\be\label{H5B3}
 H^5(B^3 (\Z_p \oplus \Z_p), \mathbb{C}^\times) = H^6 (B^3(\Z_p \oplus\Z_p), \Z) \,, 
\ee
where $B^3 (\Z_p \oplus \Z_p)$ is the Eilenberg-MacLane space 
\be\label{KZp2}
K((\Z_p)^2 ,3)   = K(\Z_p,3) \times K(\Z_p,3)\,. 
\ee
We can compute the cohomology (\ref{H5B3}) by first computing the cohomology of $K(\Z_p, 3)$ and then applying K\"unneth.
Mod $p$, the cohomology of a single factor is the free graded-commutative
algebra (exterior on the odd-degree generators, polynomial on the even ones)
\be\ba \label{eq:KZp3ring}
&H^\bullet (K(\Z_p, 3), \Z_p) \cr 
=& \Lambda\big[x_3,\, P^1 x_3,\ldots\big]\otimes
\Z_p\big[\beta x_3,\, \beta P^1 x_3,\ldots\big]\,,
\ea
\ee
where $x_3$ is the fundamental degree-$3$ class, $\beta$ is the $\Z_p$ Bockstein
($\deg\beta x_3=4$), and the Steenrod power $P^1$ raises the degree by $2(p-1)$, so
$\deg P^1x_3=2p+1\ge7$ for every odd prime.  Note $x_3^2=0$ by graded commutativity
($x_3$ has odd degree and $p$ is odd).  Hence in low degrees the only nonzero groups
are
\be\label{eq:KZp3low}
H^k(K(\Z_p,3);\Z_p)=
\left\{ \ba 
\Z_p  & \quad 
\text{for } k=0,3,4, 7, 8, \cdots \cr 
0 & \quad  \text{for }k=1,2,5,6, \cdots \,.
\ea\right.
\ee
It is convenient to pass to integral cohomology before applying K\"unneth: since all
positive-degree groups are finite $p$-torsion, the universal coefficient theorem
$H^k(X;\Z_p)\cong H^k(X;\Z)\otimes\Z_p\,\oplus\,\mathrm{Tor}(H^{k+1}(X;\Z),\Z_p)$
applied to \eqref{eq:KZp3low} gives
\be
H^k(K(\Z_p,3);\Z)= \Z,\,0,\,0,\,0,\,\Z_p,\,0,\,0 \quad\text{for }k=0,\ldots,6\,,
\label{eq:KZp3int}
\ee
the $\Z_p$ in degree $4$ being the integral Bockstein lift of $x_3$.  The integral
K\"unneth formula for \eqref{KZp2} in degree $6$ then reads
\be
\ba
H^6\big(K(\Z_p,3)^{\times2};\Z\big)
=&\bigoplus_{i+j=6} H^i\otimes H^j \\
&\oplus \bigoplus_{i+j=7}\mathrm{Tor}\big(H^i,H^j\big)\,,
\ea
\label{eq:KunnethDeg6}
\ee
and every term vanishes by \eqref{eq:KZp3int}: for $i+j=6$ each product involves a
vanishing factor ($H^{2}=H^{3}=H^{6}=0$ integrally), and for $i+j=7$ the only
candidates $\mathrm{Tor}(H^4,H^3)$ and $\mathrm{Tor}(H^3,H^4)$ vanish since
$H^3(K(\Z_p,3);\Z)=0$.  Hence $H^6(B^3(\Z_p\oplus\Z_p);\Z)=0$, proving the claim.
Note that working integrally is essential: the \emph{mod-$p$} K\"unneth formula does
have a nonzero degree-$6$ class, $x_3^{(1)}\otimes x_3^{(2)}$; it corresponds to no
integral torsion, as \eqref{eq:KZp3int} makes manifest (equivalently, it is not in the
image of the Bockstein, since $H^5$ of the product vanishes mod $p$).

\section{Higher Cup Products on a Cubic Lattice}
\label{sec: higher cup products}
In this appendix, we review the definitions of the (co)boundary operation, the cup product, and the higher cup products on a 3d cubic lattice.
We will follow the convention in \cite[Appendix B]{MengSun2026}.
See also~\cite{Chen_2023} for more details of higher cup products on hypercubic lattices.

\vspace*{\baselineskip}
\noindent{\bf Notations.}
Let us first fix the notations.
We denote a three-dimensional cube on a cubic lattice in $\mathbb{R}^3$ by
\begin{equation}
(\bullet, \bullet, \bullet) \coloneq \{(x, y, z) \mid x, y, z \in [0, 1] \}.
\end{equation}
Here, $\bullet$ in the $i$th entry indicates that the $i$th coordinate varies from $0$ to $1$.
The boundary of cube $(\bullet, \bullet, \bullet)$ consists of six faces of the form
\begin{equation}
(x, \bullet, \bullet), \quad (\bullet, y, \bullet), \quad (\bullet, \bullet, z),
\label{eq: faces}
\end{equation}
where $x, y, z \in \{0, 1\}$.
These faces are illustrated in Figure~\ref{fig: faces}.
Similarly, the boundaries of the above faces consist of the edges of the form
\begin{equation}
(x, y, \bullet), \quad (x, \bullet, z), \quad (\bullet, y, z),
\label{eq: edges}
\end{equation}
where $x, y, z \in \{0, 1\}$.
These edges are illustrated in Figure~\ref{fig: edges}.
\begin{figure*}[t]
\centering
\adjincludegraphics[valign=c]{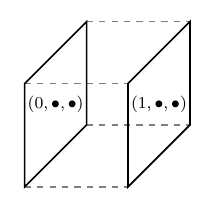} $\quad$
\adjincludegraphics[valign=c]{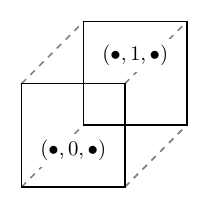} $\quad$
\adjincludegraphics[valign=c]{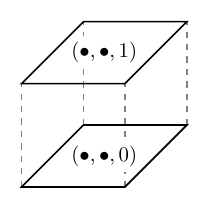}
\caption{The faces on the boundary of cube $(\bullet, \bullet, \bullet)$.}
\label{fig: faces}
\end{figure*}
\begin{figure*}[t]
\centering
\adjincludegraphics[valign=c]{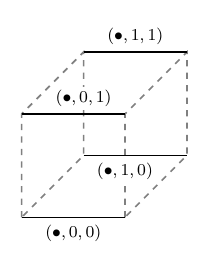} $\quad$
\adjincludegraphics[valign=c]{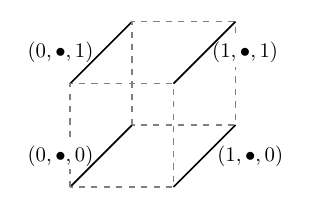}
\adjincludegraphics[valign=c]{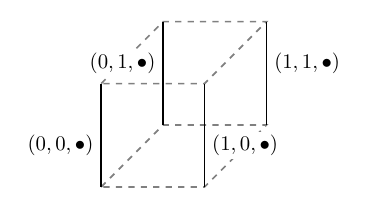}
\caption{The edges on the boundary of cube $(\bullet, \bullet, \bullet)$.}
\label{fig: edges}
\end{figure*}

\vspace{\baselineskip}
\noindent{\bf Boundary and coboundary.}
Now, we define the boundary operation on $\mathbb{Z}$-valued chains on a cubic lattice.
First of all, the boundary of a 0-chain $v$ is defined to be zero.
The boundary of a 1-chain $e \coloneq (\bullet)$ is defined by
\begin{equation}
\partial (\bullet) = (1) - (0).
\label{eq: partial e def}
\end{equation}
Here, $(\bullet)$ denotes any one of the edges of the form \eqref{eq: edges}, where we omitted the fixed coordinates.
Similarly, the boundary of a 2-chain $f \coloneq (\bullet, \bullet)$ is defined by
\begin{align}
\partial (\bullet, \bullet) = (1, \bullet) - (0, \bullet) - (\bullet, 1) + (\bullet, 0).
\label{eq: partial f def}
\end{align}
Here, $(\bullet, \bullet)$ denotes any one of the faces of the form~\eqref{eq: faces}, where we omitted the fixed coordinate.
Finally, the boundary of a 3-chain $c \coloneq (\bullet, \bullet, \bullet)$ is defined by
\begin{equation}
\begin{aligned}
\partial (\bullet, \bullet, \bullet) &= (1, \bullet, \bullet) - (0, \bullet, \bullet) - (\bullet, 1, \bullet) \\
& \quad + (\bullet, 0, \bullet) + (\bullet, \bullet, 1) - (\bullet, \bullet, 0).
\end{aligned}
\label{eq: partial c def}
\end{equation}
The boundary of a general chain is then defined by linearly extending the above definition.

Based on the above definition of the boundary, one can define the coboundary of any $\mathbb{Z}$-valued cochain $\bm{a}$ by
\begin{equation}
\delta \bm{a}(b) = \bm{a}(\partial b),
\end{equation}
where $b$ is an arbitrary $\mathbb{Z}$-valued chain.
In particular, the coboundary of any 3-cochain is zero.
For each face $f$, edge $e$, and vertex $v$, the coboundaries of the corresponding cochains can be written explicitly as
\begin{align}
\delta \bm{f} &= \sum_{c^{\prime}} \bm{f}(\partial c^{\prime}) \bm{c}^{\prime},
\label{eq: delta f def}
\\
\delta \bm{e} &= \sum_{f^{\prime}} \bm{e}(\partial f^{\prime}) \bm{f}^{\prime},
\label{eq: delta e def}
\\
\delta \bm{v} &= \sum_{e^{\prime}} \bm{v}(\partial e^{\prime}) \bm{e^{\prime}},
\label{eq: delta v def}
\end{align}
where the summations are taken over all cubes, faces, and edges, respectively.
We note that the coefficients on the right-hand side are either $0$ or $\pm 1$.
Specifically, the coefficient in \eqref{eq: delta f def} is non-zero only when $f$ is on the boundary of $c^{\prime}$.
Similarly, the coefficient in \eqref{eq: delta e def} is non-zero only when $e$ is on the boundary of $f^{\prime}$, and the coefficient in \eqref{eq: delta v def} is non-zero only when $v$ is on the boundary of $e^{\prime}$.

\vspace{\baselineskip}
\noindent{\bf Higher cup products.}
Using the above notation, the cup product and the higher cup products of $\mathbb{Z}$-valued cochains on a cubic lattice can be written as follows \cite[Appendix B]{MengSun2026}:
\begin{widetext}
\begin{align}
\bm{a}_2 \cup \bm{b}_1 (\bullet, \bullet, \bullet) &= \bm{a}_2(0, \bullet, \bullet) \bm{b}_1(\bullet, 1, 1) - \bm{a}_2(\bullet, 0, \bullet) \bm{b}_1(1, \bullet, 1) + \bm{a}_2(\bullet, \bullet, 0) \bm{b}_1(1, 1, \bullet), \\
\bm{a}_1 \cup \bm{b}_2 (\bullet, \bullet, \bullet) &= \bm{a}_1(0, 0, \bullet) \bm{b}_2(\bullet, \bullet, 1) - \bm{a}_1(0, \bullet, 0) \bm{b}_2(\bullet, 1, \bullet) + \bm{a}_1(\bullet, 0, 0) \bm{b}_2(1, \bullet, \bullet), \\
\bm{a}_2 \cup_1 \bm{b}_2 (\bullet, \bullet, \bullet) &= \bm{a}_2(\bullet, 1, \bullet) \bm{b}_2(\bullet, \bullet, 0) - \bm{a}_2(\bullet, \bullet, 1) \bm{b}_2(\bullet, 0, \bullet) +\bm{a}_2(0, \bullet, \bullet) \bm{b}_2(\bullet, \bullet, 0) \\
&\quad + \bm{a}_2(0, \bullet, \bullet) \bm{b}_2(\bullet, 1, \bullet) - \bm{a}_2(\bullet, \bullet, 1) \bm{b}_2(1, \bullet, \bullet) - \bm{a}_2(\bullet, 0, \bullet) \bm{b}_2(1, \bullet, \bullet), \\
\bm{a}_1 \cup_1 \bm{b}_3 (\bullet, \bullet, \bullet) &= \bm{a}_1(\bullet, 1, 1) \bm{b}_3(\bullet, \bullet, \bullet) + \bm{a}_1(0, \bullet, 1) \bm{b}_3(\bullet, \bullet, \bullet) + \bm{a}_1(0, 0, \bullet) \bm{b}_3(\bullet, \bullet, \bullet), \\
\bm{a}_3 \cup_1 \bm{b}_1 (\bullet, \bullet, \bullet) &= \bm{a}_1(\bullet, \bullet, \bullet) \bm{b}_1(\bullet, 0, 0) + \bm{a}_3(\bullet, \bullet, \bullet) \bm{b}_1(1, \bullet, 0) + \bm{a}_3(\bullet, \bullet, \bullet) \bm{b}_1(1, 1, \bullet), \\
\bm{a}_2 \cup_2 \bm{b}_3 (\bullet, \bullet, \bullet) &= \bm{a}_2(\bullet, \bullet, 0) \bm{b}_3(\bullet, \bullet, \bullet) + \bm{a}_2(\bullet, 1, \bullet) \bm{b}_3(\bullet, \bullet, \bullet) + \bm{a}_2(0, \bullet, \bullet) \bm{b}_3(\bullet, \bullet, \bullet), \\
\bm{a}_3 \cup_2 \bm{b}_2 (\bullet, \bullet, \bullet) &= \bm{a}_3(\bullet, \bullet, \bullet) \bm{b}_2(1, \bullet, \bullet) + \bm{a}_3(\bullet, \bullet, \bullet) \bm{b}_2(\bullet, 0, \bullet) + \bm{a}_3(\bullet, \bullet, \bullet) \bm{b}_2(\bullet, \bullet, 1), \\
\bm{a}_3 \cup_3 \bm{b}_3 (\bullet, \bullet, \bullet) &= \bm{a}_3(\bullet, \bullet, \bullet) \bm{b}_3(\bullet, \bullet, \bullet).
\end{align}
\end{widetext}
Here, the subscripts of $\bm{a}$ and $\bm{b}$ represent the degrees of these cochains.

\section{Equivalent definition of 3-fermion Kramers-Wannier-Wegner operator}
\label{sec: Equivalent definition of 3-fermion KWW operator}
In Section~\ref{sec: 3-fermion KWW}, we defined the 3-fermion Kramers-Wannier-Wegner operator $\mathsf{D}_\text{3F}$ by using the KWW operator $\mathsf{S}$ and the Tsui-Wen entangler $\mathsf{T}^2$.
In this appendix, we will provide another equivalent definition of $\mathsf{D}_\text{3F}$ by using the 3-fermion QCA in \cite[Appendix E]{Shirley_2022}.

\subsection{Kramers-Wannier-Wegner operator revisited}
\label{sec: KWW revisited}
As we will see later, our definition of the 3-fermion KWW operator is motivated by a specific representation of the ordinary KWW operator.
As such, we start by revisiting the KWW operator in this subsection.
More specifically, the goal of this subsection is to represent the KWW operator as a quantum circuit consisting of unitary operators and projections, using ancillary qubits.

%\vspace*{\baselineskip}
%\noindent{\bf Definition.}
\subsubsection{Definition of the KWW operator}
Let us first recall the definition of the KWW operator.
The KWW operator $\mathsf{S}$ acting on the tensor product Hilbert space~\eqref{eq: qubit Hilb} is defined by the composition
\begin{equation}
\mathsf{S} \coloneq t^{\frac{1}{2}} \mathsf{H} \mathsf{D}_{\mathrm{gauge}}.
\label{eq: KWW op}
\end{equation}
Here, $t^{\frac{1}{2}}$ is the half translation in the $(\frac{1}{2}, \frac{1}{2}, \frac{1}{2})$ direction, $\mathsf{H}$ is the tensor product of the on-site Haramard gates, and $\mathsf{D}_{\mathrm{gauge}}$ is the gauging operator for $\mathbb{Z}_2$ 1-form symmetry.
The gauging operator maps qubits on the faces to those on the edges, i.e.,
\begin{equation}
\mathsf{D}_{\mathrm{gauge}}: \mathcal{H}_F = \bigotimes_{f \in F} \mathbb{C}^2 \to \mathcal{H}_E = \bigotimes_{e \in E} \mathbb{C}^2.
\label{eq: gauging op}
\end{equation}
Here, $F$ and $E$ denote the sets of faces and edges, respectively.
The action of $\mathsf{D}_{\mathrm{gauge}}$ on any state in the computational basis is given by \cite{Yoshida:2015cia, Gorantla:2024ocs}
\begin{equation}
\mathsf{D}_{\mathrm{gauge}} \left( \bigotimes_{f \in F} \ket{a_f}_f \right) = \frac{1}{2^{N/2}} \bigotimes_{e \in E} \ket{a^{\prime}_e}_e, 
\end{equation}
where $N$ is the number of cubes and $a^{\prime}_e$ is the sum of the $a_f$'s on the faces around $e$, that is,
\begin{equation}
a^{\prime}_e = \sum_{f \in F \text{ s.t. } e \in \partial f} a_f \quad \mod 2.
\end{equation}
We note that $\mathsf{D}_{\mathrm{gauge}}$ is non-invertible because it annihilates states charged under (non-topological) $\mathbb{Z}_2$ 1-form symmetry.
The normalization factor in \eqref{eq: gauging op} is chosen so that $\mathsf{D}_{\mathrm{gauge}}$ satisfies
\begin{equation}
\mathsf{D}_{\mathrm{gauge}}^{\dagger} \mathsf{D}_{\mathrm{gauge}} = \frac{1}{2^N} \sum_{\Sigma: \text{2-cycles}} \eta(\Sigma),
\end{equation}
where the right-hand side is the condensation operator on the lattice~\cite{Gorantla:2024ocs}.
The normalization of $\mathsf{D}_{\mathrm{gauge}}$ is not important in later discussions.

%\vspace*{\baselineskip}
%\noindent{\bf Tensor network representation.}
\subsubsection{Tensor network representation}
The KWW operator $\mathsf{S}$ defined above can be represented by the following tensor network:
\begin{equation}
\mathsf{S} = \frac{1}{2^{N/2}}t^{\frac{1}{2}} \mathsf{H} \adjincludegraphics[valign=c, trim={10, 10, 10, 10}]{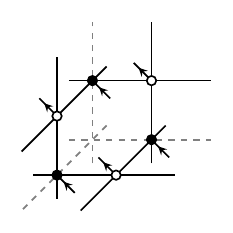}.
\label{eq: KWW TN1}
\end{equation}
Here, the gray dashed lines represent the edges of the dual lattice (i.e., the faces of the direct lattice), the black edges with arrows represent physical legs,\footnote{The incoming edges correspond to the initial state, whereas the outgoing edges correspond to the final state.} and those without arrows represent the virtual legs.
The bond Hilbert space associated with each virtual leg is $\mathbb{C}^2$, whose basis states are labeled by $0$ and $1$.
The black and white dots in the above diagram represent the copy and addition tensors defined by
\begin{align}
\adjincludegraphics[valign=c, trim={10, 10, 10, 10}]{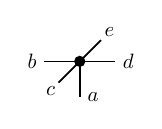} &= \delta_{a, b} \delta_{a, c} \delta_{a, d} \delta_{a, e}, \\
\adjincludegraphics[valign=c, trim={10, 10, 10, 10}]{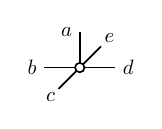} &= \delta_{a, b+c+d+e \bmod 2}.
\end{align}
We note that the copy and addition tensors satisfy the following commutation relation with the Hadamard gate $H$:
\begin{equation}
\adjincludegraphics[valign=c, trim={10, 10, 10, 10}]{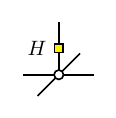} = 2^{\frac{3}{2}} \adjincludegraphics[valign=c, trim={10, 10, 10, 10}]{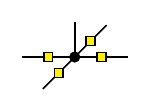}.
\end{equation}
Using this relation, we can write the KWW operator in \eqref{eq: KWW TN1} as
\begin{equation}
\mathsf{S} = 2^{4N} t^{\frac{1}{2}} \adjincludegraphics[valign=c, trim={10, 10, 10, 10}]{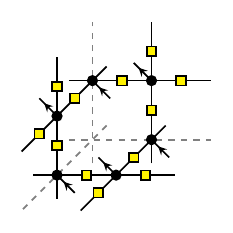}.
\label{eq: KWW TN2}
\end{equation}
Here, the yellow squares represent the Hadamard gate.
This tensor network representation was first obtained in~\cite{Gorantla:2024ocs}.

%\vspace*{\baselineskip}
%\noindent{\bf KWW in terms of unitary and projection.}
\subsubsection{Quantum circuit representation}
%To write $\mathsf{D}$ as the composition of a unitary and a projection, we further deform the tensor network representation in \eqref{eq: KWW TN2} as follows:
To write $\mathsf{S}$ as a quantum circuit, we further deform the tensor network representation in \eqref{eq: KWW TN2} as follows:
\begin{equation}
\mathsf{S} = 2^{4N} t^{\frac{1}{2}} \adjincludegraphics[valign=c, trim={10, 10, 10, 10}]{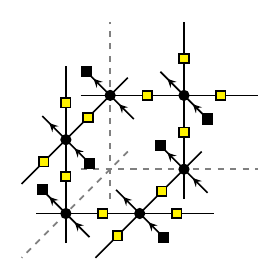}.
\label{eq: KWW TN3}
\end{equation}
Here, the black dots at the junctions are again the copy tensors, while the black squares at the endpoints of the physical legs are the local tensors defined by
\begin{equation}
\adjincludegraphics[valign=c, trim={10, 10, 10, 10}]{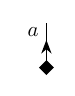} = \adjincludegraphics[valign=c, trim={10, 10, 10, 10}]{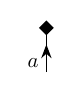} = 1, \quad \forall a \in \{0, 1\}.
\end{equation}
We note that the above tensors represent the unnormalized $\ket{+}$ state and its conjugate $\bra{+}$, i.e.,
\begin{equation}
\adjincludegraphics[valign=c, trim={10, 10, 10, 10}]{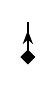} = \ket{0} + \ket{1} = \sqrt{2} \ket{+}, \quad
\adjincludegraphics[valign=c, trim={10, 10, 10, 10}]{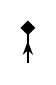} = \bra{0} + \bra{1} = \sqrt{2} \bra{+}.
\end{equation}
The tensor network representation in \eqref{eq: KWW TN3} shows that the action of $\mathsf{S}$ on any state $\ket{\psi} \in \mathcal{H}_F$ can be written as
\begin{equation}
\mathsf{S} \ket{\psi} = 2^{7N} t^{\frac{1}{2}} \bra{\{+\}_F} \sf{U}_{\mathrm{cluster}} \left(\ket{\psi} \otimes \ket{\{+\}_E}\right),
\label{eq: D unitary projection}
\end{equation}
where $\bra{\{+\}_F}$ and $\ket{\{+\}_E}$ are the trivial product states on the faces and edges, defined by
\begin{equation}
\bra{\{+\}_F} \coloneq \bigotimes_{f \in F} \bra{+}_f, \quad 
\ket{\{+\}_E} \coloneq \bigotimes_{e \in E} \ket{+}_e,
\end{equation}
and $\sf{U}_{\mathrm{cluster}}$ is a tensor network operator defined by
\begin{equation}
\sf{U}_{\mathrm{cluster}} \coloneq \adjincludegraphics[valign=c, trim={10, 10, 10, 10}]{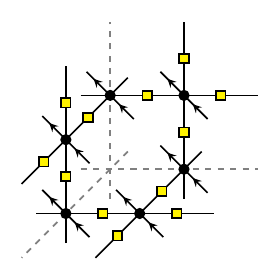}.
\label{eq: U cluster TN}
\end{equation}
As we will see below, $\sf{U}_{\mathrm{cluster}}$ is a finite-depth unitary circuit that entangles a cluster state.

To show that $\sf{U}_{\mathrm{cluster}}$ is a finite-depth unitary circuit, we resolve the copy tensor on each physical leg in \eqref{eq: U cluster TN} as follows:
\begin{equation}
\adjincludegraphics[valign=c, trim={10, 10, 10, 10}]{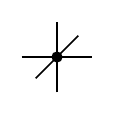} = \adjincludegraphics[valign=c, trim={10, 10, 10, 10}]{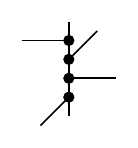}.
\label{eq: copy resolution}
\end{equation}
Here, the order of the vertices on the right-hand side does not matter because the copy tensors commute with each other.
Once we resolve the copy tensors as above, the adjacent physical legs in \eqref{eq: U cluster TN} are connected by a tensor network of the form
\begin{equation}
\adjincludegraphics[valign=c, trim={10, 10, 10, 10}]{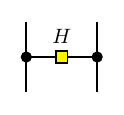} = CZ,
\end{equation}
which represents the controlled-$Z$ gate.
Thus, $\sf{U}_{\mathrm{cluster}}$ can be written as the product of local unitaries
\begin{equation}
\sf{U}_{\mathrm{cluster}} = \prod_{f \in F} \prod_{e \in \partial f} CZ_{e, f},
\label{eq: CZ product}
\end{equation}
where $CZ_{e, f}$ is the controlled-$Z$ gate acting on the qubits on a face $f$ and an edge $e$.
One can rearrange the product of $CZ$'s in~\eqref{eq: CZ product} so that $\sf{U}_{\mathrm{cluster}}$ becomes a unitary circuit of depth 4, as indicated by the decomposition~\eqref{eq: copy resolution}.
Equation~\eqref{eq: CZ product} shows that $\sf{U}_{\mathrm{cluster}}$ is the entangler of a cluster state in 3+1d, which realizes a non-trivial SPT phase with $\mathbb{Z}_2 \times \mathbb{Z}_2$ 1-form symmetry~\cite{Yoshida:2015cia}.

In summary, we found that $\mathsf{S}$ can be implemented by the following protocol: (1) we first add an ancillary qubit in the $\ket{+}$ state on each edge, (2) we then apply a finite-depth circuit $\sf{U}_{\mathrm{cluster}}$ to the entire state, (3) we then project the original qubits on the faces into the $\ket{+}$ states, and (4) we finally move the ancillas on the edges to the faces.
We note that the projection in the third step can be carried out adaptively by using measurements and feedforward if the initial state $\ket{\psi}$ is $\mathbb{Z}_2$ 1-form symmetric.
Namely, the KWW operator $\mathsf{S}$ admits an adaptive circuit representation.
We will not go into details of the adaptive circuit representation because it is analogous to the case of the Kramers-Wannier operator in 1+1d \cite{Tantivasadakarn:2021vel}.
See also \cite{Lootens:2023wnl} for adaptive circuit representations of more general duality operators in 1+1d.

\vspace*{\baselineskip}
\noindent{\bf Relation to Walker-Wang entangler.}
We note that $\sf{U}_{\mathrm{cluster}}$ can also be regarded as an entangler of the Walker-Wang ground state based on the toric code MTC \cite{Roberts:2020zmk}.
As we will see later, replacing the toric code MTC in the above construction with the 3-fermion MTC yields the 3-fermion KWW operator $\mathsf{D}_\text{3F}$.
In other words, the 3-fermion KWW operator is obtained by replacing $\sf{U}_{\mathrm{cluster}}$ in \eqref{eq: D unitary projection} with a 3-fermion QCA.
We will discuss this in more detail in Appendix~\ref{sec: 3F KWW revisited}.

\subsection{3-fermion QCA}
\label{sec: 3-fermion QCA 2-qubit version}
In this subsection, we recall the definition of the 3-fermion QCA constructed in \cite{Shirley_2022, MengSun2026}, which is equivalent to the original 3-fermion QCA defined in \cite{HaahFidkowskiHastings2023} up to finite-depth circuits.
This 3-fermion QCA will be used later to define the 3-fermion KWW operator.

To define the 3-fermion QCA, we consider a cubic lattice with two qubits on each face.
These qubits will be called an $A$-qubit and a $ B$-qubit.
The Pauli operators acting on the $A$-qubit and $B$-qubit on face $f$ are denoted by $\{X_f^{(A)}, Z_f^{(A)}\}$ and $\{X_f^{(B)}, Z_f^{(B)}\}$.
The 3-fermion QCA $\alpha_\text{3F}$ is defined by its action on these Pauli operators.
The action of $\alpha_\text{3F}$ on the Pauli $X$ operators is given by \cite{Shirley_2022, MengSun2026}
\begin{equation}
\begin{aligned}
\alpha_\text{3F}(X_f^{(A)}) &= X_f^{(A)} \prod_{e^{\prime}} \left. \widetilde{G}_{e^{\prime}}^{(B)} \right.^{\int \bm{f} \smile \bm{e}^{\prime}}
= X_f^{(A)} \widetilde{G}_{t^{\frac{1}{2}}(f)}^{(B)}, \\
\alpha_\text{3F}(X_f^{(B)}) &= X_f^{(B)} \prod_{e^{\prime}} \left. \widetilde{G}_{e^{\prime}}^{(A)} \right.^{\int \bm{e}^{\prime} \smile \bm{f}}
= X_f^{(B)} \widetilde{G}_{t^{-\frac{1}{2}}(f)}^{(A)},
\label{eq: alpha 3F X}
\end{aligned}
\end{equation}
where $\widetilde{G}_e^{(A)}$ and $\widetilde{G}_e^{(B)}$ are defined by
\begin{equation}
\widetilde{G}_e^{(\mu)} \coloneq Z_{\delta \bm{e}} \prod_{f^{\prime}} \left. X_{f^{\prime}}^{(\mu)} \right.^{\int \delta \bm{e} \smile_1 \bm{f}^{\prime}}, \quad
\mu = A, B.
\label{eq: G 3F}
\end{equation}
On the other hand, the action of $\alpha_\text{3F}$ on the Pauli $Z$ operators is given by \cite{Shirley_2022, MengSun2026}
\begin{equation}
\alpha_\text{3F}(Z_f^{(\mu)}) = \tilde{u}_f^{(\mu)} \prod_{f^{\prime}} \alpha_\text{3F}(X_{f^{\prime}}^{(\mu)})^{\int \bm{f}^{\prime} \smile_1 \bm{f}}
\end{equation}
for $\mu = A, B$.
Here, $\tilde{u}_f^{(\mu)}$ is the fermion hopping operator defined as in \eqref{eq: fermion hopping}, i.e.,
\begin{equation}
\tilde{u}_f^{(\mu)} \coloneq Z_f^{(\mu)} \prod_{f^{\prime}} \left. X_{f^{\prime}}^{(\mu)} \right.^{\int \bm{f}^{\prime} \smile_1 \bm{f}}.
\end{equation}

To write the action of $\alpha_\text{3F}$ more concisely, we now introduce the fermionic flux operator $\tilde{u}_{\delta \bm{e}}^{(\mu)}$, which is defined by the product of the hopping operators around edge $e$:
\begin{equation}
\tilde{u}_{\delta \bm{e}}^{(\mu)} \coloneq \prod_{f^{\prime}} \left. \tilde{u}_{f^{\prime}}^{(\mu)} \right.^{\delta \bm{e}(f^{\prime})} = Z_{\delta \bm{e}}^{(\mu)} \prod_{f^{\prime}} \left. X_{f^{\prime}}^{(\mu)} \right.^{\int \bm{f}^{\prime} \smile_1 \delta \bm{e}}.
\label{eq: fermion loop 3F}
\end{equation}
Here, we recall our convention that all Pauli $X$ operators act before Pauli $Z$ operators.
By comparing \eqref{eq: fermion loop 3F} with \eqref{eq: G 3F}, we find that $\tilde{u}_{\delta \bm{e}}^{(\mu)}$ is related to $\widetilde{G}_e^{(\mu)}$ as
\begin{equation}
\begin{aligned}
\tilde{u}_{\delta \bm{e}}^{(\mu)} &= \widetilde{G}_e^{(\mu)} \prod_{f^{\prime}} \left. X_{f^{\prime}}^{(\mu)} \right.^{\int \bm{f}^{\prime} \smile_1 \delta \bm{e} + \delta \bm{e} \smile_1 \bm{f}^{\prime}} \\
&= \widetilde{G}_e^{(\mu)} \prod_{f^{\prime}} \left. X_{f^{\prime}}^{(\mu)} \right.^{\int \delta \bm{f}^{\prime} \smile_1 \bm{e} + \bm{e} \smile_1 \delta \bm{f}^{\prime}} \\
&= \widetilde{G}_e^{(\mu)} \prod_{c^{\prime}} \left. X_{\partial c^{\prime}}^{(\mu)} \right.^{\int \bm{c}^{\prime} \smile_1 \bm{e} + \bm{e} \smile_1 \bm{c}^{\prime}}.
\end{aligned}
\label{eq: uG}
\end{equation}
Namely, the fermionic flux operator $\tilde{u}_{\delta \bm{e}}^{(\mu)}$ is equal to $\widetilde{G}_e^{(\mu)}$ modulo $X_{\partial c}^{(\mu)}$.

Using the relation~\eqref{eq: uG}, we can rewrite the action of $\alpha_\text{3F}$ on the Pauli $X$ operators as follows:
\begin{equation}
\begin{aligned}
\alpha_\text{3F}(X_f^{(A)}) &= X_f^{(A)} \tilde{u}_{\delta t^{\frac{1}{2}}(\bm{f})}^{(B)} \left. \prod_{c^{\prime}} X_{\partial c^{\prime}}^{(B)} \right.^{\int \bm{c}^{\prime} \cup_1 t^{\frac{1}{2}}(\bm{f}) + t^{\frac{1}{2}}(\bm{f}) \cup_1 \bm{c}^{\prime}}, \\
\alpha_\text{3F}(X_f^{(B)}) &= X_f^{(B)} \tilde{u}_{\delta t^{-\frac{1}{2}}(\bm{f})}^{(A)} \left. \prod_{c^{\prime}} X_{\partial c^{\prime}}^{(A)} \right.^{\int \bm{c}^{\prime} \cup_1 t^{-\frac{1}{2}}(\bm{f}) + t^{-\frac{1}{2}}(\bm{f}) \cup_1 \bm{c}^{\prime}}.
\end{aligned}
\label{eq: alpha 3F X2}
\end{equation}
This action can be illustrated as shown in Figure~\ref{fig: alpha 3F XA} and Figure~\ref{fig: alpha 3F XB}.
\begin{figure*}
\centering
$\alpha_\text{3F}(X_f^{(A)}) =$
\adjincludegraphics[valign=c]{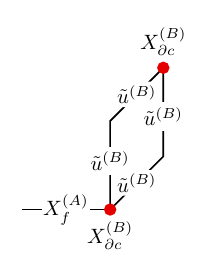}, 
\adjincludegraphics[valign=c]{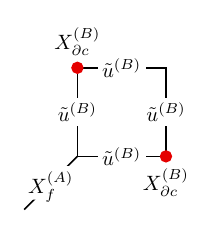},
\adjincludegraphics[valign=c]{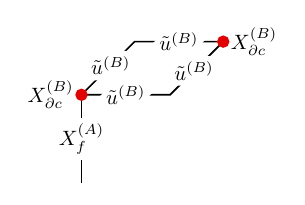}
\caption{The action of the 3-fermion QCA $\alpha_\text{3F}$ on $X_f^{(A)}$ illustrated on the dual lattice.}
\label{fig: alpha 3F XA}
\end{figure*}
\begin{figure*}
\centering
$\alpha_\text{3F}(X_f^{(B)}) =$
\adjincludegraphics[valign=c]{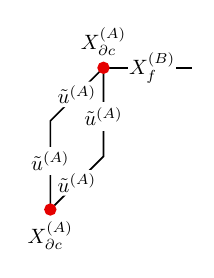}, 
\adjincludegraphics[valign=c]{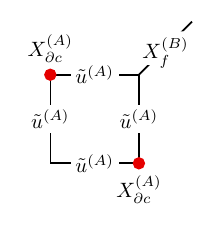}, 
\adjincludegraphics[valign=c]{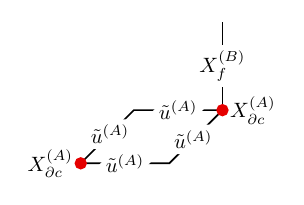}
\caption{The action of the 3-fermion QCA $\alpha_\text{3F}$ on $X_f^{(B)}$ illustrated on the dual lattice.}
\label{fig: alpha 3F XB}
\end{figure*}
On the other hand, $\alpha_\text{3F}$ leaves the fermion hopping operarors invariant \cite{MengSun2026}, i.e.,
\begin{equation}
\alpha_\text{3F}(\tilde{u}_{\delta \bm{e}}^{(\mu)}) = \tilde{u}_{\delta \bm{e}}^{(\mu)}.
\label{eq: alpha 3F u}
\end{equation}
We will use \eqref{eq: alpha 3F X2} and \eqref{eq: alpha 3F u} to compute the action of the 3-fermion KWW operator defined in the next subsection.

\subsection{3-fermion Kramers-Wannier-Wegner operator from 3-fermion QCA}
\label{sec: 3F KWW revisited}
In this subsection, we define the 3-fermion KWW operator using the 3-fermion QCA and compute its action on $\mathbb{Z}_2$ 1-form symmetric local operators.
As we will see, the 3-fermion KWW operator defined below agrees with the one defined in Section~\ref{sec: 3-fermion KWW}.
As such, the 3-fermion KWW operator defined below will be denoted by $\mathsf{D}_\text{3F}$, following the notation in Section~\ref{sec: 3-fermion KWW}.

\subsubsection{Definition of the 3-fermion KWW operator}
We first give a definition of $\mathsf{D}_\text{3F}$ as an operator on a tensor product Hilbert space.
To this end, we consider a 3d cubic lattice with a single qubit on each face.
We call this qubit an $A$-qubit.
The state space on the lattice is given by
\begin{equation}
\mathcal{H}_F^{(A)} = \bigotimes_{f \in F} \mathbb{C}^2.
\end{equation}
The 3-fermion KWW operator $\mathsf{D}_\text{3F}$ acting on this state space is defined in a similar way to the ordinary KWW operator~\eqref{eq: D unitary projection}.
More specifically, $\mathsf{D}_\text{3F}$ is defined by replacing the cluster entangler $\sf{U}_{\mathrm{cluster}}$ in \eqref{eq: D unitary projection} with the 3-fermion QCA.
In what follows, we will describe this definition in more detail.

To define $\mathsf{D}_\text{3F}$, we first add an ancillary qubit on each edge.
This qubit will be called a $B$-qubit.
The $B$-qubit on each edge is initially set to the $\ket{+}$ state as in the case of the KWW operator~\eqref{eq: D unitary projection}.
We then apply the 3-fermion QCA to the entire system consisting of the $A$- and $B$-qubits.
More precisely, we apply the unitary
\begin{equation}
\sf{U}_\text{3F}^{\prime} \coloneq t^{-\frac{1}{2}}_B \sf{U}_\text{3F} t^{\frac{1}{2}}_B,
\end{equation}
where $t_B^{\frac{1}{2}}$ is the half translation of the $B$-qubits, and $\sf{U}_\text{3F}$ is a unitary operator whose conjugation action is the 3-fermion QCA $\alpha_\text{3F}$ defined in Appendix~\ref{sec: 3-fermion QCA 2-qubit version}.
We can think of $\sf{U}_\text{3F}^{\prime}$ as a 3-fermion QCA on a cubic lattice with a single qubit on each face and a single qubit on each edge.
We then project the $A$-qubits on the faces to the $\ket{+}$ states and finally apply the half-translation to the remaining $B$-qubits on the edges.
This procedure leads to the following definition of the 3-fermion KWW operator:
\begin{equation}
\mathsf{D}_\text{3F} \ket{\psi} = \mathcal{N} t_B^{\frac{1}{2}} \bra{\{+\}_F} \sf{U}_\text{3F}^{\prime} \left( \ket{\psi} \otimes \ket{\{+\}_E} \right).
\label{eq: D3F op0}
\end{equation}
Here, $\ket{\psi} \in \mathcal{H}_F^{(A)}$ is an arbitrary state of the $A$-qubits, and $\mathcal{N}$ is a normalization factor, which is not relevant to the following discussion.
As we can see, the above definition is obtained by replacing the cluster entangler $\sf{U}_{\mathrm{cluster}}$ in~\eqref{eq: D unitary projection} with the 3-fermion QCA $\sf{U}_\text{3F}^{\prime}$.

For later convenience, we now rewrite \eqref{eq: D3F op0}.
Specifically, we first rewrite the right-hand side of \eqref{eq: D3F op0} as
\begin{equation}
\mathsf{D}_\text{3F}\ket{\psi} = \mathcal{N} \bra{\{+\}_E} t^{\frac{1}{2}} \sf{U}_\text{3F}^{\prime} \left( \ket{\psi} \otimes \ket{\{+\}_E} \right),
\label{eq: D3F op}
\end{equation}
where $t^{\frac{1}{2}}$ is the half translation acting on both the $A$-qubits and $B$-qubits.
The unitary $t^{\frac{1}{2}} \sf{U}_\text{3F}^{\prime}$ in \eqref{eq: D3F op} can also be written as
\begin{equation}
t^{\frac{1}{2}} \sf{U}_\text{3F}^{\prime} = t_A^{\frac{1}{2}} \sf{U}_\text{3F} t_B^{\frac{1}{2}},
\end{equation}
where $t_A^{\frac{1}{2}}$ is the half translation of the $A$-qubits.
By plugging this into \eqref{eq: D3F op}, we find
\begin{equation}
\begin{aligned}
\mathsf{D}_\text{3F} \ket{\psi} &= \mathcal{N} \bra{\{+\}_E} t_A^{\frac{1}{2}} \sf{U}_\text{3F} t_B^{\frac{1}{2}} \left( \ket{\psi} \otimes \ket{\{+\}_E} \right) \\
&= \mathcal{N} \bra{\{+\}_F^{(A)}} \sf{U}_\text{3F} \left( \ket{\psi} \otimes \ket{\{+\}_F^{(B)}} \right).
\end{aligned}
\label{eq: D3F op2}
\end{equation}
Here, $\ket{\{+\}_F^{(A)}}$ and $\ket{\{+\}_F^{(B)}}$ are the trivial product states of the $A$-qubits and $B$-qubits, which are both on the faces.
The above expression makes it manifest that $\mathsf{D}_\text{3F}$ maps an $A$-qubit state into a $B$-qubit state.
In what follows, we will use the following shorthand notation of \eqref{eq: D3F op2}:
\begin{equation}
\mathsf{D}_\text{3F} = \mathcal{N} \bra{\{+\}_F^{(A)}} \sf{U}_\text{3F} \ket{\{+\}_F^{(B)}}.
\label{eq: D3F op3}
\end{equation}

We note that the 3-fermion KWW operator~\eqref{eq: D3F op3} is non-invertible because it acts as zero on any states charged under the $\mathbb{Z}_2$ 1-form symmetry.
Indeed, since the 3-fermion QCA $\sf{U}_\text{3F}$ commutes with the symmetry operator $\eta(\Sigma) = \prod_{f \in \Sigma} X_f^{(A)}$ for any closed surface $\Sigma$,\footnote{This immediately follows from \eqref{eq: alpha 3F X2}.} we have
\begin{equation}
\begin{aligned}
\mathsf{D}_\text{3F} \eta(\Sigma) &= \mathcal{N} \bra{\{+\}_F^{(A)}} \sf{U}_\text{3F} \eta(\Sigma) \ket{\{+\}_F^{(B)}} \\
&= \mathcal{N} \bra{\{+\}_F^{(A)}} \eta(\Sigma) \sf{U}_\text{3F} \ket{\{+\}_F^{(B)}} \\
&= \mathsf{D}_\text{3F},
\end{aligned}
\end{equation}
where the last equality follows from $\bra{\{+\}_F^{(A)}} \eta(\Sigma) = \bra{\{+\}_F^{(A)}}$.
The above equation implies that $\mathsf{D}_\text{3F} \ket{\psi} = 0$ if $\eta(\Sigma) \ket{\psi} = -\ket{\psi}$ for some $\Sigma$, i.e., if $\ket{\psi}$ is charged under the $\mathbb{Z}_2$ 1-form symmetry.

\vspace*{\baselineskip}
\noindent{\bf Adaptive circuit representation.}
As in the case of the KWW operator, the projection to the $\ket{+}$ state can be done adaptively by using measurements and feedforward, as long as the initial state $\ket{\psi}$ is symmetric under the $\mathbb{Z}_2$ 1-form symmetry.
However, this does not mean that $\mathsf{D}_\text{3F}$ admits an adaptive circuit representation with a single round of measurement, because $\sf{U}_\text{3F}$ is not a finite-depth circuit \cite{HaahFidkowskiHastings2023, Shirley_2022}.
Nevertheless, $\mathsf{D}_\text{3F}$ does admit an adaptive circuit representation with multiple measurement rounds.
Indeed, as we will see later, $\mathsf{D}_\text{3F}$ is equal to $\mathsf{S}\mathsf{T}^2\mathsf{S}\mathsf{T}^2\mathsf{S}^{\dagger}$ as an operator, and the latter admits an adaptive circuit representation with three measurement rounds because $\mathsf{S}$ and $\mathsf{S}^{\dagger}$ can be implemented by adaptive circuits with a single round of measurement.
Therefore, $\mathsf{D}_\text{3F}$ admits an adaptive circuit representation with three measurement rounds.

%\vspace*{\baselineskip}
%\noindent{\bf Action on local operators.}
\subsubsection{Action on local operators}
Now, we compute the action of $\mathsf{D}_\text{3F}$ on $\mathbb{Z}_2$ 1-form symmetric local operators $X_f$ and $\tilde{u}_{\delta \bm{e}}$.
Based on the definition~\eqref{eq: D3F op3}, one can compute the action of $\mathsf{D}_\text{3F}$ on $X_f$ as
\begin{equation}
\begin{aligned}
\mathsf{D}_\text{3F} X_f^{(A)} &= \mathcal{N} \bra{\{+\}_F^{(A)}} \sf{U}_\text{3F} X_f^{(A)} \ket{\{+\}_F^{(B)}} \\
&= \mathcal{N} \bra{\{+\}_F^{(A)}} \alpha_\text{3F}(X_f^{(A)}) \sf{U}_\text{3F} \ket{\{+\}_F^{(B)}} \\
&= \tilde{u}_{\delta t^{\frac{1}{2}}(\bm{f})}^{(B)} \mathsf{D}_\text{3F} \,.
\end{aligned}
\label{eq: D3F X}
\end{equation}
Here, in the last equality, we used \eqref{eq: alpha 3F X2} and $\bra{\{+\}_F^{(A)}} X_{f^{\prime}}^{(A)} = \bra{\{+\}_F^{(A)}}$ for any face $f^{\prime} \in F$.
Similarly, one can also compute the action of $\mathsf{D}_\text{3F}$ on $\tilde{u}_{\delta \bm{e}}$ as
\begin{equation}
\begin{aligned}
\mathsf{D}_\text{3F} \tilde{u}_{\delta \bm{e}}^{(A)}
&= \mathcal{N} \bra{\{+\}_F^{(A)}} \sf{U}_\text{3F} \tilde{u}_{\delta \bm{e}}^{(A)} \ket{\{+\}_F^{(B)}} \\
&= \mathcal{N} \bra{\{+\}_F^{(A)}} \sf{U}_\text{3F} \tilde{u}_{\delta \bm{e}}^{(A)} X_{t^{\frac{1}{2}}(e)}^{(B)} \ket{\{+\}_F^{(B)}} \\
&= \mathcal{N} \bra{\{+\}_F^{(A)}} \alpha_\text{3F}(\tilde{u}_{\delta \bm{e}}^{(A)} X_{t^{\frac{1}{2}}(e)}^{(B)}) \sf{U}_\text{3F} \ket{\{+\}_F^{(B)}} \\
&= X_{t^{\frac{1}{2}}(e)}^{(B)} \mathsf{D}_\text{3F}.
\end{aligned}
\label{eq: D3F u}
\end{equation}
Here, in the last equality, we used
\begin{equation}
\alpha_\text{3F}(\tilde{u}_{\delta \bm{e}}^{(A)} X_{t^{\frac{1}{2}}(e)}^{(B)})
= X_{t^{\frac{1}{2}}(e)}^{(B)} \left. \prod_{c^{\prime}} X_{\partial c^{\prime}}^{(A)} \right.^{\int \bm{c}^{\prime} \smile_1 \bm{e} + \bm{e} \smile_1 \bm{c}^{\prime}},
\end{equation}
which follows from \eqref{eq: alpha 3F X2} and \eqref{eq: alpha 3F u}.
Equations~\eqref{eq: D3F X} and \eqref{eq: D3F u} show that $\mathsf{D}_\text{3F}$ acts on $\mathbb{Z}_2$ 1-form symmetric local operators as
\begin{equation}
X_f \xmapsto{\mathsf{D}_\text{3F}} \tilde{u}_{\delta t^{\frac{1}{2}}(\bm{f})}, \quad
\tilde{u}_{\delta \bm{e}} \xmapsto{\mathsf{D}_\text{3F}} X_{t^{\frac{1}{2}}(e)}.
\label{eq: D3F action}
\end{equation}
This agrees with the action of the 3-fermion KWW operator $\mathsf{S}\mathsf{T}^2\mathsf{S}\mathsf{T}^2\mathsf{S}^{\dagger}$ defined in Section~\ref{sec: 3-fermion KWW}.

\subsubsection{Operator equality}
%We have shown that
%\begin{equation}
%\mathsf{D}_\text{3F} = \mathsf{D}U\mathsf{D}U\mathsf{D}^{\dagger}
%\label{eq: operator equality}
%\end{equation}
%as an equality of QCAs on the $\mathbb{Z}_2$ 1-form symmetric algebra $\mathcal{A}_{\mathbb{Z}_2^{(1)}} / \mathfrak{I}$.
We have shown that $\mathsf{D}_\text{3F}$ and $\mathsf{S}\mathsf{T}^2\mathsf{S}\mathsf{T}^2\mathsf{S}^\dag$ act in the same way on the algebra of $\mathbb{Z}_2$ 1-form symmetric local operators.
In what follows, we will show that they also act in the same way on the entire tensor-product Hilbert space $\mathcal{H}_F$.
Namely, we will show that 
\begin{equation}
\mathsf{D}_\text{3F} \ket{\psi} = \mathsf{S}\mathsf{T}^2\mathsf{S}\mathsf{T}^2\mathsf{S}^{\dagger} \ket{\psi}
\label{eq: D3F psi}
\end{equation}
for any state $\ket{\psi} \in \mathcal{H}_F$.

To show \eqref{eq: D3F psi}, we first consider the case when $\ket{\psi}$ is charged under the $\mathbb{Z}_2$ 1-form symmetry.
In this case, both $\mathsf{D}_\text{3F}$ and $\mathsf{S}\mathsf{T}^2\mathsf{S}\mathsf{T}^2\mathsf{S}^{\dagger}$ act as zero on $\ket{\psi}$.
Thus, equation~\eqref{eq: D3F psi} holds for any $\ket{\psi}$ charged under the $\mathbb{Z}_2$ 1-form symmetry.

Next, we consider the case where $\ket{\psi}$ is symmetric under the $\mathbb{Z}_2$ 1-form symmetry.
In this case, if equation~\eqref{eq: D3F psi} holds for some symmetric state $\ket{\psi_0}$, then it also holds for a state $\ket{\psi}$ obtained by acting on $\ket{\psi_0}$ with an arbitrary number of local symmetric operators generated by $X_f$ and $Z_{\delta \bm{e}}$.
This is because $\mathsf{D}_\text{3F}$ and $\mathsf{S}\mathsf{T}^2\mathsf{S}\mathsf{T}^2\mathsf{S}^{\dagger}$ act in the same way on any local symmetric operators.
Furthermore, any symmetric state $\ket{\psi}$ can be obtained from $\ket{\psi_0}$ in that way.\footnote{We can show this as follows. First, we notice that any $\mathbb{Z}_2$ 1-form symmetric state can be written as a superposition of closed loop configurations of bosonic flux $Z_{\delta \bm{e}}$. For any such (non-zero) state $\ket{\psi_0}$, we can apply the product of local symmetric operators of the form $\frac{1}{2}(1 \pm Z_{\delta \bm{e}})$ so that $\ket{\psi_0}$ is projected into a state with a fixed configuration of bosonic flux loops. We can then create any loop configuration from that state by applying Pauli $X$ operators.}
Therefore, equation~\eqref{eq: D3F psi} holds for any state $\ket{\psi}$ if there exists a non-zero symmetric state $\ket{\psi_0}$ that satisfies
\begin{equation}
\mathsf{D}_\text{3F}\ket{\psi_0} = \mathsf{S}\mathsf{T}^2\mathsf{S}\mathsf{T}^2\mathsf{S}^{\dagger} \ket{\psi_0}.
\label{eq: D3F psi0}
\end{equation}

Now, to find a non-zero symmetric state satisfying~\eqref{eq: D3F psi0}, we choose $\ket{\psi_0}$ to be a state stabilized by the fermionic flux operators $\tilde{u}_{\delta \bm{e}}$ for all edges $e$:
\begin{equation}
\ket{\psi_0} \coloneq \left( \prod_{e \in E} \frac{1+\tilde{u}_{\delta \bm{e}}}{2} \right) \ket{\{+\}_F}.
\label{eq: psi0}
\end{equation}
We note that this state is $\mathbb{Z}_2$ 1-form symmetric because both $\tilde{u}_{\delta \bm{e}}$ and $\ket{\{+\}_F}$ are symmetric.
Since $\mathsf{D}_\text{3F}$ maps $\tilde{u}_{\delta e}$ to $X_f$ as shown in \eqref{eq: D3F action}, we have
\begin{equation}
X_f \mathsf{D}_\text{3F} \ket{\psi_0} = \mathsf{D}_\text{3F} \tilde{u}_{\delta t^{-\frac{1}{2}}(\bm{f})} \ket{\psi_0} = \mathsf{D}_\text{3F} \ket{\psi_0}.
\end{equation}
This equation shows that $\mathsf{D}_\text{3F} \ket{\psi_0}$ is stabilized by $X_f$ for all faces $f$.
Namely, we have
\begin{equation}
\mathsf{D}_\text{3F} \ket{\psi_0} \propto \ket{\{+\}_F}.
\end{equation}
For the same reason, we also have
\begin{equation}
\mathsf{S}\mathsf{T}^2\mathsf{S}\mathsf{T}^2\mathsf{S}^{\dagger} \ket{\psi_0} \propto \ket{\{+\}_F}.
\end{equation}
Thus, we find that $\ket{\psi_0}$ in \eqref{eq: psi0} satisfies \eqref{eq: D3F psi0}, provided that $\mathsf{D}_\text{3F}$ is normalized appropriately.
This in turn implies that \eqref{eq: D3F psi} holds for any state $\ket{\psi} \in \mathcal{H}_F$.

\section{Properties of the quotient algebra $\mc{A}_{\Z_2^{(1)}} / \mathfrak{I}$}\label{app:quotient_algebra}

In Section \ref{sec:qca_on_quotient_algebra} we gave a quick overview of the mathematical framework needed to define QCA on the quotient algebra $\mc{A}_{\Z_2^{(1)}} / \mathfrak{I}$. In this appendix, we give the same results (without all the surrounding explanations and context) but with all the proofs.

We begin by recalling what the setup is. We have an infinite cubical lattice (in this appendix, we will work exclusively in the thermodynamic limit), with qubits on the faces. The full algebra of operators is
\be
\mc{A} \coloneqq \langle X_f, Z_f \; \mid \; f \in F \rangle \,,
\ee
where $F$ is the set of all faces (2-cells).

That is, it is the (complex) algebra generated by the Pauli X and Z operators at all the faces. A generic element of $\mc{A}$ is a (finite, complex) linear combination of (finite) products of these Pauli operators over various faces. Without loss of generality, we may assume that within a product, all the X operators appear to the left of the Z operators. Since $X^2 = Z^2 = 1$, each appears at most once.

Similarly, the $\Z_2^{(1)}$-symmetric subalgebra is
\be
\mc{A}_{\Z_2^{(1)}} \coloneqq \left\langle X_f, Z_{\delta \bm{e}} = \prod_{f' \; \mid \; e \in \partial f'} Z_{f'} \; \mid \; e \in E, f \in F \right\rangle \,.
\ee

$F$ is again the set of all faces, and $E$ is the set of all edges (1-cells). Together with the Hermitian adjoint, $\mc{A}_{\Z_2^{(1)}}$ forms a $*$-algebra.

We define the ideal $\mathfrak{I} \subset \mc{A}_{\Z_2^{(1)}}$
\be
\mathfrak{I} \coloneqq \curpar{ \sum_{i=1}^n (1 - X_{\partial c_i}) a_i \; \mid \; n \in \mathbb{N}, c_i \in C, a_i \in \mc{A}_{\Z_2^{(1)}}} \,,
\ee
where $C$ is the set of all cubes (3-cells).

\begin{lemma}
 $\mathfrak{I} \subset \mc{A}_{\Z_2^{(1)}}$ as defined above is a two-sided $*$-ideal.
\end{lemma}

\begin{proof}
 Clearly $\mathfrak{I} + \mathfrak{I} = \mathfrak{I}$ and $-\mathfrak{I} = \mathfrak{I}$: it is closed under addition and additive inverse. It is trivially closed under right multiplication by any $a \in \mc{A}_{\Z_2^{(1)}}$. This makes it a right ideal.

 To prove that it is a left ideal, and that $\mathfrak{I}^\dag = \mathfrak{I}$, note that every element of $\mc{A}_{\Z_2^{(1)}}$ commutes with $X_{\partial c}$, for any cube $c \in C$. This is because all the generators of $\mc{A}_{\Z_2^{(1)}}$ commute with $X_{\partial c}$. Thus:
 \be
 a \sum_{i=1}^n (1 - X_{\partial c_i}) a_i = \sum_{i=1}^n (1 - X_{\partial c_i}) (a \cdot a_i) \in \mathfrak{I} \,,
 \ee
 and using the fact that $X$ is self-adjoint
 \begin{align}
  \paren{\sum_{i=1}^n (1 - X_{\partial c_i}) a_i}^\dag &= \sum_{i=1}^n a_i^\dag (1 - X_{\partial c_i}) \nonumber \\
  &= \sum_{i=1}^n (1 - X_{\partial c_i}) a_i^\dag \in \mathfrak{I} \,.
 \end{align}
\end{proof}

\begin{corollary}
 $\mc{A}_{\Z_2^{(1)}} / \mathfrak{I}$ is a $*$-algebra.
\end{corollary}

\subsection{The twisted group algebra $\mathbb{C}^\text{lk}[C_2 \oplus C_1]$}

Let us fix, as before, the infinite cubical lattice in 3d. As a CW-complex, it has an associated chain complex, and we denote
\be
C_k \coloneqq C_k(\text{cubical lattice}, \Z_2) \coloneqq \bigoplus_{k-\text{cell } \sigma} \Z_2 \,,
\ee
the space of k-chains with $\Z_2$-coefficients.

\begin{lemma}
 Let $V \coloneqq \mathbb{C}[C_2 \oplus C_1]$ be the complex vector space of (finite) linear combinations of elements in the abelian group $C_2 \oplus C_1$. Define the map $\mu : V \times V \rightarrow V$ by
 \be\label{eq:Clk_multiplication}
 \mu\paren{ (c_2, c_1), (c_2', c_1') } \coloneqq (-1)^{\delta \bm{c_1} (c_2')} (c_2 + c_2', c_1 + c_1') \,,
 \ee
 on the basis elements, and extended by bilinearity.

 Then $(V, \mu)$ is an (associative, complex, unital) algebra with multiplication $\mu$.
\end{lemma}

\begin{lemma}
 Let $(V, \mu)$ be the algebra as before. Define $(-)^{*} : V \rightarrow V$ by
 \be
 (c_2, c_1)^{*} \coloneqq (-1)^{\delta \bm{c_1} (c_2)} (c_2, c_1)
 \ee
 on the basis elements, and extended by \emph{conjugate}-linearity to the rest of $V$.

 Then $(V, \mu, {}^{*})$ is a $*$-algebra.
\end{lemma}

\begin{proof}
 Both lemmas are simply straightforward verification of the properties. One must check the following equations, where for simplicity we just write $\mu$ as multiplication
 \begin{align}
  \big[(c_2, c_1) (c_2', c_1')\big] (c_2'', c_1'') &= (c_2, c_1) \big[(c_2', c_1') (c_2'', c_1'')\big] \\
  \paren{(c_2, c_1)^{*}}^{*} &= (c_2, c_1) \\
  \paren{(c_2, c_1) (c_2', c_1')}^{*} &= (c_2', c_1')^{*} (c_2, c_1)^{*}
 \end{align}
\end{proof}

\begin{definition}
 We define $\mathbb{C}^\text{lk}[C_2 \oplus C_1]$ as the $*$-algebra described above.
\end{definition}

\subsection{Isomorphisms of algebras}

\begin{lemma}
 The map $\phi : \mathbb{C}^\text{lk}[C_2 \oplus C_1] \rightarrow \mc{A}_{\Z_2^{(1)}}$ defined by
 \be
 \phi(c_2, c_1) \coloneqq \prod_{f \in c_2} X_f \prod_{e \in c_1} Z_{\delta \bm{e}} \,,
 \ee
 and extended by linearity, is a $*$-homomorphism.
\end{lemma}

\begin{proof}
 Note that
 \be
 \prod_{e \in c_1} Z_{\delta \bm{e}} = \prod_{f \in \delta \bm{c_1}} Z_f \,.
 \ee

 This means that:
 \be
 \phi(c_2, c_1)^\dag = \prod_{f \in \delta \bm{c_1}} Z_f \prod_{f' \in c_2} X_{f'} \,,
 \ee
 and to switch the order of the operators, we must pick up as many $-1$ phases as there are overlaps between $\delta \bm{c_1}$ and $c_2$. $\delta \bm{c_1}(c_2)$ only counts the overlaps mod 2, but that is all we need, so:
 \begin{align}
  \phi(c_2, c_1)^\dag &= (-1)^{\delta \bm{c_1}(c_2)} \prod_{f' \in c_2} X_{f'} \prod_{f \in \delta \bm{c_1}} Z_f \nonumber \\
  &= \phi \paren{ (-1)^{\delta \bm{c_1}(c_2)} (c_2, c_1) } = \phi\paren{ (c_2, c_1)^{*} } \,.
 \end{align}

 The proof that $\phi$ distributes over the product is very similar. One must swap the order of two products of Pauli Z and Pauli X, and the phase in Eq.~\eqref{eq:Clk_multiplication} is precisely the one that is need to match that.
\end{proof}

\begin{proposition}\label{prop:isomorphisms1}
 We have isomorphisms
 \be
 \mc{A}_{\Z_2^{(1)}} \cong \faktor{\mathbb{C}^\text{lk}[C_2 \oplus C_1]}{\sim_\delta} \cong \mathbb{C}^\text{lk}[C_2 \oplus B^2_\text{fin.}] \,,
 \ee
 where $\sim_\delta$ is the equivalence
 \be
 (c_2, c_1) \sim_\delta (c_2', c_1') \iff c_2 = c_2' \text{ and } \delta \bm{c_1} = \delta \bm{c_1'} \,.
 \ee

 It can be extended to a vector subspace (and a two-sided $*$-ideal) by being identified with the subspace
 \be
 \mathbb{C}\big\langle (c_2, c_1) - (c_2, c_1') \; \mid \; c_2 \in C_2, c_1, c_1' \in C_1,  \delta \bm{c_1} = \delta \bm{c_1'} \big\rangle \,.
 \ee
\end{proposition}

\begin{proof}
 The isomorphism on the left is direct application of the first isomorphism theorem to $\phi$. The fact that the image is all of $\mc{A}_{\Z_2^{(1)}}$ is clear, just from the definition of that algebra.

 For the kernel, observe that the subspace defined previously is contained in the kernel. Indeed
 \be
 \phi((c_2, c_1) - (c_2, c_1')) = \prod_{c_2} X \prod_{\delta \bm{c_1}} Z - \prod_{c_2} X \prod_{\delta \bm{c_1'}} Z = 0 \,,
 \ee
 follows immediately from $\delta \bm{c_1} = \delta \bm{c_1'}$.

 Now, suppose
 \be
 \phi\paren{ \sum_{c_2, c_1} A_{c_2, c_1} (c_2, c_1) } = \sum_{c_2, c_1} A_{c_2, c_1} \prod_{c_2} X \prod_{\delta \bm{c_1}} Z = 0 \,,
 \ee
 where $A_{c_2, c_1} \in \mathbb{C}$, and only finitely many of them are non-zero. That last equality can be checked in $\mc{A}$, which we know has basis
 \be
 \curpar{ \prod_{c_2} X \prod_{c_2'} Z, \text{ for } c_2, c_2' \in C_2} \,.
 \ee

 We must have that, for any fixed $c_2, \delta \bm{c_1}$:
 \be
 \sum_{c_1' \; \mid \; \delta \bm{c_1'} = \delta \bm{c_1}} A_{c_2, c_1'} = 0 \,.
 \ee

 Only finitely many $A_{c_2, c_1'}$ are non-zero. Put them in some order, and label them by $\gamma_i = (c_2, c_{1, i})$ for $i = 1, ..., n$. Then $A_{\gamma_n} = - A_{\gamma_1} - \dots - A_{\gamma_{n-1}}$, and
 \be
 \sum_{i = 1}^n A_{\gamma_i} \gamma_i = \sum_{i = 1}^{n-1} A_{\gamma_i} \paren{\gamma_i - \gamma_n}
 \ee

 This same rearrangement can be done for every $c_2$, and every $\delta \bm{c_1}$. This shows that
 \be
 \sum_{c_2, c_1} A_{c_2, c_1} (c_2, c_1) \,,
 \ee
 is in the span of elements of the form $\gamma_i - \gamma_n = (c_2, c_{1, i}) - (c_2, c_{1, n})$, where $\delta \bm{c_{1, i}} = \delta\bm{c_{1, n}}$, which is what we wanted to show.

 The isomorphism on the right is just a direct consequence of how we interpret the isomorphism on the left. The map which gives the isomorphism is just
 \be
 (c_2, c_1) \longmapsto (c_2, \delta \bm{c_1}) \,.
 \ee
 
 Instead of labeling the basis elements by $(c_2, c_1)$, and then collapsing all the ones which have the same coboundary, we can just directly label the basis elements by $(c_2, \bm{b^2}) \in C_2 \oplus B^2_\text{fin.}$. The only caveat is that $\bm{b^2}$ is not quite an arbitrary coboundary, because it has to be finitely supported.
\end{proof}

\begin{proposition}\label{prop:isomorphisms2}
 We have isomorphisms
 \be
 \faktor{\mc{A}_{\Z_2^{(1)}}}{\mathfrak{I}} \cong \faktor{\mathbb{C}^\text{lk}[C_2 \oplus C_1]}{\sim_{\delta, \partial}} \cong \mathbb{C}^\text{lk}[B_1 \oplus B^2_\text{fin.}] \,,
 \ee
 where $\sim_{\delta, \partial}$ is the equivalence
 \be
 (c_2, c_1) \sim_\delta (c_2', c_1') \iff \partial c_2 = \partial c_2' \text{ and } \delta \bm{c_1} = \delta \bm{c_1'} \,.
 \ee

 Like before, it can be reinterpreted as a two-sided $*$-ideal, and that is how we interpret the quotient.
\end{proposition}

\begin{proof}
 This is once again an application of the first isomorphism theorem. Let $\pi$ denote the projection to the quotient
 \be
 \pi : \mc{A}_{\Z_2^{(1)}} \longrightarrow \faktor{\mc{A}_{\Z_2^{(1)}}}{\mathfrak{I}} \,,
 \ee
 which is surjective by construction. Consider $\pi \circ \phi$. It is surjective, because it is the composition of two surjective maps. Now we must show that the kernel is the one we expect. Note that $\ker (\pi \circ \phi) = \phi^{-1}(\mathfrak{I})$, which is the ideal generated by all the preimages under $\phi$ of $1 - X_{\partial c}$.

 Suppose $x \in \phi^{-1}(\mathfrak{I})$. Then
 \be
 \phi(x) = \sum_{i=1}^n (1 - X_{\partial c^{(i)}}) a_i \,,
 \ee
 for some cubes $c^{(i)} \in C$, and some $a_i \in \mc{A}_{\Z_2^{(1)}}$. Then:
 \be
 \phi(x) = \phi \paren{ \sum_{i} ((0,0) - (\partial c^{(i)}, 0)) \phi^{-1}(a_i) } \,,
 \ee
 where $\phi^{-1}(a_i)$ are some arbitrary choices of preimages of $a_i$. Therefore, the two sides differ by an element of the kernel of $\phi$:
 \be
 x = \sum_{i} ((0,0) - (\partial c^{(i)}, 0)) \phi^{-1}(a_i) + y \,,
 \ee
 with $\phi(y) = 0$. $(0, 0) - (\partial c^{(i)}, 0)$ is just one choice of preimage of $1 - X_{\partial c^{(i)}}$. Any ambiguity in this choice gets absorbed into $y$.

 The choice of $\phi^{-1}(a_i)$ will be of the form
 \be
 \phi^{-1}(a_i) = \sum_{c_2, c_1} A^{(i)}_{c_2, c_1} (c_2, c_1) \,.
 \ee

 For each term in the sum, we have
 \be
 \paren{(0, 0) - (\partial c^{(i)}, 0)} (c_2, c_1) = (c_2, c_1) - (c_2 + \partial c^{(i)}, c_1) \,.
 \ee

 Therefore, $x$ is in the span of
 \be
 (c_2, c_1) - (c_2', c_1') \text{ such that } \partial c_2 = \partial c_2', \delta \bm{c_1} = \delta \bm{c_1'} \,.
 \ee

 Indeed, any element of $\ker \phi$, such as $y$ above, is a linear combination of terms which also satisfy $c_2 = c_2'$ (so in particular $\partial c_2 = \partial c_2'$ is trivially true). Meanwhile, the terms in $\phi^{-1}(a_i)$ have $c_1 = c_1'$, and $\partial (c_2 + \partial c^{(i)}) = \partial c_2$.

 Conversely, suppose $\partial c_2 = \partial c_2'$ and $\delta \bm{c_1} = \delta \bm{c_1'}$. Then
 \begin{align}
  \phi \paren{(c_2, c_1) - (c_2', c_1')} &= \prod_{c_2} X \prod_{\delta \bm{c_1}} Z - \prod_{c_2'} X \prod_{\delta \bm{c_1'}} Z \nonumber \\
  &= \paren{1 - \prod_{c_2 + c_2'} X } \prod_{c_2} X \prod_{\delta \bm{c_1}} Z \,.
 \end{align}

 Clearly, $\partial (c_2 + c_2') = 0$. In the case of an infinite cubical lattice, the homology is trivial, so every cycle is a boundary. That is, there is some $c_3 \in C_3$ such that $\partial c_3 = c_2 + c_2'$. Then:
 \be
 1 - \prod_{c_2 + c_2'} X = 1 - \prod_{f \in \partial c_3} X_f = 1 - \prod_{c \in c_3} X_{\partial c} \,.
 \ee

 Now we wish to show that the above is in $\mathfrak{I}$, because then we will be able to conclude that $\phi \paren{(c_2, c_1) - (c_2', c_1')} \in \mathfrak{I}$.

 This is done by induction on the number of cubes in $c_3$ (which is always finite):
 \be
 1 - \prod_{i=1}^n X_{\partial c^{(i)}} = \paren{1 - X_{\partial c^{(1)}}} + \paren{1 - \prod_{i=2}^n X_{\partial c^{(i)}}} X_{\partial c^{(1)}} \,.
 \ee

 The first term is clearly in $\mathfrak{I}$, and the second one is in $\mathfrak{I}$ by induction.

 Finally, we can also conclude that the map
 \be
 (c_2, c_1) \longmapsto (\partial c_2, \delta \bm{c_1}) \,,
 \ee
 gives the second isomorphism
 \be
 \faktor{\mathbb{C}^\text{lk}[C_2 \oplus C_1]}{\sim_{\delta, \partial}} \cong \mathbb{C}^\text{lk}[B_1 \oplus B^2_\text{fin.}] \,.
 \ee
\end{proof}

\subsection{Locality}

Let us recall the definitions of support in $\mc{A}_{\Z_2^{(1)}}$ and $\mc{A}_{\Z_2^{(1)}} / \mathfrak{I}$:

\begin{definition}
 For an operator $a \in \mc{A}_{\Z_2^{(1)}}$, we define its support as
 \be
 \operatorname{Supp}(a) \coloneqq \curpar{e \; \mid \; \sqpar{Z_{\delta \bm{e}}, a} \neq 0} \coprod \curpar{f \; \mid \; \sqpar{X_f, a} \neq 0} \,,
 \ee
 which is a subset of $E \coprod F$.
\end{definition}

\begin{definition}
 For an equivalence class $a + \mathfrak{I} \in \mathcal{A}_{\Z_2^{(1)}} / \mathfrak{I}$, we define its support as
 \be
 \operatorname{Supp}(a + \mathfrak{I}) \coloneqq \bigcap_{I \in \mathfrak{I}} \operatorname{Supp}(a + I) \,.
 \ee
\end{definition}

\begin{lemma}
 Let $(c_2, c_1) \in C_2 \oplus C_1$. The support of $\phi(c_2, c_1) \in \mc{A}_{\Z_2^{(1)}}$ is
 \be
 \operatorname{Supp}(\phi(c_2, c_1)) = \partial c_2 \coprod \delta \bm{c_1} \,.
 \ee
\end{lemma}

\begin{proof}
 We must check where $\prod_{c_2} X \prod_{\delta \bm{c_1}} Z$ does not commute with $Z_{\delta \bm{e}}$ and $X_f$.

 For the check against $X_f$, what contributes is 
 \be
 \prod_{f' \in \delta \bm{c_1}} Z_{f'} \,.
 \ee

 Simply put, they do not commute if $\delta \bm{c_1}(f) = 1$. For $\Z_2$ coefficients, we identify $\delta \bm{c_1}$ with the set of faces on which it is non-zero.

 Commutation with $Z_{\delta \bm{e}}$ is a little more complicated. But note that this $Z_{\delta \bm{e}}$ is a product of 4 $Z$ operators around an edge. It will not commute with $\prod_{c_2} X$ if and only if $c_2$ has an odd number of faces around the edge $e$. In other words, if $\bm{e}(\partial c_2) = 1$, because $\partial c_2$ precisely counts, for each edge, whether the number of faces around it is even or odd. Finally, the set of edges for which $Z_{\delta \bm{e}}$ does not commute can literally be identified with $\partial c_2$.
\end{proof}

\begin{lemma}\label{lemma:quotient_support}
 Let $(c_2, c_1) \in C_2 \oplus C_1$. The support of $\phi(c_2, c_1) + \mathfrak{I} \in \mc{A}_{\Z_2^{(1)}} / \mathfrak{I}$ is
 \be
 \operatorname{Supp}(\phi(c_2, c_1) + \mathfrak{I}) = \partial c_2 \coprod \delta \bm{c_1} \,.
 \ee
\end{lemma}

\begin{proof}
 The inclusion $\subseteq$ is trivial by definition. Let us suppose it is a strict inclusion. That is, there is some $J \in \mathfrak{I}$, and some $\mu \in \partial c_2 \coprod \delta \bm{c_1}$ (we use $\mu$ to denote a face or an edge, without distinguishing which it is), such that $\mu \notin \operatorname{Supp}(\phi(c_2, c_1) + J)$.

 Let $Y_\mu$ be the operator associated to $\mu$: $X_f$ if $\mu = f$ is a face, $Z_{\delta \bm{e}}$ if $\mu = e$ is an edge.

 Then
 \be
 Y_\mu \paren{\phi(c_2, c_1) + J} = \paren{\phi(c_2, c_1) + J} Y_\mu \,.
 \ee

 But we also know that
 \be
 \phi(c_2, c_1) Y_\mu = - Y_\mu \phi(c_2, c_1) \,.
 \ee

 Therefore, we have
 \be
 Y_\mu \phi(c_2, c_1) + Y_\mu J = -Y_\mu \phi(c_2, c_1) + J Y_\mu \,,
 \ee
 and rearranging
 \be
 Y_\mu \phi(c_2, c_1) = \frac{J Y_\mu - Y_\mu J}{2} \in \mathfrak{I} \,.
 \ee

 But the operator on the left is invertible, so $\mathfrak{I} = \mc{A}_{\Z_2^{(1)}}$, which is a contradiction because
 \be
 0 \cong \faktor{\mc{A}_{\Z_2^{(1)}}}{\mathfrak{I}} \cong \mathbb{C}^\text{lk}[B_1 \oplus B^2_\text{fin.}] \,,
 \ee
 and the algebra on the right hand side clearly is not the 0 algebra.
\end{proof}

\begin{lemma}\label{lemma:support_of_sum_and_prod}
 For any $a, b \in \mc{A}_{\Z_2^{(1)}}$, the support satisfies
 \begin{align}
  \operatorname{Supp}(a + b + \mathfrak{I}) &\subseteq \operatorname{Supp}(a + \mathfrak{I}) \cup \operatorname{Supp}(b + \mathfrak{I}) \,, \\
  \operatorname{Supp}(ab + \mathfrak{I}) &\subseteq \operatorname{Supp}(a + \mathfrak{I}) \cup \operatorname{Supp}(b + \mathfrak{I}) \,.
 \end{align}
\end{lemma}

\begin{proof}
 Showing the complementary statements is easier. Suppose $\mu \notin \operatorname{Supp}(a + \mathfrak{I}) \cup \operatorname{Supp}(b + \mathfrak{I})$. Then, there are some $I, J \in \mathfrak{I}$ such that
 \begin{align}
  Y_\mu (a + I) &= (a + I) Y_\mu & Y_\mu (b + J) &= (b + J) Y_\mu \,.
 \end{align}

 Then
 \be
 Y_\mu (a + I + b + J) = (a + I + b + J) Y_\mu \,,
 \ee
 with $a + b + I + J \in a + b + \mathfrak{I}$. And similarly
 \be
 Y_\mu (a + I) (b + J) = (a + I) (b + J) Y_\mu \,,
 \ee
 with $(a + I)(b + J) \in ab + \mathfrak{I}$.
\end{proof}

\begin{lemma}\label{lemma:exact_support}
 For any $a + \mathfrak{I} \in \mc{A}_{\Z_2^{(1)}} / \mathfrak{I}$, $\exists I \in \mathfrak{I}$ such that
 \be
 \operatorname{Supp}(a + I) = \operatorname{Supp}(a + \mathfrak{I})
 \ee
\end{lemma}

\begin{proof}
 Obviously we must have $\operatorname{Supp}(a) \supseteq \operatorname{Supp}(a + \mathfrak{I})$. If they are equal, we are done.

 If not, we proceed by induction, by showing that we can always find $a' \in a + \mathfrak{I}$ with support strictly contained in $\operatorname{Supp}(a)$ (essentially, by removing at least one of the points of support). Such a procedure would terminate in a finite number of steps with the $a + I$ we wanted.

 Like before, let $\mu \in \operatorname{Supp}(a) \setminus \operatorname{Supp}(a + \mathfrak{I})$ be either a face or an edge, and denote by $Y_\mu$ either $X_f$ if $\mu = f$ is a face, or $Z_{\delta \bm{e}}$ if $\mu = e$ is an edge.

 On the one hand,
 \be
 Y_\mu a Y_\mu \neq a
 \ee
 (note that since $Y_\mu$ squares to 1, testing for commutativity is the same as the previous statement).

 On the other hand, there is some $J \in \mathfrak{I}$ such that
 \be
 Y_\mu (a + J) Y_\mu = a + J
 \ee

 Let us define
 \be
 a_\pm \coloneqq \frac{a \pm Y_\mu a Y_\mu}{2}
 \ee

 By definition, $a_{+}$ commutes with $Y_\mu$, while $a_{-}$ anticommutes. Then $a = a_{+} + a_{-}$, and $Y_\mu a Y_\mu = a_{+} - a_{-}$, and their inequality implies $a_{-} \neq 0$.

 Then, we also have
 \be
 a_{+} + a_{-} + J = Y_\mu (a + J) Y_\mu = a_{+} - a_{-} + Y_\mu J Y_\mu
 \ee
 and so
 \be
 a_{-} = \frac{Y_\mu J Y_\mu - J}{2} \in \mathfrak{I}
 \ee

 This also means that $a_{+} \in a + \mathfrak{I}$, and we shall show that this is what we were looking for. We want to verify that $\operatorname{Supp}(a_{+}) \subsetneq \operatorname{Supp}(a)$ with part of the difference being that $\mu \in \operatorname{Supp}(a)$ but $\mu \notin \operatorname{Supp}(a_{+})$. So now we must check the inclusion (with inequality given by $\mu$).

 Again it is easier to check the converse. Let $\nu \notin \operatorname{Supp}(a)$. Then
 \be
 Y_\nu a_{+} Y_\nu = \frac{Y_\nu a Y_\nu + Y_\nu Y_\mu a Y_\mu Y_\nu}{2} = \frac{a + (\pm 1)^2 Y_\mu a Y_\mu}{2} = a_{+}
 \ee

 The sign $(\pm 1)^2$ always cancels, because $Y_\mu, Y_\nu$ either commute, or anticommute. But they do so on both sides of that equation, so whichever sign it is, it always cancels.
\end{proof}

\begin{corollary}
 Let $\mc{R} \subseteq E \coprod F$ be a region. Denote by $\paren{\mc{A}_{\Z_2^{(1)}}}_{\mc{R}}$ and $\paren{\mc{A}_{\Z_2^{(1)}} / \mathfrak{I}}_{\mc{R}}$ the subalgebras of operators whose support is contained in $\mc{R}$. The quotient map restricted to a region
 \be
 \pi : \paren{\mc{A}_{\Z_2^{(1)}}}_{\mc{R}} \rightarrow \paren{\mc{A}_{\Z_2^{(1)}} / \mathfrak{I}}_{\mc{R}}
 \ee
 is a surjective $*$-homomorphism.
\end{corollary}

\begin{proof}
 The fact that they are subalgebras follows from Lemma~\ref{lemma:support_of_sum_and_prod}. $\pi$ is obviously a $*$-homomorphism, and its image is indeed contained in $\paren{\mc{A}_{\Z_2^{(1)}} / \mathfrak{I}}_{\mc{R}}$, because by definition the support of an equivalence class is always contained in the support of any representative.

 It is surjective because of the previous lemma.
\end{proof}

\subsection{Group of QCA}

We recall the definition of QCA in the quotient algebra. 

\begin{definition}
 An algebra homomorphism 
 \be
 \alpha : \faktor{\mc{A}_{\Z_2^{(1)}}}{\mathfrak{I}} \rightarrow \faktor{\mc{A}_{\Z_2^{(1)}}}{\mathfrak{I}} \,,
 \ee
 is called \emph{locality preserving} if there exists some $r \geq 0$ such that the following is true. For any region $\mc{R} \subseteq E \coprod F$ (finite or infinite),
 \be
 \alpha\paren{\paren{\mc{A}_{\Z_2^{(1)}} / \mathfrak{I}}_\mc{R}} \subseteq \paren{\mc{A}_{\Z_2^{(1)}} / \mathfrak{I}}_{\mc{R}^{+r}} \,.
 \ee

 The ``expanded region'' $\mc{R}^{+r}$ is
 \be
 \mc{R}^{+r} \coloneqq \curpar{\mu \in E \coprod F \ \mid \ \exists \nu \in \mc{R} \text{ s.t. } d(\mu, \nu) \leq r} \,,
 \ee
 with $d(\cdot, \cdot)$ the euclidean distance between centers of edges and faces.
\end{definition}

\begin{definition}
 A \emph{Quantum Cellular Automaton (QCA)} on $\mc{A}_{\Z_2^{(1)}} / \mathfrak{I}$ is a $*$-automorphism
 \be
 \alpha : \faktor{\mc{A}_{\Z_2^{(1)}}}{\mathfrak{I}} \rightarrow \faktor{\mc{A}_{\Z_2^{(1)}}}{\mathfrak{I}} \,,
 \ee
 which is locality preserving.

 The set of all QCA will be denoted $\operatorname{QCA}(\mc{A}_{\Z_2^{(1)}} / \mathfrak{I})$. A constant $r$ associated to a particular $\alpha \in \operatorname{QCA}(\mc{A}_{\Z_2^{(1)}} / \mathfrak{I})$ is an upper bound on the ``spread'' of $\alpha$.
\end{definition}

\begin{theorem}\label{thm:group_of_qca}
 The set $\operatorname{QCA}(\mc{A}_{\Z_2^{(1)}} / \mathfrak{I})$ forms a group under composition.
\end{theorem}

\begin{proof}
 We need to show that the composition of two QCA is locality preserving, and that the inverse is locality preserving. For convenience, we work in $\mathbb{C}^\text{lk}[B_1 \oplus B^2_\text{fin.}]$ which we have already shown is isomorphic to the quotient. For simplicity of the notation, define
 \be
 \mc{B} \coloneqq \mathbb{C}^\text{lk}[B_1 \oplus B^2_\text{fin.}] \,.
 \ee

 The composition is straightforward. If $\alpha, \beta$ are QCA, of range $r, s$, then $\beta \circ \alpha$ is a QCA of range $r + s$. One must simply check
 \be
 \paren{\mc{R}^{+r}}^{+s} \subseteq \mc{R}^{+(r + s)} \,,
 \ee
 which follows from the triangle inequality.

 For contradiction purposes, suppose that $\alpha$ is a QCA, of range $r$, but $\alpha^{-1}$ is not locality preserving. For any $r' \gg r$, there is some region $\mc{R}$ for which
 \be
 \alpha^{-1}\paren{\mc{B}_\mc{R}} \not\subseteq \mc{B}_{\mc{R}^{+r'}} \,.
 \ee

 From now on, fix some sufficiently big $r' \gg r$.

 Firstly, note that we can always choose some basis element $(b_1, \bm{b^2}) \in B_1 \oplus B^2_\text{fin.}$ for which
 \be
 \alpha^{-1}(b_1, \bm{b^2}) \notin \mc{B}_{\paren{b_1 \coprod b^2}^{+r'}} \,.
 \ee

 If not, it would mean all basis elements have their spread bounded by $r'$. But then that extends to a bound on the spread of $\alpha^{-1}$, which we are assuming does not exist.

 \textbf{Claim.} Let $\mu \in \operatorname{Supp}(a + \mathfrak{I})$. Then
 \be
 (Y_\mu + \mathfrak{I}) (a + \mathfrak{I}) \neq (a + \mathfrak{I}) (Y_\mu + \mathfrak{I})
 \ee

 This would be trivial in $\mc{A}_{\Z_2^{(1)}}$ (by definition $Y_\mu a \neq a Y_\mu$ for $\mu \in \operatorname{Supp}(a)$), but we claim that this also applies in the quotient.

 We pick
 \be
 \mu \in \operatorname{Supp}(\alpha^{-1}(b_1, \bm{b^2})) \setminus \paren{b_1 \coprod b^2}^{+r'} \,,
 \ee
 and using the claim we find
 \be
 (Y_\mu + \mathfrak{I}) \alpha^{-1}(b_1, \bm{b^2}) \neq \alpha^{-1}(b_1, \bm{b^2}) (Y_\mu + \mathfrak{I}) \,.
 \ee

 Now apply $\alpha$ to both sides
 \be
 \alpha(Y_\mu + \mathfrak{I}) (b_1, \bm{b^2}) \neq (b_1, \bm{b^2}) \alpha(Y_\mu + \mathfrak{I}) \,.
 \ee

 Since $\alpha$ has spread $r$, $\alpha(Y_\mu + \mathfrak{I})$ is contained in a ball of radius $r \ll r'$ around $\mu$. The original $(b_1, \bm{b^2})$ is entirely outside of this ball.

 Writing $\alpha(Y_\mu + \mathfrak{I})$ in the basis $B_1 \oplus B^2_\text{fin.}$, we find that any terms $(b_1', \bm{b}'^2)$ appearing in it can be rewritten as $(\partial c_2, \delta \bm{c^1})$, with $c_2, c_1$ entirely contained inside the ball. Therefore, when multiplying, we find that all signs are $+1$:
 \begin{align}
  (\partial c_2, \delta \bm{c^1}) (b_1, \bm{b^2}) &= (\partial c_2 + b_1, \delta \bm{c^1} + \bm{b^2}) \nonumber \\
  &= (b_1, \bm{b^2}) (\partial c_2, \delta \bm{c^1}) \,,
 \end{align}
 simply because $\bm{c^1} (b_1) = \bm{b^2}(c_2) = 0$, one is inside the ball, one is entirely outside.

 But from this we conclude that
 \be
 \alpha(Y_\mu + \mathfrak{I}) (b_1, \bm{b^2}) = (b_1, \bm{b^2}) \alpha(Y_\mu + \mathfrak{I}) \,,
 \ee
 which is a contradiction.

 \textbf{Proof.} (of the claim) Assume for contradiction that $Y_\mu a + \mathfrak{I} = a Y_\mu + \mathfrak{I}$.

 This is equivalent to $Y_\mu a - a Y_\mu = J \in \mathfrak{I}$. Now, observe that
 \be
 Y_\mu J Y_\mu = a Y_\mu - Y_\mu a = - J \implies J = \frac{1}{2} \paren{J - Y_\mu J Y_\mu}
 \ee

 Rearranging these equations a little:
 \begin{align}
  Y_\mu a + \frac{1}{2} Y_\mu J Y_\mu &= a Y_\mu + \frac{1}{2} J \,, \nonumber \\
  Y_\mu \paren{a + \frac{1}{2} J Y_\mu} &= \paren{a + \frac{1}{2} J Y_\mu} Y_\mu \,.
 \end{align}

 We find that $Y_\mu$ commutes with $a + \frac{1}{2} J Y_\mu \in a + \mathfrak{I}$, and hence
 \be
 \mu \notin \operatorname{Supp}\paren{a + \frac{1}{2} J Y_\mu} \implies \mu \notin \operatorname{Supp}(a + \mathfrak{I}) \,,
 \ee
 which is a contradiction.
\end{proof}

\subsection{Extensions of QCA}\label{app:extensions_of_qca}

Let us recall the definition of an extension of QCA from the quotient algebra to the full tensor product algebra.

\begin{definition}\label{def:QCA_extension}
 Let $\alpha : \mc{A}_{\Z_2^{(1)}} / \mathfrak{I} \rightarrow \mc{A}_{\Z_2^{(1)}} / \mathfrak{I}$ be a QCA. We say that $\tilde{\alpha} : \mc{A} \rightarrow \mc{A}$ is an extension of $\alpha$ to the full algebra if $\tilde{\alpha}$ is a QCA, $\tilde{\alpha}(\mc{A}_{\Z_2^{(1)}}) = \mc{A}_{\Z_2^{(1)}}$, and the following diagram
 \be
 \begin{tikzcd}[nodes={inner sep=5pt}]
  \mc{A}_{\Z_2^{(1)}} \arrow[r, "\tilde{\alpha}"] \arrow[d, "\pi", two heads] & \mc{A}_{\Z_2^{(1)}} \arrow[d, "\pi", two heads] \\
  \mc{A}_{\Z_2^{(1)}} / \mathfrak{I} \arrow[r, "\alpha"]      & \mc{A}_{\Z_2^{(1)}} / \mathfrak{I}
 \end{tikzcd}
 \ee
 commutes.
\end{definition}

\begin{lemma}
 Let $\tilde{\alpha} : \mc{A} \rightarrow \mc{A}$ be an extension of $\alpha : \mc{A}_{\Z_2^{(1)}} / \mathfrak{I} \rightarrow \mc{A}_{\Z_2^{(1)}} / \mathfrak{I}$. Then, the following are equivalent
 \begin{enumerate}
  \item $\tilde{\alpha}(\mc{A}_{\Z_2^{(1)}}) = \mc{A}_{\Z_2^{(1)}}$

  \item $\tilde{\alpha}(\mathfrak{I}) = \mathfrak{I}$

  \item $\mathfrak{I} \subset \tilde{\alpha}(\mc{A}_{\Z_2^{(1)}})$
 \end{enumerate}
\end{lemma}

\begin{lemma}
 Let $\tilde{\alpha} : \mc{A} \rightarrow \mc{A}$ be a QCA on the full algebra. If $\tilde{\alpha}(\mc{A}_{\Z_2^{(1)}}) = \mc{A}_{\Z_2^{(1)}}$ and $\tilde{\alpha}(\mathfrak{I}) = \mathfrak{I}$, then $\tilde{\alpha}$ defines a restricted QCA on $\mc{A}_{\Z_2^{(1)}} / \mathfrak{I}$ via
 \be
 \alpha(a + \mathfrak{I}) \coloneqq \tilde{\alpha}(a) + \mathfrak{I} \quad \text{for any } a \in \mc{A}_{\Z_2^{(1)}} \,.
 \ee
\end{lemma}

\begin{theorem}\label{thm:QCA_extendability}
 Let $\alpha : \mc{A}_{\Z_2^{(1)}} / \mathfrak{I} \rightarrow \mc{A}_{\Z_2^{(1)}} / \mathfrak{I}$ be a QCA such that $\alpha(Z_{\delta \bm{e}} + \mathfrak{I}) = Z_{\delta \bm{e}} + \mathfrak{I}$. Then, there exists an extension $\tilde{\alpha}$ to the full algebra. Furthermore, we may choose it such that $\tilde{\alpha}(Z_f) = Z_f \; \forall f \in F$.
\end{theorem}

\begin{proof}
 The proof is constructive. Let us define $\tilde{\alpha}(Z_f) = Z_f$.

 Now choose $x_f \in \alpha(X_f + \mathfrak{I})$ arbitrarily. We will later see how to make the choice correctly. Note that we have:
 \begin{align}
 (x_f + \mathfrak{I}) (Z_{\delta \bm{e}} + \mathfrak{I}) &= \alpha(X_f Z_{\delta \bm{e}} + \mathfrak{I}) \nonumber \\
 (Z_{\delta \bm{e}} + \mathfrak{I}) (x_f + \mathfrak{I}) &= \alpha(Z_{\delta \bm{e}} X_f + \mathfrak{I})
 \end{align}
 so the commutation relations of $x_f$ and $Z_{\delta \bm{e}}$ are the same as those of $X_f$ and $Z_{\delta \bm{e}}$ (or rather, their equivalence classes). This means that $X_f x_f$ commutes with all $Z_{\delta \bm{e}}$, or in other words
 \be
 X_f x_f Z_{\delta \bm{e}} - Z_{\delta \bm{e}} X_f x_f = J_{e, f} \in \mathfrak{I} \qquad \forall e, f \,.
 \ee

 From this we can also see that $Z_{\delta \bm{e}} J_{e, f} Z_{\delta \bm{e}} = - J_{e, f}$, and therefore
 \be
 \paren{X_f x_f + \frac{1}{2} Z_{\delta \bm{e}} J_{e, f}} Z_{\delta \bm{e}} = Z_{\delta \bm{e}} \paren{X_f x_f + \frac{1}{2} Z_{\delta \bm{e}} J_{e, f}} \,.
 \ee

 That is, for every edge $e \in E$, there is some $I \in \mathfrak{I}$, such that $X_f x_f + I$ commutes with $Z_{\delta \bm{e}}$. This means that $e \notin \operatorname{Supp}(X_f x_f + \mathfrak{I})$, and so
 \be
 \operatorname{Supp}(X_f x_f + \mathfrak{I}) \subseteq F \,.
 \ee

 It can be shown (Lemma \ref{lemma:exact_support}) that there is always some $a_f \in X_f x_f + \mathfrak{I}$ whose support is exactly the same as that of its equivalence class:
 \be
 \operatorname{Supp}(a_f) = \operatorname{Supp}(X_f x_f + \mathfrak{I}) \subseteq F \,.
 \ee
 In fact, we may choose a representative $a_f$ which is written entirely in terms of $Z_{\delta \bm{e}}$. Any $X_f$, would have to appear in configurations without boundary, but those are just symmetry operators, which can all be eliminated by choosing a different representative:
 \be
 \prod_{f \in c_2} X_f \prod_{e \in c_1} Z_{\delta \bm{e}} = \prod_{e \in c_1} Z_{\delta \bm{e}} - \paren{1 - \prod_{f \in c_2} X_f} \prod_{e \in c_1} Z_{\delta \bm{e}} \,,
 \ee
 where due to $\partial c_2 = 0$, the second term is in $\mathfrak{I}$.

 Finally, we define $\tilde{\alpha}(Z_f) = Z_f$ and $\tilde{\alpha}(X_f) = X_f a_f$.

 Let us check that such $\tilde{\alpha}$ really is a $*$-homomorphism. Clearly $\tilde{\alpha}(Z_f) = Z_f$ all commute between each other. Furthermore, $\tilde{\alpha}(Z_f), \tilde{\alpha}(X_f)$ commute or anticommute as they should: the sign comes from the $X_f$ in $\tilde{\alpha}(X_f) = X_f a_f$, and $[a_f, Z_p] = 0$ because $a_f$ is written entirely in terms of $Z$.

 Now, we must verify that $\tilde{\alpha}(X_f)$ commute between each other. Write
 \be
 a_f = \sum_{b^2} A^f_{b^2} \prod_{f \in b^2} Z_f \,.
 \ee

 Then
 \be
 X_f a_f X_p a_p = X_f X_p \paren{\sum_{b^2} (-1)^{\bm{b^2}(p)} A^f_{b^2} \prod_{f \in b^2} Z_f } a_p \,.
 \ee

 Similarly
 \be
 X_p a_p X_f a_f = X_p X_f a_f \paren{ \sum_{b'^2} (-1)^{\bm{b}'^2(f)} A^p_{b'^2} \prod_{f \in b'^2} Z_f } \,.
 \ee

 Due to the fact that $\alpha(X_f + \mathfrak{I})$ and $\alpha(X_p + \mathfrak{I})$ commute, we also must have $X_f a_f X_p a_p - X_p a_p X_f a_f \in \mathfrak{I}$, or equivalently if we multiply by $X_f X_p$ on the left,
 \be
 (X_p a_f X_p) a_p - a_f (X_f a_p X_f) \in \mathfrak{I} \,,
 \ee
 where $X a X$ are just the expressions in parentheses earlier.
 
 \textbf{Claim.} $\langle Z_{\delta \bm{e}} \rangle \cap \mathfrak{I} = \{ 0 \}$.

 We see that that difference, which is in $\mathfrak{I}$, can also be written entirely in terms of $Z_{\delta \bm{e}}$ (conjugating by $X_f$ does not change this, it only adds some signs). By the claim, the difference is $0$, so $\tilde{\alpha}(X_f)$ commute between each other.

 \textbf{Proof.} Clearly, every element of $\langle Z_{\delta \bm{e}} \rangle \subset \mc{A}$ commutes with every $Z_f$.

 We are actually going to show something slightly stronger. That if some $I \in \mathfrak{I}$ also commutes with all $Z_f, \forall f \in F$, then $I = 0$.

 We argue by contradiction: suppose there is some $I \in \mathfrak{I} \setminus \{0\}$ which commutes with all $Z_f$. Without loss of generality, we may assume that $I$ is written entirely in terms of $X$ operators. Indeed, expanding $I$ in the basis of $\mc{A}$, $\prod X \prod Z$ with the products running over arbitrary 2-chains, we can group the terms according to their configuration of $Z$ operators. In such an expression, each ``coefficient'' (which is actually an operator written entirely in terms of $X$ operators) must commute with all $Z_f$. At least one of those coefficients is non-zero, and we may pick that one.

 But then, if we have
 \be
 I = \sum_{c_2} I_{c_2} \prod_{f \in c_2} X_f \,,
 \ee
 this is still written as a linear combination of linearily independent terms. For any $c'_2 \neq 0$, choose a plaquette $p \in c'_2$:
 \be
 I = Z_p I Z_p = - I_{c'_2} \prod_{f \in c'_2} X_f + \sum_{c_2 \neq c'_2} (-1)^{\bm{p}(c_2)} I_{c_2} \prod_{f \in c_2} X_f \,,
 \ee
 so $I_{c'_2} = 0$. The only term which remains is $I = I_0 \id{}$. But this is invertible for any $I_0 \neq 0$, so we conclude $I = 0$. With this, the proof of the claim is complete.
 
 We have shown that
 \begin{align}
  \tilde{\alpha}(Z_f) &= Z_f \,, & \tilde{\alpha}(X_f) &= X_f a_f \,,
 \end{align}
 preserves all the commutation relations of the generators. $Z_f$ is also self-adjoint, and unitary (i.e. squares to 1).

 To show that $X_f a_f$ is self-adjoint, note that
 \be
 (X_f a_f)^\dag + \mathfrak{I} = \alpha\paren{(X_f + \mathfrak{I})^\dag} = \alpha\paren{X_f + \mathfrak{I}} = X_f a_f + \mathfrak{I} \,.
 \ee

 We manipulate the condition by multiplying by a unitary:
 \be
 (X_f a_f)^\dag - X_f a_f \in \mathfrak{I} \iff a_f^\dag - X_f a_f X_f \in \mathfrak{I} \,.
 \ee

 The operator on the right can clearly be written using only $Z_{\delta \bm{e}}$, so by using the claim it must be $0$. The self-adjointness of $X_f a_f$ holds exactly, not only at the level of equivalence classes.

 Finally, to show that $X_f a_f$ is unitary:
 \be
 (X_f a_f)^2 + \mathfrak{I} = \alpha\paren{(X_f + \mathfrak{I})^2} = 1 + \mathfrak{I} \,,
 \ee
 we manipulate the condition
 \be
 X_f a_f X_f a_f - 1 \in \mathfrak{I} \iff a_f^\dag a_f - 1 \in \mathfrak{I} \,.
 \ee

 By using the claim again, we conclude that this equality must hold exactly.

 Hence, $\tilde{\alpha}$ defines a $*$-homomorphism.

 We can construct its inverse explicitly,
 \begin{align}
  \tilde{\alpha}^{-1}(Z_f) &= Z_f \,, & \tilde{\alpha}^{-1}(X_f) &= X_f a_f^\dag \,,
 \end{align}
 to show that it is a $*$-automorphism.

 The commutative diagram in the statement of the theorem is equivalent to $\pi \circ \tilde{\alpha} = \alpha \circ \pi$, when restricted to $\mc{A}_{\Z_2^{(1)}}$. As a preliminary step, note that
 \be
 X_f \in \langle X_f a_f, Z_{\delta \bm{e}} \; \mid \; f \in F, e \in E \rangle
 \ee
 because $a_f \in \langle Z_{\delta \bm{e}} \rangle$. Hence
 \be
 \tilde{\alpha} \paren{\mc{A}_{\Z_2^{(1)}}} = \mc{A}_{\Z_2^{(1)}} \,,
 \ee
 and so the equality $\pi \circ \tilde{\alpha} = \alpha \circ \pi$ makes sense.
 
 As always we can check equality on the generators:
 \begin{align}
  \alpha(Z_{\delta \bm{e}} + \mathfrak{I}) &= Z_{\delta \bm{e}} + \mathfrak{I} \,, \\
  \tilde{\alpha}(Z_{\delta \bm{e}}) + \mathfrak{I} &= Z_{\delta \bm{e}} + \mathfrak{I} \,,
 \end{align}
 and similarly
 \begin{align}
  \alpha(X_f + \mathfrak{I}) &= x_f + \mathfrak{I} = X_f a_f + \mathfrak{I} \,, \\
  \tilde{\alpha}(X_f) + \mathfrak{I} &= X_f a_f + \mathfrak{I} \,.
 \end{align}

 Lastly, we must show that $\tilde{\alpha}$ is also locality preserving, assuming $\alpha$ has a radius of spread $r$. Obviously we may focus on $X_f a_f$, and all we really need to check is that $a_f$ has bounded support (uniformly bounded for all $f$).

 In the quotient:
 \begin{multline}
  \operatorname{Supp}(a_f + \mathfrak{I}) = \operatorname{Supp}\paren{(X_f + \mathfrak{I}) \alpha(X_f + \mathfrak{I})} \\
  \subseteq f \cup \operatorname{Supp}(\alpha(X_f + \mathfrak{I})) \subseteq B_r(f) \,,
 \end{multline}
 where $B_r(f)$ is the ball of radius $r$ with center $f$. That last inclusion follows from the assumption that $\alpha$ has bounded spread (is locality preserving).

 Now recall that we chose the representative $a_f$ exactly as that one which has the same support as its equivalence class. Hence:
 \be
 \operatorname{Supp}(a_f) = \operatorname{Supp}(a_f + \mathfrak{I}) \subseteq B_r(f) \,,
 \ee
 which also shows that $\tilde{\alpha}$ is locality preserving, and hence an ordinary QCA on a tensor product algebra.

 We remark here that this last step relies on the fact that for $a_f$ in particular, its support in $\mc{A}_{\Z_2^{(1)}}$ as defined in this paper matches the usual notion of support. This is because $a_f$ is written entirely in terms of $Z$ operators, so when checking for where it does not commute with $X_f, Z_{\delta \bm{e}}$, only $X_f$ may fail to commute. This ends up being the same as checking where $a_f$ is not proportional to the identity.
\end{proof}

\begin{proposition}\label{prop:qca_nonextendability}
 Let $\alpha : \mc{A}_{\Z_2^{(1)}} / \mathfrak{I} \rightarrow \mc{A}_{\Z_2^{(1)}} / \mathfrak{I}$ be a QCA such that $\alpha(Z_{\delta \bm{e}} + \mathfrak{I}) = X_{t^\frac{1}{2} e} + \mathfrak{I}$.
 Then, there does not exist an extension of $\alpha$ to the full algebra.
\end{proposition}

\begin{proof}
 Suppose for contradiction purposes that there does exist $\tilde{\alpha} : \mc{A} \rightarrow \mc{A}$ a QCA which is the lift of $\alpha$.

 Consider a large loop $\bm{b} \in B^2_\text{fin.}$ on the direct lattice, and another $\bm{b}' \in B^2_\text{fin.}$, such that they link exactly once, and such that every face of $b$ is at least at a distance $r$ from every other face in $b'$. This can always be done by taking for example large enough square loops in a Hopf link configuration, where one loop passes through the center of the other square.

 Consider the operator $\prod_{f \in b} Z_f$, and note that if $\bm{c}$ is such that $\delta \bm{c} = \bm{b}$, then
 \be
 \tilde{\alpha}\paren{\prod_{f \in b} Z_f} + \mathfrak{I} = \alpha \paren{\prod_{f \in b} Z_f + \mathfrak{I}} = \prod_{e \in c} X_{t^\frac{1}{2} e} + \mathfrak{I} \,.
 \ee

 In other words, $\tilde{\alpha}(\prod_{f \in b} Z_f) \in \prod_{e \in c} X_{t^\frac{1}{2} e} + \mathfrak{I}$.

 Now consider the other loop $b'$, and note that
 \begin{multline}
 \paren{\prod_{e \in c} X_{t^\frac{1}{2} e} + \mathfrak{I}} \paren{\prod_{f \in b'} Z_f + \mathfrak{I}} \\
 = - \paren{\prod_{f \in b'} Z_f + \mathfrak{I}} \paren{\prod_{e \in c} X_{t^\frac{1}{2} e} + \mathfrak{I}} \,.
 \end{multline}

 This is due to the fact that $\bm{b}'(t^\frac{1}{2} c) = 1$, or alternatively, that $b$ and $b'$ link an odd number of times.

 In particular, we see that these equivalence classes do not commute, and hence any pair of representatives also do not commute. Indeed, the contrapositive is easier to see. Suppose $ab = ba$. Then
 \be
 ab + \mathfrak{I} = ba + \mathfrak{I} \implies (a + \mathfrak{I}) (b + \mathfrak{I}) = (b + \mathfrak{I}) (a + \mathfrak{I}) \,.
 \ee

 From this we see that, since $\tilde{\alpha}(\prod_{f \in b} Z_f) \in \prod_{e \in c} X_{t^\frac{1}{2} e} + \mathfrak{I}$, $\tilde{\alpha}(\prod_{f \in b} Z_f)$ does not commute with $\prod_{f \in b'} Z_f$. In other words, as an operator in $\mc{A}$, it is supported (not proportional to the identity) on at least one of the faces in $b'$. But
 \be
 \tilde{\alpha}\paren{\prod_{f \in b} Z_f} = \prod_{f \in b} \tilde{\alpha}(Z_f) \,,
 \ee
 hence at least one of $\tilde{\alpha}(Z_f)$ is supported somewhere on $b'$. Recall that $b, b'$ were chosen such that the distance between them is at least $r$. Hence, the hypothetical QCA $\tilde{\alpha}$ has spread at least $r$.

 But $r$ can be chosen arbitrarily large, and the previous construction will always work. Therefore, $\tilde{\alpha}$ has unbounded spread, which is a contradiction.
\end{proof}

%%%%%%%%%%%%%%%%%%%%%%%%%%%%%%%%%%%%
%%%%%%%%%%%%%%%%%%%%%%%%%%%%%%%%%%%%
%%%%%%%%%%%%%%%%%%%%%%%%%%%%%%%%%%%%diagram

\twocolumngrid
\bibliographystyle{ytphys}
\small 
\baselineskip=.94\baselineskip
\let\bbb\bibitem\def\bibitem{\itemsep4pt\bbb}
\bibliography{GenSym}

%%%%%%%%%%%%%%%%%%%%%%%%%%%%%%%%%%%%
%%%%%%%%%%%%%%%%%%%%%%%%%%%%%%%%%%%%
%%%%%%%%%%%%%%%%%%%%%%%%%%%%%%%%%%%%

\end{document}